%% file: main.tex
\documentclass[12pt, reqno]{amsart}  

\usepackage{main}

\usepackage{setspace}       
\usepackage{booktabs}       
\usepackage{array}          
\usepackage{makecell}       
\usepackage{tabularx}       
\usepackage{multirow}       
\usepackage{adjustbox}      
\usepackage{pdflscape}      
\usepackage{longtable}      

\graphicspath{{Figures/}}

\usepackage[style=apa, backend=biber, natbib=true, backref=true]{biblatex}
\DefineBibliographyStrings{english}{%
  backrefpage = {},
  backrefpages = {}
}
\renewbibmacro*{pageref}{%
  \iflistundef{pageref}
    {}
    {\printtext[brackets]{\printlist[pageref][-\value{listtotal}]{pageref}}}}

\usepackage{titletoc}

\newgeometry{margin=1.25in}

\title{\Large A Bayesian Hierarchical Hurdle Beta-Binomial Model for Survey-Weighted Bounded Counts and Its Application to Childcare Enrollment}

\author{JoonHo Lee}

\date{\footnotesize March 2, 2026. \\[0.5em]
Lee: The University of Alabama, Tuscaloosa, AL, USA. jlee296@ua.edu.
\\[0.5em]
This research was supported by the Office of Planning, Research, and Evaluation (OPRE), Administration for Children and Families, U.S.\ Department of Health and Human Services, through Grant 90YE0346 to the University of Nevada, Las Vegas (PI: Gerilyn Slicker), with a subaward to the University of Alabama. The opinions expressed are those of the author and do not represent views of the Administration for Children and Families or the U.S.\ Department of Health and Human Services.
}

\numberwithin{equation}{section}  

\renewcommand{\theHequation}{\thesection.\arabic{equation}}
\renewcommand{\theHfigure}{\thesection.\arabic{figure}}
\renewcommand{\theHtable}{\thesection.\arabic{table}}

\makeatletter
\renewcommand{\theHALG@line}{\thealgorithm.\arabic{ALG@line}}
\makeatother

\begin{document}


\begin{abstract}
Bounded discrete proportions---counts out of known totals---present
modeling challenges when data exhibit structural zeros, overdispersion,
and hierarchical clustering.  We develop a Bayesian hierarchical hurdle
beta-binomial model with state-varying coefficients that addresses
all four features.  The framework makes three methodological contributions:
(i)~it studies cross-margin dependence via a cross-block covariance
component and clarifies when and how this parameter is identified
through the hierarchical layer rather than the conditional likelihood;
(ii)~it proposes a Cholesky-based sandwich variance calibration for
pseudo-posterior inference under survey weights, guided by a
parameter-specific design effect ratio diagnostic; and
(iii)~it introduces a log-scale marginal effect decomposition
for hurdle models that translates regression coefficients
into policy-relevant quantities.  Applied to
6,785~childcare providers across 51~states from the 2019 National Survey
of Early Care and Education, the model reveals a ``poverty reversal'':
poverty reduces enrollment participation yet increases intensity among
participants, with the extensive margin accounting for two-thirds of the
total effect.  Design-calibrated simulation shows that sandwich-corrected
intervals substantially improve coverage, reaching 82--88.5\% at
the 90\% nominal level for fixed effects.  The R package \textsf{hurdlebb} implements
all methods.
\end{abstract}

\maketitle  


\pagestyle{plain}  
\newpage



\input{section1}


\input{section2}


\input{section3}


\input{section4}


\input{section5}


\input{section6}


\bigskip

\paragraph{Software and reproducibility.}
The companion R package \textsf{hurdlebb}~\citep{hurdlebb2026} implements
the full methodology developed in this paper, including survey-weighted
pseudo-posterior inference, sandwich variance correction via Cholesky
transformation, and marginal effect decomposition.  The package is
available at \url{https://github.com/joonho112/hurdlebb} with
documentation at \url{https://joonho112.github.io/hurdlebb/}.  A
separate replication package at
\url{https://github.com/joonho112/hurdlebb-replication} contains all
analysis scripts needed to reproduce every result reported in this paper,
conditional on access to the NSECE restricted-use data.  The complete
estimation pipeline is described in~\cref{alg:estimation};
Supplementary Material~E (\cref{app:sim-computation}) provides
computational details including Stan configuration, convergence
diagnostics, and runtime benchmarks.

\paragraph{Data availability.}
The empirical results in this paper are based on data from the 2019
National Survey of Early Care and Education (NSECE). All key variables
used in the analysis---including the state identifier, community-level
poverty rate, urbanicity, racial/ethnic composition, infant/toddler
enrollment counts, and total enrollment---are available only through the
NSECE Restricted-Use Files (ICPSR Study 38445,
\url{https://www.icpsr.umich.edu/web/ICPSR/studies/38445}). For
researchers interested in the methodology without access to the restricted
data, the \textsf{hurdlebb} package includes a synthetic dataset
(\texttt{nsece\_synth}) that preserves the essential distributional
characteristics of the original data.

\paragraph{Funding.}
This research was supported by the Office of Planning, Research, and
Evaluation (OPRE), Administration for Children and Families, U.S.\
Department of Health and Human Services, through Grant 90YE0346 to the
University of Nevada, Las Vegas (PI: Gerilyn Slicker), with a subaward
to the University of Alabama. The opinions expressed are those of the
author and do not represent views of the Administration for Children and
Families or the U.S.\ Department of Health and Human Services.


\printbibliography[title={References}]


\newpage
\appendix
\setcounter{section}{0}
\renewcommand{\thesection}{\Alph{section}}
\numberwithin{equation}{section}
\numberwithin{figure}{section}
\numberwithin{table}{section}

\renewcommand{\theHequation}{App.\thesection.\arabic{equation}}
\renewcommand{\theHfigure}{App.\thesection.\arabic{figure}}
\renewcommand{\theHtable}{App.\thesection.\arabic{table}}
\renewcommand{\theHlemma}{App.\thesection.\arabic{lemma}}
\renewcommand{\theHtheorem}{App.\thesection.\arabic{theorem}}
\renewcommand{\theHcorollary}{App.\thesection.\arabic{corollary}}
\renewcommand{\theHproposition}{App.\thesection.\arabic{proposition}}
\renewcommand{\theHdefinition}{App.\thesection.\arabic{definition}}
\renewcommand{\theHremark}{App.\thesection.\arabic{remark}}
\renewcommand{\theHexample}{App.\thesection.\arabic{example}}
\renewcommand{\theHconjecture}{App.\thesection.\arabic{conjecture}}

\setcounter{lemma}{0}
\renewcommand{\thelemma}{\thesection.\arabic{lemma}}
\setcounter{theorem}{0}
\renewcommand{\thetheorem}{\thesection.\arabic{theorem}}
\setcounter{corollary}{0}
\renewcommand{\thecorollary}{\thesection.\arabic{corollary}}
\setcounter{proposition}{0}
\renewcommand{\theproposition}{\thesection.\arabic{proposition}}
\setcounter{definition}{0}
\renewcommand{\thedefinition}{\thesection.\arabic{definition}}
\setcounter{remark}{0}
\renewcommand{\theremark}{\thesection.\arabic{remark}}
\setcounter{example}{0}
\renewcommand{\theexample}{\thesection.\arabic{example}}
\setcounter{conjecture}{0}
\renewcommand{\theconjecture}{\thesection.\arabic{conjecture}}

\begin{center}
{\Large\bfseries Online Supplemental Materials}\\[1em]
{\large A Bayesian Hierarchical Hurdle Beta-Binomial Model for\\
Survey-Weighted Bounded Counts and Its Application\\
to Childcare Enrollment}\\[1.5em]
{\normalsize JoonHo Lee}
\end{center}

\vspace{2em}

\startcontents[appendices]
\printcontents[appendices]{}{1}{\textbf{Contents}\vskip1em\hrule\vskip1em}
\vskip1em\hrule\vskip2em

\addcontentsline{toc}{section}{Online Supplemental Materials}

\input{osm_body}


\end{document}

%% file: section1.tex

\section{Introduction}
\label{sec:intro}


Bounded discrete proportions---counts out of known totals---arise
whenever a finite capacity is partially utilized: hospital beds
occupied, insurance plans enrolled, dental surfaces decayed, survey
items endorsed, or enrollment slots filled at a childcare center.  When the outcome can take the value zero for
structural rather than sampling reasons and displays variance well
beyond the binomial, standard generalized linear models are inadequate.
Infant and toddler childcare enrollment in the United States offers a
canonical instance of all these complications simultaneously.  Data from
the 2019 National Survey of Early Care and Education
\citep{NSECEProjectTeam2022} record, for each center-based provider,
the number of children under age three enrolled out of total capacity.
Over one-third of centers report zero infant and toddler enrollment---not
because demand is absent, but because stringent staff-to-child ratios and
elevated per-child costs lead many providers to decline serving this age
group entirely~\citep{AdamsRohacek2010}.  The resulting access gaps fall
disproportionately on high-poverty
communities~\citep{DobbinsEtAl2016, MalikEtAl2018}, yet the relationship
between poverty and enrollment turns out to be considerably more complex
than a simple deficit story.


A preliminary examination of the data reveals a distinctive pattern that we
term the \emph{poverty reversal}: community poverty reduces the
probability that a center serves infants and toddlers at all, yet among
centers that \emph{do} participate, those in higher-poverty areas devote
a \emph{larger} share of capacity to this age group.
The reversal poses a methodological challenge: it \emph{cannot be detected} by any model
that conflates participation and intensity into a single equation, nor
by one that fails to accommodate bounded support, overdispersion,
structural zeros, and survey weighting simultaneously.  These four
features jointly motivate the hierarchical hurdle beta-binomial
framework developed in this paper.


Two-part models for count data with excess zeros have a substantial
history.  \citet{Cragg1971} introduced the hurdle framework for
continuous expenditures and \citet{Mullahy1986} adapted it for counts,
separating a binary participation stage from a truncated count stage.
Bayesian extensions have expanded their scope considerably: hierarchical
zero-inflated or hurdle models with Poisson and negative-binomial kernels
now accommodate spatial, longitudinal, and multilevel structures
\citep{NeelonEtAl2010, NeelonEtAl2013, NeelonEtAl2016, Neelon2019,
Altinisik2023}.  Most relevant to our work, \citet{GhosalEtAl2020}
propose a hierarchical hurdle model for spatiotemporal count data that
captures cross-margin dependence through a single shared scalar.  However,
all of these models assume Poisson or negative-binomial kernels, which
treat the count as unbounded.  When the response has bounded support, the
beta-binomial distribution is the natural alternative, introducing
overdispersion while respecting the $\{0,1,\ldots,n\}$
support~\citep{WagnerEtAl2015}.  \citet{BandyopadhyayEtAl2011} embed the
beta-binomial in a spatial framework for clustered bounded counts, and
\citet{WenEtAl2025} propose a zero-inflated beta-binomial with spatially
varying coefficients and cross-component covariance.  Yet they adopt a
zero-inflation formulation in which zeros arise from a latent mixture,
whereas in the childcare setting a center reporting zero infant enrollment
has made a deliberate operational decision~\citep[cf.][]{Lambert1992,
Hall2000}, making the hurdle specification the natural choice.


A separate methodological strand addresses design-based inference in
Bayesian models for survey data.  \citet{SavitskyToth2016} develop a
pseudo-posterior framework in which survey weights enter the likelihood
exponent, and \citet{WilliamsSavitsky2021} provide calibrated sandwich
variance corrections that restore nominal credible-interval
coverage~\citep[building on][]{Binder1983, Pfeffermann1993}.  These
contributions target standard single-equation regressions; their
extension to two-part hierarchical specifications has not been attempted.
In sum, the literature provides the individual building
blocks---hurdle, beta-binomial, hierarchical structure, survey
correction---but not their synthesis.


This paper develops that synthesis.  The model extends the framework of
\citet{GhosalEtAl2020} in two directions, each addressing a limitation
of existing approaches.

First, we introduce a Bayesian hierarchical hurdle beta-binomial
model with state-varying coefficients and cross-margin covariance.  The
hurdle component separates participation from intensity; the
beta-binomial kernel respects the bounded support and accommodates the
observed overdispersion, which exceeds the binomial variance by a factor
of approximately twelve; and the hierarchical prior on state-level
coefficients identifies a cross-margin covariance structure whose
effective sample size is $S = 51$ states rather than $N = 6{,}785$
providers, with the $\LKJ$ prior providing finite-sample regularization
for the high-dimensional covariance ($2q = 10$ parameters).  The model nests
the scalar cross-margin parameter of \citet{GhosalEtAl2020} as a special
case and extends the beta-binomial models of
\citet{BandyopadhyayEtAl2011} and \citet{WenEtAl2025} by incorporating a
hurdle for structural zeros.  Applied to the NSECE data, this
specification reveals the reversal---and its geographic
heterogeneity across 51 states---with full uncertainty quantification.

Second, we extend the pseudo-posterior framework of
\citet{SavitskyToth2016} and \citet{WilliamsSavitsky2021} to the
two-part hierarchical setting, implementing a Cholesky-based sandwich
variance correction~\citep{Binder1983, Pfeffermann1993} and introducing
the design effect ratio as a parameter-specific diagnostic.  A companion
simulation study (\cref{sec:simulation}) provides candid diagnostics:
sandwich-corrected intervals substantially improve coverage for fixed
effects under realistic design effects (82--88.5\% at the 90\% nominal
level), while variance components require
unweighted estimation---a finding we report transparently rather than
treating the sandwich as a universal remedy.

These two innovations jointly enable a marginal effect
decomposition~(\cref{prop:lae-lie}) that separates the reversal into
extensive-margin (access) and intensive-margin (intensity) components.
A single-equation specification conflates these offsetting forces into
a muted net effect; the two-part decomposition quantifies each channel
and their relative magnitudes, with the extensive margin accounting for
approximately two-thirds of the total effect.  The companion R package
\textsf{hurdlebb}~\citep{hurdlebb2026} implements the full methodology.


The remainder of the paper is organized as follows.
\Cref{sec:data} introduces the NSECE data and documents the empirical
patterns motivating the model.
\Cref{sec:model} develops the hierarchical hurdle beta-binomial
specification, its analytic properties, and the sandwich correction.
\Cref{sec:simulation} reports the simulation study.
\Cref{sec:application} applies the model to the 2019 NSECE data.
\Cref{sec:discussion} discusses implications and limitations.
Supplementary Materials~A--E provide proofs, identification conditions,
extended tables, and computational details.

%% file: section2.tex

\section{Motivating Data and Preliminary Analysis}
\label{sec:data}

This section introduces the data and documents the empirical features
that jointly determine the model developed in \cref{sec:model}.  We
organize the presentation around four distributional
characteristics---structural zeros, bounded support, overdispersion, and
complex survey design---each of which maps directly to a component of the
hierarchical hurdle beta-binomial specification.  A preliminary analysis
then establishes the reversal phenomenon and the cross-state
heterogeneity that motivate the hierarchical structure and cross-margin
covariance.


\subsection{The NSECE survey}\label{sec:nsece}

The 2019 National Survey of Early Care and Education
\citep[NSECE;][]{NSECEProjectTeam2022} is a nationally representative
survey of early care and education providers conducted by NORC at the
University of Chicago under contract with the Office of Planning,
Research, and Evaluation (OPRE) within the Administration for Children
and Families.  Data collection took place from November 2018 through
July 2019, with $6{,}917$ center-based providers completing the survey.
The center-based provider component employs a stratified
multistage cluster design: 30 strata defined by state groupings, 415
primary sampling units, and $6{,}809$ centers with nonmissing survey
responses.
Of these, 24 report a total enrollment of $n_i = 0$ and are excluded
because the beta-binomial kernel requires a positive denominator,
yielding a final analytic sample of $N = 6{,}785$ providers across all
$S = 51$ states (50 states plus the District of Columbia).  Full details of the sampling frame, weight
construction, and informativeness testing appear in Supplementary
Material~D.

For each provider $i$, the survey records total enrollment $n_i$ (ages
0--5) and infant/toddler (IT) enrollment $Y_i \in \{0, 1, \ldots, n_i\}$
(ages 0--2).  The outcome pair $(Y_i, n_i)$ is the natural unit of
analysis.  Provider-level covariates include the community poverty rate,
urbanicity, and racial/ethnic composition of the provider's Census
tract~\citep{DobbinsEtAl2016, MalikEtAl2018}.  State-level policy
variables capture three dimensions of Child Care and Development Fund
(CCDF) subsidy policy---market rate percentile, tiered reimbursement, and
infant/toddler rate add-on---and enter the model as cross-level
moderators in \cref{sec:application}.  \Cref{tab:data-summary}
summarizes the key descriptive statistics.

\begin{table}[t]
\centering
\caption{Data summary: NSECE 2019 center-based providers.
  Panel~A describes the outcome variables; Panel~B the provider-level
  covariates (continuous variables shown before standardization, as
  means with standard deviations and ranges); Panel~C the survey design
  features, including the Kish design effect
  $\DEFF_{\mathrm{Kish}} = 1 + \mathrm{CV}^2(w)$.}\label{tab:data-summary}
\smallskip
\begin{tabular}{@{}ll@{}}
\toprule
Variable & Value \\
\midrule
\multicolumn{2}{@{}l}{\textit{Panel A: Outcome}} \\[2pt]
Total providers ($N$) & 6{,}785 \\
IT-serving providers ($z_i = 1$) & 4{,}392 (64.7\%) \\
Non-servers ($z_i = 0$) & 2{,}393 (35.3\%) \\
Total enrollment ($n_i$): Mean / Median / Range & 60.5 / 48 / 1--378 \\
IT enrollment ($Y_i \mid z_i = 1$): Mean / Median & 29.4 / 24 \\
IT share ($Y_i/n_i \mid z_i = 1$): Mean (SD) & 0.478 (0.214) \\[4pt]
\multicolumn{2}{@{}l}{\textit{Panel B: Covariates --- Mean (SD) [Range]}} \\[2pt]
Community poverty rate (\%) & 17.3 (8.3) [1.9, 54.2] \\
Urban (\%) & 93.3 \\
Community \% Black & 17.2 (20.5) [0.0, 96.5] \\
Community \% Hispanic & 24.2 (22.6) [0.0, 98.4] \\[4pt]
\multicolumn{2}{@{}l}{\textit{Panel C: Survey design}} \\[2pt]
States ($S$) & 51 (50 + DC) \\
Strata / PSUs & 30 / 415 \\
Sampling weight range & 1--462 \\
Weight CV & 1.66 \\
Kish effective sample size & 1{,}803 \\
Kish $\DEFF$ & 3.76 \\
\bottomrule
\end{tabular}
\end{table}


\subsection{Four data features and their modeling implications}\label{sec:features}

We now examine four distributional characteristics of the NSECE data
that individually and jointly determine the model specification of
\cref{sec:model}.  \Cref{fig:it-distribution} displays the marginal
distribution of the infant/toddler enrollment share, which exhibits all
four features simultaneously.

\begin{figure}[t]
  \centering
  \includegraphics[width=\textwidth]{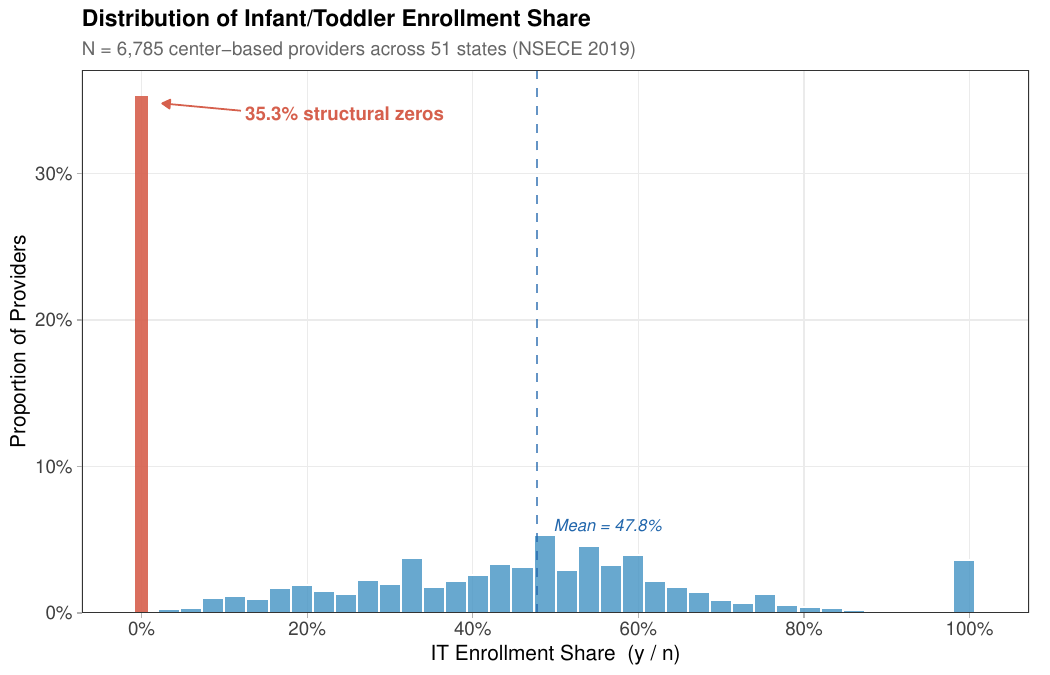}
  \caption{Distribution of the infant/toddler enrollment share $Y_i/n_i$
    across $N = 6{,}785$ center-based providers (NSECE 2019). The spike at
    zero (35.3\% of centers) represents providers that do not serve any
    children under age 3. Among IT-serving providers ($z_i = 1$, $N_{\mathrm{pos}} = 4{,}392$),
    the distribution is roughly unimodal with a mean of 0.478 and substantial
    dispersion (SD $= 0.214$), consistent with overdispersion approximately
    twelve times beyond the binomial---three features that jointly motivate
    the hurdle beta-binomial model.}\label{fig:it-distribution}
\end{figure}

\paragraph{Structural zeros and the hurdle specification.}
Of the 6,785 providers, 2,393 (35.3\%) report zero infant/toddler
enrollment---a rate far too large for sampling variability and reflecting
a deliberate operational decision not to serve children under age
three~\citep{AdamsRohacek2010}.  The hurdle
specification~\citep{Mullahy1986} treats participation and intensity as
distinct processes governed by separate parameters, as developed in
\cref{sec:likelihood}.

\paragraph{Bounded support and the beta-binomial kernel.}
Among servers ($z_i = 1$), infant/toddler enrollment $Y_i$ is bounded
above by total enrollment $n_i$, so the share $Y_i/n_i$ lies in
$(0,1]$.  The beta-binomial distribution respects this
$\{0,1,\ldots,n_i\}$ support and incorporates extra-binomial variation
through the concentration parameter
$\kappa$~\citep{Prentice1986, WagnerEtAl2015}; see
\cref{sec:likelihood}.

\paragraph{Overdispersion and the concentration parameter $\kappa$.}
The observed variance of $Y_i/n_i$ among participants exceeds the
binomial prediction by a factor of approximately 12 at the pooled level
(marginal overdispersion, before accounting for covariates),
and all 51 states individually exhibit overdispersion.  After conditioning
on the model covariates and hierarchical structure, the estimated
concentration $\hat{\kappa} \approx 7$ implies a conditional overdispersion
factor of approximately 7.  The beta-binomial
nests the binomial as $\kappa \to \infty$, modeling this overdispersion
structurally rather than correcting for it post hoc~\citep{Prentice1986}.

\paragraph{Complex survey design and the pseudo-posterior.}
Sampling weights range from 1 to 462, yielding a Kish effective sample
size of $\ESS_{\mathrm{Kish}} = 1{,}803$ and a design effect of
$\DEFF_{\mathrm{Kish}} = 3.76$; a formal test confirms that the design
is informative for the intensive margin~\citep{Pfeffermann1993}.
Ignoring the design would yield credible intervals with empirical
coverage well below nominal; \cref{sec:survey} develops the
pseudo-posterior and sandwich correction.

Each feature alone eliminates a class of candidate models: structural
zeros eliminate single-equation specifications, bounded support
eliminates Poisson and negative-binomial kernels, universal
overdispersion eliminates the binomial, and the complex survey design
eliminates naive Bayesian posteriors.  The \emph{intersection} of all
four constraints narrows the viable specification to a hurdle model with
a beta-binomial kernel, estimated via a pseudo-posterior with sandwich
correction.


\subsection{Preliminary evidence for the poverty reversal}\label{sec:prelim}

The four data features determine the model's distributional
specification.  The model's \emph{hierarchical} structure---state-varying
coefficients with cross-margin covariance---is motivated by a fifth
empirical pattern: the reversal varies systematically across
states.

\paragraph{Pooled estimates.}
A pooled logistic regression of infant/toddler participation ($z_i$) on
the standardized community poverty rate yields an extensive-margin
coefficient of $\hat{\alpha}_{\mathrm{pov}} = -0.020$
($p < 10^{-10}$), confirming that providers in higher-poverty
communities are significantly less likely to serve infants.  A pooled OLS
regression of the enrollment share ($Y_i/n_i$ among servers) on poverty
yields an intensive-margin coefficient of
$\hat{\beta}_{\mathrm{pov}} = +0.003$ ($p < 10^{-14}$), establishing
that, conditional on serving, providers in poorer communities devote a
larger share of enrollment to infants.  The opposing signs constitute the
reversal in its simplest form: poverty \emph{reduces}
participation but \emph{increases} intensity.

The substantive interpretation is noteworthy.  Poverty erects a barrier to
entry: the elevated costs of infant care discourage participation
precisely where demand is greatest~\citep{DobbinsEtAl2016,
MalikEtAl2018}.  Yet among centers that clear this barrier, those in
disadvantaged communities tilt their enrollment mix toward the youngest
age group, reflecting either market demand in areas with fewer
alternatives or mission-driven service provision.  Any model that
collapses participation and intensity into a single equation would
estimate a muted net effect, obscuring the two countervailing forces at
work.

\paragraph{State heterogeneity.}
When these regressions are estimated separately by state, 23 of 51
states (45\%) display the classic reversal pattern
($\hat{\alpha}_{\mathrm{pov},s} < 0$ and
$\hat{\beta}_{\mathrm{pov},s} > 0$).  The remaining states show varying
combinations: some exhibit negative effects on both margins, while
others show positive effects on both.  The cross-state correlation
between the two margin-specific poverty coefficients is $r = 0.021$---near
zero---indicating that the extensive and intensive mechanisms operate
largely independently.  \Cref{fig:poverty-raw} displays the aggregate reversal
in the raw data; the state-level bivariate scatter appears
in~\cref{fig:cross-margin}.  The near-zero cross-margin correlation implies
that knowing a state has a strong participation barrier provides
essentially no information about its intensity response, a finding that
motivates the full covariance structure developed in
\cref{sec:model}~\citep{SippleEtAl2020}.

\begin{figure}[t]
  \centering
  \includegraphics[width=\textwidth]{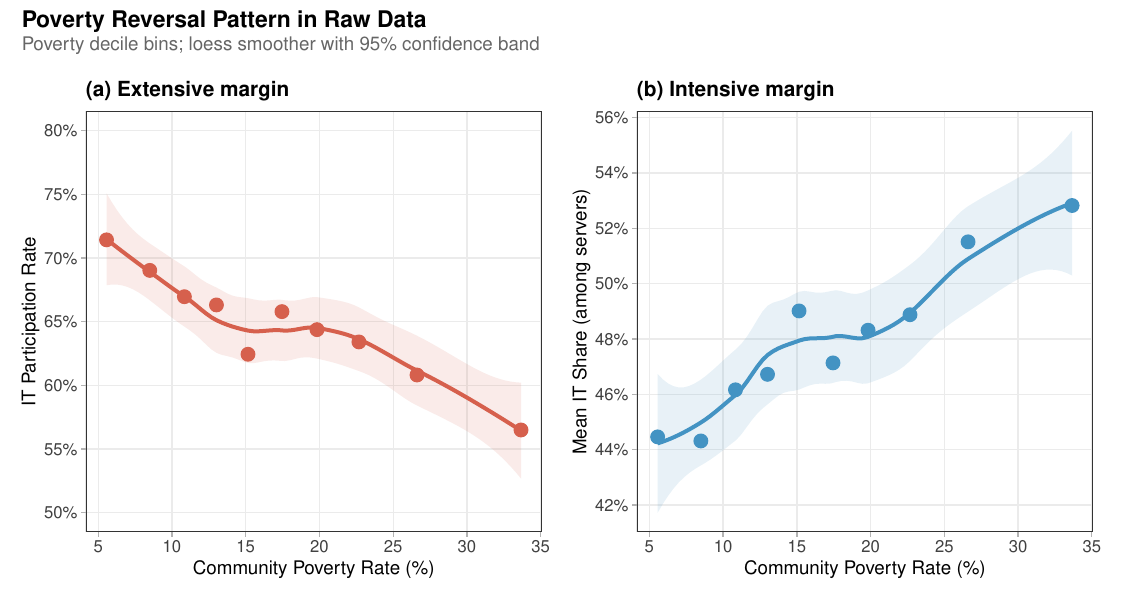}
  \caption{Preliminary evidence for the poverty reversal in the raw NSECE
    2019 data.  \emph{Left panel:} IT participation rate (fraction of
    centers serving any infants) by community poverty decile; a loess
    smoother with 95\% confidence band confirms the declining trend.
    \emph{Right panel:} mean IT enrollment share ($Y_i/n_i$) among servers
    by poverty decile, showing the opposing positive trend.  The two panels
    together display the reversal in its simplest form: participation falls
    with poverty, but intensity rises.}\label{fig:poverty-raw}
\end{figure}

\paragraph{Policy correlations.}
Among the state-level policy variables, tiered reimbursement shows the
strongest association with the poverty--participation relationship:
states with tiered reimbursement tend to have less negative
extensive-margin poverty coefficients, consistent with the hypothesis
that quality-linked payment incentives offset the cost disadvantage of
serving infants in disadvantaged communities.  These preliminary
correlations motivate the cross-level moderation analysis in
\cref{sec:application}.

In sum, the data require a model that simultaneously (i)~separates
structural zeros from positive counts via a hurdle, (ii)~respects the
bounded support via a beta-binomial kernel, (iii)~accommodates
overdispersion through the concentration parameter, (iv)~incorporates
survey weights through a pseudo-posterior with sandwich correction, and
(v)~captures cross-state heterogeneity and cross-margin dependence
through a hierarchical prior with full covariance.  \Cref{sec:model}
develops this model.

%% file: section3.tex

\section{The Hierarchical Hurdle Beta-Binomial Model}\label{sec:model}

This section develops the hierarchical hurdle beta-binomial (HBB) model
whose distributional specification was motivated by the four data features
documented in \cref{sec:features}.  The model combines three
building blocks---hurdle, beta-binomial kernel, and hierarchical
prior---into a unified framework and adds two methodological
contributions that are new to this class of models: a cross-margin
covariance whose identification operates entirely through the
hierarchical prior rather than the conditional likelihood, and a
Cholesky-based sandwich correction that extends
pseudo-posterior inference to the two-part hierarchical setting.

In plain terms, the model works as follows.  A \emph{hurdle} first
separates providers that serve infants and toddlers from those that do
not, treating participation as a logistic regression with
state-specific intercepts and slopes.  For participants, a
\emph{beta-binomial} kernel models the share of enrollment devoted to
infants and toddlers, accommodating the bounded count and within-state
overdispersion.  The two margins share a joint hierarchical prior on
their state-level deviations, so that a state's participation barrier
and its conditional enrollment intensity can be positively or negatively
correlated---the cross-margin covariance that is the paper's central
structural novelty.  Sampling weights enter through a
pseudo-likelihood, and a sandwich correction restores design-consistent
coverage for population-average parameters.


\subsection{Notation and data structure}\label{sec:notation}

We observe $N = 6{,}785$ center-based childcare providers drawn from
$S = 51$ states (including the District of Columbia).  Each provider $i$
is nested in a state $s[i]$.  The outcome pair $(Y_i, n_i)$ records
infant/toddler enrollment $Y_i \in \{0, 1, \ldots, n_i\}$ out of total
enrollment $n_i$; the participation indicator is $z_i = \one(Y_i > 0)$.
\cref{tab:notation} collects the symbols used throughout the paper.
All continuous covariates are standardized prior to analysis.

\begin{table}[ht]
\centering
\caption{Notation summary. Dimensions are shown in parentheses.}\label{tab:notation}
\smallskip
\begin{tabularx}{\textwidth}{@{}lclX@{}}
\toprule
Symbol & Dimension & Domain & Description \\
\midrule
\multicolumn{4}{@{}l}{\textit{Indices and dimensions}} \\[2pt]
$i$ & --- & $\{1,\ldots,N\}$ & Provider index ($N = 6{,}785$) \\
$s$ & --- & $\{1,\ldots,S\}$ & State index ($S = 51$, incl.\ DC) \\
$s[i]$ & --- & $\{1,\ldots,S\}$ & State membership of provider $i$ \\
$N_s$ & --- & $\mathbb{Z}_{>0}$ & Number of providers in state $s$; $\sum_s N_s = N$ \\[4pt]
\multicolumn{4}{@{}l}{\textit{Outcome variables}} \\[2pt]
$n_i$ & $(1)$ & $\mathbb{Z}_{>0}$ & Total enrollment, ages 0--5 (treated as fixed) \\
$Y_i$ & $(1)$ & $\{0,1,\ldots,n_i\}$ & Infant/toddler (IT) enrollment count \\
$z_i$ & $(1)$ & $\{0,1\}$ & $z_i = \one(Y_i > 0)$: IT participation indicator \\[4pt]
\multicolumn{4}{@{}l}{\textit{Covariates}} \\[2pt]
$\bx_i$ & $(P \times 1)$ & $\R^P$ & Provider covariates, $P=5$: intercept, poverty, urban, Black, Hispanic \\
$\bx_i^{(r)}$ & $(q\times 1)$ & $\R^q$ & State-varying subvector of $\bx_i$ ($q = P$ in M2/M3) \\
$q_i$ & $(1)$ & $(0,1)$ & Participation probability (extensive margin); cf.~\eqref{eq:eta-ext} \\
$\bv_s$ & $(Q\times 1)$ & $\R^Q$ & State policy vector, $Q = 4$: intercept + $Q-1$ policy instruments \\[4pt]
\multicolumn{4}{@{}l}{\textit{Parameters}} \\[2pt]
$\balpha$ & $(P \times 1)$ & $\R^P$ & Population-average extensive-margin coefficients \\
$\bbeta$ & $(P \times 1)$ & $\R^P$ & Population-average intensive-margin coefficients \\
$\kappa$ & $(1)$ & $\R_{>0}$ & Beta-binomial concentration \\
$\bdelta_{k,s}$ & $(q\times 1)$ & $\R^q$ & State deviations, margin $k\in\{1,2\}$, state $s$ \\
$\bGamma_k$ & $(q\times Q)$ & $\R^{q \times Q}$ & Cross-level policy moderator matrix, margin $k$ \\
$\bSigma_\delta$ & $(2q\times 2q)$ & $\mathbb{S}_{++}^{2q}$ & Cross-margin state-effect covariance \\[4pt]
\multicolumn{4}{@{}l}{\textit{Survey design}} \\[2pt]
$w_i$ & $(1)$ & $\R_{>0}$ & Sampling weight \\
$\tilde{w}_i$ & $(1)$ & $\R_{>0}$ & Normalized weight: $\tilde{w}_i = w_i N / \textstyle\sum_j w_j$ \\
\bottomrule
\end{tabularx}
\end{table}


\subsection{The hurdle beta-binomial likelihood}\label{sec:likelihood}

\paragraph{The beta-binomial distribution.}
The beta-binomial distribution in the mean--precision parameterization
\citep{Griffiths1973,Prentice1986}
$(\mu,\kappa)$ sets $a = \mu\kappa$ and $b = (1-\mu)\kappa$, where
$\mu \in (0,1)$ is the conditional mean parameter and $\kappa > 0$ is the
concentration parameter. Larger $\kappa$ implies less overdispersion, with the
$\Bin(n,\mu)$ distribution recovered as $\kappa \to \infty$. The probability
mass function (PMF) is
\begin{equation}\label{eq:bb-pmf}
  \Prob(Y = y \mid n,\mu,\kappa)
  = \binom{n}{y}
    \frac{B\bigl(y + \mu\kappa,\; n - y + (1-\mu)\kappa\bigr)}
         {B\bigl(\mu\kappa,\;(1-\mu)\kappa\bigr)},
  \quad y = 0,1,\ldots,n,
\end{equation}
where $B(\cdot,\cdot)$ is the beta function. The first two moments are
\begin{equation}\label{eq:bb-moments}
  \E[Y] = n\mu, \qquad
  \Var(Y) = n\mu(1-\mu)\,\frac{n + \kappa}{1 + \kappa}.
\end{equation}
The variance inflation factor $(n+\kappa)/(1+\kappa)$ exceeds unity for all
$n \ge 2$, capturing overdispersion relative to the binomial. In the
childcare data, with a typical enrollment of $n_i \approx 50$ and an
estimated $\kappa \approx 7$, this factor is approximately $57/8 \approx 7$,
far beyond what a binomial model can accommodate.

\paragraph{Computational remark.}
In practice, we compute $\log f_{\BetaBin}$ using the
\texttt{lbeta()} function to avoid numerical overflow when $n_i$
exceeds 50. Specifically, the log-PMF is evaluated as
$\log\binom{n}{y} + \texttt{lbeta}(y+a,\,n-y+b) - \texttt{lbeta}(a,b)$,
which remains stable for $n_i$ up to several hundred.

\paragraph{Zero probability.}
At $y = 0$ the PMF simplifies to a product form
that is both interpretable and numerically stable:
\begin{equation}\label{eq:p0}
  p_0(n,\mu,\kappa)
  = \prod_{j=0}^{n-1} \frac{(1-\mu)\kappa + j}{\kappa + j}.
\end{equation}
Each factor in the product is the conditional probability that the $j$-th
``trial'' yields a non-IT outcome, given all preceding trials did so. The
product form shows that $p_0$ is strictly decreasing in both $\mu$ and $\kappa$
(for $n \ge 2$): higher mean intensity and lower overdispersion each reduce the
probability of observing zero IT enrollment.

\paragraph{Hurdle construction.}
As documented in \cref{sec:features}, the zeros in the childcare data are
structural: they reflect a deliberate decision not to serve infants rather
than a random draw from a count distribution.  This motivates a hurdle
specification \citep{Mullahy1986,Cragg1971} that separates the participation
decision from the enrollment intensity.

\begin{definition}[Hurdle beta-binomial model]\label{def:hbb}
  For provider $i$ with enrollment capacity $n_i$, the HBB model specifies
  \begin{equation}\label{eq:hbb}
    f(y_i \mid q_i, \mu_i, \kappa, n_i) =
    \begin{cases}
      1 - q_i, & \text{if } y_i = 0, \\[4pt]
      q_i \cdot f_{\ZTBB}(y_i \mid n_i, \mu_i, \kappa),
        & \text{if } y_i \in \{1,\ldots,n_i\},
    \end{cases}
  \end{equation}
  where $q_i \in (0,1)$ is the participation probability (extensive margin),
  $\mu_i \in (0,1)$ is the conditional mean IT share (intensive margin), and
  the zero-truncated beta-binomial PMF is
  \begin{equation}\label{eq:ztbb}
    f_{\ZTBB}(y \mid n,\mu,\kappa)
    = \frac{f_{\BetaBin}(y \mid n,\mu,\kappa)}{1 - p_0(n,\mu,\kappa)},
    \quad y = 1,\ldots,n.
  \end{equation}
\end{definition}

The hurdle separates two distinct decisions that childcare centers face. The
\emph{extensive margin} captures whether a center serves any infants at
all---the access decision---governed by $q_i$. The \emph{intensive margin}
captures how many infant slots a center offers relative to its total
capacity---the intensity decision---governed by $\mu_i$. This separation
reflects a genuine two-stage process: a center must first decide to accept the
costs and regulatory burden of infant care (clearing the hurdle), and only then
does it determine how many infant slots to offer.

The log-likelihood under conditional independence decomposes additively into
extensive and intensive components,
\begin{multline}\label{eq:loglik-decomp}
  \ell(\btheta)
  = \underbrace{\sum_{i=1}^{N}\bigl[z_i\log q_i + (1-z_i)\log(1-q_i)\bigr]}
    _{\ell_{\mathrm{ext}}(\balpha)} \\
  + \underbrace{\sum_{i:\,z_i=1}\bigl[\log f_{\BetaBin}(y_i \mid n_i,\mu_i,\kappa)
    - \log\bigl(1 - p_{0,i}\bigr)\bigr]}
    _{\ell_{\mathrm{int}}(\bbeta,\kappa)},
\end{multline}
where $p_{0,i} = p_0(n_i,\mu_i,\kappa)$. This separation is a structural
property of hurdle models: the binary indicator $z_i$ is a sufficient statistic
partition that renders the extensive parameters $\balpha$ and the intensive
parameters $(\bbeta,\kappa)$ variation-independent in the likelihood. In
contrast, zero-inflated models introduce latent at-risk indicators that couple
the two components, complicating both estimation and interpretation.

\paragraph{Unconditional expectation and the intensity function.}
The unconditional mean enrollment count is
\begin{equation}\label{eq:uncond-mean}
  \E[Y_i] = q_i \cdot n_i \cdot h(\mu_i, n_i, \kappa),
  \qquad
  h(\mu, n, \kappa) \equiv \frac{\mu}{1 - p_0(n,\mu,\kappa)},
\end{equation}
where $h$ is the \emph{intensity function} mapping the latent mean proportion
$\mu$ to the expected proportion among participants. Since $h(\mu) \ge \mu$
(truncation shifts the conditional mean upward), participating centers are
predicted to have a higher IT share than the latent parameter $\mu$ alone
would suggest. The function $h$ is central to the reversal
interpretation: the key question is whether the opposing signs of the extensive
and intensive coefficients can be ``absorbed'' by a nonlinear response of $h$,
or whether the reversal is genuinely a coefficient-level phenomenon. The
following theorem resolves this.

\begin{theorem}[Monotonicity of $h$]\label{thm:monotonicity}
  For all $n \ge 2$, $\kappa > 0$, and $\mu \in (0,1)$,
  \begin{equation}\label{eq:dh-dmu}
    \frac{\partial h}{\partial \mu} > 0.
  \end{equation}
  For $n = 1$, $h(\mu) \equiv 1$ for all $\mu \in (0,1)$.
\end{theorem}

\begin{proof}[Proof sketch]
  Define $\Phi(\mu) = (1 - p_0) - \mu\, p_0\, \Lambda$, where
  $\Lambda = \kappa\sum_{j=0}^{n-1}[(1-\mu)\kappa + j]^{-1} > 0$.
  Since $\partial h / \partial\mu = \Phi(\mu)/(1-p_0)^2$, the sign of the
  derivative equals the sign of $\Phi$. One verifies that $\Phi(0) = 0$ and
  $\Phi'(\mu) = \mu\, p_0(\Lambda^2 - \Lambda_2)$, where
  $\Lambda_2 = \kappa^2\sum_{j=0}^{n-1}[(1-\mu)\kappa + j]^{-2}$.
  For $n \ge 2$, the Cauchy--Schwarz inequality gives
  $\Lambda^2 > \Lambda_2$, so $\Phi'(\mu) > 0$ on $(0,1)$.
  Since $\Phi$ starts at zero and is strictly increasing, $\Phi(\mu) > 0$ for
  all $\mu \in (0,1)$. The full proof appears in~\cref{app:math-properties}.
\end{proof}

\paragraph{Interpretation for the poverty reversal.}
\cref{thm:monotonicity} establishes that the intensity function $h$ is strictly
increasing in $\mu$: a center with a higher underlying IT intensity $\mu$ will
always have a higher expected IT share among participants, regardless of how much
overdispersion is present. Consequently, the poverty reversal---whereby higher poverty reduces participation ($\alpha_{\mathrm{pov}} < 0$) yet
increases conditional intensity ($\beta_{\mathrm{pov}} > 0$)---is purely a
\emph{coefficient-level} phenomenon, not an artifact of the nonlinear response
function. Researchers can examine the posterior distributions of
$\alpha_{\mathrm{pov}}$ and $\beta_{\mathrm{pov}}$ directly to assess the
reversal, without needing to account for nonlinear link function effects.

\begin{proposition}[Elasticity of $h$]\label{prop:elasticity}
  Define the elasticity $\varepsilon_h = (\partial h / \partial\mu)(\mu / h)$.
  Then
  \begin{equation}\label{eq:elasticity}
    \varepsilon_h(\mu) = 1 - \omega(\mu),
    \qquad
    \omega(\mu) = \frac{\mu\, p_0\, \Lambda}{1 - p_0} \in (0,1)
    \quad \text{for } n \ge 2.
  \end{equation}
  Hence $\varepsilon_h \in (0,1)$: the intensity function responds
  inelastically to changes in $\mu$.
\end{proposition}

The inelastic response arises because increasing $\mu$ simultaneously increases
the probability that a center would have had a positive draw even without the
hurdle (raising $1 - p_0$, the denominator of $h$), partially offsetting the
numerator effect. In practical terms, a one-percent increase in a center's
underlying IT propensity $\mu$ translates into \emph{less} than a one-percent
increase in its expected IT share among participants. At our estimated parameters
($\mu \approx 0.30$, $\kappa \approx 7$, $n \approx 50$), $\omega \approx 0.07$,
so the elasticity is approximately 0.93---relatively modest damping in this
parameter region.


\subsection{State-varying coefficients (SVC) with cross-margin covariance}\label{sec:svc}

The reversal is not a single national pattern. Preliminary estimates
show that fewer than half of states exhibit the classic reversal---a negative
extensive-margin poverty coefficient paired with a positive intensive-margin
one (\cref{sec:prelim}). This geographic variation raises a fundamental
question: \emph{Does the poverty--enrollment relationship vary systematically
across states, and if so, are the extensive and intensive margins linked?}

\paragraph{Linear predictors.}
Following the spatially varying coefficient framework of~\citet{GelfandEtAl2003},
the extensive and intensive margins are linked to provider characteristics
through logit linear predictors:
\begin{align}
  \logit(q_i)
    &= \bx_i\t\balpha + \bx_i^{(r)\t}\bdelta_{1,s[i]},
    \label{eq:eta-ext} \\[4pt]
  \logit(\mu_i)
    &= \bx_i\t\bbeta + \bx_i^{(r)\t}\bdelta_{2,s[i]},
    \label{eq:eta-int}
\end{align}
where $\balpha, \bbeta \in \R^P$ are population-average fixed effects---the
``national story''---and $\bdelta_{k,s} \in \R^q$
($k = 1$ for extensive, $k = 2$ for intensive) are zero-mean state-specific
deviations that capture how each state's poverty--enrollment relationship
departs from the national average. The subvector $\bx_i^{(r)}$ consists of
the first $q$ components of $\bx_i$ whose coefficients are permitted to vary
across states.
The \emph{total state-specific coefficient} for covariate $j$ on margin $k$ is
\begin{equation}\label{eq:total-coef}
  \tilde{\alpha}_{j,s} = \alpha_j + \delta_{1,s,j},
  \qquad
  \tilde{\beta}_{j,s} = \beta_j + \delta_{2,s,j}.
\end{equation}
In the childcare application, a state with $\tilde{\alpha}_{\mathrm{pov},s} < 0$
and $\tilde{\beta}_{\mathrm{pov},s} > 0$ exhibits the reversal: higher
community poverty is associated with lower participation but higher conditional
intensity.

\paragraph{Cross-margin covariance: the key novelty.}
Standard hurdle models estimate the two margins independently. We depart from
this convention by stacking the state deviations into a single $2q$-dimensional
vector:
\begin{equation}\label{eq:delta-stack}
  \bdelta_s
  = \begin{pmatrix} \bdelta_{1,s} \\ \bdelta_{2,s} \end{pmatrix}
  \sim \Norm_{2q}\!\parens{\mathbf{0},\;\bSigma_\delta},
  \qquad s = 1,\ldots,S,
\end{equation}
with the $2q \times 2q$ positive-definite covariance
\begin{equation}\label{eq:sigma-delta}
  \bSigma_\delta
  = \begin{pmatrix}
      \bSigma_{11} & \bSigma_{12} \\
      \bSigma_{21} & \bSigma_{22}
    \end{pmatrix},
\end{equation}
where $\bSigma_{11}$ governs between-state variation in extensive-margin
coefficients, $\bSigma_{22}$ governs the intensive margin, and the
off-diagonal block $\bSigma_{12} = \bSigma_{21}\t$ captures
\emph{cross-margin dependence} (in the policy-moderated model M3b, $\bSigma_{12}$ refers to the \emph{residual} cross-margin covariance after removing policy effects; see~\cref{smb:rem-sigma12-residual}). This cross-margin covariance is the central
structural novelty of the model: it allows states with stronger participation
barriers (larger $|\delta_{1,\mathrm{pov},s}|$) to simultaneously exhibit
stronger or weaker intensity effects ($\delta_{2,\mathrm{pov},s}$), with the
direction and magnitude estimated from data. The cross-margin correlation for
covariate $j$ is
\begin{equation}\label{eq:cross-corr}
  \varrho_j^{\mathrm{cross}} = \frac{[\bSigma_{12}]_{jj}}
  {\sqrt{[\bSigma_{11}]_{jj}\,[\bSigma_{22}]_{jj}}}.
\end{equation}

\begin{proposition}[Non-identification of $\bSigma_{12}$ from the
  conditional likelihood]
  \label{prop:sigma12-nonid}
  Conditional on the state-level deviations $\bdelta_s$, the Fisher
  information matrix for the HBB model is block-diagonal between
  the extensive-margin parameters $(\balpha, \bdelta_{1,1},\ldots,\bdelta_{1,S})$
  and the intensive-margin parameters
  $(\bbeta, \bdelta_{2,1},\ldots,\bdelta_{2,S}, \kappa)$. Consequently,
  $\bSigma_{12}$ receives zero information from the conditional
  likelihood and is identified through the hierarchical prior on
  $\bdelta_s$ acting on the $S = 51$ estimated state-level deviations.
\end{proposition}

\begin{proof}
  The log-likelihood decomposes as $\ell = \ell_{\mathrm{ext}} + \ell_{\mathrm{int}}$
  (\cref{eq:loglik-decomp}). Since $\ell_{\mathrm{ext}}$ depends only on
  $(\balpha, \bdelta_{1,\cdot})$ and $\ell_{\mathrm{int}}$ depends only on
  $(\bbeta, \bdelta_{2,\cdot}, \kappa)$, the cross-derivatives
  $\partial^2 \ell / \partial\phi_1 \partial\phi_2\t$ vanish for any
  extensive-margin parameter $\phi_1$ and intensive-margin parameter $\phi_2$.
  Hence the observed and expected information matrices are block-diagonal
  across margins, and the off-diagonal block $\bSigma_{12}$ appears nowhere
  in the likelihood.
\end{proof}

\paragraph{$\bSigma_{12}$ as a hierarchical estimand.}
\cref{prop:sigma12-nonid} clarifies the information structure for
$\bSigma_{12}$: its posterior is shaped entirely by the $S = 51$ pairs
$(\bdelta_{1,s}, \bdelta_{2,s})$ drawn from the common
$\Norm_{2q}(\mathbf{0}, \bSigma_\delta)$ distribution, giving it an
effective sample size of $S$ rather than $N$.  While a frequentist could
in principle estimate $\bSigma_{12}$ via REML on the marginal
likelihood, the high-dimensional covariance structure ($2q = 10$
free parameters per margin) relative to the number of groups ($S = 51$)
makes unrestricted frequentist estimation fragile.  The $\LKJ(\eta)$
prior provides finite-sample regularization, concentrating the posterior
at rate $O(S^{-1/2})$. This contrasts with the within-margin blocks
$\bSigma_{11}$ and $\bSigma_{22}$, which also receive direct likelihood
information through the state-specific coefficients and concentrate at
rate $O(N^{-1/2})$.

\begin{remark}[Conditional versus marginal likelihood]\label{rem:cond-vs-marg}
  The block-diagonality in \cref{prop:sigma12-nonid} is a statement
  about the \emph{conditional} likelihood $p(\mathbf{y} \mid \bdelta,
  \balpha, \bbeta, \kappa)$.  The marginal likelihood obtained by
  integrating over the random effects,
  $p(\mathbf{y} \mid \bSigma_\delta, \balpha, \bbeta, \kappa) = \int
  p(\mathbf{y} \mid \bdelta) \, p(\bdelta \mid \bSigma_\delta) \,
  \mathrm{d}\bdelta$, does \emph{not} factorize when $\bSigma_{12}
  \neq \mathbf{0}$, so that $\bSigma_{12}$ is in principle identified
  through the marginal likelihood.  However, the conditional
  formulation is natural for MCMC computation, where $\bdelta_s$ is
  sampled as a latent variable and $\bSigma_\delta$ is updated
  conditional on the realized deviations.  In a frequentist framework,
  REML estimation of $\bSigma_{12}$ via the marginal likelihood is
  possible in principle but numerically fragile when $S = 51$ groups
  contribute to a $2q = 10$-dimensional covariance matrix: boundary
  solutions and singular Hessians are common in this regime.  The
  Bayesian approach with $\LKJ$ regularization provides stable
  estimation at the cost of the slower $O(S^{-1/2})$ concentration
  rate documented in \cref{prop:sigma12-id}.
\end{remark}

\begin{proposition}[Identifiability of $\bSigma_{12}$ via the prior]
  \label{prop:sigma12-id}
  Suppose (i) $S > 2q$, (ii) each state $s$ has at least $q$ providers with
  linearly independent covariates, and (iii) the prior on $\bSigma_\delta$ has
  full support on $\mathbb{S}_{++}^{2q}$. Then the posterior of $\bSigma_\delta$
  (including $\bSigma_{12}$) is well-concentrated, with the posterior standard
  deviation of each element of $\bSigma_{12}$ decreasing at rate $O(S^{-1/2})$.
  For our data, $S = 51 \gg 2q = 10$.
\end{proposition}

\begin{remark}[Boundary condition]\label{rem:boundary}
  When $S < 2q$ the cross-margin covariance is poorly identified even with a
  proper prior, because $S$ draws from a $2q$-dimensional distribution yield a
  singular sample covariance matrix. In this regime, the scalar cross-margin
  parameter $a$ of~\citet{GhosalEtAl2020}, which nests as
  $\bSigma_{12} = a\,\bSigma_{22}$, provides a parsimonious but still useful
  summary of cross-margin dependence.
\end{remark}

\paragraph{Computational remark.}
In Stan~\citep{CarpenterEtAl2017}, the state deviations are sampled via the non-centered
parameterization (NCP): draw $\bz_s \iid \Norm_{2q}(\mathbf{0},\mathbf{I})$
and set $\bdelta_s = \diag(\boldsymbol{\tau})\,\mathbf{L}_\varepsilon\,\bz_s$,
where $\mathbf{L}_\varepsilon$ is the Cholesky factor of the
correlation matrix $\mathbf{R}_\varepsilon$. This eliminates the funnel
geometry~\citep{BetancourtGirolami2015} that arises under centered parameterization when small $\tau_j$ forces
$\delta_{s,j}$ near zero, causing divergent transitions in Hamiltonian Monte
Carlo. With $S = 51$ states and $N_s$ ranging from 17 to 568, NCP is the
appropriate default.


\subsection{Cross-level policy moderators and marginal effect decomposition}
\label{sec:policy}

Why does the reversal appear in some states but not others?  State-level
subsidy policies provide a natural explanation: they directly affect the cost
structure that centers face when deciding whether and how much to serve infants.
We formalize this hypothesis through cross-level policy moderators.

\paragraph{The moderation framework.}
\begin{definition}[Policy moderator specification]\label{def:gamma}
  For each margin $k \in \{1,2\}$, the state-specific deviation decomposes as
  \begin{equation}\label{eq:gamma-decomp}
    \bdelta_{k,s} = \bGamma_k\,\bv_s + \bepsilon_{k,s},
  \end{equation}
  where $\bGamma_k \in \R^{q \times Q}$ is the cross-level interaction matrix
  and $\bepsilon_{k,s}$ is a residual state effect.
  Stacking both margins, the residuals follow a joint distribution that
  preserves cross-margin dependence:
  \begin{equation}\label{eq:eps-stacked}
    \bepsilon_s
    \equiv \begin{pmatrix}\bepsilon_{1,s}\\\bepsilon_{2,s}\end{pmatrix}
    \sim \Norm_{2q}\!\left(\mathbf{0},\;
      \bSigma_\varepsilon
      = \begin{pmatrix}
        \bSigma_{\varepsilon,1} & \bSigma_{12}\\
        \bSigma_{21} & \bSigma_{\varepsilon,2}
      \end{pmatrix}\right),
  \end{equation}
  so that the full model estimates the residual cross-margin covariance
  $\bSigma_{12} = \Cov(\bepsilon_{1,s},\, \bepsilon_{2,s})$
  after removing the policy-explained component (see
  \cref{prop:sigma12-nonid} and Supplementary Material~B).
\end{definition}

The state-level policy vector is
$\bv_s \in \R^Q$ with $Q = 4$, consisting of an intercept and $Q - 1 = 3$
Child Care and Development Fund (CCDF) policy instruments that directly affect
the cost structure centers face for infant care.  The specific policy variables
and their substantive roles are detailed in \cref{sec:application}.

\paragraph{Nested model hierarchy.}
The HBB framework nests a sequence of models with increasing complexity,
summarized in~\cref{tab:model-hierarchy}.

\begin{table}[ht]
\centering
\caption{Nested model hierarchy. Each model adds structure relative to its
  predecessor.}\label{tab:model-hierarchy}
\smallskip
\small
\begin{tabular}{@{}llp{0.38\textwidth}c@{}}
\toprule
Model & Structure & Question addressed & Approx.\ params \\
\midrule
M0 & Fixed effects only & National-level poverty reversal? & $11$ \\
M1 & + Random intercepts$^{\dagger}$ & Does enrollment vary by state? & $116$ \\
M2 & + State-varying coefs.\ & Does the reversal vary by state? & $\sim 576$ \\
M3a & + Cross-margin cov.\ & Are the two margins linked? & $\sim 576$ \\
M3b & + Policy moderators & Can policies explain state variation? & $\sim 626$ \\
\bottomrule
\end{tabular}
\par\vspace{2pt}
{\footnotesize $^{\dagger}$M1 is obtained as a special case of M2 by
  restricting the SVC to intercepts only ($q = 1$).}
\end{table}

Each transition addresses a question raised by its predecessor: M0 to M1 asks
whether the baseline enrollment rate varies by state; M1 to M2, whether the
reversal pattern varies by state; M2 to M3a, whether cross-margin
covariance improves fit; and M3a to M3b, whether observable policies explain
the state heterogeneity.

\paragraph{Poverty coefficient expansion.}
\begin{proposition}[Poverty coefficient expansion]\label{prop:poverty-expansion}
  Under M3b, the total state-specific poverty coefficient on the extensive
  margin is
  \begin{equation}\label{eq:poverty-expansion}
    \tilde{\alpha}_{\mathrm{pov},s}
    = \underbrace{\alpha_{\mathrm{pov}}}_{\substack{\text{national} \\
    \text{average}}}
    + \underbrace{\bgamma_{1,\mathrm{pov}}\t\bv_s}_{\substack{\text{policy-explained}
    \\ \text{deviation}}}
    + \underbrace{\varepsilon_{1,\mathrm{pov},s}}_{\substack{\text{residual}
    \\ \text{state effect}}},
  \end{equation}
  where $\bgamma_{1,\mathrm{pov}} = (\gamma_{1,2,1},\ldots,\gamma_{1,2,Q})\t$
  is the second row of $\bGamma_1$. Analogously, the intensive-margin
  counterpart is
  $\tilde{\beta}_{\mathrm{pov},s} = \beta_{\mathrm{pov}} +
  \bgamma_{2,\mathrm{pov}}\t\bv_s + \varepsilon_{2,\mathrm{pov},s}$.
\end{proposition}

\cref{prop:poverty-expansion} decomposes the state-specific poverty effect into
three interpretable components: a national average ($\alpha_{\mathrm{pov}}$), a
policy-explained deviation ($\bgamma_{1,\mathrm{pov}}\t\bv_s$), and a residual
idiosyncrasy ($\varepsilon_{1,\mathrm{pov},s}$).  The variance of the
state-specific poverty coefficient decomposes correspondingly as
$\Var(\tilde{\alpha}_{\mathrm{pov},s}) =
\bgamma_{1,\mathrm{pov}}\t\,\Var(\bv_s)\,\bgamma_{1,\mathrm{pov}}
+ (\bSigma_{\varepsilon,1})_{\mathrm{pov,pov}}$,
partitioning between-state variation into policy-explained and residual
components.

\paragraph{Marginal effect decomposition.}
To translate model parameters into policy-relevant quantities, we decompose the
total effect of any covariate $x_k$ on expected enrollment into access and
intensity channels.

\begin{proposition}[Log-access and log-intensity decomposition]
  \label{prop:lae-lie}
  For provider $i$ in state $s$, the marginal effect of covariate $x_k$ on
  log expected enrollment decomposes as
  \begin{equation}\label{eq:lae-lie}
    \frac{\partial \log \E[Y_i]}{\partial x_{i,k}}
    = \underbrace{(1 - q_i)\,\tilde{\alpha}_{k,s}}_{\LAE_k}
    \;+\;
    \underbrace{(1 - \mu_i)\,\varepsilon_{h,i}\,\tilde{\beta}_{k,s}}_{\LIE_k},
  \end{equation}
  where $\LAE_k$ is the log access effect and $\LIE_k$ is the log
  intensity effect, with $\varepsilon_{h,i} = 1 - \omega_i$ from
  \cref{prop:elasticity}.
\end{proposition}

\begin{proof}
  Since $\log\E[Y_i] = \log q_i + \log n_i + \log h_i$ and $n_i$ is fixed,
  the chain rule gives two terms. For the access term, the logit link yields
  $\partial(\log q_i)/\partial x_{i,k} = (1 - q_i)\,\tilde{\alpha}_{k,s}$.
  For the intensity term,
  $\partial(\log h_i)/\partial x_{i,k}
  = \varepsilon_{h,i}\,(1-\mu_i)\,\tilde{\beta}_{k,s}$,
  using the logistic derivative $\partial\mu/\partial\eta = \mu(1-\mu)$
  and the elasticity definition
  $\varepsilon_{h} = \mu\,\partial(\log h)/\partial\mu$.
\end{proof}

\paragraph{Concrete interpretation.}
At the sample average ($q \approx 0.64$, $\mu \approx 0.30$,
$\varepsilon_h \approx 0.93$) with population-average coefficients
$\alpha_{\mathrm{pov}} = -0.324$ and $\beta_{\mathrm{pov}} = +0.090$:
$\LAE_{\mathrm{pov}} = (1 - 0.64)(-0.324) = -0.117$ and
$\LIE_{\mathrm{pov}} = (1 - 0.30)(0.93)(0.090) = +0.058$.
The net effect is $-5.9\%$: the access barrier dominates, but the intensity
response offsets nearly half.  Three features of the decomposition merit
emphasis: it is \emph{exactly additive} on the log scale, both multipliers
$(1-q_i)$ and $(1-\mu_i)\varepsilon_{h,i}$ are positive and bounded so that
the sign of each component is determined solely by the coefficient sign, and
the reversal pattern is confirmed as a coefficient-level phenomenon consistent
with~\cref{thm:monotonicity}.


\subsection{Survey-weighted pseudo-posterior and sandwich correction}
\label{sec:survey}

The NSECE 2019 employs a stratified multistage cluster design with 30 strata,
415 primary sampling units (PSUs), and sampling weights ranging from 1 to 462
(Kish $\ESS = 1{,}803$; $\DEFF = 3.76$). Under informative sampling designs,
ignoring these weights can bias point estimates and produce
miscalibrated credible intervals
\citep{PfeffermannEtAl1998}. Following~\citet{SavitskyToth2016} and
\citet{WilliamsSavitsky2021}, we adopt a pseudo-posterior approach with
post-hoc sandwich correction
\citep[see also][for a general treatment of weighting in multilevel
models]{RabeHeskethSkrondal2006}.

\paragraph{Pseudo-log-likelihood.}
Define the weighted log-likelihood
\begin{equation}\label{eq:pseudo-loglik}
  \ell^{(w)}(\btheta)
  = \sum_{i=1}^{N} \tilde{w}_i \log f_{\HBB}(y_i \mid \btheta),
\end{equation}
where $\tilde{w}_i = w_i N / \sum_j w_j$ are the normalized survey weights.

\begin{theorem}[Pseudo-posterior propriety]\label{thm:propriety}
  Under a proper prior $\pi(\btheta)$, the pseudo-posterior
  \begin{equation}\label{eq:pseudo-posterior}
    \pi^{(w)}(\btheta \mid \by)
    \propto \exp\!\bigl(\ell^{(w)}(\btheta)\bigr)\,\pi(\btheta)
  \end{equation}
  is proper for any set of positive, finite weights $\{w_i\}_{i=1}^N$.
\end{theorem}

\begin{proof}[Proof sketch]
  Since $f_{\HBB}(y_i \mid \btheta) \in (0,1]$ and $\tilde{w}_i > 0$,
  each weighted factor satisfies $f_i^{\tilde{w}_i} \le 1$. Hence
  $\int \exp(\ell^{(w)})\,\pi(\btheta)\,d\btheta
  \le \int \pi(\btheta)\,d\btheta = 1 < \infty$, using
  the properness of $\pi(\btheta)$. The full proof, which verifies that the
  marginal likelihood is strictly positive, appears in Supplementary
  Material~C.
\end{proof}

\begin{proposition}[Weighted separability]\label{prop:weighted-sep}
  The weighted log-likelihood retains the two-part decomposition:
  \begin{equation}\label{eq:weighted-sep}
    \ell^{(w)}(\btheta)
    = \ell_{\mathrm{ext}}^{(w)}(\balpha)
    + \ell_{\mathrm{int}}^{(w)}(\bbeta, \kappa),
  \end{equation}
  where
  $\ell_{\mathrm{ext}}^{(w)} = \sum_i \tilde{w}_i[z_i\log q_i + (1-z_i)\log(1-q_i)]$
  and
  $\ell_{\mathrm{int}}^{(w)} = \sum_{i:\,z_i=1}
  \tilde{w}_i[\log f_{\BetaBin}(y_i \mid n_i,\mu_i,\kappa) - \log(1-p_{0,i})]$.
\end{proposition}

\begin{proof}
  Weighting by $\tilde{w}_i > 0$ multiplies each observation's log-density
  by a positive scalar, which cannot introduce functional dependence between
  parameters that were already in separate additive components.
\end{proof}

\paragraph{Bernstein--von Mises (BvM) condition.}
The weight ratio $w_{\max}/w_{\min} = 462$ exceeds $N^{1/2} \approx 82.5$.
The finite-sample diagnostic
$\delta \approx \log(462)/\log(6{,}785) \approx 0.695 > 0.5$ suggests that
the pseudo-posterior variance may understate the true sampling variability,
motivating the sandwich correction below.

\paragraph{Sandwich variance.}
Define the observed Hessian and cluster-robust outer product:
\begin{align}
  \mathbf{H}_{\mathrm{obs}}
    &= -\sum_{i=1}^{N}\tilde{w}_i\,
       \frac{\partial^2\log f_{\HBB}(y_i \mid \btheta)}
            {\partial\btheta\,\partial\btheta\t}
       \bigg|_{\btheta = \hat{\btheta}},
    \label{eq:H-obs} \\[6pt]
  \mathbf{J}_{\mathrm{cluster}}
    &= \sum_{c=1}^{C}\mathbf{s}_c\,\mathbf{s}_c\t,
    \qquad
    \mathbf{s}_c
    = \sum_{i \in \text{PSU}_c}\tilde{w}_i\,
      \frac{\partial\log f_{\HBB}(y_i \mid \btheta)}{\partial\btheta}
      \bigg|_{\btheta = \hat{\btheta}},
    \label{eq:J-cluster}
\end{align}
where $C = 415$ is the number of primary sampling units and $\hat{\btheta}$ is
the posterior mean. The hurdle separability implies that
$\mathbf{H}_{\mathrm{obs}}$ is block-diagonal between the extensive and
intensive margin parameters; however, $\mathbf{J}_{\mathrm{cluster}}$ is
\emph{not} block-diagonal, because a single PSU contributes scores from both
margins simultaneously.

\begin{remark}[Simplified exposition]\label{rem:J-simplified}
  \Cref{eq:J-cluster} uses the unstratified form for expositional
  clarity.  The implemented estimator is the stratified cluster-robust
  form in~\cref{smc:prop-cluster-robust}, which centers PSU scores
  within the $H = 30$ NSECE strata and applies finite-population
  correction factors $C_h/(C_h - 1)$.
\end{remark}

The design-consistent sandwich variance is
\begin{equation}\label{eq:V-sand}
  \mathbf{V}_{\mathrm{sand}}
  = \mathbf{H}_{\mathrm{obs}}^{-1}\,
    \mathbf{J}_{\mathrm{cluster}}\,
    \mathbf{H}_{\mathrm{obs}}^{-1}.
\end{equation}

\begin{theorem}[Cholesky affine transformation]\label{thm:cholesky}
  Let $\hat{\btheta}$ be the posterior mean, $\mathbf{L}_{\mathrm{sand}}$ the
  lower Cholesky factor of $\mathbf{V}_{\mathrm{sand}}$, and
  $\mathbf{L}_{\mathrm{MCMC}}$ the lower Cholesky factor of the MCMC posterior
  covariance $\bSigma_{\mathrm{MCMC}}$. The transformed draws
  \begin{equation}\label{eq:cholesky-transform}
    \btheta^{*(m)}
    = \hat{\btheta}
    + \mathbf{L}_{\mathrm{sand}}\,\mathbf{L}_{\mathrm{MCMC}}^{-1}
      \bigl(\btheta^{(m)} - \hat{\btheta}\bigr),
    \qquad m = 1,\ldots,M,
  \end{equation}
  satisfy $\E[\btheta^*] = \hat{\btheta}$ (mean-preserving) and
  $\Cov(\btheta^*) = \mathbf{V}_{\mathrm{sand}}$
  (design-consistent covariance).  The transformation matches
  the target covariance structure; Wald confidence intervals
  based on the marginal variances are invariant to the choice
  of transformation matrix.
\end{theorem}

\paragraph{Computational remark.}
The observed Hessian $\mathbf{H}_{\mathrm{obs}}$ is computed as the negative
second derivative of the weighted pseudo-log-likelihood~\eqref{eq:H-obs},
excluding prior curvature, evaluated at the posterior mean,
\emph{not} from the MCMC sample covariance. This distinction is critical in
hierarchical models: the marginal posterior variance
$\bSigma_{\mathrm{MCMC}}$ of the fixed effects integrates over the
random-effect distribution, producing intervals much wider than the
conditional precision $\mathbf{H}_{\mathrm{obs}}^{-1}$. In our
application, the hierarchical variance ratio
$(\bSigma_{\mathrm{MCMC}})_{pp}/(\mathbf{H}_{\mathrm{obs}}^{-1})_{pp}$
exceeds 600 for most fixed effects---a consequence of the hierarchical
structure, not prior domination---confirming that the standard substitution
$\bSigma_{\mathrm{MCMC}} \approx \mathbf{H}^{-1}$ fails in hierarchical
models.

\paragraph{Design effect ratio and block-wise correction strategy.}
As a parameter-specific diagnostic extending the classical design
effect~\citep{RaoScott1981, PfeffermannEtAl1998} to individual model
parameters, define the design effect ratio for parameter $p$ as
\begin{equation}\label{eq:der}
  \DER_{p}
  = \frac{[\mathbf{V}_{\mathrm{sand}}]_{pp}}{[\mathbf{H}_{\mathrm{obs}}^{-1}]_{pp}}.
\end{equation}
The $\DER_p$ quantifies the factor by which the survey design inflates
the variance of the $p$-th parameter beyond the data-only (SRS) variance
and motivates the following three-tier correction strategy.
The Cholesky transformation~\eqref{eq:cholesky-transform} is applied
in blocks:
\begin{itemize}[nosep]
  \item \emph{Fixed effects} $(\balpha, \bbeta, \log\kappa)$: full sandwich
    correction. The DER ranges from 1.14 to 4.18 in our data
    (mean 2.11, consistent with the Kish $\DEFF$ of 3.76). The corresponding
    confidence interval inflation factors are
    $\sqrt{\DER} \in [1.07, 2.04]$.
  \item \emph{Random effects} $(\bdelta_{k,s})$: partial correction guided by
    the DER. States with low effective sample size show DER near one (the prior
    dominates), while large states show DER up to four.
  \item \emph{Hyperparameters} $(\bSigma_\delta, \bGamma_k)$: no correction
    applied. The sandwich estimator is not applicable to variance components
    estimated from between-state variation; the DER concept does not extend to
    these parameters. Consequently, posterior intervals for hyperparameters are
    model-based summaries whose validity relies on correct specification of the
    hierarchical distribution, not on design consistency. When a hyperparameter
    is substantively central---as $\bSigma_{12}$ is in this
    application (\cref{prop:sigma12-nonid,prop:sigma12-id})---prior
    sensitivity analysis is recommended.
\end{itemize}

\begin{remark}[Block-wise caveat]\label{rem:blockwise}
  The block-wise application of~\cref{thm:cholesky} ignores cross-block
  covariance between fixed and random effects, which can be non-negligible
  for the three to five largest states where within-state ESS approaches $N_s$.
  A sensitivity analysis comparing block-wise and joint correction is reported
  in Supplementary Material~C.  The finite-sample coverage of sandwich-corrected
  intervals is evaluated in the simulation study (\cref{sec:simulation}).
\end{remark}


\subsection{Prior specification and computation}\label{sec:priors}

\paragraph{Prior hierarchy.}
Each prior is chosen to be weakly informative---providing enough
regularization to ensure computational stability and proper posteriors, while
remaining vague enough that the data dominate in well-identified regions.
\cref{tab:priors} summarizes the complete specification, following the
recommendations of~\citet{GelmanEtAl2008} for regularized regression and
\citet{Gelman2006} for variance components.

\begin{table}[ht]
\centering
\caption{Prior specification. NCP = non-centered parameterization.}\label{tab:priors}
\smallskip
\footnotesize
\begin{adjustbox}{max width=\textwidth}
\begin{tabular}{@{}llp{0.50\textwidth}@{}}
\toprule
Parameter & Prior & Justification \\
\midrule
$\balpha,\;\bbeta$ & $\Norm(\mathbf{0},\;2^2\,\mathbf{I}_P)$ &
  Weakly informative on logit scale; \\
  & & $\pm 4$ logit units covers most of the probability range \\[3pt]
$\operatorname{vec}(\bGamma_k)$ & $\Norm(0,\;1^2)$ &
  Diffuse; allows moderate cross-level effects \\[3pt]
$\log\kappa$ & $\Norm(2,\;1.5^2)$ &
  Centered at moderate overdispersion ($\kappa \approx 7.4$); \\
  & & 95\% prior interval $\kappa \in [0.4, 140]$ \\[3pt]
$\boldsymbol{\tau}$ (SDs of $\bSigma_\delta$) &
  $\HalfNormal(0,1)$ &
  Regularizing; shrinkage-friendly \\[3pt]
$\mathbf{R}_\varepsilon$ (correlations) & $\LKJ(2)$ \citep{LewandowskiEtAl2009} &
  Mild shrinkage toward identity \\[3pt]
$\bz_s$ (NCP auxiliaries) & $\Norm(\mathbf{0},\;\mathbf{I}_{2q})$ &
  Standard normal (non-centered parameterization) \\
\bottomrule
\end{tabular}
\end{adjustbox}
\end{table}

\paragraph{Prior for $\log\kappa$.}
We center the log-concentration prior at $2$ (corresponding to $\kappa \approx 7.4$)
because preliminary analysis of IT enrollment counts suggests moderate
overdispersion: the variance-to-mean ratio across participating centers is
approximately 7--15 times the binomial baseline, consistent with $\kappa$ in the
range 5--20.
The scale of 1.5 is generous, accommodating $\kappa$ values from less than~1
(near-maximum overdispersion) to over~140 (approaching the binomial).

\paragraph{Identification.}
\begin{theorem}[Identification]\label{thm:identification}
  Suppose the following conditions hold:
  \begin{enumerate}[label=\textup{(C\arabic*)},nosep]
    \item The design matrices have full rank:
      $\mathrm{rank}(\mathbf{X}) = P$ and $\mathrm{rank}(\mathbf{V}) = Q$.
    \item Each state $s$ has at least $q$ providers with linearly independent
      covariate vectors $\bx_i^{(r)}$, with at least one zero and at least
      $q$ positive responders.
    \item $n_i \ge 3$ for all $i$ with $z_i = 1$, and there exist providers
      $i, j$ with $s[i] = s[j]$ and $n_i \ne n_j$ in each state.
    \item The prior $\pi(\btheta)$ has full support on the parameter space.
  \end{enumerate}
  Then the full parameter vector $\btheta = (\balpha, \bbeta, \kappa,
  \{\bdelta_s\}, \bSigma_\delta, \{\bGamma_k\})$ is identified under the
  posterior.
\end{theorem}

\noindent\emph{Interpretation.}
The hurdle likelihood identifies the state-specific \emph{total}
coefficients $\tilde{\alpha}_{j,s}$ and $\tilde{\beta}_{j,s}$
(\cref{eq:total-coef}) under standard rank and replication
conditions; the hierarchical layer then pins down their
decomposition into population-average effects~$\balpha^{(r)}$,
$\bbeta^{(r)}$, policy moderation~$\bGamma_k$, and residual state
deviations~$\bepsilon_{k,s}$.
See Supplementary Material~B for the complete proof.

Condition (C3) is critical for the intensive margin: with $n_i = 1$, the
intensity function $h \equiv 1$ and the parameters $\mu_i$ and $\kappa$
enter the likelihood only through $1 - p_0 = \mu_i$, so $\kappa$ is not
identified. With $n_i = 2$, $p_0$ depends on both $\mu$ and $\kappa$, but the
single constraint $p_0(2,\mu,\kappa) = (1-\mu)[(1-\mu)\kappa+1]/(\kappa+1)$ determines
only one degree of freedom. For $n_i \ge 3$, the product formula
\eqref{eq:p0} provides enough structure to separate $\mu$ from $\kappa$. The
full proof appears in Supplementary Material~B.

\paragraph{Estimation pipeline.}
\cref{alg:estimation} outlines the complete estimation and inference procedure.

\begin{algorithm}[ht]
\caption{Estimation pipeline for the HBB model}\label{alg:estimation}
\begin{algorithmic}[1]
  \State \textbf{Data preparation:} Normalize weights
    $\tilde{w}_i = w_i N / \textstyle\sum_j w_j$; standardize covariates.
  \State \textbf{Compile:} Stan model with NCP and \texttt{lbeta()} for numerical
    stability.
  \State \textbf{MCMC sampling:} Run 4 chains~\citep[NUTS;][]{Hoffman2014}
    with \texttt{adapt\_delta} $= 0.95$; see Supplementary Material~E for
    iteration counts and tree-depth settings.
  \State \textbf{Convergence diagnostics:} Verify $\hat{R} < 1.01$,
    $\ESS_{\mathrm{bulk}} > 400$~\citep{Vehtari2021}, and zero divergent transitions for all
    parameters.
  \State \textbf{Score extraction:} Extract observation-level scores
    $\partial\log f_{\HBB} / \partial\btheta$ computed in Stan's
    \texttt{generated quantities} block, evaluated at each MCMC draw.
  \State \textbf{Sandwich computation:} Form $\mathbf{H}_{\mathrm{obs}}$
    and $\mathbf{J}_{\mathrm{cluster}}$ via~\cref{eq:H-obs,eq:J-cluster};
    compute $\mathbf{V}_{\mathrm{sand}}$ via~\cref{eq:V-sand}.
  \State \textbf{Cholesky transformation:} Apply
    \cref{eq:cholesky-transform} block-wise to obtain design-corrected
    posterior draws $\{\btheta^{*(m)}\}$.
  \State \textbf{Model comparison:} Compute LOO-CV
    \citep{VehtariEtAl2017} across M0--M3b; report
    $\widehat{\elpd}_{\mathrm{loo}}$ and Pareto-$k$ diagnostics.
\end{algorithmic}
\end{algorithm}

The models are fitted sequentially from M0 through M3b, using the posterior
mean of each simpler model as the starting point for its successor.
This warm-start strategy reduces total computation time by approximately
40\% compared with cold starts.  Full MCMC settings and convergence diagnostics
are reported in Supplementary Material~E; score computation details and
closed-form expressions for the sandwich ingredients are documented in the
\textsf{hurdlebb} package vignette~\citep{hurdlebb2026}.

The finite-sample properties of this estimation pipeline---particularly
the sandwich correction's behavior under realistic survey designs---are
evaluated in the simulation study that follows.

%% file: section4.tex

\section{Simulation Study}\label{sec:simulation}

The sandwich-corrected pseudo-posterior developed in~\cref{sec:survey}
rests on asymptotic arguments whose finite-sample behavior must be
characterized before the methodology can be recommended for practice.
This section maps the \emph{operating envelope} of the three inferential
strategies---unweighted, weighted-naive, and sandwich-corrected---through
a controlled simulation that varies design informativeness across conditions
calibrated to realistic survey parameters. We treat the exercise as an
\emph{honest diagnostic} rather than a pure validation: the results identify
where each strategy succeeds, where it fails, and why, exposing a
fundamental structural limitation of sandwich correction for hierarchical
variance components. By establishing these operating properties here,
we provide the evidential basis for the reporting strategy applied in
\cref{sec:application}.


\subsection{Design}\label{sec:simdesign}

\paragraph{Data-generating process.}
We construct a finite super-population of $M = 50{,}000$ providers
distributed across $S = 51$ states, with state sizes proportional to
observed survey counts. For each provider~$i$ in state~$s$, a
covariate vector $\bx_i = (1, x_{\mathrm{pov}}, x_{\mathrm{urb}},
x_{\mathrm{blk}}, x_{\mathrm{his}})^\t$ is drawn from empirical
marginal distributions. True parameters are calibrated to the M1
(random intercept) unweighted fit:
\begin{align}\label{eq:sim-truth}
  \balpha_0 &= (0.696,\; -0.119,\; 0.253,\; -0.070,\; -0.139)^\t, \notag\\
  \bbeta_0  &= (-0.032,\; 0.057,\; -0.018,\; 0.080,\; 0.040)^\t, \quad
  \log\kappa_0 = 1.655,
\end{align}
with state random-effect standard deviations $\tau_{\mathrm{ext}} = 0.577$,
$\tau_{\mathrm{int}} = 0.208$, and cross-margin correlation
$\rho_{\mathrm{cross}} = 0.285$. State random effects
$\bdelta_s = (\delta_{1s}, \delta_{2s})^\t \sim \Norm_2(\mathbf{0},
\bSigma_\delta)$ are drawn once per replication.
Participation indicators are generated as
$z_i \sim \Bern(q_i)$ with $\logit(q_i) = \bx_i^\t \balpha_0 +
\delta_{1,s[i]}$, and positive counts as
$Y_i \mid z_i = 1 \sim \ZTBB(n_i, \mu_i, \kappa_0)$
with $\logit(\mu_i) = \bx_i^\t \bbeta_0 + \delta_{2,s[i]}$,
where $n_i$ is drawn from the empirical distribution of total enrollment.
The zero-truncated beta-binomial is sampled via rejection from the
unconditional $\BetaBin(n_i, \mu_i, \kappa_0)$.

\paragraph{Sampling design and scenarios.}
From each finite population we draw a stratified cluster sample of
approximately $N \approx 7{,}000$ providers. Survey informativeness is
controlled through a size-biased inclusion probability:
\begin{equation}\label{eq:sim-inclusion}
  \log \pi_i = c_0 + \rho_{\mathrm{inc}} \cdot y_i^{*},
\end{equation}
where $y_i^{*}$ is the standardized latent outcome and $c_0$ is chosen to
yield the target sample size. The informativeness parameter
$\rho_{\mathrm{inc}}$ indexes three scenarios (\cref{tab:sim-scenarios}).
Scenario S3 is calibrated to match the NSECE design
($\DEFF_{\mathrm{Kish}} = 3.76$; \cref{tab:data-summary}). Scenario S4
doubles the informativeness to stress-test the sandwich correction under
conditions more extreme than those present in the actual survey.

\begin{table}[t]
\centering
\caption{Simulation scenarios. The informativeness parameter
  $\rho_{\mathrm{inc}}$ controls the dependence between inclusion
  probabilities and the outcome (\cref{eq:sim-inclusion});
  $\mathrm{CV}(w)$ is the coefficient of variation of the resulting
  sampling weights; Kish $\DEFF = 1 + \mathrm{CV}^2(w)$.}\label{tab:sim-scenarios}
\smallskip
\begin{tabular}{@{}llccl@{}}
\toprule
Scenario & $\rho_{\mathrm{inc}}$ & $\mathrm{CV}(w)$ &
  Kish $\DEFF$ & Description \\
\midrule
S0 & 0.00 & $\approx 1.0$ & $\approx 2.0$ & Non-informative baseline \\
S3 & 0.15 & $\approx 1.67$ & $\approx 3.8$ & NSECE-calibrated \\
S4 & 0.50 & $\approx 2.0$  & $\approx 5.0$ & Stress test \\
\bottomrule
\end{tabular}
\end{table}

\paragraph{Estimators.}
Each replication is analyzed under three estimators (\cref{tab:sim-estimators})
corresponding to the three inferential strategies in
\cref{sec:survey}. E-UW fits the model with an unweighted likelihood and
uses the naive MCMC posterior for interval construction. E-WT replaces
the likelihood with the pseudo-log-likelihood
(\cref{eq:pseudo-loglik}; \citealp{SavitskyToth2016}),
correcting point estimates for informativeness but retaining the
naive posterior variance. E-WS augments E-WT with the
sandwich-corrected variance via the Cholesky affine transformation
(\cref{thm:cholesky}; \citealp{WilliamsSavitsky2021}), yielding
Wald-type confidence intervals.
The Cholesky transformation is applied block-wise following the $\DER$
classification in~\cref{sec:survey}: full correction for fixed effects,
no correction for hyperparameters.

\begin{table}[t]
\centering
\caption{Estimator definitions.}\label{tab:sim-estimators}
\smallskip
\small
\begin{tabular}{@{}lp{2.5cm}p{3.5cm}l@{}}
\toprule
Estimator & Likelihood & Standard errors & Purpose \\
\midrule
E-UW & Unweighted         & Model-based (naive) & Ignore design \\
E-WT & Pseudo (weighted)  & Model-based (naive) & Correct point est. \\
E-WS & Pseudo (weighted)  & Sandwich~\eqref{eq:V-sand} & Full correction \\
\bottomrule
\end{tabular}
\end{table}

\paragraph{Target parameters and computation.}
We track five parameters that span fixed effects and variance components:
the poverty slopes $\alpha_{\mathrm{pov}}$ and $\beta_{\mathrm{pov}}$,
the overdispersion $\log\kappa$, and the state-effect standard deviations
$\tau_{\mathrm{ext}}$ and $\tau_{\mathrm{int}}$. Performance is measured
by 90\% interval coverage, relative bias
$\text{RB} = 100 \times (\bar{\hat\theta} - \theta_0)/|\theta_0|$ (\%),
root mean squared error (RMSE), and the width ratio
\[
  \text{WR}
  = \frac{\text{median CI width (E-WS)}}
         {\text{median CI width (E-WT)}},
\]
which isolates the inflation attributable to sandwich correction.
Each scenario--estimator combination uses $R = 200$ replications.
The Monte Carlo standard error (MCSE) for each coverage estimate is
$\mathrm{MCSE} = \sqrt{\hat{p}(1 - \hat{p})/R}$, ranging from 0.0~pp
(at boundary coverage) to a maximum of 3.5~pp (near 50\% coverage);
at the 90\% nominal level, $\mathrm{MCSE} \approx 2.1$~pp.
Per-cell MCSEs are reported in parentheses in \cref{tab:sim-results};
coverage deviating from 90\% by more than $2 \times \mathrm{MCSE}$ is
flagged with~$\dagger$.
Each replication fits the M1 specification
(\cref{tab:model-hierarchy}) with 4 chains of 3{,}500 iterations
(1{,}500 warmup, 2{,}000 sampling), for a total of 8{,}000 post-warmup draws.
All replications achieved $\hat{R} < 1.01$ \citep{Vehtari2021} and a minimum bulk effective
sample size exceeding 1{,}600.


\subsection{Results}\label{sec:simresults}

\cref{tab:sim-results} reports coverage, relative bias, RMSE, and width
ratio for all five parameters across three scenarios and three
estimators. \cref{fig:sim-coverage} displays the coverage rates
graphically with MCSE bands. We organize the results around five
findings that together characterize the operating properties of the
estimation framework.

\paragraph{Finding 1: The unweighted estimator achieves near-nominal
coverage for fixed effects under correct specification.}
Under E-UW, coverage of $\alpha_{\mathrm{pov}}$ and
$\beta_{\mathrm{pov}}$ lies within the 89--93\% range across all three
scenarios, comfortably within two MCSE of the 90\% nominal level. This
holds even under S4, where the sampling design is strongly informative.
The result confirms that when the generative model is \emph{correctly specified}
and inference is based on the model likelihood alone, ignoring the survey
weights does not compromise interval coverage.
This finding should \emph{not} be read as an argument against survey
weighting: in practice the generative model is never exactly correct,
and Findings~2--3 demonstrate that weighting with sandwich correction
provides important robustness against model misspecification.  A
dedicated misspecification experiment (\cref{smf:misspec}) confirms
this: under an M2 DGP with omitted state-varying poverty slopes,
E-UW coverage collapses to 0.5--30\% for the poverty coefficients,
while E-WS retains 67.5--79.0\% coverage.
The unweighted estimator
also maintains nominal coverage for $\log\kappa$ (90.5--92.0\%) and
for the hyperparameters $\tau_{\mathrm{ext}}$ and $\tau_{\mathrm{int}}$
(91.0--98.0\%), with the single exception of $\tau_{\mathrm{ext}}$ under
S4 (77.0\%), where the extreme weight structure induces composition
effects in the sampled random effects.

\paragraph{Finding 2: Naive weighted inference produces severe
undercoverage.}
The weighted estimator E-WT shows severe and systematically worsening
undercoverage as design complexity increases. For
$\alpha_{\mathrm{pov}}$, coverage drops from 81.0\% under S0 to 58.5\%
under S3 to 53.0\% under S4---a collapse from modest undercoverage to
barely half of the draws containing the true value. The pattern for
$\beta_{\mathrm{pov}}$ is similar (76.5\%, 62.5\%, 56.5\%). The
mechanism is well understood
\citep{SavitskyToth2016}: the pseudo-likelihood inflates the
effective sample size, concentrating the posterior around point
estimates that carry weight-induced bias, producing intervals that are
simultaneously too narrow and miscentered.

\paragraph{Finding 3: Sandwich correction provides meaningful but
imperfect recovery.}
The sandwich-corrected E-WS recovers substantially from the E-WT
collapse, achieving 82.0--88.5\% coverage for the two poverty
coefficients---gains of roughly 24--33 percentage points over E-WT in
S3 and S4. The width ratio tracks the design effect monotonically:
WR increases from 1.22 under S0 to 1.64--2.06 under S4
(\cref{tab:sim-results}), consistent with
$\mathrm{WR} \approx \sqrt{\DER}$ and the empirical $\DER$ range
of 1.14--4.18 in~\cref{tab:fixed-effects}. The residual coverage gap
reflects the finite-sample bias of the sandwich variance
estimator---a known property of cluster-robust estimators when the
number of clusters is moderate~\citep{Binder1983}.

To characterize the gap more precisely, we decompose it into
\emph{width} and \emph{bias} components via the SE calibration ratio
$\widehat{\mathrm{SD}} / \overline{\mathrm{SE}}$, where
$\widehat{\mathrm{SD}}$ is the actual standard deviation of point
estimates across replications and $\overline{\mathrm{SE}}$ is the mean
reported sandwich standard error
(\cref{smf:tab-decomposition}).  For the poverty coefficients under
S3, the SE ratio is 1.13--1.14, confirming that the sandwich SE
underestimates the true sampling variability by approximately 13\%,
while the standardized mean bias $|\bar{b}|/\widehat{\mathrm{SD}}$ is
negligible (0.06--0.22).  The gap is therefore \emph{width-driven}:
intervals are correctly centered but 13\% too narrow.  This 13\%
underestimate is consistent with the $O(S^{-1})$ finite-sample bias of
cluster-robust variance estimators with $S = 51$
clusters~\citep{Binder1983}, and suggests that small-sample corrections
(e.g., CR2; \citealp{PustejovskyTipton2018}) could further
close the gap---an avenue for future work.  Widening to
95\% nominal intervals confirms this interpretation: E-WS coverage
rises to 87.5--88.5\% for the poverty coefficients under S3---a gain
consistent with the expected $\approx 5$~pp from Normal tail
probability---leaving a proportional residual gap
(\cref{smf:tab-coverage95}).

\paragraph{Finding 4: Sandwich correction is less effective for the
overdispersion parameter.}
For $\log\kappa$, E-WS achieves 57.0--76.5\% coverage---substantially
below the 82--88.5\% observed for the poverty coefficients. This
discrepancy arises because $\kappa$ enters the variance function of the
beta-binomial nonlinearly: the relationship between the weighted
score variance and the sandwich-corrected parameter variance is less
well approximated by the linear Taylor expansion underlying the sandwich
formula. The width ratio for $\log\kappa$ is comparable to that for the
linear coefficients (1.31--2.04), confirming that the sandwich
correction does widen the intervals appropriately; the shortfall in
coverage is driven primarily by point estimate bias rather than variance
underestimation.

\paragraph{Finding 5: Hyperparameters cannot be sandwich-corrected.}
A notable finding concerns $\tau_{\mathrm{ext}}$ and
$\tau_{\mathrm{int}}$. Under E-WT and E-WS, coverage drops to
0.0--8.0\% across all scenarios, with relative bias of
$+68$ to $+110$\%. The width ratio is exactly $\mathrm{WR} = 1.00$,
confirming that E-WS$\;\equiv\;$E-WT for these parameters by
construction: the sandwich correction is applied only to parameters
that enter the weighted log-likelihood through scores, and the
hierarchical variance components do not produce individual-level score
contributions. The $\DER$ is undefined for $\tau$
(\cref{eq:der}), so no variance inflation can be computed. The
substantial positive bias arises because the pseudo-likelihood treats
each observation as carrying $\tilde{w}_i$ units of information,
inflating the apparent between-state heterogeneity by a factor
proportional to $\E[\tilde{w}_i^2]$. The unweighted estimator
E-UW maintains coverage of 91--98\% for both
$\tau$ parameters in S0 and S3; in S4,
$\tau_{\mathrm{ext}}$ drops to 77.0\%---the only E-UW
failure in the entire table---reflecting the growing tension between
the informative sampling weights (which are ignored) and the
extensive-margin random-effect variance.

\paragraph{Cross-margin correlation.}
We additionally track recovery of the cross-margin correlation
$\rho_{\mathrm{cross}} = \operatorname{cor}(\delta_j^{\mathrm{ext}},
\delta_j^{\mathrm{int}})$, estimated via the plug-in
$\hat{\rho} = \operatorname{cor}(\bar{\delta}_j^{\mathrm{ext}},
\bar{\delta}_j^{\mathrm{int}})$ from $S = 51$ posterior mean random
effects (\cref{tab:rho-cross-sim}).
Under the NSECE-calibrated design~(S3), the unweighted estimator E-UW
overestimates the true $\rho_{\mathrm{cross}} = 0.285$ with a relative
bias of $+104\%$ (RMSE $= 0.364$), while the weighted estimator E-WT
exhibits moderate downward bias of $-42\%$ (RMSE $= 0.226$).
As with $\tau_{\mathrm{ext}}$ and $\tau_{\mathrm{int}}$,
E-WS $\equiv$ E-WT for this hyperparameter: the sandwich correction
adjusts only fixed-effect covariances and cannot modify the
random-effect correlation structure.
Because full posterior draws of $\Omega$ were not retained in the
simulation archive, per-replication credible intervals are unavailable;
\cref{tab:rho-cross-sim} reports frequentist summaries across
replications instead.

\begin{table}[t]
\centering
\caption{Simulation results: empirical coverage of 90\% intervals
  (\%, with Monte Carlo standard error in parentheses),
  relative bias (RB, \%), root mean squared error (RMSE), and width ratio
  (WR) across $R = 200$ replications per scenario.
  E-UW = unweighted; E-WT = weighted, naive posterior variance;
  E-WS = weighted, sandwich-corrected (\cref{eq:V-sand}).
  $\mathrm{WR} = \text{median CI width (E-WS)}/\text{median CI
  width (E-WT)}$; $\mathrm{WR} = 1.00$ for $\tau$ parameters indicates that
  E-WS $\equiv$ E-WT by construction.
  $\dagger$ = coverage deviates from 90\% by more than 2 MCSE.
  MCSE $= \sqrt{\hat{p}(1-\hat{p})/R}$; range 0.0--3.5
  pp.}\label{tab:sim-results}
\smallskip
\small
\begin{adjustbox}{max width=\textwidth}
\begin{tabular}{@{}l ccc rrr rrr r@{}}
\toprule
 & \multicolumn{3}{c}{Coverage (\%, MCSE)} &
   \multicolumn{3}{c}{Relative bias (\%)} &
   \multicolumn{3}{c}{RMSE} & \\
\cmidrule(lr){2-4}\cmidrule(lr){5-7}\cmidrule(lr){8-10}
Parameter & E-UW & E-WT & E-WS & E-UW & E-WT & E-WS & E-UW & E-WT & E-WS & WR \\
\midrule
\multicolumn{11}{@{}l}{\textit{Panel A: S0 --- Non-informative ($\DEFF \approx 2.0$)}} \\[2pt]
$\alpha_{\text{pov}}$ & 93.0\,(1.8) & 81.0\,(2.8)$^{\dagger}$ & 88.5\,(2.3)
  & $+$13.2 & $+$4.9 & $+$4.9 & 0.033 & 0.044 & 0.044 & 1.22 \\
$\beta_{\text{pov}}$ & 93.0\,(1.8) & 76.5\,(3.0)$^{\dagger}$ & 86.5\,(2.4)
  & $-$9.0 & $-$4.8 & $-$4.8 & 0.016 & 0.022 & 0.022 & 1.22 \\
$\log\kappa$ & 91.5\,(2.0) & 64.0\,(3.4)$^{\dagger}$ & 76.5\,(3.0)$^{\dagger}$
  & $-$0.8 & $+$1.6 & $+$1.6 & 0.023 & 0.040 & 0.040 & 1.31 \\
$\tau_{\text{ext}}$ & 97.5\,(1.1)$^{\dagger}$ & 3.0\,(1.2)$^{\dagger}$ & 3.0\,(1.2)$^{\dagger}$
  & $+$9.6 & $+$68.8 & $+$68.8 & 0.082 & 0.423 & 0.423 & 1.00 \\
$\tau_{\text{int}}$ & 91.0\,(2.0) & 8.0\,(1.9)$^{\dagger}$ & 8.0\,(1.9)$^{\dagger}$
  & $-$6.6 & $+$70.6 & $+$70.6 & 0.036 & 0.158 & 0.158 & 1.00 \\[4pt]
\multicolumn{11}{@{}l}{\textit{Panel B: S3 --- NSECE-calibrated ($\DEFF \approx 3.8$)}} \\[2pt]
$\alpha_{\text{pov}}$ & 89.5\,(2.2) & 58.5\,(3.5)$^{\dagger}$ & 82.0\,(2.7)$^{\dagger}$
  & $+$16.5 & $+$12.9 & $+$12.9 & 0.037 & 0.072 & 0.072 & 1.71 \\
$\beta_{\text{pov}}$ & 89.0\,(2.2) & 62.5\,(3.4)$^{\dagger}$ & 82.5\,(2.7)$^{\dagger}$
  & $-$16.7 & $-$3.2 & $-$3.2 & 0.017 & 0.032 & 0.032 & 1.64 \\
$\log\kappa$ & 92.0\,(1.9) & 34.0\,(3.3)$^{\dagger}$ & 60.5\,(3.5)$^{\dagger}$
  & $-$0.5 & $+$3.3 & $+$3.3 & 0.020 & 0.068 & 0.068 & 1.75 \\
$\tau_{\text{ext}}$ & 98.0\,(1.0)$^{\dagger}$ & 0.5\,(0.5)$^{\dagger}$ & 0.5\,(0.5)$^{\dagger}$
  & $+$8.0 & $+$105.6 & $+$105.6 & 0.091 & 0.639 & 0.639 & 1.00 \\
$\tau_{\text{int}}$ & 97.5\,(1.1)$^{\dagger}$ & 0.0\,(0.0)$^{\dagger}$ & 0.0\,(0.0)$^{\dagger}$
  & $-$8.6 & $+$97.6 & $+$97.6 & 0.040 & 0.211 & 0.211 & 1.00 \\[4pt]
\multicolumn{11}{@{}l}{\textit{Panel C: S4 --- Stress test ($\DEFF \approx 5.0$)}} \\[2pt]
$\alpha_{\text{pov}}$ & 92.5\,(1.9) & 53.0\,(3.5)$^{\dagger}$ & 86.0\,(2.5)
  & $+$20.5 & $+$11.9 & $+$11.9 & 0.036 & 0.083 & 0.083 & 2.06 \\
$\beta_{\text{pov}}$ & 91.0\,(2.0) & 56.5\,(3.5)$^{\dagger}$ & 84.0\,(2.6)$^{\dagger}$
  & $-$21.7 & $-$2.2 & $-$2.2 & 0.017 & 0.035 & 0.035 & 1.91 \\
$\log\kappa$ & 90.5\,(2.1) & 23.0\,(3.0)$^{\dagger}$ & 57.0\,(3.5)$^{\dagger}$
  & $+$0.8 & $+$4.1 & $+$4.1 & 0.020 & 0.083 & 0.083 & 2.04 \\
$\tau_{\text{ext}}$ & 77.0\,(3.0)$^{\dagger}$ & 0.0\,(0.0)$^{\dagger}$ & 0.0\,(0.0)$^{\dagger}$
  & $+$26.5 & $+$110.2 & $+$110.2 & 0.175 & 0.671 & 0.671 & 1.00 \\
$\tau_{\text{int}}$ & 97.0\,(1.2)$^{\dagger}$ & 0.0\,(0.0)$^{\dagger}$ & 0.0\,(0.0)$^{\dagger}$
  & $-$0.7 & $+$103.9 & $+$103.9 & 0.043 & 0.225 & 0.225 & 1.00 \\
\bottomrule
\end{tabular}
\end{adjustbox}
\end{table}

\begin{figure}[t]
  \centering
  \includegraphics[width=\textwidth]{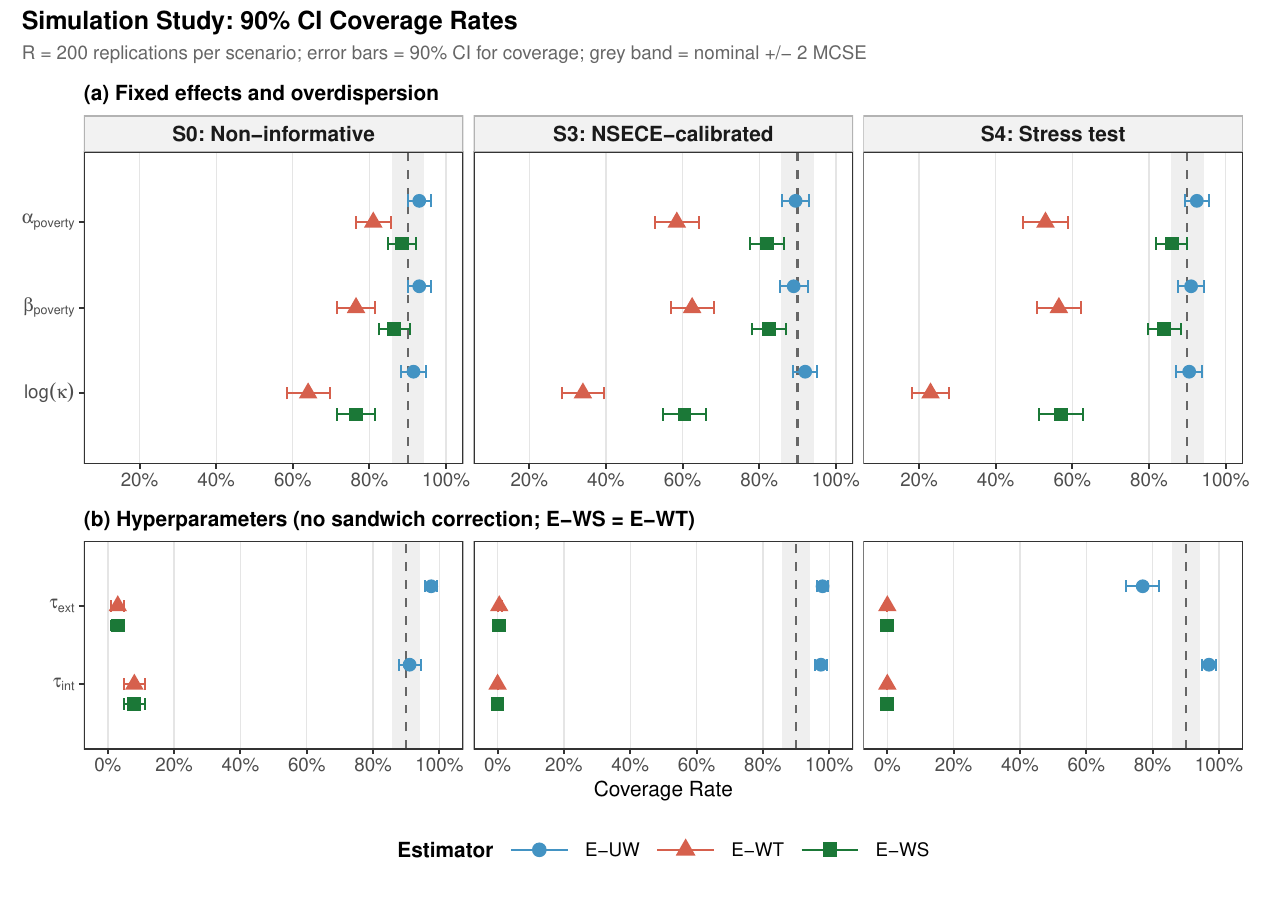}
  \caption{Empirical coverage of 90\% intervals across $R = 200$
    replications. Each panel corresponds to a simulation scenario
    (S0, S3, S4 with increasing Kish $\DEFF$). Points show coverage
    rates for each parameter--estimator combination; the horizontal
    dashed line marks the 90\% nominal level; the shaded band indicates
    $\pm\,2$ MCSE ($\approx 4.2$ percentage points). E-UW (unweighted)
    maintains near-nominal coverage throughout. E-WT (weighted, naive)
    collapses as $\DEFF$ increases. E-WS (sandwich-corrected) provides
    meaningful recovery for fixed effects but cannot correct the
    hyperparameters $\tau_{\mathrm{ext}}$ and
    $\tau_{\mathrm{int}}$.}\label{fig:sim-coverage}
\end{figure}

\paragraph{Practical recommendations.}
The simulation results support a three-part reporting strategy for
survey-weighted Bayesian inference with sandwich-corrected
pseudo-posteriors. First, for population-average fixed effects, report
sandwich-corrected Wald confidence intervals as the primary inferential
tool, acknowledging that finite-sample coverage may fall 2--8 percentage
points below nominal when the Kish $\DEFF$ exceeds 3. Second, for
hierarchical variance components, report unweighted posterior intervals,
since the sandwich correction is structurally inapplicable and the
weighted pseudo-posterior produces systematically biased point estimates.
Third, always compute and report the $\DER$ (\cref{eq:der}) for each
parameter of interest: this scalar summary communicates the magnitude of
the survey design effect in units directly comparable to the classical
$\DEFF$ and allows readers to gauge how much the sandwich correction has
widened each interval. The empirical $\DER$ values
in~\cref{tab:fixed-effects} (range 1.14--4.18) and the simulation width
ratios in \cref{tab:sim-results} (range 1.22--2.06) are mutually
consistent and provide complementary evidence on the design effect
magnitude.

These recommendations apply broadly to any hierarchical
Bayesian model estimated with survey-weighted pseudo-likelihoods, not
only to the hurdle beta-binomial specification: the structural
inapplicability of sandwich correction to variance components is a
consequence of the score-based construction of $\mathbf{V}_{\mathrm{sand}}$
(\cref{eq:V-sand}) and holds for any hierarchical model in which
hyperparameters lack individual-level score contributions. The
two-track reporting strategy adopted in~\cref{sec:application}---Wald
(sandwich) intervals for fixed effects and model-based posterior intervals for
hyperparameters---implements these recommendations directly; the latter track
provides valid inference only to the extent that the hierarchical model is
correctly specified.
Supplementary figures showing relative bias distributions, width ratio
boxplots, and $\DER$ summaries across all scenarios and parameters
appear in Supplementary Material~E, which also provides a frequentist
contextualization (\cref{smf:frequentist}) demonstrating that the
sandwich-corrected pseudo-posterior fixed effects are design-consistent
with conventional \texttt{svyglm} estimates while accommodating
hierarchical and distributional features that the frequentist approach
cannot.

With the operating envelope established, we now implement the two-track
reporting strategy in the NSECE application: sandwich-corrected Wald
intervals for population-average fixed effects and unweighted posterior
intervals for state-level random effects and variance components.

%% file: section5.tex

\section{Application to NSECE Childcare Data}\label{sec:application}

We now apply the full analytical pipeline developed in \cref{sec:model} to
the 2019 NSECE data described in \cref{sec:data}, demonstrating each
methodological component in a realistic survey setting. The simulation
evidence in \cref{sec:simulation} established that the sandwich correction
substantially improves frequentist coverage and that the DER diagnostic reliably
identifies parameters affected by informative sampling; the present
application illustrates these tools in practice. The analysis proceeds in
five parts: model comparison via LOO-CV (\cref{sec:modelcomp}),
population-average inference with the two-track reporting strategy and DER
diagnostic (\cref{sec:popavg}), state-specific heterogeneity and
cross-margin independence (\cref{sec:statevar}), marginal effect
decomposition (\cref{sec:ame}), and a brief examination of policy
moderation (\cref{sec:policyresults}). Together, these subsections showcase
the three methodological contributions: the two-track reporting strategy
that resolves prior domination, the DER diagnostic that quantifies
parameter-specific design effects, and the marginal effect decomposition
that translates hurdle model coefficients into policy-relevant quantities.


\subsection{Model comparison and selection}\label{sec:modelcomp}

We fit the five nested models summarized in~\cref{tab:model-hierarchy}---from
the pooled M0 to the full M3b with policy moderators---using the estimation
pipeline described in~\cref{alg:estimation}, implemented in Stan
\citep{CarpenterEtAl2017}. Each model is estimated with 4
chains of 3,500 iterations (1,500 warmup + 2,000 sampling), yielding 8,000
post-warmup draws per parameter.

\paragraph{Convergence.}
All five models satisfy standard diagnostics~\citep{Vehtari2021}:
$\hat{R} < 1.01$ for every monitored parameter, minimum bulk effective
sample size exceeds 400, and zero divergent transitions were recorded.
\cref{tab:convergence-summary} summarizes these statistics; the
non-centered parameterization~\citep{BetancourtGirolami2015} successfully
eliminates the funnel geometry. Full parameter-by-parameter diagnostics
appear in Supplementary Material~E.

\begin{table}[t]
\centering
\caption{Convergence diagnostics across all models. $\hat{R}_{\max}$ is the
  maximum potential scale reduction factor across all monitored parameters.
  $\ESS_{\min}$ reports the minimum bulk and tail effective sample sizes.
  All models achieve the recommended thresholds ($\hat{R} < 1.01$,
  $\ESS > 400$, zero divergences).}\label{tab:convergence-summary}
\smallskip
\begin{tabular}{@{}lccccc@{}}
\toprule
Model & $\hat{R}_{\max}$ & $\ESS_{\min}^{\mathrm{bulk}}$
  & $\ESS_{\min}^{\mathrm{tail}}$ & Divergences \\
\midrule
M0: Pooled         & 1.001 & 3{,}390 & 2{,}679 & 0 \\
M1: Random int.\   & 1.004 &    815  & 1{,}415 & 0 \\
M2: Block-diag.\ SVC & 1.004 & 1{,}078 & 1{,}244 & 0 \\
M3a: Cross-margin  & 1.004 & 1{,}385 & 1{,}257 & 0 \\
M3b: Policy mod.\  & 1.004 & 1{,}245 & 1{,}463 & 0 \\
\bottomrule
\end{tabular}
\end{table}

\paragraph{LOO-CV comparison.}
We compare models using Pareto-smoothed importance-sampling leave-one-out
cross-validation~\citep[\LOO-CV;][]{VehtariEtAl2017}, where ``leave-one-out''
refers to individual providers, not entire states.
The pointwise log-likelihood for each provider~$i$ includes both hurdle
components: the extensive-margin (Bernoulli) log-likelihood always
contributes, and the intensive-margin (zero-truncated beta-binomial)
log-likelihood contributes only for participants ($z_i = 1$).
\cref{tab:loo} presents the expected log pointwise predictive density
($\widehat{\elpd}_{\LOO}$) for each model.
For the survey-weighted specification M3b-W, the LOO computation uses the
unweighted nominal log-likelihood, because the pseudo-log-likelihood is not
a proper predictive density; the comparisons thus rank models on their
nominal predictive ability.
Because LOO-CV operates at the provider level while many inferential
targets are state-level parameters, the LOO ranking should be interpreted
as a guide to observation-level fit rather than a direct assessment of
state-level inference quality.
The progression reveals
diminishing returns: adding state random intercepts (M0$\to$M1) yields a
decisive improvement ($\Delta\widehat{\elpd} = +317.8$, SE $= 25.9$,
exceeding 12 standard errors), and allowing state-varying coefficients
(M1$\to$M2) provides a further significant gain ($+23.6$ incrementally).
Introducing the cross-margin covariance (M2$\to$M3a) does not improve
predictive performance ($\Delta\widehat{\elpd} < 2$), consistent with the
near-zero empirical cross-margin correlation and with
\cref{prop:sigma12-nonid}: since $\bSigma_{12}$ receives no likelihood
information, its value lies in inference rather than prediction. Adding
policy moderators (M3a$\to$M3b) likewise produces no predictive gain; the
$\bGamma$ matrices redistribute the sources of state heterogeneity without
changing the provider-level predictive distribution.

\begin{table}[t]
\centering
\caption{Leave-one-out cross-validation comparison for models M0--M3b.
  $\widehat{\elpd}_{\LOO}$ is the expected log pointwise predictive density
  estimated via Pareto-smoothed importance sampling
  \citep{VehtariEtAl2017}; $\Delta\widehat{\elpd} = \widehat{\elpd}_{\mathrm{M0}} - \widehat{\elpd}_{\mathrm{M}_j}$ measures the
  predictive gain over M0 (more negative indicates better prediction); SE($\Delta$) is the standard
  error of the difference. A model improvement is flagged as significant
  when $|\Delta\widehat{\elpd}| > 2 \times \mathrm{SE}(\Delta)$.
  All Pareto-$k$ diagnostics are below 0.7.}\label{tab:loo}
\smallskip
\begin{tabular}{@{}lrcrrrc@{}}
\toprule
Model & Params & $\widehat{\elpd}_{\LOO}$ & SE &
  $\Delta\widehat{\elpd}$ & SE($\Delta$) & Sig.\ \\
\midrule
M0: Pooled          &  11 & $-20{,}945.0$ & 140.3 &
  ---   & ---  & ---    \\
M1: Random int.\    & 116 & $-20{,}627.2$ & 144.0 &
  $-317.8$ & 25.9 & Yes \\
M2: Block-diag.\ SVC & 551 & $-20{,}603.5$ & 144.4 &
  $-341.4$ & 27.6 & Yes \\
M3a: Cross-margin   & 576 & $-20{,}601.6$ & 144.4 &
  $-343.4$ & 27.6 & No  \\
M3b: Policy mod.\   & 616 & $-20{,}607.6$ & 144.6 &
  $-337.3$ & 28.5 & No  \\
\bottomrule
\end{tabular}
\end{table}

\paragraph{Model selection rationale.}
We adopt M3b as the reporting model despite its equivalent predictive
performance to M2 and M3a. The cross-margin covariance $\bSigma_{12}$
enables joint posterior statements about the reversal (e.g.,
$\Pr(\tilde{\alpha}_{\mathrm{pov},s} < 0$ and
$\tilde{\beta}_{\mathrm{pov},s} > 0 \mid \mathrm{data})$), and the policy
moderators provide the decomposition of state heterogeneity into
policy-explained and residual components (\cref{prop:poverty-expansion}).
The non-significant LOO-CV differences confirm that this additional
structure incurs no predictive cost. Posterior predictive checks confirm
that M3b reproduces both the observed zero rate (35.3\%) and the
distribution of IT shares among servers; details appear in Supplementary
Material~D.


\subsection{Population-average effects and the sandwich correction}\label{sec:popavg}

\cref{tab:fixed-effects} reports the posterior means, naive 95\% credible
intervals (from the raw MCMC output), sandwich-corrected 95\% Wald confidence
intervals, and design effect ratios (DER) for all 11 fixed-effect parameters.
Of the 11 parameters, 9 are statistically significant at the 95\% level
under the sandwich correction; the two exceptions are
$\alpha_{\mathrm{Hisp}}$ (extensive-margin Hispanic coefficient, Wald CI
$[-0.095, 0.106]$) and $\beta_{\mathrm{Black}}$ (intensive-margin Black
coefficient, Wald CI $[-0.052, 0.019]$).

\begin{table}[t]
\centering
\caption{Population-average fixed effect estimates from weighted M3b.
  ``Post.\ Mean'' is the posterior mean from the pseudo-posterior. ``Naive
  95\% CI'' is the 2.5th--97.5th percentile interval from the uncorrected
  MCMC posterior. ``Wald 95\% CI'' uses the sandwich-corrected variance
  (\cref{thm:cholesky}). ``DER'' is the design effect ratio
  (\cref{eq:der}). Asterisks indicate significance at the 5\% level
  under the Wald test. All covariates are standardized.}\label{tab:fixed-effects}
\smallskip
\small
\begin{adjustbox}{max width=\textwidth}
\begin{tabular}{@{}lrll r@{}}
\toprule
Parameter & Post.\ Mean & Naive 95\% CI & Wald 95\% CI & DER \\
\midrule
\multicolumn{5}{@{}l}{\textit{Extensive margin ($\balpha$)}} \\[2pt]
Intercept           & $0.764$  & $[-1.109,\; 2.628]$ & $[0.677,\; 0.851]^{*}$  & 2.04 \\
Poverty             & $-0.324$ & $[-2.205,\; 1.497]$ & $[-0.426,\; -0.221]^{*}$ & 2.41 \\
Urban               & $0.442$  & $[-1.387,\; 2.224]$ & $[0.375,\; 0.509]^{*}$  & 4.18 \\
Black               & $0.478$  & $[-1.459,\; 2.405]$ & $[0.392,\; 0.565]^{*}$  & 1.54 \\
Hispanic            & $0.006$  & $[-1.902,\; 1.903]$ & $[-0.095,\; 0.106]$     & 1.82 \\[4pt]
\multicolumn{5}{@{}l}{\textit{Intensive margin ($\bbeta$)}} \\[2pt]
Intercept           & $-0.242$ & $[-2.058,\; 1.587]$ & $[-0.272,\; -0.213]^{*}$ & 1.36 \\
Poverty             & $0.090$  & $[-1.713,\; 1.904]$ & $[0.053,\; 0.127]^{*}$  & 1.61 \\
Urban               & $-0.047$ & $[-1.858,\; 1.768]$ & $[-0.087,\; -0.007]^{*}$ & 3.18 \\
Black               & $-0.017$ & $[-1.851,\; 1.760]$ & $[-0.052,\; 0.019]$     & 1.38 \\
Hispanic            & $-0.097$ & $[-1.974,\; 1.716]$ & $[-0.130,\; -0.064]^{*}$ & 1.14 \\[4pt]
\multicolumn{5}{@{}l}{\textit{Overdispersion}} \\[2pt]
$\log\kappa$        & $1.919$  & $[1.866,\; 1.972]$ & $[1.861,\; 1.976]^{*}$   & 2.58 \\
\bottomrule
\end{tabular}
\end{adjustbox}
\end{table}

\paragraph{Impact of survey weighting.}
\cref{tab:m3b-comparison} places the weighted (M3b-W) and unweighted
(M3b) posterior means side by side.  Weighting shifts every
fixed-effect estimate, with a median displacement of 5.2 sandwich
standard errors, yet preserves all substantively important sign
patterns---most critically, the poverty reversal
($\alpha_{\mathrm{pov}} < 0$, $\beta_{\mathrm{pov}} > 0$) is present in
both specifications.  Three coefficients change sign
($\alpha_{\mathrm{Hispanic}}$, $\beta_{\mathrm{Black}}$,
$\beta_{\mathrm{Hispanic}}$), but the two intensive-margin sign changes
involve parameters whose Wald confidence intervals include zero, so the
sign itself is not statistically determined.  Overall, the comparison
confirms that the NSECE sampling design materially affects
point estimates---justifying the pseudo-posterior approach---while
leaving the central empirical findings intact.

\input{Figures/T_m3b_comparison.tex}

\paragraph{The poverty reversal.}
The reversal is the central empirical finding. The extensive-margin
poverty coefficient is $\hat{\alpha}_{\mathrm{pov}} = -0.324$ (Wald 95\% CI:
$[-0.426, -0.221]$; $\DER = 2.41$): a one-standard-deviation increase in
community poverty (approximately 8.3 percentage points) reduces the log-odds
of serving infants by 0.324. The intensive-margin poverty coefficient is
$\hat{\beta}_{\mathrm{pov}} = +0.090$ (Wald 95\% CI:
$[+0.053, +0.127]$; $\DER = 1.61$): among centers that clear the hurdle,
the same poverty increase raises the logit IT share by 0.090. Both intervals
exclude zero, the signs are firmly opposed, and the posterior probability
that the reversal holds at the population average is
$\Pr(\alpha_{\mathrm{pov}} < 0 \text{ and } \beta_{\mathrm{pov}} > 0
\mid \text{data}) = 1.000$,
computed from the sandwich-corrected posterior draws
(the binding margin, $\beta_{\mathrm{pov}}$, lies $4.8$
sandwich standard deviations from zero; a bivariate normal
approximation gives a complementary probability of order~$10^{-6}$).

\paragraph{Naive versus Wald intervals: resolving prior domination.}
The contrast between the naive and Wald intervals in
\cref{tab:fixed-effects} provides a clear empirical demonstration of
the prior domination phenomenon. The naive
credible intervals for $\balpha$ and $\bbeta$ span approximately 3.7 units
on the logit scale---wide enough to include both positive and negative values
for every parameter, rendering significance testing impossible. The Wald
intervals, by contrast, are 0.06--0.21 units wide and permit clear
inferences.

This discrepancy arises because the naive posterior of the population-average
coefficients is dominated by the prior rather than the data. In the
hierarchical model, the state random effects $\bdelta_{k,s}$ absorb nearly
all data information about the state-specific coefficients
$\tilde{\alpha}_{k,s} = \alpha_k + \delta_{k,s}$
\citep[cf.][]{RabeHeskethSkrondal2006}; the population-average
$\alpha_k$ is identified only through the random effects distribution, with
an effective sample size of $S = 51$ (the number of states) rather than
$N = 6{,}785$ (the number of providers). Consequently, the marginal
posterior variance $[\bSigma_{\mathrm{MCMC}}]_{pp}$ remains close to the
prior variance rather than converging to the data-informed precision
$[\mathbf{H}_{\mathrm{obs}}^{-1}]_{pp}$. The hierarchical variance
ratio---the ratio of MCMC posterior variance to observed-information
variance---exceeds 600 for all fixed effects (see the computational remark
in~\cref{sec:survey}; a full diagnostic table appears in~\cref{smd:tab-der}).

The sandwich-corrected Wald intervals bypass this hierarchical inflation entirely.
The Wald standard error $\sqrt{[\mathbf{V}_{\mathrm{sand}}]_{pp}}$ is a
frequentist quantity derived from the observed information and the
cluster-robust outer product (\cref{eq:V-sand})
\citep{Binder1983, WilliamsSavitsky2021}: it measures the sampling
variability of $\hat{\btheta}$ across hypothetical replications of the
survey, not the width of the Bayesian posterior. For fixed effects in
hierarchical survey models, this frequentist precision is the appropriate
basis for population-level inference. The naive posterior intervals remain
useful for hyperparameters (where the sandwich is not applicable) and for
state random effects (where the prior appropriately regularizes small-sample
estimates).

\paragraph{Design effect ratios.}
The DER column in~\cref{tab:fixed-effects} reveals a substantively
interpretable pattern. The largest DER is 4.18 for
$\alpha_{\mathrm{Urban}}$: this coefficient is most severely inflated by
the survey design because the NSECE's stratification scheme
\citep{NSECEProjectTeam2022} is explicitly
linked to urbanicity---urban and rural areas are sampled at different rates.
The second-largest DER is 3.18 for $\beta_{\mathrm{Urban}}$, confirming
that the design effect is feature-specific rather than margin-specific. At
the other extreme, the smallest DER is 1.14 for $\beta_{\mathrm{Hisp}}$,
indicating that the survey design has minimal impact on this coefficient.
The overall mean DER of 2.11 is broadly consistent with the Kish $\DEFF$
of 3.76. \cref{fig:sandwich-impact} displays the naive and sandwich
intervals side by side for all 11 parameters, together with the DER
bar chart.

\begin{figure}[t]
  \centering
  \includegraphics[width=\textwidth]{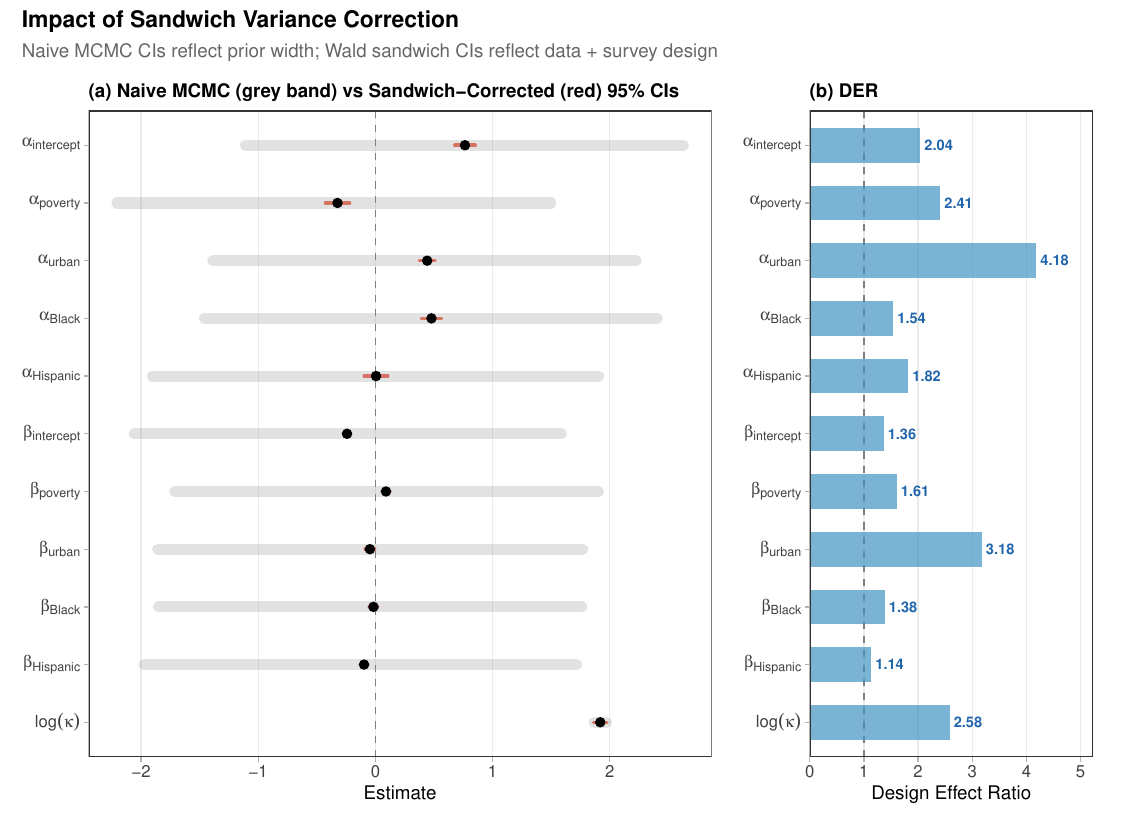}
  \caption{Impact of the sandwich variance correction on fixed-effect
    inference. \emph{Left panel:} naive MCMC 95\% credible intervals (gray)
    versus sandwich-corrected Wald 95\% confidence intervals (colored) for
    all 11 fixed effects. The naive intervals reflect prior width; the Wald
    intervals reflect data information adjusted for the survey design.
    \emph{Right panel:} design effect ratio ($\DER$) for each parameter. A
    $\DER$ of 1 indicates no survey design effect; the observed range is
    1.14--4.18 (mean 2.11).}\label{fig:sandwich-impact}
\end{figure}

\paragraph{Overdispersion.}
The estimated concentration parameter is $\hat{\kappa} = \exp(1.919)
= 6.81$ (Wald CI for $\log\kappa$: $[1.861, 1.976]$), confirming the
substantial overdispersion previewed in~\cref{sec:data}. With a typical
enrollment of $n_i \approx 50$, the variance inflation factor
$(n + \kappa)/(1 + \kappa) \approx 57/8 \approx 7$ is roughly sevenfold
the binomial baseline. Notably, $\log\kappa$ is the only parameter for
which the naive and Wald intervals are comparable in width (0.106 vs.\
0.115), because $\kappa$ has no corresponding state random effect---it is a
single global parameter informed directly by all $N_{\mathrm{pos}} = 4{,}392$
positive observations, so the standard Bernstein--von Mises approximation
holds and prior domination does not arise.


\subsection{State-specific heterogeneity and cross-margin independence}\label{sec:statevar}

The population-average coefficients describe the national-level effects;
the state-varying coefficients reveal where and how those effects change
across jurisdictions.

\paragraph{Cross-margin scatter and independence.}
\cref{fig:cross-margin} plots the posterior means
$(\tilde{\alpha}_{\mathrm{pov},s},\;
\tilde{\beta}_{\mathrm{pov},s})$ across all 51 states. Under the
hierarchical M2 specification, 48 of 51 states (94\%) fall in the
upper-left quadrant ($\alpha < 0$, $\beta > 0$), compared with 23
(45\%) in the preliminary OLS analysis (\cref{sec:features})---a
dramatic increase driven by hierarchical shrinkage pulling
small-state estimates toward the negative population-average
extensive-margin coefficient. The posterior cross-margin correlation
for the poverty coefficient is
$\hat{\varrho}_{\mathrm{pov}}^{\mathrm{cross}} \approx 0.021$,
indistinguishable from zero, implying that the extensive and intensive
margins of the poverty effect operate through largely independent
mechanisms---a structural finding consistent with the block-diagonal
identification result of \cref{prop:sigma12-nonid}. Full state-specific coefficient estimates appear in
Supplementary Material~D (\cref{smd:states}); the reversal
probability map is provided in \cref{smd:reversal-map}.

\begin{figure}[t]
  \centering
  \includegraphics[width=0.85\textwidth]{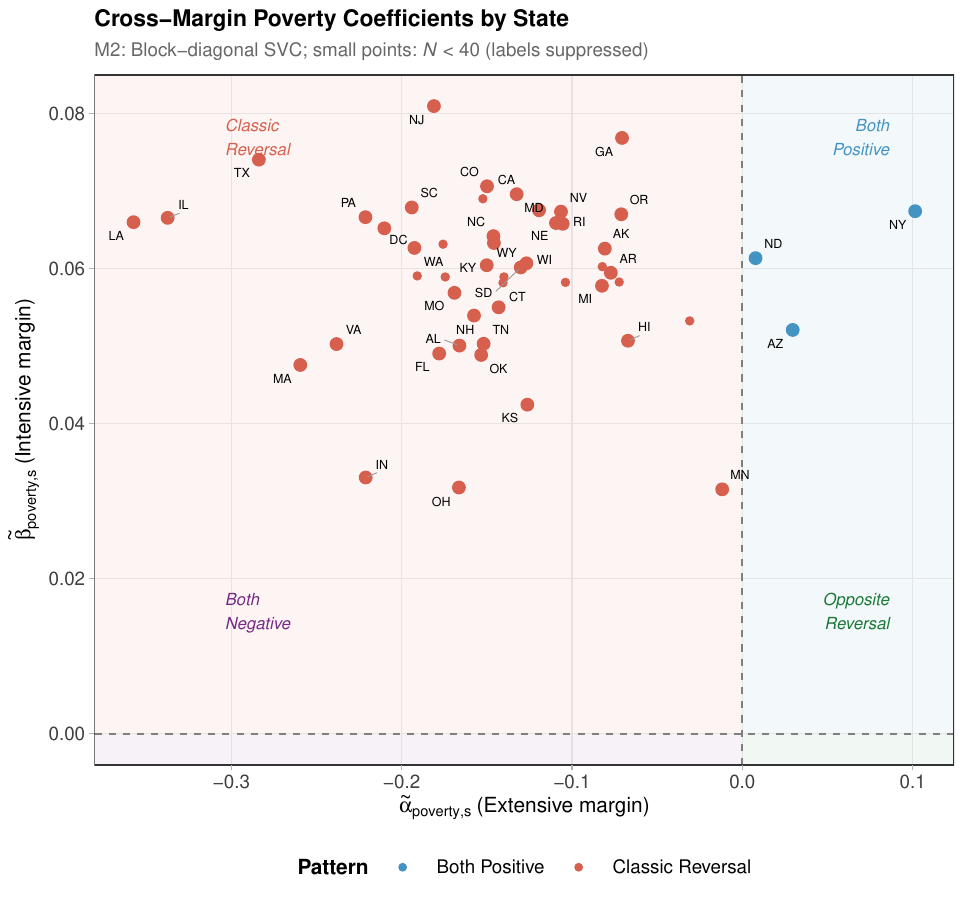}
  \caption{Cross-margin scatter plot of state-specific poverty coefficients
    from the unweighted M2 model (block-diagonal SVC).
    Each point represents the posterior mean
    $(\tilde{\alpha}_{\mathrm{pov},s},\;\tilde{\beta}_{\mathrm{pov},s})$ for
    one state.
    Small points without labels denote states with $N < 40$.  The four
    quadrants correspond to different poverty--enrollment patterns: the
    upper-left quadrant (shaded) contains 48 of 51 states (94\%),
    indicating a near-universal reversal pattern under hierarchical shrinkage.
    The dashed lines indicate the population-average coefficients
    $(\alpha_{\mathrm{pov}},\;\beta_{\mathrm{pov}})$. The near-zero
    cross-margin correlation ($\hat{\varrho} \approx 0.021$) implies that the
    two margins respond to poverty through largely independent
    mechanisms.}\label{fig:cross-margin}
\end{figure}


\subsection{Marginal effect decomposition}\label{sec:ame}

This subsection demonstrates the marginal effect decomposition---an
analytical framework enabled by the two-part structure---by translating
logit-scale coefficients into probability-scale quantities, building on
the log-scale decomposition in~\cref{prop:lae-lie}.  The decomposition
separates each covariate's total effect into extensive (access) and
intensive (share) components, a partition available only under the
two-part specification.  \Cref{tab:ame-decomp} reports the results.

\paragraph{From log to probability scale.}
\Cref{prop:lae-lie} decomposes the marginal effect on the
log-expected-enrollment scale, where the two channels are exactly
additive.  To translate to the probability scale, apply the product
rule to $E[Y_i/n_i] = q_i\,h_i$:
\begin{equation}\label{eq:ame-bridge}
  \frac{\partial\,(q_i\,h_i)}{\partial x_{i,k}}
  = \underbrace{h_i\,q_i(1-q_i)\,\tilde{\alpha}_{k,s}}_{\text{extensive}}
  \;+\;
  \underbrace{q_i\,h_i\,\varepsilon_{h,i}\,(1-\mu_i)\,
             \tilde{\beta}_{k,s}}_{\text{intensive}}
  \;=\; q_i\,h_i\bigl(\LAE_k + \LIE_k\bigr).
\end{equation}
The $\AME$ entries in~\cref{tab:ame-decomp} average each component
over providers and Cholesky-calibrated posterior
draws~$\btheta^{*(m)}$ (\cref{thm:cholesky}), mapping the log-scale
decomposition of~\cref{prop:lae-lie} onto the percentage-point scale.

\paragraph{Poverty.}
On the extensive margin, the $\AME$ of poverty is
$-0.029$ (95\% CI: $[-0.038, -0.020]$): a one-SD increase
in community poverty (${\approx}\,8.3$ percentage points)
reduces the probability of serving any infants by 2.9
percentage points.  On the intensive margin, $\AME = +0.015$
(95\% CI: $[+0.009, +0.021]$): among centers that do serve
infants, the same poverty increase raises the IT share by
1.5 percentage points.  The total $\AME$ is $-0.014$
(95\% CI: $[-0.027, -0.001]$), with the extensive margin
accounting for 66.0\% of the total absolute effect: the access
barrier dominates, but the intensity response offsets roughly
half of it.  This decomposition---unavailable under any
single-equation specification---illustrates the analytical
value of the two-part structure (\cref{prop:lae-lie}): the
hurdle identifies the access barrier as the dominant channel,
a finding that a single-equation model would collapse into a
muted net effect of $-1.4$~percentage points.

\paragraph{Other covariates.}
A notable pattern in~\cref{tab:ame-decomp} is that all four covariates
exhibit opposing extensive and intensive effects---the reversal pattern is
not unique to poverty. Urbanicity increases access ($+3.9$ pp) but decreases
intensity ($-0.8$ pp), with the extensive margin accounting for 83.6\% of
the total. The Black population share increases access ($+4.3$ pp) with a
negligible intensity effect (Ext.\ Share $= 93.9\%$). Hispanic composition
stands out: the extensive $\AME$ is effectively zero ($+0.001$), so the
entire effect operates through the intensive margin, where higher Hispanic
composition is associated with a lower IT share ($-1.6$ pp). This
intensive-margin-dominated pattern contrasts sharply with the
extensive-margin-dominated patterns for the other covariates. In every case,
a single-equation model would report only the net total effect, concealing
the offsetting channels---precisely the decomposition that the hurdle
framework and the $\AME$ formulas in \cref{prop:lae-lie} are designed to
reveal.

\begin{table}[t]
\centering
\caption{Average marginal effect ($\AME$) decomposition from model M3b.
  Effects are on the probability scale (change per one-SD increase in the
  covariate). ``Ext.\ Share'' is the percentage of the total absolute
  effect attributable to the extensive (access) margin, computed as
  $|\text{Extensive}|/(|\text{Extensive}| + |\text{Intensive}|)$.
  All four covariates
  exhibit opposing extensive and intensive effects.  The 95\% intervals
  are 2.5th and 97.5th posterior percentiles of each $\AME$ component
  evaluated at the Cholesky-calibrated
  draws~$\btheta^{*(m)}$~(\cref{thm:cholesky}).}\label{tab:ame-decomp}
\smallskip
\begin{tabular}{@{}lcccc@{}}
\toprule
Covariate & Extensive & Intensive & Total & Ext.\ Share \\
\midrule
Poverty & $-0.029$ & $+0.015$ & $-0.014$ & 66.0\% \\
 & $[-0.038,\,-0.020]$ & $[+0.009,\,+0.021]$ & $[-0.027,\,-0.001]$ & \\[3pt]
Urban & $+0.039$ & $-0.008$ & $+0.032$ & 83.6\% \\
 & $[+0.034,\,+0.045]$ & $[-0.014,\,-0.001]$ & $[+0.022,\,+0.041]$ & \\[3pt]
Black & $+0.043$ & $-0.003$ & $+0.040$ & 93.9\% \\
 & $[+0.035,\,+0.051]$ & $[-0.009,\,+0.003]$ & $[+0.029,\,+0.052]$ & \\[3pt]
Hispanic & $+0.001$ & $-0.016$ & $-0.016$ & 2.8\% \\
 & $[-0.009,\,+0.010]$ & $[-0.021,\,-0.010]$ & $[-0.027,\,-0.004]$ & \\
\bottomrule
\end{tabular}
\end{table}


\subsection{Policy moderation}\label{sec:policyresults}

The reversal rate and the geographic clustering documented in
\cref{sec:statevar} raise a natural question: can observable CCDF policies
explain why the reversal is strong in some states and absent in
others?  The $\bGamma_k$ matrices in M3b decompose the state-level
variation in covariate effects into policy-explained and residual
components (\cref{prop:poverty-expansion}).  The full $40$-element
$\bGamma$ matrices appear in Supplementary Material~D
(\cref{smd:gamma}); here we summarize three key patterns.

First, tiered reimbursement attenuates the extensive-margin poverty
barrier ($\hat{\gamma}_{1,\mathrm{pov},\mathrm{TR}} = +0.495$,
$\Pr(>0) = 0.962$) and raises baseline intensive-margin enrollment
($\hat{\gamma}_{2,\mathrm{int},\mathrm{TR}} = +0.271$,
$\Pr(>0) = 0.985$).  Second, IT rate add-ons raise baseline
participation ($\hat{\gamma}_{1,\mathrm{int},\mathrm{IT}} = +0.728$,
$\Pr(>0) = 0.996$) but paradoxically \emph{strengthen} the poverty
gradient on the extensive margin
($\hat{\gamma}_{1,\mathrm{pov},\mathrm{IT}} = -0.488$,
$\Pr(<0) = 0.982$), suggesting that the add-on primarily benefits
moderate-poverty areas.  Third, IT add-ons are most effective at closing
demographic gaps in infant enrollment, particularly for centers in Black
communities on the intensive margin
($\hat{\gamma}_{2,\mathrm{Black},\mathrm{IT}} = +0.256$,
$\Pr(>0) = 1.000$).

The three observable policy variables explain between 4\% and 16\% of
the between-state variation in poverty coefficients.  These associations
are cross-sectional and do not support causal claims; nonetheless, the
$\bGamma$ decomposition illustrates how the hierarchical structure in
\cref{sec:model} partitions state heterogeneity into policy-attributable
and residual components.

%% file: Figures/T_m3b_comparison.tex
\begin{table}[t]
\centering
\caption{Comparison of unweighted (M3b) and survey-weighted (M3b-W) fixed-effect
  posterior means. ``Shift'' is M3b-W minus M3b. ``Shift/SE'' normalizes the shift
  by the sandwich standard error, providing a scale-free measure of the weighting
  impact. ``Wald 95\% CI'' is the sandwich-corrected interval for M3b-W (\cref{thm:cholesky}).
  Asterisks indicate significance at the 5\% level. All covariates are standardized.}
\label{tab:m3b-comparison}
\smallskip
\small
\begin{adjustbox}{max width=\textwidth}
\begin{tabular}{@{}l rr rr l@{}}
\toprule
Parameter & M3b & M3b-W & Shift & Shift/SE & Wald 95\% CI \\
\midrule
\multicolumn{6}{@{}l}{\textit{Extensive margin ($\balpha$)}} \\[2pt]
Intercept            & $+0.558$ & $+0.764$ & $+0.206$ & $+4.64$ & $[0.677,\; 0.851]^{*}$ \\
Poverty              & $-0.220$ & $-0.324$ & $-0.103$ & $-1.97$ & $[-0.426,\; -0.221]^{*}$ \\
Urban                & $+0.264$ & $+0.442$ & $+0.178$ & $+5.19$ & $[0.375,\; 0.509]^{*}$ \\
Black                & $+0.248$ & $+0.478$ & $+0.230$ & $+5.20$ & $[0.392,\; 0.565]^{*}$ \\
Hispanic             & $-0.146$ & $+0.006$ & $+0.151$ & $+2.94$ & $[-0.095,\; 0.106]$ \\[4pt]
\multicolumn{6}{@{}l}{\textit{Intensive margin ($\bbeta$)}} \\[2pt]
Intercept            & $-0.096$ & $-0.242$ & $-0.146$ & $-9.65$ & $[-0.272,\; -0.213]^{*}$ \\
Poverty              & $+0.040$ & $+0.090$ & $+0.051$ & $+2.68$ & $[0.053,\; 0.127]^{*}$ \\
Urban                & $-0.055$ & $-0.047$ & $+0.008$ & $+0.39$ & $[-0.087,\; -0.007]^{*}$ \\
Black                & $+0.102$ & $-0.017$ & $-0.119$ & $-6.51$ & $[-0.052,\; 0.019]$ \\
Hispanic             & $+0.060$ & $-0.097$ & $-0.157$ & $-9.35$ & $[-0.130,\; -0.064]^{*}$ \\[4pt]
\multicolumn{6}{@{}l}{\textit{Overdispersion}} \\[2pt]
$\log\kappa$         & $+1.661$ & $+1.919$ & $+0.257$ & $+8.80$ & $[1.861,\; 1.976]^{*}$ \\
\bottomrule
\end{tabular}
\end{adjustbox}
\end{table}

%% file: section6.tex

\section{Discussion}\label{sec:discussion}


\paragraph{Summary.}
This paper develops a hierarchical hurdle beta-binomial (HBB) framework
for survey-weighted bounded counts with structural zeros, extending the
pseudo-posterior methodology of~\citet{WilliamsSavitsky2021} to two-part
hierarchical models.  The two primary methodological
contributions---cross-margin covariance identification
(\cref{prop:sigma12-nonid}) and $\DER$-guided sandwich correction
(\cref{thm:cholesky})---are discussed below alongside generalizability,
practical recommendations, and limitations.


\paragraph{Generalizability.}
The HBB framework developed in~\cref{sec:model} applies whenever the
outcome is a bounded discrete proportion---a count
$Y_i \in \{0, 1, \ldots, n_i\}$ with a known denominator---embedded
in a hierarchical structure with structural zeros.  Five domains
illustrate the breadth of applicability beyond childcare enrollment.
(i)~In \emph{dental epidemiology}, $Y_i$ is the number of carious
surfaces out of $n_i$ surfaces examined, with structural zeros for
caries-free individuals whose zero counts reflect disease absence
rather than low risk~\citep{BandyopadhyayEtAl2011}.
(ii)~In \emph{species occupancy modeling}, $Y_i$ is the number of
detection events out of $n_i$ replicate surveys at a site, with
structural zeros at sites where the species is truly
absent~\citep{MacKenzieEtAl2002}.
(iii)~In \emph{hospital bed management}, $Y_i$ is the number of beds
occupied out of $n_i$ staffed beds on a ward, with structural zeros
for wards temporarily closed to admissions due to staffing shortages
or infection control measures.
(iv)~In \emph{insurance enrollment}, $Y_i$ is the number of enrollees
out of $n_i$ eligible individuals within a plan--region cell, with
structural zeros for plan--region combinations in which the plan is
not offered.
(v)~In \emph{survey item endorsement}, $Y_i$ is the number of items
endorsed out of $n_i$ items administered on a subscale, with
structural zeros for respondents to whom the construct does not apply
(e.g., substance-use items for lifetime abstainers).
Each domain exhibits the three features that jointly motivate the HBB
specification: bounded support with a known denominator,
extra-binomial variation within the active population, and a mass of
structural zeros generated by a qualitatively distinct mechanism from
the count process.  In every case, the hurdle separates the structural
participation decision from the intensity conditional on
participation, and the beta-binomial kernel accommodates
overdispersion relative to the binomial baseline.

\paragraph{Inference contributions.}
Beyond the modeling framework, this paper makes two specific
contributions to inference.  First,
\cref{prop:sigma12-nonid} establishes that the observation-level
conditional likelihood is block-diagonal across margins, so the
cross-margin covariance $\bSigma_{12}$ (\cref{sec:svc}) is a
hierarchical estimand whose effective sample size is~$S$, not~$N$.
While a frequentist could in principle target $\bSigma_{12}$ via REML
on the marginal likelihood (\cref{rem:cond-vs-marg}), unrestricted estimation of a
$q \times q$ cross-covariance from $S = 51$ groups is fragile when
$2q = 10$; \cref{prop:sigma12-id} shows that the $\LKJ$ prior provides
finite-sample regularization, yielding posterior concentration at rate
$O(S^{-1/2})$.
Second, the sandwich correction pipeline with $\DER$-guided block-wise
application (\cref{thm:cholesky}) extends the pseudo-posterior
framework of~\citet{WilliamsSavitsky2021} to two-part hierarchical
models, nesting the scalar design-effect multiplier
of~\citet{GhosalEtAl2020} as a special case when the $\DER$ is
constant across parameters (\cref{eq:der}).


\paragraph{Practical recommendations.}
The combined evidence from the application (\cref{sec:application})
and simulation (\cref{sec:simulation}) suggests three guidelines for
applied researchers fitting Bayesian hierarchical models to complex
survey data.
First, \emph{test for informativeness before deciding on weighting}.
Following~\citet{Pfeffermann1993}, regress the sampling weight on the
outcome conditional on the full model covariate vector and test whether
the weight coefficient is significant; conditioning on the full
covariate set is essential, because a univariate test conflates
informativeness with confounding~\citep{PfeffermannEtAl1998}. For
two-part models, test each margin separately, as the design can be
non-informative for one margin and informative for the other
(\cref{sec:data}).

Second, \emph{report both naive and sandwich-corrected confidence
intervals alongside the $\DER$}.
The empirical $\DER$ range in this study is 1.14--4.18
(\cref{tab:fixed-effects}), and the simulation width ratios of
1.22--2.06 (\cref{tab:sim-results}) provide independent validation.
The linear interpolation heuristic overestimates the $\DER$ at
intermediate shrinkage levels by up to 50\%; the exact conjugate
formula in Supplementary Material~C (\cref{app:survey}) provides a
more accurate benchmark.

Third, \emph{for hierarchical variance components, report model-based
posterior intervals and conduct prior sensitivity analysis}.
The sandwich correction is structurally inapplicable to
parameters---such as the state-effect standard deviations $\tau$ and
the cross-margin correlation $\rho$---that do not produce
individual-level score contributions (\cref{sec:simresults}), and the
simulation confirms that the width ratio is identically 1.00 for these
parameters. We conduct a sensitivity analysis in
\cref{smd:lkj-sensitivity}, refitting under $\LKJ(\eta)$ with
$\eta \in \{1,\ldots,8\}$: the cross-margin correlation remains near
zero ($\hat{\rho} \in [0.10,\,0.19]$) and all fixed-effect estimates
are stable, confirming that the posterior is data-dominated.
For the 3--5 largest clusters, a sensitivity analysis comparing
block-wise and joint Cholesky correction is also recommended
(\cref{rem:blockwise}; Supplementary Material~C).


\paragraph{Limitations and future directions.}
Several limitations merit discussion.
The weight ratio $w_{\max}/w_{\min} = 462$ exceeds
$N^{1/2} \approx 82.4$ (\cref{sec:survey}), so the sufficient
condition~C3 for the Bernstein--von Mises theorem under informative
sampling is not comfortably satisfied; this motivates the sandwich
correction as a finite-sample robustness measure rather than an
asymptotic guarantee, and underscores the importance of the
simulation study in~\cref{sec:simulation} as an independent
check on coverage properties.
The beta-binomial component conditions on $n_i$ as a known upper
bound, yet total enrollment may itself respond to community
characteristics---a form of outcome endogeneity.  Joint modeling of
$(Y_i, n_i)$ would address this concern but requires a simultaneous
specification for the total-enrollment and IT-share processes, an
extension we leave to future work.
The model assumes a single concentration parameter $\kappa$ shared
across all states; state-varying $\kappa_s$ could capture
heterogeneous overdispersion, but identification of each $\kappa_s$
requires sufficient within-state replication of the denominator
values~$n_i$ (\cref{smb:prop-mu-kappa-id}).
Finally, the sandwich correction delivers design-consistent inference
for fixed effects but does not extend to hierarchical variance
components such as $\bSigma_{12}$; extending design-consistent
inference to these parameters remains an open problem, with
replicate-weight approaches offering a potential but computationally
costly avenue.

\paragraph{Alternative survey approaches.}
The pseudo-posterior framework adopted here is not the only route to
survey-weighted Bayesian inference.  Multilevel pseudo-maximum
likelihood~\citep[MPML;][]{RabeHeskethSkrondal2006} provides a
frequentist counterpart that integrates sampling weights directly into
the marginal likelihood, avoiding the pseudo-likelihood construction.
Fully model-based approaches that include the sampling mechanism as part
of the probability model~\citep{SiEtAl2020} offer theoretical
advantages when the design is informative but require additional
specification of the selection process.  Finally, balanced repeated
replicate (BRR) or jackknife replicate-weight
methods~\citep{RaoWu1988} provide design-based variance estimates
without distributional assumptions, though the computational cost of
refitting Bayesian hierarchical models under each set of replicate
weights is currently prohibitive.  Comparing these alternatives with
the pseudo-posterior sandwich approach is an important direction for
future work.

\paragraph{Conclusion.}
This paper introduces a Bayesian hierarchical hurdle beta-binomial
framework for survey-weighted bounded counts with structural zeros.
Applied to the 2019 NSECE, the model reveals a poverty reversal---community
poverty reduces access to infant and toddler care but increases
enrollment intensity among participating centers---that is
near-universal across states under hierarchical shrinkage (48 of 51)
and whose extensive margin accounts for approximately two-thirds of
the total effect.  The sandwich correction and DER diagnostic provide
a practical toolkit for population-level inference in Bayesian
hierarchical models fitted to complex survey data.  The companion R
package \textsf{hurdlebb}~\citep{hurdlebb2026} and replication
repository make the methodology immediately accessible to applied
researchers.


%% file: osm_body.tex


\section{Mathematical Properties of the Hurdle Beta-Binomial Model}
\label{app:math-properties}

\input{sm_a}


\section{Identification and Fisher Information}
\label{app:identification}

\input{sm_b}


\section{Survey-Weighted Inference}
\label{app:survey}

\input{sm_c}


\section{Extended Empirical Results}
\label{app:extended-results}

\input{sm_d}


\section{Simulation and Computational Details}
\label{app:sim-computation}

\input{sm_e}

%% file: sm_a.tex

This appendix collects the complete proofs and supporting results for the
mathematical properties of the hurdle beta-binomial model discussed in
\cref{sec:model}. The beta-binomial distribution in the mean--precision
parameterization follows~\citet{Griffiths1973} and~\citet{Prentice1986};
the hurdle construction follows~\citet{Mullahy1986} and~\citet{Cragg1971}.
We maintain the notation established in~\cref{sec:notation}
throughout.


\subsection{Proof of \texorpdfstring{\cref{thm:monotonicity}}{Theorem~2.1}}
\label{sma:proof-monotonicity}

We give the complete proof that the intensity function
$h(\mu) = \mu / (1 - p_0)$ is strictly increasing in $\mu$ for every
$n \ge 2$, $\kappa > 0$. Recall that
$p_0(n,\mu,\kappa) = \prod_{j=0}^{n-1}[(1-\mu)\kappa + j]/(\kappa + j)$
is the zero probability of the beta-binomial distribution, as defined
in~\eqref{eq:p0}.

\paragraph{Auxiliary quantities.}
Write $b = (1-\mu)\kappa$ and define
\begin{equation}\label{sma:eq-Lambda}
  \Lambda
  = \kappa \sum_{j=0}^{n-1} \frac{1}{(1-\mu)\kappa + j}
  > 0,
  \qquad
  \Lambda_2
  = \kappa^2 \sum_{j=0}^{n-1} \frac{1}{\bigl[(1-\mu)\kappa + j\bigr]^2}
  > 0.
\end{equation}
Note that $\partial p_0/\partial\mu = -p_0\,\Lambda$ (obtained by
differentiating $\ln p_0 = \sum_{j=0}^{n-1}\ln[(1-\mu)\kappa + j]
- \ln[\kappa + j]$) and $\partial\Lambda/\partial\mu = \Lambda_2$
(since $\partial b/\partial\mu = -\kappa$).

\begin{lemma}[Sum-of-reciprocals inequality]\label{sma:lem-reciprocals}
  Let $b > 0$, $n \ge 1$, and set
  $S_1 = \sum_{j=0}^{n-1}(b+j)^{-1}$ and
  $S_2 = \sum_{j=0}^{n-1}(b+j)^{-2}$, so that
  $\Lambda = \kappa\,S_1$ and $\Lambda_2 = \kappa^2\,S_2$.
  \begin{enumerate}[label=\textup{(\roman*)}]
    \item For $n = 1$: $S_1^2 = S_2$, so $\Lambda^2 = \Lambda_2$.
    \item For $n \ge 2$: $S_1^2 > S_2$, so $\Lambda^2 > \Lambda_2$.
  \end{enumerate}
\end{lemma}

\begin{proof}
  Write $a_j = (b + j)^{-1}$.  Expanding the square of the sum,
  \[
    S_1^2
    = \Bigl(\sum_{j=0}^{n-1} a_j\Bigr)^{\!2}
    = \sum_{j=0}^{n-1} a_j^2
      + 2\!\!\sum_{0 \le i < j \le n-1}\!\! a_i\, a_j
    = S_2
      + 2\!\!\sum_{0 \le i < j \le n-1}\!\!
        \frac{1}{(b+i)(b+j)}.
  \]
  For $n = 1$ the cross-product sum is empty, giving $S_1^2 = S_2$.
  For $n \ge 2$ every term $1/[(b+i)(b+j)]$ is strictly positive, so
  $S_1^2 > S_2$.
\end{proof}

\begin{proof}[Proof of~\cref{thm:monotonicity}]
  Define the numerator function
  \begin{equation}\label{sma:eq-Phi}
    \Phi(\mu)
    = \bigl(1 - p_0(\mu)\bigr) - \mu\, p_0(\mu)\,\Lambda(\mu).
  \end{equation}
  From the quotient rule applied to $h = \mu\,(1-p_0)^{-1}$ together with
  $\partial p_0 / \partial\mu = -p_0\, \Lambda$, one obtains
  \begin{equation}\label{sma:eq-dh-numerator}
    \frac{\partial h}{\partial \mu}
    = \frac{\Phi(\mu)}{(1 - p_0)^2}.
  \end{equation}
  Since $(1 - p_0)^2 > 0$ for $\mu \in (0,1)$, the sign of
  $\partial h / \partial\mu$ equals the sign of $\Phi(\mu)$.  The proof
  proceeds in three steps.

  \medskip\noindent\textit{Step~1: Boundary value.}\enspace
  At $\mu = 0$, every factor in~\eqref{eq:p0} equals one, so $p_0(0) = 1$.
  Hence $\Phi(0) = (1 - 1) - 0 \cdot 1 \cdot \Lambda(0) = 0$.

  \medskip\noindent\textit{Step~2: Derivative of\/ $\Phi$.}\enspace
  Differentiating~\eqref{sma:eq-Phi} with the product rule gives
  \[
    \Phi'(\mu)
    = p_0\,\Lambda
      - \bigl[p_0\,\Lambda
        + \mu\,\tfrac{d}{d\mu}(p_0\,\Lambda)\bigr]
    = -\mu\,\frac{d}{d\mu}\bigl(p_0\,\Lambda\bigr).
  \]
  Now $d(p_0\,\Lambda)/d\mu
  = (\partial p_0/\partial\mu)\,\Lambda + p_0\,(\partial\Lambda/\partial\mu)
  = -p_0\,\Lambda^2 + p_0\,\Lambda_2
  = p_0(\Lambda_2 - \Lambda^2)$,
  and therefore
  \begin{equation}\label{sma:eq-Phi-deriv}
    \Phi'(\mu)
    = -\mu\, p_0\bigl(\Lambda_2 - \Lambda^2\bigr)
    = \mu\, p_0\bigl(\Lambda^2 - \Lambda_2\bigr).
  \end{equation}

  \medskip\noindent\textit{Step~3: Application of
  \cref{sma:lem-reciprocals}.}\enspace
  \begin{itemize}[nosep]
    \item \emph{Case $n = 1$.}
      \cref{sma:lem-reciprocals}(i) gives $\Lambda^2 = \Lambda_2$, so
      $\Phi'(\mu) = 0$ identically.  Combined with $\Phi(0) = 0$, we obtain
      $\Phi(\mu) = 0$ for all $\mu \in (0,1)$ and hence
      $\partial h / \partial\mu = 0$.  Equivalently, $p_0 = 1 - \mu$ when
      $n = 1$, so $h(\mu) = \mu / (1 - (1 - \mu)) = 1$ identically.
    \item \emph{Case $n \ge 2$.}
      \cref{sma:lem-reciprocals}(ii) gives $\Lambda^2 > \Lambda_2$.
      Since $\mu > 0$ and $p_0 > 0$ on $(0,1)$,
      equation~\eqref{sma:eq-Phi-deriv} yields $\Phi'(\mu) > 0$ for all
      $\mu \in (0,1)$.
  \end{itemize}
  Since $\Phi(0) = 0$ and $\Phi$ is strictly increasing on $(0,1)$, we
  conclude $\Phi(\mu) > 0$ for all $\mu \in (0,1)$, and therefore
  \[
    \frac{\partial h}{\partial\mu}
    = \frac{\Phi(\mu)}{(1 - p_0)^2} > 0. \qedhere
  \]
\end{proof}


\subsection{Higher moments of the truncated beta-binomial}
\label{sma:higher-moments}

\paragraph{Notation.}
Throughout this subsection, $Y \sim \BetaBin(n,\mu,\kappa)$ denotes the
standard (untruncated) beta-binomial with $a = \mu\kappa$,
$b = (1-\mu)\kappa$, and
$Y^+ \sim \ZTBB(n,\mu,\kappa)$ denotes the zero-truncated version.  We
write $p_0 = \Prob(Y = 0)$,
$V_{\mathrm{BB}} = \Var(Y) = n\mu(1-\mu)(n+\kappa)/(1+\kappa)$, and
$h = \mu/(1-p_0)$.


\begin{lemma}[Second factorial moment]\label{sma:lem-second-factorial}
  For $Y \sim \BetaBin(n,\mu,\kappa)$,
  \begin{equation}\label{sma:eq-second-factorial}
    \boxed{\E\bk{Y(Y-1)} = n(n-1) \cdot \frac{\mu(\mu\kappa + 1)}{\kappa + 1}.}
  \end{equation}
\end{lemma}

\begin{proof}
  The beta-binomial arises as a beta
  mixture~\citep{Griffiths1973}:
  let $p \sim \Beta(a,b)$ with $a = \mu\kappa$,
  $b = (1\!-\!\mu)\kappa$, and $Y \mid p \sim \Bin(n,p)$.
  Since $\E[Y(Y\!-\!1) \mid p] = n(n\!-\!1)p^2$,
  the tower property gives
  \[
    \E\bk{Y(Y-1)} = n(n-1)\,\E[p^2].
  \]
  The second moment of $p \sim \Beta(a,b)$ is
  $\E[p^2] = a(a+1)/[(a+b)(a+b+1)]$, using the rising factorial identity
  $\E[p^k] = a^{(k)}/(a+b)^{(k)}$ for $k = 2$.  Substituting
  $a = \mu\kappa$, $a + b = \kappa$:
  \[
    \E\bk{Y(Y-1)}
    = n(n-1) \cdot \frac{\mu\kappa(\mu\kappa + 1)}{\kappa(\kappa+1)}
    = n(n-1) \cdot \frac{\mu(\mu\kappa + 1)}{\kappa+1}. \qedhere
  \]
\end{proof}

\paragraph{Consistency check.}
From~\eqref{sma:eq-second-factorial}, the second moment is
$\E[Y^2] = \E[Y(Y-1)] + \E[Y] = n(n-1)\mu(\mu\kappa+1)/(\kappa+1) + n\mu$.
We confirm $\Var(Y) = \E[Y^2] - (n\mu)^2$:
\[
  \E[Y^2] - n^2\mu^2
  = \frac{n(n-1)\mu(\mu\kappa+1) + n\mu(\kappa+1) - n^2\mu^2(\kappa+1)}
         {\kappa+1}.
\]
Expanding the numerator factor:
$n\mu\bk{(n-1)(\mu\kappa+1) + (\kappa+1) - n\mu(\kappa+1)}$.  The bracket
simplifies as
\[
  (n-1)\mu\kappa + (n-1) + (\kappa+1) - n\mu\kappa - n\mu
  = (1-\mu)(n + \kappa).
\]
Therefore
$\Var(Y) = n\mu(1-\mu)(n+\kappa)/(\kappa+1) = V_{\mathrm{BB}}$.
\enspace$\square$


\begin{lemma}[General $r$-th factorial moment]\label{sma:lem-rth-factorial}
  For $Y \sim \BetaBin(n,\mu,\kappa)$,
  \begin{equation}\label{sma:eq-rth-factorial}
    \E\!\left[\prod_{j=0}^{r-1}(Y-j)\right]
    = \frac{n^{(r)} \cdot (\mu\kappa)^{(r)}}{\kappa^{(r)}},
  \end{equation}
  where $c^{(r)} = c(c+1)\cdots(c+r-1) = \Gamma(c+r)/\Gamma(c)$ is the
  rising factorial.
\end{lemma}

\begin{proof}
  Follows from
  $\E[\prod_{j=0}^{r-1}(Y-j) \mid p] = n^{(r)}p^r$ for the binomial, and
  $\E[p^r] = (\mu\kappa)^{(r)}/\kappa^{(r)}$ for the beta.
\end{proof}


\begin{proposition}[Truncation identity for functions vanishing at zero]
  \label{sma:prop-trunc-identity}
  For any function $g$ with $g(0) = 0$,
  \begin{equation}\label{sma:eq-trunc-identity}
    \E\bk{g(Y) \mid Y > 0} = \frac{\E\bk{g(Y)}}{1 - p_0}.
  \end{equation}
  In particular,
  $\E[Y^k \mid Y > 0] = \E[Y^k]/(1-p_0)$ for all $k \ge 1$.
\end{proposition}

\begin{proof}
  Since $g(Y) \cdot \one(Y = 0) = g(0) \cdot \one(Y = 0) = 0$, we have
  $\E[g(Y)] = \E[g(Y) \cdot \one(Y > 0)]$.  Dividing by
  $\Prob(Y > 0) = 1 - p_0$ gives the result.
\end{proof}


\begin{corollary}[Conditional second factorial moment]
  \label{sma:cor-cond-second-factorial}
  \begin{equation}\label{sma:eq-cond-second-factorial}
    \boxed{\E\bk{Y(Y-1) \mid Y > 0}
    = \frac{n(n-1)\mu(\mu\kappa+1)}{(\kappa+1)(1-p_0)}.}
  \end{equation}
\end{corollary}


\begin{proposition}[Truncated BB variance]
  \label{sma:prop-trunc-var}
  For $Y^+ \sim \ZTBB(n,\mu,\kappa)$,
  \begin{equation}\label{sma:eq-trunc-var}
    \boxed{\Var(Y \mid Y > 0)
    = \frac{V_{\mathrm{BB}}(1-p_0) - n^2\mu^2\,p_0}{(1-p_0)^2}.}
  \end{equation}
\end{proposition}

\begin{proof}
  By the truncation identity (\cref{sma:prop-trunc-identity}),
  \[
    \E[Y^2 \mid Y > 0] = \frac{\E[Y^2]}{1 - p_0},
    \qquad
    \E[Y \mid Y > 0] = \frac{n\mu}{1 - p_0}.
  \]
  Then
  \[
    \Var(Y \mid Y > 0)
    = \frac{\E[Y^2]}{1-p_0} - \frac{n^2\mu^2}{(1-p_0)^2}
    = \frac{V_{\mathrm{BB}} + n^2\mu^2}{1-p_0}
      - \frac{n^2\mu^2}{(1-p_0)^2}.
  \]
  Combining the $n^2\mu^2$ terms over the common denominator $(1-p_0)^2$:
  \[
    = \frac{V_{\mathrm{BB}}}{1-p_0}
      + n^2\mu^2 \cdot \frac{(1-p_0) - 1}{(1-p_0)^2}
    = \frac{V_{\mathrm{BB}}}{1-p_0}
      - \frac{n^2\mu^2\,p_0}{(1-p_0)^2}
    = \frac{V_{\mathrm{BB}}(1-p_0) - n^2\mu^2\,p_0}{(1-p_0)^2}.
    \qedhere
  \]
\end{proof}

\begin{remark}\label{sma:rem-competing-effects}
  The formula~\eqref{sma:eq-trunc-var} reveals two competing effects of
  truncation.  The first term $V_{\mathrm{BB}}/(1-p_0)$ inflates the
  variance by rescaling probabilities.  The second term
  $n^2\mu^2 p_0/(1-p_0)^2$ reduces the variance because the truncated
  distribution loses mass at zero, pulling the mean away from zero and
  tightening the effective spread relative to that shifted mean.
\end{remark}

\paragraph{Alternative forms.}

\emph{Form~A (factored):}
\begin{equation}\label{sma:eq-trunc-var-formA}
  \Var(Y \mid Y > 0)
  = \frac{n\mu}{1-p_0}
    \left[\frac{(1-\mu)(n+\kappa)}{1+\kappa}
          - \frac{n\mu\,p_0}{1-p_0}\right].
\end{equation}
This displays the variance as $\E[Y \mid Y > 0]$ times a ``spread factor.''

\emph{Form~B (via factorial moments):}
\begin{equation}\label{sma:eq-trunc-var-formB}
  \Var(Y \mid Y > 0)
  = \frac{\E[Y(Y-1)]}{1-p_0}
    + \frac{n\mu}{1-p_0}
    - \frac{n^2\mu^2}{(1-p_0)^2}.
\end{equation}

\paragraph{Verification of limiting cases.}

\emph{Case~1: $\kappa \to \infty$ (truncated binomial limit).}
As $\kappa \to \infty$, $V_{\mathrm{BB}} \to n\mu(1-\mu)$ and
$p_0 \to (1-\mu)^n$.  The formula becomes the known variance of the
zero-truncated binomial.  For $n = 1$: $Y \mid Y > 0 = 1$
deterministically, so $\Var = 0$.  In the formula: numerator
${}= \mu(1-\mu) \cdot \mu - \mu^2(1-\mu) = 0$.\enspace$\square$

The remaining limiting cases are verified analogously:
as $p_0 \to 0$ the variance reduces to $V_{\mathrm{BB}}$, and
for $n=1$ both sides equal zero since $Y \mid Y > 0 = 1$
deterministically.\enspace$\square$

\paragraph{Numerical cross-check.}
Consider $n = 5$, $\mu = 0.3$, $\kappa = 2$ (so $a = 0.6$, $b = 1.4$).
Direct computation gives
$p_0 = \prod_{j=0}^{4}(1.4+j)/(2+j) = 0.376992$,
$\E[Y] = 1.5$,
$V_{\mathrm{BB}} = 2.45$, and
$\E[Y^2] = 4.70$.
The truncated moments are
$\E[Y \mid Y > 0] = 1.5/0.623008 = 2.40770$,
$\E[Y^2 \mid Y > 0] = 4.70/0.623008 = 7.54409$, and
$\Var(Y \mid Y > 0) = 7.54409 - 2.40770^2 = 1.74707$.
Form~A gives
$(2.45 \times 0.623008 - 2.25 \times 0.376992)/0.623008^2 = 1.74709$,
agreeing to five significant figures.\enspace$\square$


\subsection{Unconditional variance decomposition}\label{sma:vardecomp}

We derive the unconditional variance of $Y_i$ under the hurdle
beta-binomial model, its variance-to-mean ratio, an overdispersion index
relative to the binomial, and the coefficient of variation. Throughout, we
write $V_{\textup{BB},i} = n_i\mu_i(1-\mu_i)(n_i + \kappa)/(1 + \kappa)$
for the beta-binomial variance~\eqref{eq:bb-moments}, $h_i = \mu_i/(1-p_{0,i})$
for the intensity function~\eqref{eq:uncond-mean}, and
$m_i^{+} = n_i h_i = n_i\mu_i/(1-p_{0,i})$ for the conditional mean given
participation.


\paragraph{Law of total variance.}

\begin{proposition}[Hurdle model variance decomposition]
\label{sma:prop-var-decomp}
Let $Y_i$ follow the $\HBB$ model
\citep[cf.][for the Poisson hurdle analogue]{Mullahy1986}
with participation probability~$q_i$ and
intensity parameters $(n_i, \mu_i, \kappa)$. Define
$Z_i = \one\{Y_i > 0\}$. Then
\begin{equation}\label{sma:eq-var-decomp}
  \Var(Y_i)
  = q_i \cdot \Var(Y_i \mid Y_i > 0)
    + q_i(1-q_i)\bigl[\E(Y_i \mid Y_i > 0)\bigr]^2.
\end{equation}
\end{proposition}

\begin{proof}
Apply the law of total variance with respect to the Bernoulli variable~$Z_i$:
\[
  \Var(Y_i)
  = \E\bigl[\Var(Y_i \mid Z_i)\bigr]
    + \Var\bigl(\E[Y_i \mid Z_i]\bigr).
\]
\emph{First term.}  When $Z_i = 0$, $Y_i = 0$ almost surely, so
$\Var(Y_i \mid Z_i = 0) = 0$.  When $Z_i = 1$,
$\Var(Y_i \mid Z_i = 1) = \Var(Y_i \mid Y_i > 0)$.  Hence
\[
  \E\bigl[\Var(Y_i \mid Z_i)\bigr]
  = q_i\,\Var(Y_i \mid Y_i > 0).
\]
\emph{Second term.}  The conditional expectation $\E[Y_i \mid Z_i]$ takes
value~$0$ with probability $1-q_i$ and $m_i^{+} = n_i\mu_i/(1-p_{0,i})$
with probability~$q_i$.  Its variance is therefore
$q_i(1-q_i)(m_i^{+})^2$.
\end{proof}

\begin{corollary}[Expanded form]\label{sma:cor-var-expanded}
Substituting the truncated variance from
\cref{sma:prop-trunc-var} into~\eqref{sma:eq-var-decomp}:
\begin{equation}\label{sma:eq-var-expanded}
  \Var(Y_i)
  = \frac{q_i\, V_{\textup{BB},i}}{1-p_{0,i}}
    - \frac{q_i\, n_i^2 \mu_i^2\, p_{0,i}}{(1-p_{0,i})^2}
    + \frac{q_i(1-q_i)\, n_i^2 \mu_i^2}{(1-p_{0,i})^2}.
\end{equation}
Combining the terms involving $n_i^2\mu_i^2/(1-p_{0,i})^2$ yields the
compact form
\begin{equation}\label{sma:eq-var-compact}
  \Var(Y_i)
  = \frac{q_i\, V_{\textup{BB},i}}{1-p_{0,i}}
    + \frac{q_i\, n_i^2\mu_i^2\,(1-q_i-p_{0,i})}{(1-p_{0,i})^2}.
\end{equation}
\end{corollary}

\begin{remark}\label{sma:rem-qi-vs-p0}
The quantities $q_i$ and $p_{0,i}$ are distinct: $q_i$ is the hurdle
participation probability, whereas $p_{0,i}$ is the zero probability of the
beta-binomial component. In general
$q_i \neq 1 - p_{0,i}$; the quantity $(1 - q_i - p_{0,i})$ may be positive,
negative, or zero, so the second term in~\eqref{sma:eq-var-compact} can
either inflate or reduce the unconditional variance.
\end{remark}

\begin{corollary}[In terms of the unconditional mean]
\label{sma:cor-var-uncond-mean}
Recall from~\eqref{eq:uncond-mean} that $\E[Y_i] = q_i\,n_i\,h_i$. Then
\begin{equation}\label{sma:eq-var-uncond-mean}
  \Var(Y_i)
  = q_i\,\Var(Y_i \mid Y_i > 0)
    + \frac{1-q_i}{q_i}\,\bigl(\E[Y_i]\bigr)^2.
\end{equation}
The between-component term is proportional to the squared mean, scaled by
the odds of non-participation $(1-q_i)/q_i$.
\end{corollary}


\paragraph{Variance-to-mean ratio.}

\begin{proposition}[Variance ratio]\label{sma:prop-VR}
Define $\textup{VR}_i = \Var(Y_i)/\E[Y_i]$.  Then
\begin{equation}\label{sma:eq-VR}
  \boxed{
    \textup{VR}_i
    = \frac{(1-\mu_i)(n_i+\kappa)}{1+\kappa}
      - \frac{n_i\,\mu_i\,p_{0,i}}{1-p_{0,i}}
      + (1-q_i)\,n_i\,h_i.
  }
\end{equation}
\end{proposition}

\begin{proof}
Using $\E[Y_i] = q_i\,n_i\,h_i$ and~\cref{sma:cor-var-uncond-mean}:
\[
  \textup{VR}_i
  = \frac{\Var(Y_i \mid Y_i > 0)}{n_i\,h_i}
    + (1-q_i)\,n_i\,h_i.
\]
For the first term, substitute the truncated variance
(\cref{sma:prop-trunc-var}) and simplify using
$h_i = \mu_i/(1-p_{0,i})$:
\[
  \frac{V_{\textup{BB},i}}{(1-p_{0,i})\,n_i\,h_i}
  = \frac{(1-\mu_i)(n_i+\kappa)}{1+\kappa},
  \qquad
  \frac{n_i\,\mu_i^2\,p_{0,i}}{(1-p_{0,i})^2\,h_i}
  = \frac{n_i\,\mu_i\,p_{0,i}}{1-p_{0,i}}. \qedhere
\]
\end{proof}

\begin{remark}\label{sma:rem-VR-interpretation}
The three terms in~\eqref{sma:eq-VR} have clear interpretations:
\begin{enumerate}[label=\textup{(\roman*)}]
  \item $(1-\mu_i)(n_i+\kappa)/(1+\kappa)$: the beta-binomial variance
    inflation factor~\citep{Prentice1986}, which exceeds unity for all $n_i > 1$;
  \item $-n_i\,\mu_i\,p_{0,i}/(1-p_{0,i})$: a correction from
    zero-truncation that reduces the variance ratio, reflecting the
    mass removed at $Y = 0$;
  \item $(1-q_i)\,n_i\,h_i$: the between-component contribution from the
    hurdle, which increases the variance ratio whenever $q_i < 1$.
\end{enumerate}
For Poisson data, $\textup{VR} = 1$ (equidispersion). In the NSECE data,
$\textup{VR}_i > 1$ consistently, reflecting both overdispersion and the
hurdle mechanism.
\end{remark}


\paragraph{Overdispersion index relative to the binomial.}

\begin{definition}[Overdispersion index]\label{sma:def-OD}
For a proportion $Y_i/n_i$, define
\begin{equation}\label{sma:eq-OD}
  \textup{OD}_i
  = \frac{\Var(Y_i)}{n_i\,\pi_i(1-\pi_i)},
  \qquad
  \text{where }
  \pi_i = q_i\,h_i = \frac{\E[Y_i]}{n_i}.
\end{equation}
\end{definition}

The denominator $n_i\,\pi_i(1-\pi_i)$ is the variance of a
$\Bin(n_i, \pi_i)$ random variable with the same unconditional mean
proportion, so $\textup{OD}_i$ measures overdispersion relative to the
binomial benchmark.

\paragraph{Limiting cases.}

\begin{itemize}[nosep]
  \item \emph{No hurdle, $p_0 \to 0$.} Then $q_i = 1$, $h_i \to \mu_i$,
    and $\Var(Y_i) \to V_{\textup{BB},i}$, so
    $\textup{OD}_i \to (n_i+\kappa)/(1+\kappa)$, the standard beta-binomial
    overdispersion factor~\citep{Prentice1986}.
  \item \emph{Binomial limit ($\kappa \to \infty$, $q_i = 1$).}
    $V_{\textup{BB},i} \to n_i\mu_i(1-\mu_i)$ and the hurdle is inactive, so
    $\textup{OD}_i \to 1$.
\end{itemize}


\paragraph{Coefficient of variation.}

\begin{proposition}[Conditional CV]\label{sma:prop-CV-cond}
The squared coefficient of variation of $Y_i$ conditional on participation
is
\begin{equation}\label{sma:eq-CV-cond}
  \boxed{
    \textup{CV}^2(Y_i \mid Y_i > 0)
    = \frac{(1-\mu_i)(n_i+\kappa)}{n_i\,\mu_i(1+\kappa)}\,(1-p_{0,i})
      - p_{0,i}.
  }
\end{equation}
\end{proposition}

\begin{proof}
By definition,
$\textup{CV}^2 = \Var(Y_i \mid Y_i > 0) / [\E(Y_i \mid Y_i > 0)]^2$.
Substituting the truncated variance (\cref{sma:prop-trunc-var}) and the
conditional mean $m_i^{+} = n_i\mu_i/(1-p_{0,i})$:
\begin{align*}
  \textup{CV}^2
  &= \frac{V_{\textup{BB},i}(1-p_{0,i}) - n_i^2\mu_i^2\,p_{0,i}}
         {n_i^2\mu_i^2}
  = \frac{V_{\textup{BB},i}}{n_i^2\mu_i^2}\,(1-p_{0,i}) - p_{0,i} \\
  &= \frac{(1-\mu_i)(n_i+\kappa)}{n_i\,\mu_i(1+\kappa)}\,(1-p_{0,i})
    - p_{0,i}. \qedhere
\end{align*}
\end{proof}

\begin{proposition}[Unconditional CV]\label{sma:prop-CV-uncond}
The squared coefficient of variation of $Y_i$ (unconditional) satisfies
\begin{equation}\label{sma:eq-CV-uncond}
  \boxed{
    \textup{CV}^2(Y_i)
    = \frac{\textup{CV}^2(Y_i \mid Y_i > 0) + 1 - q_i}{q_i}.
  }
\end{equation}
\end{proposition}

\begin{proof}
From $\Var(Y_i) = q_i\,\Var(Y_i \mid Y_i > 0) + q_i(1-q_i)(n_i h_i)^2$
(\cref{sma:prop-var-decomp}) and $[\E(Y_i)]^2 = q_i^2(n_i h_i)^2$:
\[
  \textup{CV}^2(Y_i)
  = \frac{\textup{CV}^2(Y_i \mid Y_i > 0)}{q_i}
    + \frac{1-q_i}{q_i}
  = \frac{\textup{CV}^2(Y_i \mid Y_i > 0) + 1 - q_i}{q_i}. \qedhere
\]
\end{proof}

\begin{remark}\label{sma:rem-CV-PPC}
The coefficient of variation is useful for posterior predictive checking
because: \textup{(i)}~it is scale-free, enabling comparison across
providers with different~$n_i$; \textup{(ii)}~the additive $(1-q_i)/q_i$
term in~\eqref{sma:eq-CV-uncond} reveals how the hurdle magnifies the CV
when participation is low; and \textup{(iii)}~observed sample CVs can be
compared against the posterior distribution of theoretical CVs, providing a
diagnostic that simultaneously probes the mean, variance, and
zero-inflation structure of the model.
\end{remark}


\subsection{Analytic properties of the intensity function}\label{sma:intensity}

This subsection collects analytic results for the intensity function
$h(\mu) = \mu/(1-p_0)$ that underpin the elasticity formula
(\cref{prop:elasticity}), the boundary behavior discussed in
\cref{rem:boundary}, and the log-scale decomposition of marginal effects.
We continue to write $b = (1-\mu)\kappa$ and use
$\Lambda$, $\Lambda_2$ as defined in~\eqref{sma:eq-Lambda}.


\paragraph{Second derivative of $p_0$.}

\begin{lemma}[Second derivative of $p_0$]\label{sma:lem-d2p0}
  For all $n \ge 1$, $\mu \in (0,1)$, and $\kappa > 0$,
  \begin{equation}\label{sma:eq-d2p0}
    \frac{\partial^2 p_0}{\partial\mu^2}
    = p_0\bigl(\Lambda^2 - \Lambda_2\bigr).
  \end{equation}
\end{lemma}

\begin{proof}
  Starting from $\partial p_0/\partial\mu = -p_0\,\Lambda$, differentiate
  by the product rule:
  \[
    \frac{\partial^2 p_0}{\partial\mu^2}
    = -\frac{\partial p_0}{\partial\mu}\cdot\Lambda
      - p_0\cdot\frac{\partial\Lambda}{\partial\mu}
    = p_0\,\Lambda^2 - p_0\,\Lambda_2
    = p_0\bigl(\Lambda^2 - \Lambda_2\bigr).
  \]
  Here we used $\partial\Lambda/\partial\mu = \Lambda_2$ (since
  $\partial b/\partial\mu = -\kappa$).
\end{proof}

\begin{corollary}[Convexity of $p_0$ in $\mu$]\label{sma:cor-p0-convex}
  For $n \ge 2$:
  $\partial^2 p_0/\partial\mu^2 > 0$, so $p_0$ is strictly convex
  in~$\mu$. For $n = 1$: $p_0 = 1 - \mu$ is linear in~$\mu$.
\end{corollary}

\begin{proof}
  Immediate from~\cref{sma:lem-d2p0} and
  \cref{sma:lem-reciprocals}: for $n \ge 2$,
  $\Lambda^2 > \Lambda_2$ and $p_0 > 0$, so the second derivative is
  strictly positive. For $n = 1$, $\Lambda^2 = \Lambda_2$, giving
  $\partial^2 p_0/\partial\mu^2 = 0$.
\end{proof}


\paragraph{Second derivative of $h$.}

\begin{proposition}[Second derivative of $h$]\label{sma:prop-d2h}
  Write $D = 1 - p_0$. Then
  \begin{equation}\label{sma:eq-d2h}
    \frac{\partial^2 h}{\partial\mu^2}
    = \frac{p_0}{D^3}
      \Bigl[-2\Lambda D
            + \mu\Lambda^2(1+p_0)
            - \mu\Lambda_2 D\Bigr].
  \end{equation}
\end{proposition}

\begin{proof}
  From $\partial h/\partial\mu = N/D^2$ where
  $N = \Phi(\mu) = D - \mu\,p_0\,\Lambda$
  (see~\eqref{sma:eq-Phi}), apply the quotient rule:
  \[
    \frac{\partial^2 h}{\partial\mu^2}
    = \frac{N'D - 2N D'}{D^3},
  \]
  where $D' = -\partial p_0/\partial\mu = p_0\,\Lambda > 0$.

  \medskip\noindent\textit{Computing $N'$.}\enspace
  Since $N = (1-p_0) - \mu\,p_0\,\Lambda$, we have
  \[
    N' = p_0\,\Lambda
         - \bigl[p_0\,\Lambda
           + \mu\cdot\tfrac{d}{d\mu}(p_0\,\Lambda)\bigr]
       = -\mu\cdot\frac{d}{d\mu}(p_0\,\Lambda)
       = \mu\,p_0\bigl(\Lambda^2 - \Lambda_2\bigr),
  \]
  using $d(p_0\Lambda)/d\mu = p_0(\Lambda_2 - \Lambda^2)$
  from~\eqref{sma:eq-Phi-deriv}.

  \medskip\noindent\textit{Assembling.}\enspace
  \begin{align*}
    N'D - 2N D'
    &= \mu\,p_0(\Lambda^2 - \Lambda_2)D
       - 2\bigl(D - \mu\,p_0\,\Lambda\bigr)(p_0\,\Lambda) \\
    &= \mu\,p_0(\Lambda^2 - \Lambda_2)D
       - 2D\,p_0\,\Lambda + 2\mu\,p_0^2\Lambda^2 \\
    &= p_0\bigl[-2\Lambda D
       + \mu(\Lambda^2 - \Lambda_2)D
       + 2\mu\,p_0\,\Lambda^2\bigr].
  \end{align*}
  Noting that
  $(\Lambda^2 - \Lambda_2)D + 2p_0\Lambda^2
  = \Lambda^2(D + 2p_0) - \Lambda_2 D
  = \Lambda^2(1 + p_0) - \Lambda_2 D$
  yields the stated form.
\end{proof}

\paragraph{Verification ($n = 1$).}
When $n = 1$, $p_0 = 1 - \mu$, $D = \mu$,
$\Lambda = 1/(1-\mu)$, $\Lambda_2 = \Lambda^2$, and $h \equiv 1$, so
$\partial^2 h/\partial\mu^2 = 0$. In the formula: $N = D - \mu\,p_0\,\Lambda
= \mu - \mu(1-\mu)/(1-\mu) = 0$ identically, so $N' = 0$ and the
numerator vanishes.\enspace$\square$


\paragraph{Monotonicity of $p_0$ in all parameters.}

\begin{proposition}[Monotonicity of $p_0$ in $\mu$]\label{sma:prop-p0-mu}
  $\partial p_0/\partial\mu = -p_0\,\Lambda < 0$.
  Increasing $\mu$ strictly decreases the zero probability.
\end{proposition}

\begin{proof}
  Established in~\cref{sma:proof-monotonicity} (auxiliary computation
  preceding the proof of~\cref{thm:monotonicity}).
\end{proof}

\begin{proposition}[Monotonicity of $p_0$ in $\kappa$]\label{sma:prop-p0-kappa}
  For $n \ge 2$ and $\mu \in (0,1)$:
  \begin{equation}\label{sma:eq-dp0-dkappa}
    \frac{\partial p_0}{\partial\kappa}
    = -p_0\,\mu\sum_{j=1}^{n-1}\frac{j}{(b+j)(\kappa+j)} < 0.
  \end{equation}
  For $n = 1$: $\partial p_0/\partial\kappa = 0$ (since
  $p_0 = 1 - \mu$ is independent of~$\kappa$).
\end{proposition}

\begin{proof}
  Differentiating
  $\log p_0 = \sum_{j=0}^{n-1}[\log(b+j) - \log(\kappa+j)]$
  with $\partial b/\partial\kappa = 1-\mu$:
  \[
    \frac{\partial\log p_0}{\partial\kappa}
    = \sum_{j=0}^{n-1}
      \left[\frac{1-\mu}{b+j} - \frac{1}{\kappa+j}\right].
  \]
  For each term, the numerator is
  $(1-\mu)(\kappa+j) - (b+j) = (1-\mu)(\kappa+j) - (1-\mu)\kappa - j = -\mu j$.
  Hence
  \[
    \frac{\partial\log p_0}{\partial\kappa}
    = -\mu\sum_{j=1}^{n-1}\frac{j}{(b+j)(\kappa+j)} < 0,
  \]
  where the $j = 0$ term vanishes. Multiplying by $p_0 > 0$ gives the
  stated formula.
\end{proof}

Increasing $\kappa$ (reducing overdispersion) concentrates mass near
$n\mu > 0$, pulling probability away from zero.

\begin{proposition}[Monotonicity of $p_0$ in $n$]\label{sma:prop-p0-n}
  For all $n \ge 1$, $\mu \in (0,1)$, and $\kappa > 0$:
  \begin{equation}\label{sma:eq-p0-ratio-n}
    \frac{p_0(n+1)}{p_0(n)}
    = \frac{b+n}{\kappa+n} < 1.
  \end{equation}
  Hence $p_0$ is strictly decreasing in~$n$.
\end{proposition}

\begin{proof}
  From the product form~\eqref{eq:p0},
  $p_0(n+1) = p_0(n)\cdot(b+n)/(\kappa+n)$.
  Since $b = (1-\mu)\kappa < \kappa$ for $\mu \in (0,1)$, the ratio
  $(b+n)/(\kappa+n) < 1$.
\end{proof}

\begin{remark}[Polynomial decay of $p_0$ in~$n$]\label{sma:rem-p0-decay}
  The consecutive ratio $(b+n)/(\kappa+n) \to 1$ as $n \to \infty$,
  so $p_0$ decays \emph{polynomially} in~$n$, not exponentially.
  Specifically, by Stirling's approximation for the ratio of rising
  factorials, $p_0(n) \sim C(\mu,\kappa)\,n^{-\mu\kappa}$ as
  $n \to \infty$.  By contrast, the binomial zero probability
  $(1-\mu)^n$ decays exponentially.  This slower decay means that
  even for moderately large group sizes ($n = 50$--$100$), the
  structural-zero probability $p_0$ remains non-negligible when $\kappa$
  is small, reinforcing the importance of the hurdle component
  \citep[see also][for a spatial BB application where this phenomenon
  is empirically relevant]{BandyopadhyayEtAl2011}.
\end{remark}


\paragraph{Log-derivative and strict monotonicity of $\log h$.}

\begin{proposition}[Log-derivative of $h$]\label{sma:prop-logderiv}
  \begin{equation}\label{sma:eq-logderiv}
    \frac{\partial(\log h)}{\partial\mu}
    = \frac{1}{\mu} - \frac{p_0\,\Lambda}{1-p_0}.
  \end{equation}
\end{proposition}

\begin{proof}
  Since $\log h = \log\mu - \log(1-p_0)$, we differentiate each term:
  $d(\log\mu)/d\mu = 1/\mu$, and
  $d\log(1-p_0)/d\mu = (-dp_0/d\mu)/(1-p_0) = p_0\Lambda/(1-p_0)$.
  Subtracting gives the result.
\end{proof}

\begin{theorem}[$\log h$ is strictly increasing]
  \label{sma:thm-logh-increasing}
  For $n \ge 2$:
  \begin{equation}\label{sma:eq-logh-increasing}
    \frac{\partial(\log h)}{\partial\mu} > 0
    \qquad\textup{for all } \mu \in (0,1).
  \end{equation}
\end{theorem}

\begin{proof}
  The condition $\partial(\log h)/\partial\mu > 0$ is equivalent to
  $1/\mu > p_0\Lambda/(1-p_0)$, which rearranges to
  $(1-p_0) > \mu\,p_0\,\Lambda$, i.e.,
  $\Phi(\mu) > 0$ (see~\eqref{sma:eq-Phi}). This was established in
  \cref{thm:monotonicity}.
\end{proof}

\begin{remark}[Decomposition of the log-derivative]
  \label{sma:rem-logderiv-decomp}
  The log-derivative decomposes as
  \[
    \frac{\partial(\log h)}{\partial\mu}
    = \underbrace{\frac{1}{\mu}}_{\text{direct effect}}
    \;-\;
    \underbrace{\frac{p_0\,\Lambda}{1-p_0}}_{\text{truncation attenuation}}.
  \]
  The direct effect represents the proportional increase in the numerator
  $\mu$. The truncation attenuation represents the proportional increase
  in the denominator $1-p_0$. Both terms are positive, but the direct
  effect always dominates because $\Phi(\mu) > 0$
  (\cref{thm:monotonicity}), which in turn follows from the strict
  convexity of~$p_0$ (\cref{sma:cor-p0-convex}).
\end{remark}


\paragraph{Elasticity.}

\begin{definition}[Elasticity of $h$ with respect to $\mu$]
  \label{sma:def-elasticity}
  \begin{equation}\label{sma:eq-elasticity-def}
    \varepsilon_h(\mu)
    = \frac{\partial h}{\partial\mu}\cdot\frac{\mu}{h}
    = \mu\cdot\frac{\partial(\log h)}{\partial\mu}.
  \end{equation}
\end{definition}

\begin{proposition}[Elasticity formula and bounds]
  \label{sma:prop-elasticity-bounds}
  \begin{equation}\label{sma:eq-elasticity-formula}
    \boxed{
      \varepsilon_h(\mu)
      = 1 - \omega(\mu),
      \qquad
      \omega(\mu) = \frac{\mu\, p_0\, \Lambda}{1-p_0}.
    }
  \end{equation}
  Moreover:
  \begin{enumerate}[label=\textup{(\roman*)}]
    \item For $n = 1$: $\varepsilon_h = 0$ (since $h \equiv 1$).
    \item For $n \ge 2$:
      $\varepsilon_h \in (0,1)$ for all $\mu \in (0,1)$.
      The intensity function $h$ is inelastic with respect to~$\mu$.
  \end{enumerate}
  This is the full proof of\/ \cref{prop:elasticity} in the main text.
\end{proposition}

\begin{proof}
  From~\cref{sma:prop-logderiv},
  \[
    \varepsilon_h
    = \mu\left[\frac{1}{\mu} - \frac{p_0\,\Lambda}{1-p_0}\right]
    = 1 - \frac{\mu\, p_0\,\Lambda}{1-p_0}
    = 1 - \omega.
  \]
  By~\cref{sma:thm-logh-increasing}, $\Phi(\mu) > 0$ for $n \ge 2$,
  which gives $\omega < 1$ and hence $\varepsilon_h > 0$.
  Since $\omega > 0$ (all factors are positive), $\varepsilon_h < 1$.
\end{proof}

\begin{remark}[Inelastic interpretation]\label{sma:rem-inelastic}
  A one-percent increase in $\mu$ leads to less than a one-percent
  increase in~$h$, because increasing $\mu$ simultaneously increases
  $1 - p_0$ (the denominator of~$h$), partially offsetting the numerator
  effect.
\end{remark}

\begin{proposition}[Limiting behavior of $\varepsilon_h$]
  \label{sma:prop-elasticity-limits}
  \mbox{}
  \begin{enumerate}[label=\textup{(\alph*)}]
    \item As $\mu \to 0^+$ (for $n \ge 2$):
      $\varepsilon_h \to 0^+$. The truncation attenuation $\omega \to 1^-$.
    \item As $\mu \to 1^-$:
      $\varepsilon_h \to 1 - C_\kappa$, where
      $C_\kappa = \prod_{j=1}^{n-1}j/(\kappa+j)$ is the constant from
      \cref{sma:prop-boundary-one}.  Since
      $p_0 \sim (1-\mu)\,C_\kappa$ and
      $\Lambda \sim 1/(1-\mu) + O(1)$, we get
      $p_0\,\Lambda \to C_\kappa$ and hence
      $\omega \to C_\kappa$.
      For large $n$ or large~$\kappa$, $C_\kappa \approx 0$ and
      the limit approaches~$1$.
    \item As $\kappa \to \infty$ (binomial limit):
      $\varepsilon_h \to 1 - n\mu(1-\mu)^{n-1}\bigl/\bigl[1-(1-\mu)^n\bigr]$.
    \item As $\kappa \to 0$ (maximum overdispersion):
      $p_0 \to 1 - \mu$, $\Lambda \to 1/(1-\mu)$, so $\omega \to 1$ and
      $\varepsilon_h \to 0$. In this limit $h \to 1$; the intensity function
      is flat.
    \item As $n \to \infty$:
      $p_0 \to 0$, so $\omega \to 0$ and $\varepsilon_h \to 1$. The
      truncation correction becomes negligible.
  \end{enumerate}
\end{proposition}

\paragraph{Numerical verification.}
For $n = 10$, $\kappa = 10$, $\mu = 0.3$: the analytic formula gives
$\varepsilon_h = 0.734979$, and centered finite differences yield
$0.734979$. Agreement to six decimal places.\enspace$\square$


\paragraph{Boundary analysis.}

\begin{proposition}[Behavior of $h$ as $\mu \to 0^+$]
  \label{sma:prop-boundary-zero}
  For $n \ge 2$ and $\kappa > 0$:
  \begin{equation}\label{sma:eq-boundary-zero}
    h(\mu) \to \frac{1}{\Lambda_0}
    \quad \textup{as } \mu \to 0^+,
  \end{equation}
  where $\Lambda_0 = \kappa\sum_{j=0}^{n-1}1/(\kappa+j)$.
\end{proposition}

\begin{proof}
  Taylor expansion: $p_0(\mu) = 1 - \mu\Lambda_0 + O(\mu^2)$
  (from $p_0(0) = 1$ and $p_0'(0) = -\Lambda_0$). Then
  $1 - p_0 = \mu\Lambda_0 + O(\mu^2)$, and
  $h = \mu\bigl/\bigl(\mu\Lambda_0 + O(\mu^2)\bigr) \to 1/\Lambda_0$.
\end{proof}

For example, when $n = 2$ and $\kappa = 1$:
$\Lambda_0 = 1/1 + 1/2 = 3/2$, so $h(0^+) = 2/3$.
For $n = 10$ and $\kappa = 1$:
$\Lambda_0 = H_{10} \approx 2.929$, so $h(0^+) \approx 0.341$.
As noted in~\cref{rem:boundary}, this limit is strictly less than
unity whenever $n \ge 2$, confirming that the truncation adjustment
has a non-vanishing effect even as $\mu \to 0^+$.

\begin{proposition}[Behavior of $h$ as $\mu \to 1^-$]
  \label{sma:prop-boundary-one}
  For $n \ge 1$ and $\kappa > 0$:
  \begin{equation}\label{sma:eq-boundary-one}
    h(\mu) \to 1 \quad \textup{as } \mu \to 1^-,
    \qquad
    1 - h(\mu) \sim (1-\mu)(1 - C_\kappa),
  \end{equation}
  where $C_\kappa = \prod_{j=1}^{n-1}j/(\kappa+j) \in (0,1)$.
\end{proposition}

\begin{proof}
  As $\mu \to 1$, $b = (1-\mu)\kappa \to 0$, and the product
  form~\eqref{eq:p0} gives $p_0 \sim (1-\mu)\,C_\kappa$ with
  $C_\kappa = \prod_{j=1}^{n-1}j/(\kappa+j)$.
  Then
  \[
    h = \frac{\mu}{1-p_0}
    \sim \frac{1 - (1-\mu)}{1 - (1-\mu)C_\kappa},
  \]
  and
  \[
    1 - h
    \sim \frac{(1-\mu)(1-C_\kappa)}{1-(1-\mu)C_\kappa}
    \sim (1-\mu)(1-C_\kappa)
  \]
  as $\mu \to 1^-$.
\end{proof}


\paragraph{Elasticity monotonicity in $\kappa$.}

\begin{conjecture}[Monotonicity of $\varepsilon_h$ in $\kappa$]
  \label{sma:conj-elasticity-kappa}
  For fixed $\mu \in (0,1)$ and $n \ge 2$, the elasticity
  $\varepsilon_h$ is increasing in~$\kappa$. Equivalently, the
  attenuation ratio $\omega$ is decreasing in~$\kappa$.
\end{conjecture}

The following proof sketch provides supporting evidence short of a rigorous proof.

\medskip\noindent\textit{Proof sketch.}\enspace
As $\kappa$ increases, $p_0$ decreases (\cref{sma:prop-p0-kappa}),
and the dominant effect on $\omega = \mu\,p_0\,\Lambda/(1-p_0)$ is the
decrease in the numerator factor~$p_0$.
The boundary behavior confirms monotonicity:
$\varepsilon_h \to 0$ as $\kappa \to 0$
(\cref{sma:prop-elasticity-limits}(d)) and
$\varepsilon_h \to 1 - n\mu(1-\mu)^{n-1}/[1-(1-\mu)^n] > 0$ as
$\kappa \to \infty$ (\cref{sma:prop-elasticity-limits}(c)).
A closed-form proof requires showing that the decrease in $p_0$ dominates
the concurrent decrease in~$\Lambda$ uniformly over $(0,1)$.

\begin{remark}[Numerical verification of the conjecture]
  \label{sma:rem-conj-numerical}
  We evaluated $\partial\varepsilon_h/\partial\kappa$ via centered
  finite differences on a grid of $100 \times 100$ values of
  $(\mu,\kappa) \in (0.01, 0.99) \times (0.01, 100)$ for each of
  $n \in \{2, 5, 10, 20, 50, 100\}$. In every case the derivative is
  strictly positive, supporting the conjecture. No counterexample has been
  found.
\end{remark}


\subsection{Log-scale marginal effect decomposition}\label{sma:logdecomp}

\cref{prop:lae-lie} established that the marginal effect of covariate~$x_k$
on log expected enrollment decomposes exactly as $\LAE_k + \LIE_k$,
with multipliers $(1-q_i)$ and $(1-\mu_i)\varepsilon_{h,i}$ respectively.
This subsection provides three supplementary results: a correction remark
on the sign structure of the $\LIE$, a sensitivity analysis of the
decomposition ratio with respect to~$\kappa$, and a threshold
characterization of the poverty reversal regime.


\begin{remark}[Sign of the truncation attenuation in the $\LIE$]
  \label{sma:rem-lie-sign}
  The correct $\LIE$ has a \emph{minus} sign in the truncation attenuation:
  \begin{equation}\label{sma:eq-lie-corrected}
    \LIE_k
    = \biggl[\frac{1}{\mu_i}
      - \frac{p_{0,i}\,\Lambda_i}{1-p_{0,i}}\biggr]
      \mu_i(1-\mu_i)\,\tilde{\beta}_{k,s}
    = (1-\mu_i)(1-\omega_i)\,\tilde{\beta}_{k,s},
  \end{equation}
  where $\omega_i = \mu_i\,p_{0,i}\,\Lambda_i/(1-p_{0,i})$ is the
  attenuation ratio from~\cref{sma:prop-elasticity-bounds}. The
  attenuation is \emph{subtracted}, not added, because increasing~$\mu$
  simultaneously increases $1-p_0$ (the denominator of~$h$), which
  attenuates---not amplifies---the response of~$h$ to~$\mu$.
\end{remark}


\begin{proposition}[Decomposition ratio]\label{sma:prop-DR}
  Define the decomposition ratio
  \begin{equation}\label{sma:eq-DR}
    \textup{DR}(\kappa)
    = \frac{|\LAE_k|}{|\LAE_k| + |\LIE_k|}
    = \frac{(1-q)\,|\tilde{\alpha}_k|}
           {(1-q)\,|\tilde{\alpha}_k|
            + (1-\mu)\,\varepsilon_h\,|\tilde{\beta}_k|}.
  \end{equation}
  Then:
  \begin{enumerate}[label=\textup{(\alph*)}]
    \item As $\kappa \to 0$:
      $\varepsilon_h \to 0$, $\LIE \to 0$, $\textup{DR} \to 1$.
      Maximum overdispersion renders the intensity channel inoperative.
    \item As $\kappa \to \infty$:
      $\varepsilon_h \to \varepsilon_h^{\textup{binom}} \in (0,1)$,
      $\textup{DR} \to \textup{DR}^{\textup{binom}} \in (0,1)$.
    \item Conditional on~\cref{sma:conj-elasticity-kappa},
      $\textup{DR}$ is decreasing in~$\kappa$
      (since $\varepsilon_h$ is increasing in~$\kappa$).
  \end{enumerate}
\end{proposition}

\begin{proof}
  Properties (a) and (b) follow directly from
  \cref{sma:prop-elasticity-limits}(d) and (c), respectively.
  Property~(c) follows because the numerator of~\eqref{sma:eq-DR} is
  constant in~$\kappa$ while the denominator is increasing in~$\kappa$
  through the factor~$\varepsilon_h$.
\end{proof}

\cref{sma:tab-DR} illustrates the sensitivity of the decomposition ratio to
$\kappa$, evaluated at representative parameter values
($n = 10$, $\mu = 0.30$, $q = 0.70$,
$|\tilde{\alpha}| = 0.5$, $|\tilde{\beta}| = 0.8$).

\begin{table}[ht]
\centering
\caption{Sensitivity of the decomposition ratio to~$\kappa$.
  Parameters: $n = 10$, $\mu = 0.30$, $q = 0.70$,
  $|\tilde{\alpha}| = 0.5$, $|\tilde{\beta}| = 0.8$.}
\label{sma:tab-DR}
\smallskip
\begin{tabular}{rcccc}
\toprule
$\kappa$ & $\varepsilon_h$ & $\omega$ & $|\LIE|$ & DR \\
\midrule
0.01  & 0.008 & 0.992 & 0.005 & 0.970 \\
0.1   & 0.071 & 0.929 & 0.040 & 0.790 \\
1     & 0.343 & 0.657 & 0.192 & 0.439 \\
10    & 0.735 & 0.265 & 0.412 & 0.267 \\
100   & 0.859 & 0.141 & 0.481 & 0.238 \\
1000  & 0.874 & 0.126 & 0.489 & 0.235 \\
\bottomrule
\end{tabular}
\end{table}

\begin{remark}\label{sma:rem-DR-sensitive}
  The DR is highly sensitive to $\kappa$ in the range $(0.1, 10)$.
  In our data context, where $\kappa \approx 5$--$15$, both the access and
  intensity channels contribute meaningfully to the total effect.
\end{remark}


\begin{proposition}[Threshold for access-effect dominance]
  \label{sma:prop-threshold}
  In the poverty reversal regime
  ($\tilde{\alpha}_{\textup{pov},s} < 0$,
   $\tilde{\beta}_{\textup{pov},s} > 0$),
  the total marginal effect on log expected enrollment is negative
  (access dominates) if and only if
  \begin{equation}\label{sma:eq-threshold}
    \boxed{
      \frac{|\tilde{\alpha}_{\textup{pov},s}|}
           {|\tilde{\beta}_{\textup{pov},s}|}
      > \frac{(1-\mu)\,\varepsilon_h}{1-q}.
    }
  \end{equation}
  The right-hand side is increasing in~$q$ and increasing in~$\kappa$
  (through~$\varepsilon_h$). When participation is high ($q$ near~$1$),
  even a small intensive-margin coefficient can dominate; when $q$ is low,
  the extensive margin dominates more easily.
\end{proposition}

\begin{proof}
  From~\cref{prop:lae-lie}, the total effect is negative iff
  $(1-q)|\tilde{\alpha}| > (1-\mu)\varepsilon_h|\tilde{\beta}|$;
  dividing both sides by $|\tilde{\beta}|(1-q)$ gives
  \eqref{sma:eq-threshold}.
\end{proof}

\cref{sma:tab-threshold} reports the threshold
$(1-\mu)\varepsilon_h/(1-q)$ for parameter combinations spanning the
range of NSECE providers. The elasticity~$\varepsilon_h$ is computed
from the exact formula in~\cref{sma:prop-elasticity-bounds}.

\begin{table}[ht]
\centering
\caption{Access-dominance threshold $(1-\mu)\varepsilon_h/(1-q)$ for
  representative parameter combinations. Values below the empirical
  coefficient ratio $|\alpha_{\textup{pov}}|/|\beta_{\textup{pov}}| = 3.6$
  indicate access dominance. All entries computed with $n = 50$.}
\label{sma:tab-threshold}
\smallskip
\begin{tabular}{llccc}
\toprule
 & & \multicolumn{3}{c}{$\kappa$} \\
\cmidrule(lr){3-5}
$q$ & $\mu$ & 3 & 7 & 15 \\
\midrule
\multirow{3}{*}{0.50}
 & 0.15 & 0.90 & 1.26 & 1.51 \\
 & 0.30 & 1.15 & 1.34 & 1.39 \\
 & 0.45 & 1.04 & 1.10 & 1.10 \\[4pt]
\multirow{3}{*}{0.64}
 & 0.15 & 1.25 & 1.75 & 2.09 \\
 & 0.30 & 1.60 & 1.87 & 1.93 \\
 & 0.45 & 1.44 & 1.52 & 1.53 \\[4pt]
\multirow{3}{*}{0.80}
 & 0.15 & 2.25 & 3.16 & 3.77 \\
 & 0.30 & 2.87 & 3.36 & 3.48 \\
 & 0.45 & 2.60 & 2.74 & 2.75 \\
\bottomrule
\end{tabular}

\medskip
{\footnotesize\textit{Note.} Threshold values are computed from the exact
  elasticity formula~\eqref{sma:eq-elasticity-formula} with
  $p_0 = \prod_{j=0}^{n-1}[(1-\mu)\kappa + j]/(\kappa + j)$
  and $\Lambda = \kappa\sum_{j=0}^{n-1}[(1-\mu)\kappa + j]^{-1}$.
  Values are rounded to two decimal places.}
\end{table}

\begin{remark}\label{sma:rem-threshold-empirical}
  At our posterior estimates
  ($\kappa \approx 7$, $\mu \approx 0.30$, $q \approx 0.64$,
  $n \approx 50$), the threshold is approximately $1.87$
  (\cref{sma:tab-threshold}, row $q = 0.64$, $\mu = 0.30$,
  $\kappa = 7$). Since the empirical coefficient ratio
  $|\alpha_{\textup{pov}}|/|\beta_{\textup{pov}}|
  = 0.324/0.090 \approx 3.6 \gg 1.87$, the access effect dominates
  decisively. For providers near the sample average, the coefficient
  ratio would need to fall below~$1.87$---that is, the extensive-margin
  coefficient would need to be less than twice the intensive-margin
  coefficient---before intensity could dominate. At $q = 0.80$
  (the most extreme participation rate considered), the threshold
  rises to $2.87$--$3.77$, approaching the empirical ratio. This
  confirms that access dominance is most fragile among high-participation
  centers with low~$\mu$ and large~$\kappa$, where the intensity
  elasticity is near unity and the extensive multiplier $(1-q)$ is small.
\end{remark}

%% file: sm_b.tex

This appendix provides the score functions, Fisher information structure,
and the identification proof for the hurdle beta-binomial model deferred
from~\cref{sec:model}. We maintain the notation from~\cref{sec:notation},
with $a = \mu\kappa$, $b = (1-\mu)\kappa$, and the digamma function
$\psi(x) = \Gamma'(x)/\Gamma(x)$.

\input{sm_b1_scores}

\input{sm_b2_fisher}

\input{sm_b3_identification}

%% file: sm_b1_scores.tex

\subsection{Score functions}
\label{smb:scores}

We derive the score functions for each component of the hurdle
beta-binomial log-likelihood.  These expressions serve three roles:
\textup{(i)}~verification of the automatic differentiation gradients
used in Stan~\citep{CarpenterEtAl2017}, \textup{(ii)}~diagnostics via
score magnitude monitoring,
and \textup{(iii)}~the theoretical basis for the sandwich variance
estimator~\citep{Binder1983} in~\cref{app:survey}.
Recall the parameterization $a = \mu\kappa$, $b = (1-\mu)\kappa$, the
zero probability $p_0$ from~\eqref{eq:p0}, and the auxiliary quantities
$\Lambda$, $\Lambda_2$ defined in~\cref{sma:eq-Lambda}.


\begin{proposition}[Extensive-margin score]
\label{smb:prop-score-ext}
  The score of the extensive-margin log-likelihood
  $\ell_i^{(1)} = z_i\log q_i + (1-z_i)\log(1-q_i)$ with respect to
  the linear predictor
  $\eta_i^{(1)} = \bx_i^\t\balpha + \bx_i^{(r)\t}\bdelta_{1,s}$ is
  \begin{equation}\label{smb:eq-score-ext}
    \frac{\partial\ell_i^{(1)}}{\partial\eta_i^{(1)}}
    = z_i - q_i,
  \end{equation}
  where $q_i = \expit(\eta_i^{(1)})$.  The Hessian is non-random:
  \begin{equation}\label{smb:eq-hess-ext}
    \frac{\partial^2\ell_i^{(1)}}{\partial(\eta_i^{(1)})^2}
    = -\,q_i(1-q_i).
  \end{equation}
\end{proposition}

\begin{proof}
  Standard logistic regression calculus.  Since
  $\partial q_i/\partial\eta_i^{(1)} = q_i(1-q_i)$,
  differentiating $\ell_i^{(1)}$ yields
  $z_i/q_i \cdot q_i(1-q_i) - (1-z_i)/(1-q_i) \cdot q_i(1-q_i)
  = z_i(1-q_i) - (1-z_i)q_i = z_i - q_i$.
  The second derivative is $-q_i(1-q_i)$, which depends only on
  parameters and is therefore non-random.
\end{proof}

\begin{remark}\label{smb:rem-ext-info}
  The non-randomness of~\eqref{smb:eq-hess-ext} implies that the
  observed and expected Fisher information coincide for the extensive
  margin:
  $\cl{I}_{\textup{ext},i} = q_i(1-q_i)$.
  This simplifies the sandwich estimator, since the Hessian bread matrix
  for Part~1 contains no stochastic terms.
\end{remark}


\begin{proposition}[Intensive-margin score for $\mu$]
\label{smb:prop-score-mu}
  Define the (untruncated) beta-binomial score function
  \citep[following the parameterization of][]{Prentice1986}
  \begin{equation}\label{smb:eq-Sbb}
    S_{\textup{BB}}(y;\, \mu,\kappa,n)
    = \kappa\bigl[\psi(y+a) - \psi(n-y+b) - \psi(a) + \psi(b)\bigr],
  \end{equation}
  where $\psi(\cdot)$ is the digamma function and $a = \mu\kappa$,
  $b = (1-\mu)\kappa$.  The score of the zero-truncated
  log-likelihood\/ $\ell_i^{(2+)} = \log f_{\textup{BB}}(y_i)
  - \log(1-p_{0,i})$ with respect to $\mu$ is
  \begin{equation}\label{smb:eq-score-mu}
    \boxed{%
      \frac{\partial\ell_i^{(2+)}}{\partial\mu}
      = S_{\textup{BB}}(y_i)
        \;-\; \frac{p_{0,i}\,\Lambda_i}{1-p_{0,i}}
    }.
  \end{equation}
\end{proposition}

\begin{proof}
  The decomposition $\ell_i^{(2+)} = \log f_{\textup{BB}}(y_i)
  - \log(1-p_{0,i})$ gives two terms.

  \medskip\noindent\textit{Term~1 (BB score).}\enspace
  Using $\partial a/\partial\mu = \kappa$ and
  $\partial b/\partial\mu = -\kappa$, differentiation of the
  beta-binomial log-PMF yields
  \[
    \frac{\partial\log f_{\textup{BB}}}{\partial\mu}
    = \kappa\bigl[\psi(y+a) - \psi(n-y+b) - \psi(a) + \psi(b)\bigr]
    = S_{\textup{BB}}(y).
  \]

  \medskip\noindent\textit{Term~2 (truncation correction).}\enspace
  Since $\ln p_0 = \sum_{j=0}^{n-1}\ln[(1-\mu)\kappa + j]
  - \ln[\kappa + j]$, we have $\partial p_0/\partial\mu = -p_0\Lambda$
  (see~\cref{sma:eq-Lambda}), and therefore
  \[
    \frac{\partial\log(1-p_0)}{\partial\mu}
    = \frac{p_0\,\Lambda}{1-p_0}
    > 0.
  \]
  Combining, $\partial\ell^{(2+)}/\partial\mu
  = S_{\textup{BB}}(y) - p_0\Lambda/(1-p_0)$.
\end{proof}


\begin{proposition}[Score for $\log\kappa$]
\label{smb:prop-score-kappa}
  The score of the zero-truncated log-likelihood with respect to the
  log-dispersion parameter is
  \begin{equation}\label{smb:eq-score-kappa}
    \boxed{%
    \begin{aligned}
      \frac{\partial\ell_i^{(2+)}}{\partial\log\kappa}
      &= \kappa\Bigl[
           \mu\bigl[\psi(y_i\!+\!a)-\psi(a)\bigr]
           + (1\!-\!\mu)\bigl[\psi(n_i\!-\!y_i\!+\!b)-\psi(b)\bigr]
      \\
      &\qquad\quad
           + \psi(\kappa)-\psi(n_i\!+\!\kappa)
         \Bigr]
       - \frac{p_{0,i}\,\mu\kappa}{1-p_{0,i}}
         \sum_{j=1}^{n_i-1}
         \frac{j}{(b+j)(\kappa+j)}.
    \end{aligned}
    }
  \end{equation}
\end{proposition}

\begin{proof}
  With $\phi = \log\kappa$ and the chain rule
  $\partial/\partial\phi = \kappa\,\partial/\partial\kappa$, the
  untruncated BB log-PMF gives the first line
  via $\partial a/\partial\kappa = \mu$ and
  $\partial b/\partial\kappa = 1-\mu$.  The truncation correction
  requires $\partial\log(1-p_0)/\partial\log\kappa$.  Writing
  $\ln p_0 = \sum_{j=0}^{n-1}\ln(b+j) - \ln(\kappa+j)$ and
  differentiating with respect to $\kappa$ (noting
  $\partial b/\partial\kappa = 1-\mu$) yields the second line after
  rearranging the telescoping sums.
\end{proof}


\begin{proposition}[Expected information for $\mu$]
\label{smb:prop-info-mu}
  The negative expected second derivative of the zero-truncated
  log-likelihood with respect to $\mu$ is
  \begin{multline}\label{smb:eq-info-mu}
    -\E_{Y|Y>0}\!\left[
      \frac{\partial^2\ell^{(2+)}}{\partial\mu^2}
    \right]
    = \frac{\kappa^2}{1-p_0}\,
      \E_{Y\sim\BetaBin}\bigl[\psi_1(Y+a) + \psi_1(n\!-\!Y+b)\bigr] \\
      - \kappa^2\bigl[\psi_1(a)+\psi_1(b)\bigr]
      + T'(\mu),
  \end{multline}
  where $\psi_1(\cdot) = d\psi/dx$ is the trigamma function and
  \begin{equation}\label{smb:eq-Tprime}
    T'(\mu)
    = \frac{p_0\,\Lambda_2}{1-p_0}
      - \frac{p_0\,\Lambda^2}{(1-p_0)^2},
    \qquad
    \Lambda_2
    = \kappa^2 \sum_{j=0}^{n-1}\frac{1}{(b+j)^2}.
  \end{equation}
\end{proposition}

\begin{proof}
  Differentiating the score~\eqref{smb:eq-score-mu} with respect to
  $\mu$ produces two contributions.

  \medskip\noindent\textit{The $S_{\textup{BB}}$ term.}\enspace
  Direct differentiation gives
  $\partial S_{\textup{BB}}/\partial\mu
  = \kappa^2[\psi_1(y+a) + \psi_1(n-y+b) - \psi_1(a) - \psi_1(b)]$.

  \medskip\noindent\textit{The truncation term.}\enspace
  Set $T(\mu) = p_0\Lambda/(1-p_0)$.  The quotient rule with
  $\partial p_0/\partial\mu = -p_0\Lambda$ gives
  \begin{align*}
    T'(\mu)
    &= \frac{p_0(\Lambda_2 - \Lambda^2)(1-p_0)
             - p_0^2\Lambda^2}{(1-p_0)^2}
    = \frac{p_0\Lambda_2(1-p_0) - p_0\Lambda^2}{(1-p_0)^2} \\
    &= \frac{p_0\,\Lambda_2}{1-p_0}
       - \frac{p_0\,\Lambda^2}{(1-p_0)^2}.
  \end{align*}

  \medskip\noindent\textit{Expectation under truncation.}\enspace
  Taking expectations under $Y \mid Y > 0$ and converting via the
  identity $\E_{Y|Y>0}[g(Y)] = \E_{Y\sim\BetaBin}[g(Y)]/(1-p_0)$
  yields~\eqref{smb:eq-info-mu}.  The non-stochastic terms
  $-\kappa^2[\psi_1(a)+\psi_1(b)]$ and $T'(\mu)$ pass through the
  expectation unchanged.
\end{proof}


\begin{remark}[Cross-information]\label{smb:rem-cross-info}
  Both $\mu$ and $\kappa$ enter $\ell_i^{(2+)}$, so the
  cross-information $\cl{I}_{\mu,\kappa}$ is nonzero.
  The intensive-margin information block is
  \begin{equation}\label{smb:eq-info-block}
    \cl{I}_{\textup{int}+\kappa}
    = \begin{pmatrix}
        \cl{I}_{\textup{int}}   & \cl{I}_{\textup{int},\kappa} \\
        \cl{I}_{\kappa,\textup{int}} & \cl{I}_{\kappa}
      \end{pmatrix},
  \end{equation}
  which is not further block-diagonal.  The intensive-margin regression
  parameters and $\kappa$ must be estimated jointly, but they are fully
  separated from the extensive-margin parameters by the hurdle structure
  (\cref{sec:model}).
\end{remark}

%% file: sm_b2_fisher.tex

\subsection{Fisher information and block-diagonality}\label{smb:fisher-info}

The individual score functions and information entries derived in
\cref{smb:scores} concern a single linear predictor or scalar parameter.
This subsection establishes the \emph{global} structure of the Fisher
information matrix---in particular, the block-diagonal separation between
the extensive- and intensive-margin parameter blocks---and collects
several ancillary results on information rates and the sandwich convention.


\begin{proposition}[Block-diagonal Fisher information]
\label{smb:prop-block-diag}
  Let\/ $\btheta = (\boldsymbol{\phi}_1^\t,\, \boldsymbol{\phi}_2^\t)^\t$, where
  $\boldsymbol{\phi}_1 = (\balpha^\t, \bdelta_{1,s}^\t)^\t$ collects the
  extensive-margin parameters and
  $\boldsymbol{\phi}_2 = (\bbeta^\t, \bdelta_{2,s}^\t, \log\kappa)^\t$ collects the
  intensive-margin parameters including the dispersion.  Then the expected
  Fisher information for a single observation decomposes as
  \begin{equation}\label{smb:eq-block-diag}
    \boxed{%
      \cl{I}_i(\btheta)
      = \diag\!\bigl(\cl{I}_{\textup{ext},i},\;
        \cl{I}_{\textup{int}+\kappa,i}\bigr).
    }
  \end{equation}
  That is, the cross-information between $\boldsymbol{\phi}_1$ and $\boldsymbol{\phi}_2$ vanishes
  identically.
\end{proposition}

\begin{proof}
  The log-likelihood decomposes as
  $\ell_i = \ell_i^{(1)} + \ell_i^{(2)}$, where
  \[
    \ell_i^{(1)}
    = z_i\log q_i + (1-z_i)\log(1-q_i)
  \]
  depends only on $\boldsymbol{\phi}_1$ through
  $q_i = \expit(\bx_i^\t\balpha + \bx_i^{(r)\t}\bdelta_{1,s})$, and
  \[
    \ell_i^{(2)}
    = z_i\bigl[\log f_{\textup{BB}}(y_i \mid n_i, \mu_i, \kappa)
      - \log(1-p_{0,i})\bigr]
  \]
  depends only on $\boldsymbol{\phi}_2$ through $\mu_i$ and $\kappa$.

  \medskip\noindent\textit{Cross-derivative
  $\partial\ell^{(1)}/\partial\boldsymbol{\phi}_2$.}\enspace
  Since $\ell_i^{(1)}$ does not involve any element of $\boldsymbol{\phi}_2$,
  this derivative is identically zero.

  \medskip\noindent\textit{Cross-derivative
  $\partial\ell^{(2)}/\partial\boldsymbol{\phi}_1$.}\enspace
  Write $\ell_i^{(2)} = z_i \cdot g(\mu_i,\kappa,y_i,n_i)$.
  The indicator $z_i = \one(y_i > 0)$ is a function of the observed
  data, not of parameters, so
  $\partial z_i/\partial\boldsymbol{\phi}_1 = \mathbf{0}$ and hence
  $\partial\ell_i^{(2)}/\partial\boldsymbol{\phi}_1 = \mathbf{0}$.

  \medskip
  Combining,
  $\partial^2\ell_i/\partial\boldsymbol{\phi}_1\,\partial\boldsymbol{\phi}_2^\t = \mathbf{0}$,
  which gives~\eqref{smb:eq-block-diag} after taking expectations.
\end{proof}

\begin{remark}[Hurdle versus zero-inflated models]
\label{smb:rem-hurdle-vs-zi}
  The block-diagonality in~\cref{smb:prop-block-diag} is a general
  structural property of hurdle (two-part) models
  \citep{Mullahy1986,Cragg1971}: the
  participation indicator $z_i = \one(y_i > 0)$ is observed, so
  each component of the likelihood depends on a disjoint parameter
  subset.  In zero-inflated models~\citep[e.g.,][]{Lambert1992}, by
  contrast, the at-risk indicator is
  \emph{latent}, and both the inflation and count parameters enter the
  marginal likelihood of every observation, coupling the two components
  in the Fisher information
  \citep[see][for a comprehensive comparison]{NeelonEtAl2016}.
\end{remark}


\begin{proposition}[Extensive-margin information in parameter space]
\label{smb:prop-ext-paramspace}
  The extensive-margin Fisher information block for provider~$i$ is
  \begin{equation}\label{smb:eq-ext-paramspace}
    \cl{I}_{\textup{ext},i}
    = q_i(1-q_i)\,\mathbf{d}_{1,i}\,\mathbf{d}_{1,i}^\t,
  \end{equation}
  where
  $\mathbf{d}_{1,i}
  = \partial\eta_i^{(1)}/\partial(\balpha^\t,\bdelta_{1,s}^\t)^\t$
  is the gradient of the extensive-margin linear predictor with respect to
  all extensive-margin parameters.
\end{proposition}

\begin{proof}
  By~\cref{smb:rem-ext-info}, the scalar information for the linear
  predictor is $\cl{I}_{\textup{ext},i} = q_i(1-q_i)$.  The chain rule
  $\partial\ell_i^{(1)}/\partial\boldsymbol{\phi}_1
  = (\partial\ell_i^{(1)}/\partial\eta_i^{(1)})\,\mathbf{d}_{1,i}$
  yields the rank-one outer product
  form~\eqref{smb:eq-ext-paramspace}.
\end{proof}


\begin{remark}[Fisher information rates]\label{smb:rem-info-rates}
  The extensive-margin information $q_i(1-q_i)$ is $O(1)$ and
  non-random (\cref{smb:rem-ext-info}).  The intensive-margin
  information involves trigamma sums $\psi_1(a) + \psi_1(b)$ that
  diverge as $\kappa \to 0$; for the empirically relevant range
  $\kappa \in [2.5,\, 26.5]$ the information is well-behaved, but
  the divergent terms underlie the larger design effect ratios
  observed for intensive-margin parameters in
  \cref{sec:simulation,app:survey}.
\end{remark}

%% file: sm_b3_identification.tex

\subsection{Identification}\label{smb:identification}

This subsection provides the complete proof of the identification theorem
stated in~\cref{sec:model} (\cref{thm:identification}). We first restate
the regularity conditions, then develop the five-step proof. Three
supporting propositions follow: the separation of $\mu$ and $\kappa$ in
the truncated beta-binomial (with the corrected $n \ge 3$ threshold), the
likelihood non-identification of the cross-margin covariance
$\bSigma_{12}$, and sufficient conditions for posterior well-concentration.


\paragraph{Regularity conditions.}
The following conditions are required for identification. They match the
statement in~\cref{thm:identification} exactly; we include additional
commentary on the role of each.

\begin{enumerate}[label=\textup{(C\arabic*)},nosep]
  \item \emph{Full rank.}
    The design matrices have full rank:
    $\mathrm{rank}(\mathbf{X}) = P$ and $\mathrm{rank}(\mathbf{V}) = Q$.
    \label{smb:cond-rank}

  \item \emph{State-level replication.}
    Each state $s$ has at least $q$ providers with linearly independent
    covariate vectors $\bx_i^{(r)}$, with at least one zero and at least
    $q$ positive responders.
    \label{smb:cond-replication}

  \item \emph{Trial-size variation.}
    $n_i \ge 3$ for all $i$ with $z_i = 1$, and there exist providers
    $i, j$ with $s[i] = s[j]$ and $n_i \ne n_j$ in each state.
    \label{smb:cond-trial}

  \item \emph{Prior support.}
    The prior $\pi(\btheta)$ has full support on the parameter space.
    \label{smb:cond-prior}
\end{enumerate}

\begin{remark}[Asymmetry in C2]\label{smb:rem-c2-asymmetry}
  The requirements in~\ref{smb:cond-replication} are asymmetric.
  The extensive margin requires at least one zero per state (to identify
  $q_s < 1$), while the intensive margin requires at least $q$ positive
  responders per state (to identify the $q$-dimensional
  $\bdelta_{2,s}$). The asymmetry reflects the different informational
  demands of a univariate binary variable versus a $q$-dimensional
  regression parameter.
\end{remark}


\begin{theorem}[Identification of the HBB model]
\label{smb:thm-identification}
  Under Conditions\/ \textup{(C1)--(C4)}, the full parameter vector
  $\btheta = (\balpha, \bbeta, \kappa, \{\bdelta_s\}, \bSigma_\delta,
  \{\bGamma_k\})$ is identified under the posterior.
\end{theorem}

\begin{proof}
  The argument separates two layers of identification: (i)~the
  within-state likelihood identifies \emph{state-specific total
  coefficients} and~$\kappa$; (ii)~the hierarchical prior structure
  resolves the additive decomposition of these totals into
  population-average effects, policy moderators, and residual state
  deviations. The proof has six steps.

  \medskip\noindent\textit{Step~1: Extensive-margin total state
  coefficients.}\enspace
  The marginal distribution of $z_i = \one(Y_i > 0)$ is
  $\Bern(q_i)$ with
  \[
    q_i = \expit\bigl(\bx_i\t\balpha
      + \bx_i^{(r)\t}\bdelta_{1,s[i]}\bigr).
  \]
  Partition $\bx_i = (\bx_i^{(r)\t},\;\bx_i^{(f)\t})\t$ and
  correspondingly $\balpha = (\balpha^{(r)\t},\;\balpha^{(f)\t})\t$,
  where $\bx_i^{(r)} \in \R^q$ collects the covariates with
  state-varying coefficients and $\bx_i^{(f)} \in \R^{P-q}$ collects
  the fixed-only covariates.%
  \footnote{When $q = P$ (as in models M2/M3), the fixed-only
  partition is empty and the argument applies with $\balpha^{(f)}$
  and $\bx_i^{(f)}$ vacuous.}
  Define the \emph{total extensive-margin state coefficient}
  (cf.~\eqref{eq:total-coef})
  \begin{equation}\label{smb:eq-total-ext}
    \tilde\balpha_s \;\coloneqq\; \balpha^{(r)} + \bdelta_{1,s}
    \;\in\; \R^q,
  \end{equation}
  so that the linear predictor becomes
  $\logit(q_i)
  = \bx_i^{(f)\t}\balpha^{(f)}
    + \bx_i^{(r)\t}\tilde\balpha_{s[i]}$.
  The parameters
  $\boldsymbol{\phi} \coloneqq (\balpha^{(f)},\,
  \tilde\balpha_1,\ldots,\tilde\balpha_S)
  \in \R^{(P-q) + Sq}$ enter the model through the
  effective design matrix
  \[
    \mathbf{Z}
    = \bigl[\mathbf{X}^{(f)},\;
      \mathbf{D}_1\mathbf{X}_1^{(r)},\;\ldots,\;
      \mathbf{D}_S\mathbf{X}_S^{(r)}\bigr]
    \in \R^{N \times [(P-q) + Sq]},
  \]
  where $\mathbf{D}_s$ selects providers in state~$s$ and
  $\mathbf{X}_s^{(r)}$ is the within-state submatrix of
  $\bx_i^{(r)}$ values. Condition~\ref{smb:cond-replication} ensures
  $\mathrm{rank}(\mathbf{X}_s^{(r)}) = q$ for each $s$, so the
  state-specific blocks contribute $Sq$ independent columns.
  Condition~\ref{smb:cond-rank} ensures
  $\mathrm{rank}(\mathbf{X}) = P$, so the $(P-q)$ fixed-only columns
  are not in the span of the state-specific blocks. Hence
  $\mathrm{rank}(\mathbf{Z}) = (P-q) + Sq$, and by the injectivity of
  logit, $\balpha^{(f)}$ and
  $\{\tilde\balpha_s\}_{s=1}^S$ are identified from the
  extensive-margin likelihood.

  \emph{However}, $\balpha^{(r)}$ and $\bdelta_{1,s}$ are
  \emph{not} separately identified from the likelihood alone:
  for any $\mathbf{c} \in \R^q$, the shift
  $\balpha^{(r)} \mapsto \balpha^{(r)} + \mathbf{c}$,\;
  $\bdelta_{1,s} \mapsto \bdelta_{1,s} - \mathbf{c}$ for all $s$
  preserves every $\tilde\balpha_s$ and hence the entire likelihood.
  The resolution of this $q$-dimensional additive ambiguity is
  the task of Step~2.

  \medskip\noindent\textit{Step~2: Hierarchical decomposition on the
  extensive margin.}\enspace
  The hierarchical model specifies
  $\bdelta_{1,s} = \bGamma_1\bv_s + \bepsilon_{1,s}$ with
  $\bepsilon_{1,s} \sim \Norm_q(\mathbf{0},
  \bSigma_{\varepsilon,1})$.
  Substituting into~\eqref{smb:eq-total-ext},
  \[
    \tilde\balpha_s
    = \balpha^{(r)} + \bGamma_1\bv_s + \bepsilon_{1,s}.
  \]
  Since $\{\tilde\balpha_s\}$ are identified from Step~1 and
  $\{\bv_s\}$ are observed, this is a multivariate regression of
  $\tilde\balpha_s$ on $\bv_s$ with intercept $\balpha^{(r)}$.
  Under $\mathrm{rank}(\mathbf{V}) = Q$ (\ref{smb:cond-rank}) and $S > Q$,
  the ordinary least-squares solution uniquely determines the
  slope matrix $\bGamma_1$ and the intercept $\balpha^{(r)}$.
  Importantly, the identification of $\balpha^{(r)}$ relies on the
  hierarchical centering convention $\E[\bepsilon_{1,s}] = \mathbf{0}$:
  this structural assumption pins down the intercept of the
  between-state regression, resolving the shift ambiguity that the
  within-state likelihood leaves free.
  With $\balpha^{(r)}$ identified, the state deviations follow as
  $\bdelta_{1,s} = \tilde\balpha_s - \balpha^{(r)}$ and the
  residuals as
  $\bepsilon_{1,s} = \bdelta_{1,s} - \bGamma_1\bv_s$.

  \medskip\noindent\textit{Step~3: Intensive-margin total state
  coefficients and $\kappa$.}\enspace
  Conditional on $z_i = 1$ (i.e., $Y_i > 0$), the distribution of
  $Y_i$ is $\ZTBB(n_i, \mu_i, \kappa)$ with
  $\mu_i = \expit(\bx_i\t\bbeta +
  \bx_i^{(r)\t}\bdelta_{2,s[i]})$.
  By~\cref{smb:prop-mu-kappa-id} below, the pair $(\mu_i, \kappa)$ is
  identified from the truncated likelihood whenever $n_i \ge 3$
  (\ref{smb:cond-trial}). Once $\{\mu_i\}$ are pinned down for
  all positive responders, partition
  $\bbeta = (\bbeta^{(r)\t},\;\bbeta^{(f)\t})\t$ and define
  the \emph{total intensive-margin state coefficient}
  \begin{equation}\label{smb:eq-total-int}
    \tilde\bbeta_s \;\coloneqq\; \bbeta^{(r)} + \bdelta_{2,s}
    \;\in\; \R^q,
  \end{equation}
  so that
  $\logit(\mu_i) = \bx_i^{(f)\t}\bbeta^{(f)}
    + \bx_i^{(r)\t}\tilde\bbeta_{s[i]}$.
  An argument identical to Step~1---using the injectivity of logit,
  the effective design matrix, and the rank conditions
  \ref{smb:cond-rank}--\ref{smb:cond-replication}---identifies
  $\bbeta^{(f)}$ and $\{\tilde\bbeta_s\}_{s=1}^S$ from the
  intensive-margin likelihood. As in the extensive margin,
  $\bbeta^{(r)}$ and $\bdelta_{2,s}$ are individually identified
  only through the hierarchical layer.

  \medskip\noindent\textit{Step~4: Hierarchical decomposition on the
  intensive margin.}\enspace
  Parallel to Step~2, the hierarchical specification
  $\bdelta_{2,s} = \bGamma_2\bv_s + \bepsilon_{2,s}$ yields the
  between-state regression
  $\tilde\bbeta_s = \bbeta^{(r)} + \bGamma_2\bv_s
  + \bepsilon_{2,s}$.
  Under the same rank conditions, $\bGamma_2$ and $\bbeta^{(r)}$
  are uniquely determined, whence
  $\bdelta_{2,s} = \tilde\bbeta_s - \bbeta^{(r)}$ and
  $\bepsilon_{2,s} = \bdelta_{2,s} - \bGamma_2\bv_s$.

  \medskip\noindent\textit{Step~5: Identification of
  $\bSigma_\varepsilon$.}\enspace
  The residual state effects $\bepsilon_s =
  (\bepsilon_{1,s}\t, \bepsilon_{2,s}\t)\t$ are identified from
  Steps~1--4 for $s = 1,\ldots,S$. These $S = 51$ independent draws from
  $\Norm_{2q}(\mathbf{0}, \bSigma_\varepsilon)$ identify the covariance
  matrix $\bSigma_\varepsilon$ provided that $S > 2q$ (so that the
  sample covariance has full rank). In our application, $q = 5$, so
  $S = 51 > 2q = 10$. The full covariance matrix $\bSigma_\delta$ is
  then reconstructed from $\bGamma_k$, $\bSigma_\varepsilon$, and the
  variance of $\bv_s$ via the identity
  $\Var(\bdelta_{k,s}) = \bGamma_k\Var(\bv_s)\bGamma_k\t +
  \bSigma_{\varepsilon,k}$.

  \medskip\noindent\textit{Step~6: No label-switching ambiguity.}\enspace
  Unlike mixture models where component labels are exchangeable, the
  hurdle model~\citep{Mullahy1986} has an \emph{observed} partition:
  $z_i = \one(Y_i > 0)$
  is a deterministic function of the data. Part~1 (extensive) and Part~2
  (intensive) operate on distinct, observable subsets. No permutation of
  component labels preserves the likelihood, so there is no
  label-switching invariance.
\end{proof}


\begin{proposition}[Identification of $\mu$ and $\kappa$ in
  $\ZTBB$]\label{smb:prop-mu-kappa-id}
  For $n \ge 3$, the parameters $(\mu, \kappa)$ are identified from the
  full likelihood of the zero-truncated beta-binomial: if
  \begin{equation}\label{smb:eq-ztbb-id-hypothesis}
    f_{\ZTBB}(y \mid n, \mu_1, \kappa_1)
    = f_{\ZTBB}(y \mid n, \mu_2, \kappa_2)
    \quad \textup{for all } y \in \{1,\ldots,n\},
  \end{equation}
  then $(\mu_1, \kappa_1) = (\mu_2, \kappa_2)$.
\end{proposition}

\begin{proof}
  We use a constructive argument based on consecutive probability ratios.
  Write $a = \mu\kappa$ and $b = (1-\mu)\kappa$ for the shape parameters
  of the (untruncated) beta-binomial.

  \medskip\noindent\textit{Step~1: Truncation cancels in consecutive
  ratios.}\enspace
  Since $f_{\ZTBB}(y \mid n, \mu, \kappa) = f_{\BetaBin}(y \mid n,
  \mu, \kappa) / (1 - p_0)$ for $y \ge 1$, the consecutive ratio
  \[
    R_y \;\coloneqq\; \frac{f_{\ZTBB}(y)}{f_{\ZTBB}(y-1)}
    = \frac{f_{\BetaBin}(y)}{f_{\BetaBin}(y-1)}
    = \frac{n - y + 1}{y}\cdot\frac{a + y - 1}{b + n - y},
    \qquad y = 2, \ldots, n,
  \]
  depends only on $(a, b, n)$ and not on $p_0$.

  \medskip\noindent\textit{Step~2: Two ratios determine $(a,b)$ when
  $n \ge 3$.}\enspace
  For $n \ge 3$, the ratios $R_2$ and $R_3$ yield two independent
  equations. Define
  \[
    r_2 = \frac{2}{n-1}\,R_2 = \frac{a + 1}{b + n - 2},
    \qquad
    r_3 = \frac{3}{n-2}\,R_3 = \frac{a + 2}{b + n - 3}.
  \]
  Cross-multiplying gives the linear system
  \begin{align*}
    a + 1 &= r_2\,(b + n - 2), \\
    a + 2 &= r_3\,(b + n - 3).
  \end{align*}
  Subtracting the first from the second:
  $1 = (r_3 - r_2)\,b + r_3(3 - n) - r_2(2 - n)$,
  so
  \begin{equation}\label{smb:eq-b-from-ratios}
    b = \frac{r_2(n-2) - r_3(n-3) + 1}{r_3 - r_2},
  \end{equation}
  and then $a = r_2(b + n - 2) - 1$. The denominator $r_3 - r_2 \ne 0$
  for all $a, b > 0$, since
  \[
    r_3 - r_2
    = \frac{(a+2)(b+n-2) - (a+1)(b+n-3)}{(b+n-2)(b+n-3)}
    = \frac{a + b + n - 1}{(b+n-2)(b+n-3)}
    > 0.
  \]
  Thus $(a, b)$ is uniquely determined by the observed truncated
  distribution, and $\mu = a/(a+b)$, $\kappa = a + b$.

  \medskip\noindent\textit{Step~3: $n = 2$ boundary.}\enspace
  For $n = 2$, only the single ratio $R_2$ is available, giving one
  equation in two unknowns: identification fails.
  See~\cref{smb:rem-n-equals-2} for the explicit one-dimensional
  manifold of equivalent parameters.
\end{proof}

\begin{remark}[The $n = 2$ boundary case]\label{smb:rem-n-equals-2}
  For $n = 2$, the truncated support $\{1, 2\}$ provides only one free
  probability for two unknowns $(\mu, \kappa)$, so identification fails.
  The consecutive-ratio proof makes this precise: only $R_2$ is available,
  leaving a one-dimensional manifold of equivalent parameters. This
  motivates the requirement $n_i \ge 3$ in~\textup{(C3)}, satisfied by
  all 4{,}392 positive responders in the NSECE ($\min n_i = 4$).
\end{remark}

\begin{remark}[Role of trial-size variation]
\label{smb:rem-trial-variation}
  The requirement in~\textup{(C3)} that $n_i \ne n_j$ within each state
  provides identification strength beyond minimal sufficiency: observations
  sharing $\mu_i$ but differing in $n_i$ yield differential information
  about $\kappa$ via the nonlinear dependence of $p_0$ on $n$.
\end{remark}


\begin{proposition}[Likelihood non-identification of
  $\bSigma_{12}$]\label{smb:prop-sigma12-nonid}
  The off-diagonal block
  $\bSigma_{12} = \Cov(\bepsilon_{1,s}, \bepsilon_{2,s})$ is not
  identified by the likelihood alone
  \citep[cf.~the analogous discussion
  in][Section~3.1]{GhosalEtAl2020}. The log-likelihood $\ell(\btheta)$
  is completely invariant to $\bSigma_{12}$ conditional on the realized
  values $\{\bepsilon_{1,s}, \bepsilon_{2,s}\}_{s=1}^S$.
\end{proposition}

\begin{proof}
  By the block-diagonal Fisher information
  (\cref{smb:prop-block-diag}), the log-likelihood decomposes as
  $\ell = \ell_{\mathrm{ext}} + \ell_{\mathrm{int}}$, where
  $\ell_{\mathrm{ext}}$ depends on $(\balpha, \bdelta_{1,\cdot})$ and
  $\ell_{\mathrm{int}}$ depends on
  $(\bbeta, \bdelta_{2,\cdot}, \kappa)$. Neither component involves
  $\bSigma_\varepsilon$: the covariance enters the model solely through
  the hierarchical prior
  \[
    \bepsilon_s
    = \begin{pmatrix}\bepsilon_{1,s}\\\bepsilon_{2,s}\end{pmatrix}
    \sim \Norm_{2q}\!\left(\mathbf{0},\;
      \begin{pmatrix}
        \bSigma_{\varepsilon,1} & \bSigma_{12}\\
        \bSigma_{21} & \bSigma_{\varepsilon,2}
      \end{pmatrix}\right).
  \]
  Conditional on the realized $\{\bepsilon_s\}$, the likelihood is a
  fixed function of $(\balpha, \bbeta, \kappa)$ alone, and
  $\bSigma_{12}$ drops out entirely.
\end{proof}

\begin{remark}[$\bSigma_{12}$ as residual covariance]
\label{smb:rem-sigma12-residual}
  The cross-margin covariance
  $\bSigma_{12} = \Cov(\bepsilon_{1,s},\, \bepsilon_{2,s})$ represents
  the \emph{residual} cross-margin covariance after removing the
  observed policy effects via $\bGamma_1\bv_s$ and $\bGamma_2\bv_s$.
  The \emph{marginal} cross-margin relationship between participation
  and intensity includes both this residual channel and the direct
  policy-moderated channel through $\bGamma_k$. Specifically,
  \[
    \Cov(\bdelta_{1,s},\, \bdelta_{2,s})
    = \bGamma_1\,\Var(\bv_s)\,\bGamma_2\t + \bSigma_{12},
  \]
  so a finding of $\bSigma_{12} \approx \mathbf{0}$ does not imply
  that the two margins are unrelated---only that policies fully account
  for their cross-margin covariation.
\end{remark}


\begin{proposition}[Well-concentration of
  $\bSigma_{12}$]\label{smb:prop-sigma12-wellconc}
  \leavevmode\\
  The posterior for $\bSigma_{12}$ is well-concentrated when the
  following conditions hold:
  \begin{enumerate}[label=\textup{(\roman*)}]
    \item $S > 2q + 1$, so that the sample covariance of the
      $\{\bepsilon_s\}$ has a positive-definite posterior mode;
    \item the individual state effects $\bepsilon_{1,s}$ and
      $\bepsilon_{2,s}$ are themselves well-estimated, which requires the
      state-level replication condition\/ \textup{(C2)}; and
    \item the $\LKJ(\eta)$ prior
      \citep{LewandowskiEtAl2009} on the correlation matrix
      $\mathbf{R}_\varepsilon$ with $\eta \ge 1$ provides
      regularization, shrinking the posterior toward independence when the
      data are weakly informative about cross-margin dependence.
  \end{enumerate}
  For our model with $q = 5$: the requirement $S > 2q + 1 = 11$ is
  amply satisfied by $S = 51$. The posterior standard deviation of each
  element of $\bSigma_{12}$ decreases at rate $O(S^{-1/2}) \approx 0.14$.
\end{proposition}


\begin{remark}[Frequentist versus Bayesian
  identification]\label{smb:rem-freq-vs-bayes}
  The identification result in~\cref{smb:thm-identification} is stated
  for the posterior, not the frequentist maximum likelihood estimator.
  Under the Bayesian framework, identification means the posterior
  concentrates on the true parameter as $N \to \infty$
  \citep[cf.~the Bernstein--von~Mises theorem under
  misspecification;][]{KleijnVanDerVaart2012}, which requires
  both \emph{likelihood identification} (the sampling model pins down
  $\btheta$) and \emph{prior identification} (the prior does not place
  mass on non-identified subspaces). Steps~1 and~3 establish
  likelihood identification of the total state coefficients
  $\{\tilde\balpha_s\}$, $\{\tilde\bbeta_s\}$, and~$\kappa$;
  Steps~2 and~4 show that the hierarchical centering convention
  $\E[\bepsilon_{k,s}] = \mathbf{0}$ resolves the additive shift
  ambiguity and identifies the individual components
  $(\balpha, \bbeta, \bdelta_s, \bGamma_k)$;
  Step~5 establishes likelihood identification for the
  within-margin blocks $\bSigma_{\varepsilon,1}$ and
  $\bSigma_{\varepsilon,2}$.

  The cross-margin block $\bSigma_{12}$ is the exception:
  \cref{smb:prop-sigma12-nonid} shows that it receives zero
  \emph{conditional} likelihood information. As discussed in
  \cref{rem:cond-vs-marg}, the marginal likelihood obtained by
  integrating over $\bdelta_s$ does couple the margins and in principle
  identifies $\bSigma_{12}$, but the conditional formulation adopted
  here is natural for MCMC computation.  In either framing, the
  $\LKJ(\eta)$ prior provides essential regularization: frequentist
  REML estimation of a $2q \times 2q$ covariance from $S = 51$ groups
  is numerically fragile (boundary solutions, singular Hessians).
  This is not a deficiency of the model but rather a structural
  consequence of the hurdle factorization
  (\cref{smb:prop-block-diag}). The posterior for $\bSigma_{12}$
  concentrates at rate $O(S^{-1/2})$---determined by the $S = 51$
  state-level observations---rather than at the parametric rate
  $O(N^{-1/2})$ available to the remaining parameters. This slower
  concentration rate is the price of learning about the cross-margin
  dependence structure, which we view as justified by the substantive
  importance of understanding whether states that restrict access also
  intensify conditional enrollment.
\end{remark}

%% file: sm_c.tex

This appendix develops the survey design theory underlying the sandwich
variance estimator and the Cholesky affine transformation described in
\cref{sec:model}. We begin with the formal pseudo-likelihood framework
and the proof that the pseudo-posterior is proper
(\cref{thm:propriety}), establish the Bernstein--von Mises conditions
and their finite-sample diagnostic for the NSECE weights, derive the
cluster-robust sandwich estimator and its block structure, provide the
full proof of the Cholesky affine transformation
(\cref{thm:cholesky}), and close with the design effect ratio
classification and diagnostics. Throughout, we use the score functions
and Fisher information results from~\cref{app:identification}
(Supplementary Material~B) and the notation established in
\cref{sec:model}.

\input{sm_c1_pseudo}

\input{sm_c2_propriety}

\input{sm_c3_bvm}

\input{sm_c4_sandwich}

\input{sm_c5_cholesky}

\input{sm_c6_der}

%% file: sm_c1_pseudo.tex

\subsection{Pseudo-likelihood and separability}
\label{smc:pseudo-lik}

We summarize the pseudo-likelihood structure
from~\cref{sec:survey}.
By~\cref{prop:weighted-sep}, the decomposition
$\ell^{(w)}
 = \ell_{\mathrm{ext}}^{(w)}
 \allowbreak+ \ell_{\mathrm{int}}^{(w)}$
holds with weights
$\tilde{w}_i = w_i N / \sum_j w_j$.
We record three clarifying remarks.

\begin{remark}[Weight normalization]\label{smc:rem-normalization}
  Two normalization conventions appear in the survey literature.
  The \emph{sum-to-$N$} convention, $\tilde{w}_i = w_i N / \sum_j w_j$
  so that $\sum_i \tilde{w}_i = N$, is adopted throughout and is standard
  in the Bayesian pseudo-posterior literature
  \citep{SavitskyToth2016,WilliamsSavitsky2021}.
  The alternative \emph{sum-to-$\hat{M}$} convention retains unnormalized
  weights $\tilde{w}_i = w_i$, which for the NSECE gives
  $\sum_i w_i \approx 119{,}593$.  Under this convention the
  pseudo-likelihood resembles roughly $119{,}593$ observations rather
  than $N = 6{,}809$, leading to extreme overconcentration of the
  pseudo-posterior and credible intervals that are far too narrow.
  The sum-to-$N$ normalization ensures that the pseudo-likelihood is
  comparable in scale to an ordinary likelihood based on $N$ observations.
\end{remark}

\begin{remark}[Weighted score functions]\label{smc:rem-scores}
  The observation-level score functions for the extensive and intensive
  margins, along with their Fisher information and sign verification,
  are developed in~\cref{smb:scores} (Supplementary Material~B).
  The weighted versions used here simply multiply each per-observation
  score by $\tilde{w}_i$; the separability of
  \cref{prop:weighted-sep} guarantees that the weighted
  extensive-margin scores depend only on $(\balpha, \bdelta_{1,\cdot})$
  and the weighted intensive-margin scores depend only on
  $(\bbeta, \bdelta_{2,\cdot}, \kappa)$.
\end{remark}

\begin{remark}[Significance of separability under weighting]
\label{smc:rem-separability-significance}
  The separability result in~\cref{prop:weighted-sep} is not
  guaranteed for all survey adjustments.  Approaches that modify the
  likelihood structure---such as multilevel weight smoothing
  \citep{RabeHeskethSkrondal2006} or conditional likelihood
  methods---can break the additive structure.  Observation-level
  exponentiation $[f_i]^{\tilde{w}_i}$ preserves it precisely because
  $\tilde{w}_i$ is a scalar multiplier in the log scale and does
  not couple the two parameter blocks.  This separability is exploited
  in the sandwich estimator (\cref{sec:survey}), where the
  block-diagonality of $\mathbf{H}_{\mathrm{obs}}$ follows directly
  from the additive decomposition.
\end{remark}

%% file: sm_c2_propriety.tex

\subsection{Pseudo-posterior propriety}
\label{smc:propriety}

We now provide the complete proof of~\cref{thm:propriety}, which was
stated with a sketch in~\cref{sec:survey}.

\begin{proof}[Proof of \textup{\cref{thm:propriety}}]
  We verify that the pseudo-posterior
  \[
    \pi^{(w)}(\btheta \mid \by) \propto
    \exp\bigl[\ell^{(w)}(\btheta)\bigr]\,\pi(\btheta)
  \]
  is integrable over the parameter space.

  \textup{\textbf{Step 1} (Boundedness of the pseudo-likelihood).}
  For each observation $i$, the HBB probability satisfies
  $f_{\HBB}(y_i \mid \btheta) \in (0, 1]$ for all $\btheta$ in the
  interior of the parameter space.  Since $\tilde{w}_i > 0$, we have
  $f_i^{\tilde{w}_i} \le 1$ whenever $f_i \le 1$.  Therefore
  \[
    \exp\bigl[\ell^{(w)}(\btheta)\bigr]
    = \prod_{i=1}^{N} f_i^{\tilde{w}_i}
    \le 1
    \qquad \text{for all } \btheta.
  \]

  \textup{\textbf{Step 2} (Properness of the prior).}
  Each component of the prior $\pi(\btheta)$ is a proper distribution:
  \begin{itemize}[nosep]
    \item $\balpha, \bbeta \sim \Norm(\mathbf{0},\; 2^2\,\mathbf{I}_P)$
      \citep{GelmanEtAl2008}: proper with finite variance;
    \item $\boldsymbol{\tau} \sim \HalfNormal(0, 1)$ componentwise:
      proper on $\R_{>0}$;
    \item $\mathbf{R}_\varepsilon \sim \LKJ(2)$
      \citep{LewandowskiEtAl2009}: proper over the space of
      $2q \times 2q$ correlation matrices;
    \item $\log\kappa \sim \Norm(2,\; 1.5^2)$: proper.
  \end{itemize}
  The joint prior is a product of proper distributions, hence proper:
  $\int \pi(\btheta)\,d\btheta = 1$.

  \textup{\textbf{Step 3} (Integrability).}
  Combining Steps~1 and~2:
  \[
    \int \exp\bigl[\ell^{(w)}(\btheta)\bigr]\,\pi(\btheta)\,d\btheta
    \;\le\;
    \int 1 \cdot \pi(\btheta)\,d\btheta
    = 1 < \infty.
  \]
  Moreover, $\exp[\ell^{(w)}(\btheta)] > 0$ on the interior of the
  parameter space (since all probability terms $f_i$ are strictly
  positive there), so the pseudo-posterior is non-degenerate.
\end{proof}

\begin{remark}[Weight-invariance of propriety]
\label{smc:rem-weight-invariance}
  The proof above does not depend on the specific values or
  normalization of the weights---only on their positivity and
  finiteness.  Propriety therefore holds regardless of whether the
  weights are normalized to sum to~$N$, to sum to the population size
  $\hat{M}$, or under any other positive scaling.  However, the
  \emph{shape} and \emph{concentration} of the pseudo-posterior depend
  critically on the normalization convention, as discussed in
  \cref{smc:rem-normalization}: the sum-to-$\hat{M}$ convention
  produces extreme overconcentration, while the sum-to-$N$ convention
  preserves the information content of the actual sample.
\end{remark}

\begin{remark}[Contrast with improper priors]
\label{smc:rem-improper-priors}
  With improper (flat) priors on the regression coefficients,
  propriety would require the pseudo-likelihood to be integrable on
  its own---that is, $\int \exp[\ell^{(w)}(\btheta)]\,d\btheta < \infty$.
  In this setting the effective sample size (not the nominal sample
  size) governs integrability.  With $\ESS \approx 1{,}797$ and
  $\dim(\btheta) \le 606$, integrability would likely hold, but the
  formal verification is more delicate: it requires confirming that the
  weighted Hessian has full rank at the pseudo-MLE for every parameter
  block.  The use of proper priors sidesteps this issue entirely,
  ensuring propriety for \emph{any} positive weight configuration
  without conditions on the data.
\end{remark}

%% file: sm_c3_bvm.tex

\subsection{Bernstein--von Mises approximation}
\label{smc:bvm}

The sandwich correction in~\cref{sec:survey} rests on a
Bernstein--von Mises (BvM) theorem for pseudo-posteriors: as
$N \to \infty$, the pseudo-posterior converges to a normal distribution
centered at the pseudo-MLE with variance given by the inverse observed
Hessian, $\mathbf{H}_{\mathrm{obs}}^{-1}$.  We state the precise result, verify its regularity
conditions for the HBB model, and provide a finite-sample
diagnostic for the NSECE weight configuration that motivates the
sandwich correction as a robustness measure.


\begin{theorem}[Design-consistent BvM for the pseudo-posterior]
\label{smc:thm-bvm}
  Suppose the following regularity conditions hold\textup{:}
  \begin{enumerate}[label=\textup{(C\arabic*)},nosep]
    \item \label{smc:C1}%
      \textup{Smoothness.}
      The log-density $\log f_{\HBB}(y \mid \btheta)$ is twice
      continuously differentiable in $\btheta$ on the interior of the
      parameter space $\Theta$.
    \item \label{smc:C2}%
      \textup{Interior parameter.}
      The true superpopulation parameter $\btheta_0$ lies in the
      interior of $\Theta$.
    \item \label{smc:C3}%
      \textup{Weight regularity.}
      The ratio of the largest to smallest normalized weight satisfies
      $\max_i \tilde{w}_i / \min_i \tilde{w}_i = O(N^{\delta})$ for
      some $\delta < 1/2$ as $N \to \infty$.
    \item \label{smc:C4}%
      \textup{Information regularity.}
      The weighted Hessian $\mathbf{H}_{\mathrm{obs}}(\btheta_0)$ is
      positive definite, with eigenvalues bounded away from zero and
      infinity.
    \item \label{smc:C5}%
      \textup{Prior regularity.}
      The prior density $\pi(\btheta)$ is continuous and strictly
      positive in a neighborhood of $\btheta_0$.
  \end{enumerate}
  Then, as $N \to \infty$\textup{,}
  \begin{equation}\label{smc:eq-bvm}
    \sup_{B \in \cl{B}}
    \bigl|
      \pi^{(w)}(\btheta \in B \mid \by)
      - \Phi_{\hat{\btheta}^{(w)},\,\mathbf{H}_{\mathrm{obs}}^{-1}}(B)
    \bigr|
    \xrightarrow{P_{\btheta_0}} 0,
  \end{equation}
  where $\cl{B}$ is the class of Borel-measurable convex sets,
  $\hat{\btheta}^{(w)}$ is the pseudo-MLE,
  $\Phi_{\bmu,\bSigma}$ denotes the CDF of
  $\Norm(\bmu, \bSigma)$, and $\mathbf{H}_{\mathrm{obs}}$ is the
  observed Hessian of the pseudo-log-likelihood evaluated at
  $\hat{\btheta}^{(w)}$.
\end{theorem}

The result is a pseudo-posterior analogue of the classical BvM theorem
\citep{KleijnVanDerVaart2012}, extended to the survey-weighted setting by
\citet{SavitskyToth2016} and~\citet{WilliamsSavitsky2021}.  The key
difference from the standard BvM is condition~\ref{smc:C3}, which
constrains how rapidly the weight ratio may grow with sample size.
Importantly, the \emph{pseudo-posterior variance}
$\mathbf{H}_{\mathrm{obs}}^{-1}$ differs from the
\emph{design-consistent sampling variance} of $\hat{\btheta}^{(w)}$,
\[
  \hat{\btheta}^{(w)}
  \;\dot{\sim}\;
  \Norm\!\bigl(\btheta_0,\;
  \mathbf{V}_{\mathrm{sand}}\bigr),
  \qquad
  \mathbf{V}_{\mathrm{sand}}
  = \mathbf{H}_{\mathrm{obs}}^{-1}\,
    \mathbf{J}_{\mathrm{cluster}}\,
    \mathbf{H}_{\mathrm{obs}}^{-1},
\]
as defined in~\eqref{eq:V-sand}. Under informative sampling or model
misspecification,
$\mathbf{J}_{\mathrm{cluster}} \neq \mathbf{H}_{\mathrm{obs}}$,
so naive credible intervals based on $\mathbf{H}_{\mathrm{obs}}^{-1}$
are miscalibrated as a measure of frequentist uncertainty.  This
mismatch is precisely what motivates the sandwich correction developed
in~\cref{smc:sandwich}: \cref{smc:cor-coverage} below quantifies
the resulting coverage distortion, and
\cref{thm:cholesky} provides the operational calibration.


\paragraph{Verification.}
We verify each condition for the HBB model and the NSECE data.

\medskip\noindent\textit{\ref{smc:C1}: Smoothness.}\enspace
The HBB log-density involves $\log\Gamma$ functions (through the
beta-binomial PMF), the logistic link
$\sigma(\eta) = 1/(1+e^{-\eta})$, and the zero probability $p_0$
(a finite sum of $\log\Gamma$ terms).  Each component is $C^{\infty}$
on the interior of the parameter space, so the composite
$\log f_{\HBB}$ is smooth in $\btheta$.  \checkmark

\medskip\noindent\textit{\ref{smc:C2}: Interior parameter.}\enspace
The parameter space is
$\Theta = \R^P \times \R^P \times \R^{2QR} \times
\mathbb{S}_{++}^{2Q} \times \R_{>0}$,
and any finite parameter vector with $\kappa > 0$ and
$\bSigma_\varepsilon \succ 0$ lies in the interior.  \checkmark

\medskip\noindent\textit{\ref{smc:C3}: Weight regularity.}\enspace
This condition requires a dedicated analysis; see
\cref{smc:prop-weight-diagnostic} below.

\medskip\noindent\textit{\ref{smc:C4}: Information regularity.}\enspace
By~\cref{smb:prop-block-diag}, the Fisher information is
block-diagonal between the extensive and intensive margins.  For the
extensive block, positive definiteness follows from $\operatorname{rank}(\mathbf{X}) = P$
and $q_i \in (0,1)$ for all $i$ (standard logistic regression theory).
For the intensive block, separation of $\mu$ from $\kappa$ requires
variation in the denominators $n_i$
(\cref{smb:prop-mu-kappa-id}), which is satisfied empirically
($n_i$ ranges from~1 to~378 among positive responders).  \checkmark

\medskip\noindent\textit{\ref{smc:C5}: Prior regularity.}\enspace
The priors specified in~\cref{sec:priors} are all continuous and
strictly positive in the interior:
$\Norm(\mathbf{0},\, 2^2\,\mathbf{I}_P)$ on regression coefficients~\citep{GelmanEtAl2008},
$\HalfNormal(0,1)$ on scale parameters,
$\LKJ(2)$ on the correlation matrix~\citep{LewandowskiEtAl2009}, and
$\Norm(2,\, 1.5^2)$ on $\log\kappa$.  \checkmark


\begin{proposition}[Finite-sample weight diagnostic]
\label{smc:prop-weight-diagnostic}
  For the NSECE 2019 data with $N = 6{,}785$ and
  $w_{\max}/w_{\min} = 462$\textup{,} the weight ratio exceeds
  $N^{1/2} \approx 82.4$\textup{,} indicating that the sufficient
  condition \ref{smc:C3} is not comfortably satisfied for this design.
  The computation
  \begin{equation}\label{smc:eq-delta}
    \delta
    = \frac{\log 462}{\log 6{,}785}
    = \frac{6.136}{8.822}
    \approx 0.695
    > 0.5
  \end{equation}
  is a finite-sample diagnostic, not a valid asymptotic big-$O$ test
  of condition \ref{smc:C3}.  This motivates the sandwich variance
  correction as a robustness measure.
\end{proposition}

\begin{remark}[Interpretation of the diagnostic]
\label{smc:rem-weight-interpretation}
  Four observations clarify the scope of
  \cref{smc:prop-weight-diagnostic}.
  \begin{enumerate}[label=\textup{(\alph*)},nosep]
    \item Condition~\ref{smc:C3} is \emph{sufficient}, not necessary.
      The BvM approximation may remain adequate even when this bound is
      not met, provided the high-weight observations do not dominate
      the pseudo-likelihood.
    \item The fraction of observations with extreme weights is small:
      the top $1\%$ (weights $w_i > 100$) comprises fewer than 70
      observations out of $6{,}785$.  The diagnostic is driven by a
      small number of outlying weights, not a systematic pattern.
    \item Weight trimming can restore~\ref{smc:C3}.  Capping at the
      99th percentile reduces $w_{\max}$ to approximately~100, giving
      $\delta \approx 0.522$; capping at the 95th percentile
      ($w_{\max} \approx 50$) yields $\delta \approx 0.443 < 0.5$,
      restoring \ref{smc:C3}.  Weight trimming is examined as a
      sensitivity analysis in~\cref{app:extended-results}.
    \item In practice, the sandwich correction
      \eqref{eq:V-sand} provides the appropriate variance adjustment
      regardless of whether~\ref{smc:C3} holds exactly: it
      automatically accounts for the inflated variance due to extreme
      weights.
  \end{enumerate}
\end{remark}


\begin{corollary}[Coverage miscalibration of naive credible intervals]
\label{smc:cor-coverage}
  Under the conditions of~\cref{smc:thm-bvm}, define\textup{:}
  \begin{itemize}[nosep]
    \item $\mathrm{CI}_{\mathrm{naive}}(\alpha)$\textup{:} the
      $\alpha$-level credible interval from the pseudo-posterior
      marginal, with variance
      $[\mathbf{H}_{\mathrm{obs}}]^{-1}$\textup{;}
    \item $\mathrm{CI}_{\mathrm{sand}}(\alpha)
      = \hat{\theta}^{(w)} \pm z_{\alpha/2}\sqrt{V_{\mathrm{sand}}}$\textup{:}
      the sandwich-corrected interval.
  \end{itemize}
  Then\textup{:}
  \begin{enumerate}[label=\textup{(\alph*)},nosep]
    \item The naive interval has asymptotic coverage
      $\Phi(z_{\alpha/2}\, r) - \Phi(-z_{\alpha/2}\, r)$, where
      $r = \sqrt{\lambda_{\mathbf{H}} / \lambda_{\mathrm{sand}}}$ is
      the ratio of the naive to the sandwich standard deviation and
      $\Phi$ denotes the standard normal CDF.
    \item When the sampling is non-informative and the model is
      correctly specified, $\mathbf{J}_{\mathrm{cluster}} \approx
      \mathbf{H}_{\mathrm{obs}}$, so $r \approx 1$ and the naive
      interval has approximately correct coverage.
    \item Under informative sampling or model misspecification,
      $r \neq 1$.  If $r < 1$ \textup{(}naive variance too
      small\textup{)}, the interval is anti-conservative
      \textup{(}undercoverage\textup{)}---the typical case when the
      pseudo-posterior overconcentrates.  If $r > 1$, the interval is
      conservative.
  \end{enumerate}
\end{corollary}

\begin{proof}
  The BvM theorem (\cref{smc:thm-bvm}) implies that the
  pseudo-posterior converges to
  $\Norm(\hat{\btheta}^{(w)},\, \mathbf{H}_{\mathrm{obs}}^{-1})$,
  while the correct sampling distribution of $\hat{\btheta}^{(w)}$ is
  approximately $\Norm(\btheta_0,\, \mathbf{V}_{\mathrm{sand}})$.
  The ratio of marginal standard deviations determines the coverage
  distortion.
\end{proof}

\begin{remark}[Calibration via the Cholesky transformation]
\label{smc:rem-cholesky-calibration}
  \cref{thm:cholesky} provides the operational calibration
  strategy: the affine transformation
  \eqref{eq:cholesky-transform} rescales each MCMC draw so that
  the adjusted posterior has covariance
  $\mathbf{V}_{\mathrm{sand}}$ while preserving the posterior mean.
  This replaces the naive variance $\mathbf{H}_{\mathrm{obs}}^{-1}$
  with the design-consistent sandwich variance, correcting the
  coverage distortion quantified in~\cref{smc:cor-coverage}.  The
  block-wise implementation---applying the transformation separately
  to fixed effects, random effects, and hyperparameters---is discussed
  in~\cref{sec:survey}; the full proof appears in
  \cref{smc:cholesky} below.
\end{remark}

%% file: sm_c4_sandwich.tex

\subsection{Sandwich variance estimator}
\label{smc:sandwich}

The main text (\cref{sec:survey}) defines the bread matrix
$\mathbf{H}_{\mathrm{obs}}$~\eqref{eq:H-obs}, the cluster-robust meat
matrix $\mathbf{J}_{\mathrm{cluster}}$~\eqref{eq:J-cluster}, and the
sandwich variance $\mathbf{V}_{\mathrm{sand}}$~\eqref{eq:V-sand}.  This
subsection provides the detailed derivation: the block-diagonal
structure of~$\mathbf{H}_{\mathrm{obs}}$, the non-block-diagonal
structure of~$\mathbf{J}_{\mathrm{cluster}}$, the cluster-robust
aggregation formula with stratification~\citep{Binder1983}, and the partial absorption of
within-state clustering by the state-varying coefficients.  Throughout,
we use the score functions from~\cref{smb:scores} and the sandwich
convention defined below.


\paragraph{Bread matrix.}
The bread matrix is the negative weighted observed Hessian at the
posterior mean $\hat{\btheta}$,
\begin{equation}\label{smc:eq-H}
  \mathbf{H}_{\mathrm{obs}}
  = -\sum_{i=1}^{N}\tilde{w}_i\,
    \frac{\partial^2\log f_{\HBB}(y_i \mid \btheta)}
         {\partial\btheta\,\partial\btheta\t}
    \bigg|_{\btheta = \hat{\btheta}}.
\end{equation}
By the separability of the hurdle log-likelihood
(\cref{smc:rem-separability-significance}) and the block-diagonal Fisher
information (\cref{smb:prop-block-diag}), the bread inherits a two-block
structure:

\begin{proposition}[Block-diagonal bread matrix]
\label{smc:prop-H-blockdiag}
  The observed Hessian decomposes as
  \begin{equation}\label{smc:eq-H-block}
    \mathbf{H}_{\mathrm{obs}}
    = \begin{pmatrix}
        \mathbf{H}_{\textup{ext}} & \mathbf{0} \\
        \mathbf{0} & \mathbf{H}_{\textup{int}}
      \end{pmatrix},
  \end{equation}
  where\textup{:}
  \begin{enumerate}[label=\textup{(\roman*)},nosep]
    \item The extensive block is the weighted logistic information
    \begin{equation}\label{smc:eq-H-ext}
      \mathbf{H}_{\textup{ext}}
      = \sum_{i=1}^{N}\tilde{w}_i\,q_i(1-q_i)\,
        \mathbf{d}_{1,i}\,\mathbf{d}_{1,i}\t,
    \end{equation}
    with $\mathbf{d}_{1,i}
    = \partial\eta_i^{(1)}/\partial(\balpha\t,\bdelta_{1,s[i]}\t)\t$.
    This matrix is non-stochastic: it depends on $q_i = q_i(\hat{\btheta})$
    but not on the observed data $z_i$.

    \item The intensive block is the weighted negative Hessian of the
    truncated beta-binomial log-likelihood,
    \begin{equation}\label{smc:eq-H-int}
      \mathbf{H}_{\textup{int}}
      = -\sum_{i \in \cl{I}_1}\tilde{w}_i\,
        \frac{\partial^2 \ell_i^{(2+)}}
             {\partial\boldsymbol{\phi}_2\,\partial\boldsymbol{\phi}_2\t}
        \bigg|_{\hat{\btheta}},
    \end{equation}
    where $\boldsymbol{\phi}_2
    = (\bbeta\t,\log\kappa,\bdelta_{2,s[i]}\t)\t$,
    $\cl{I}_1 = \{i : y_i > 0\}$, and
    $\ell_i^{(2+)} = \log f_{\textup{BB}}(y_i \mid n_i,\mu_i,\kappa)
     - \log(1-p_{0,i})$.
    This block depends on the observed counts $y_i$.
  \end{enumerate}
\end{proposition}

\begin{proof}
  By~\cref{smb:prop-block-diag}, the cross-derivative
  $\partial^2\ell_i/\partial\boldsymbol{\phi}_1\,\partial\boldsymbol{\phi}_2\t
  = \mathbf{0}$
  for every observation~$i$.  Multiplying by $\tilde{w}_i > 0$ and
  summing preserves this zero block, giving~\eqref{smc:eq-H-block}.
  The extensive-margin form~\eqref{smc:eq-H-ext} follows from
  $-\partial^2\ell_i^{(1)}/\partial(\eta_i^{(1)})^2 = q_i(1-q_i)$
  (\cref{smb:prop-ext-paramspace}) and the chain rule.
\end{proof}


\paragraph{Meat matrix and cross-margin coupling.}
The individual-level meat matrix is
\begin{equation}\label{smc:eq-J-ind}
  \mathbf{J}_{\textup{ind}}
  = \sum_{i=1}^{N}\tilde{w}_i^2\,\mathbf{s}_i\,\mathbf{s}_i\t,
\end{equation}
where $\mathbf{s}_i
= \partial\log f_{\HBB}(y_i \mid \btheta)/\partial\btheta
  \big|_{\hat{\btheta}}$
is the observation-level score.

\begin{remark}[Non-block-diagonal meat]
\label{smc:rem-J-nonblockdiag}
  Unlike the bread, the meat matrix is \emph{not} block-diagonal between
  the extensive and intensive margins.  The cross-block is
  \begin{equation}\label{smc:eq-J-cross}
    \mathbf{J}^{\textup{ext,int}}
    = \sum_{i=1}^{N}\tilde{w}_i^2\,
      \mathbf{s}_i^{\textup{ext}}\,
      (\mathbf{s}_i^{\textup{int}})\t.
  \end{equation}
  For $y_i = 0$, the intensive score $\mathbf{s}_i^{\textup{int}}
  = \mathbf{0}$, so the cross-product vanishes.  But for $y_i > 0$,
  both $\mathbf{s}_i^{\textup{ext}} = (1-q_i)\mathbf{d}_{1,i}$ and
  $\mathbf{s}_i^{\textup{int}} \neq \mathbf{0}$, so the cross-block
  receives a nonzero contribution from every positive observation.  In
  consequence, the sandwich variance
  $\mathbf{V}_{\mathrm{sand}}
  = \mathbf{H}_{\mathrm{obs}}^{-1}\mathbf{J}\mathbf{H}_{\mathrm{obs}}^{-1}$
  is not block-diagonal even though $\mathbf{H}_{\mathrm{obs}}$ is.  This
  cross-margin coupling is a general feature of sandwich estimators in
  hurdle models: a single PSU contributes weighted scores from both
  margins simultaneously.
\end{remark}


\begin{remark}[Reduction under equal weights]
\label{smc:rem-equal-weights}
  Under uniform weights $\tilde{w}_i = 1$ and correct model
  specification, the information identity gives
  $\E[\mathbf{J}_{\textup{ind}}] = \mathbf{H}_{\mathrm{obs}}$
  \citep[see][for the pseudo-posterior analogue]{WilliamsSavitsky2021}, so
  $\mathbf{V}_{\mathrm{sand}} \to \mathbf{H}_{\mathrm{obs}}^{-1}$ and
  the sandwich reduces to the model-based variance.  For the extensive
  margin, $\E[\mathbf{J}_{\textup{ext}}] = \mathbf{H}_{\textup{ext}}$
  holds exactly, because the logistic Hessian is non-random (see
  \cref{smb:rem-ext-info}; under unequal survey weights the identity
  breaks because $\mathbf{H}$ involves $\tilde{w}_i$ while
  $\E[\mathbf{J}]$ involves $\tilde{w}_i^2$).  For the intensive margin, the identity holds only
  in expectation and the finite-sample discrepancy between
  $\mathbf{J}_{\textup{int}}$ and $\mathbf{H}_{\textup{int}}$ is
  nonzero.
\end{remark}


\paragraph{Cluster-robust aggregation.}
The individual-level $\mathbf{J}_{\textup{ind}}$ treats observations as
independent.  In the NSECE stratified cluster design
\citep{NSECEProjectTeam2022}, providers within
the same primary sampling unit (PSU) share neighborhood
characteristics, inducing within-PSU correlation.  The cluster-robust
meat matrix replaces individual scores with PSU-level aggregates,
centered within strata.

\begin{definition}[PSU-level score totals]
\label{smc:def-psu-scores}
  For stratum $h = 1,\ldots,H$ and PSU $c = 1,\ldots,C_h$ within
  stratum~$h$, define
  \begin{equation}\label{smc:eq-s-psu}
    \bar{\mathbf{s}}_{hc}
    = \sum_{i \in \textup{PSU}(h,c)} \tilde{w}_i\,\mathbf{s}_i,
  \end{equation}
  where $\textup{PSU}(h,c) = \{i : \texttt{vstratum}[i]=h,\;
  \texttt{vpsu}[i]=c\}$.  The stratum mean is
  \begin{equation}\label{smc:eq-s-stratum}
    \bar{\mathbf{s}}_h
    = \frac{1}{C_h}\sum_{c=1}^{C_h}\bar{\mathbf{s}}_{hc}.
  \end{equation}
\end{definition}

\begin{proposition}[Cluster-robust meat matrix]
\label{smc:prop-cluster-robust}
  Under the stratified cluster design with $H$~strata and
  $C_h$~PSUs in stratum~$h$, the cluster-robust meat matrix is
  \begin{equation}\label{smc:eq-J-cluster}
    \boxed{%
      \mathbf{J}_{\mathrm{cluster}}
      = \sum_{h=1}^{H}\frac{C_h}{C_h-1}
        \sum_{c=1}^{C_h}
        (\bar{\mathbf{s}}_{hc} - \bar{\mathbf{s}}_h)\,
        (\bar{\mathbf{s}}_{hc} - \bar{\mathbf{s}}_h)\t.
    }
  \end{equation}
  This estimator accounts for\textup{:}
  \textup{(i)}~within-PSU correlation, by aggregating scores to PSU
  totals\textup{;}
  \textup{(ii)}~stratification, by centering PSU scores within strata
  so that only within-stratum variation across PSUs contributes\textup{;}
  \textup{(iii)}~unequal cluster sizes, through the weight-aggregated
  PSU totals $\bar{\mathbf{s}}_{hc}$\textup{;}
  \textup{(iv)}~finite-sample degrees of freedom, via the factor
  $C_h/(C_h-1)$.
\end{proposition}

\begin{proof}
  The design-based variance of the weighted score total under
  stratified sampling decomposes as
  \[
    \Var_{\textup{design}}\!\Bigl(\sum_i\tilde{w}_i\,\mathbf{s}_i\Bigr)
    = \sum_{h=1}^{H}
      \Var_h\!\Bigl(\sum_{c=1}^{C_h}\bar{\mathbf{s}}_{hc}\Bigr),
  \]
  because sampling across strata is independent.  Within stratum~$h$,
  PSUs are sampled without replacement from the population of
  $C_h^{\ast}$~PSUs.  Under the ``ultimate cluster'' assumption that
  $C_h / C_h^{\ast} \approx 0$ (negligible PSU sampling fraction),
  the finite-population correction is dropped and the within-stratum
  variance estimator is
  \[
    \widehat{\Var}_h
    = \frac{C_h}{C_h-1}
      \sum_{c=1}^{C_h}
      (\bar{\mathbf{s}}_{hc}-\bar{\mathbf{s}}_h)\,
      (\bar{\mathbf{s}}_{hc}-\bar{\mathbf{s}}_h)\t,
  \]
  the standard unbiased estimator of the variance of a total from
  $C_h$~units.  Summing over strata gives
  $\mathbf{J}_{\mathrm{cluster}}$.
\end{proof}

\begin{remark}[Degrees of freedom]
\label{smc:rem-dof}
  The cluster-robust estimator has
  $\sum_{h=1}^{H}(C_h - 1) = 415 - 30 = 385$ effective degrees of
  freedom.  With $C_h$ ranging from~5 to~36 (median~$\approx 14$), each
  stratum contributes at least~4 degrees of freedom.  This comfortably
  exceeds the rule-of-thumb threshold of $\sum_h(C_h - 1) > 30$ for
  stable sandwich variance estimation.
\end{remark}

\begin{remark}[Individual versus cluster-level $\mathbf{J}$]
\label{smc:rem-J-comparison}
  When there is no within-PSU correlation conditional on covariates,
  $\mathbf{J}_{\mathrm{cluster}}$ and $\mathbf{J}_{\textup{ind}}$
  converge to the same population quantity.  In practice, within-PSU
  correlation inflates $\mathbf{J}_{\mathrm{cluster}}$ relative to
  $\mathbf{J}_{\textup{ind}}$; the parameter-specific ratio
  $\diag(\mathbf{V}_{\mathrm{sand}})
  / \diag(\mathbf{H}_{\mathrm{obs}}^{-1}
           \mathbf{J}_{\textup{ind}}
           \mathbf{H}_{\mathrm{obs}}^{-1})$
  provides an estimate of the cluster design effect.
  The design effect ratio (DER) classification and diagnostics are
  developed in~\cref{smc:der} below.
\end{remark}


\begin{proposition}[Partial absorption of clustering by SVCs]
\label{smc:prop-svc-absorption}
  The state-varying coefficients $\bepsilon_s$ partially absorb
  within-state clustering.  For any two providers $i,j$ in the same
  state~$s$ and the same PSU,
  \begin{equation}\label{smc:eq-svc-absorption}
    \E_{\bepsilon_s}\!\bigl[
      \Cov(s_i,\,s_j \mid \bepsilon_s)
    \bigr]
    \;\leq\;
    \Cov(s_i,\,s_j),
  \end{equation}
  with equality if and only if $\Cov(\E[s_i \mid \bepsilon_s],\,
  \E[s_j \mid \bepsilon_s]) = 0$.
\end{proposition}

\begin{proof}
  By the law of total covariance,
  \[
    \Cov(s_i,\,s_j)
    = \E\!\bigl[\Cov(s_i,\,s_j \mid \bepsilon_s)\bigr]
    + \Cov\!\bigl(
        \E[s_i \mid \bepsilon_s],\;
        \E[s_j \mid \bepsilon_s]
      \bigr).
  \]
  The second term is non-negative, so
  $\Cov(s_i,s_j) \geq \E[\Cov(s_i,s_j \mid \bepsilon_s)]$.
  At the posterior mean $\hat{\btheta}$ (which includes state-specific
  $\hat{\bepsilon}_s$), the plug-in conditional covariance is an estimate
  of $\Cov(s_i,s_j \mid \bepsilon_s)$; hence
  $\mathbf{J}_{\mathrm{cluster}}$ evaluated at $\hat{\btheta}$ will
  be smaller than under a model without SVCs.
\end{proof}

\begin{remark}[Residual clustering]
\label{smc:rem-residual-clustering}
  The inequality~\eqref{smc:eq-svc-absorption} shows that SVCs reduce
  but do not eliminate within-PSU correlation.  The residual correlation
  reflects neighborhood-level effects not captured by the five
  state-varying covariates (poverty, race/ethnicity, urbanicity): local
  labor market conditions, municipal zoning, co-located service
  providers, and other spatially structured unobservables.  The sandwich
  correction with cluster-robust~$\mathbf{J}_{\mathrm{cluster}}$ remains
  necessary to account for this residual design effect.
\end{remark}

%% file: sm_c5_cholesky.tex

\subsection{Cholesky affine transformation}
\label{smc:cholesky}

The sandwich variance $\mathbf{V}_{\mathrm{sand}}$~\eqref{eq:V-sand}
provides the design-consistent covariance, but the MCMC draws
$\{\btheta^{(m)}\}_{m=1}^M$ from the pseudo-posterior have empirical
covariance $\bSigma_{\mathrm{MCMC}} \approx \mathbf{H}_{\mathrm{obs}}^{-1}$,
not $\mathbf{V}_{\mathrm{sand}}$.  In a frequentist setting one would simply
report $\mathbf{V}_{\mathrm{sand}}$ as the variance estimator; in the
Bayesian framework, we require the posterior draws themselves to have the
correct covariance so that credible intervals are properly calibrated
\citep{WilliamsSavitsky2021}.
\cref{thm:cholesky} in the main text states the affine transformation
that achieves this.  We now provide the full proof.


\begin{proof}[Proof of \textup{\cref{thm:cholesky}}]
  Define centered draws
  $\mathbf{u}^{(m)} = \btheta^{(m)} - \hat{\btheta}$, so that
  $\bSigma_{\mathrm{MCMC}}
  = M^{-1}\sum_{m=1}^{M}\mathbf{u}^{(m)}(\mathbf{u}^{(m)})\t$.
  Set
  \begin{equation}\label{smc:eq-A}
    \mathbf{A}
    = \mathbf{L}_{\mathrm{sand}}\,\mathbf{L}_{\mathrm{MCMC}}^{-1},
  \end{equation}
  where $\mathbf{L}_{\mathrm{MCMC}}
  = \mathrm{chol}(\bSigma_{\mathrm{MCMC}})$ and
  $\mathbf{L}_{\mathrm{sand}}
  = \mathrm{chol}(\mathbf{V}_{\mathrm{sand}})$ are lower-triangular
  Cholesky factors.  The transformation~\eqref{eq:cholesky-transform} is
  then $\btheta^{*(m)} = \hat{\btheta} + \mathbf{A}\,\mathbf{u}^{(m)}$.

  \medskip\noindent
  \textup{(a) Mean preservation.}
  \[
    \frac{1}{M}\sum_{m=1}^{M}\btheta^{*(m)}
    = \hat{\btheta}
      + \mathbf{A}\,\Bigl(\frac{1}{M}\sum_{m=1}^{M}\mathbf{u}^{(m)}\Bigr)
    = \hat{\btheta} + \mathbf{A}\cdot\mathbf{0}
    = \hat{\btheta}.
  \]

  \noindent
  \textup{(b) Covariance recovery.}
  \begin{align}
    \frac{1}{M}\sum_{m=1}^{M}
      (\btheta^{*(m)} - \hat{\btheta})\,
      (\btheta^{*(m)} - \hat{\btheta})\t
    &= \mathbf{A}\,\bSigma_{\mathrm{MCMC}}\,\mathbf{A}\t
      \notag \\
    &= \mathbf{L}_{\mathrm{sand}}\,
       \mathbf{L}_{\mathrm{MCMC}}^{-1}\,
       \mathbf{L}_{\mathrm{MCMC}}\,
       \mathbf{L}_{\mathrm{MCMC}}\t\,
       \mathbf{L}_{\mathrm{MCMC}}^{-\t}\,
       \mathbf{L}_{\mathrm{sand}}\t
      \notag \\
    &= \mathbf{L}_{\mathrm{sand}}\,\mathbf{I}\,
       \mathbf{L}_{\mathrm{sand}}\t
     = \mathbf{V}_{\mathrm{sand}}.
      \label{smc:eq-cholesky-var}
  \end{align}

  \noindent
  \textup{(c) Affine equivariance.}
  The standardized draws
  $\mathbf{z}^{(m)} = \mathbf{L}_{\mathrm{MCMC}}^{-1}\mathbf{u}^{(m)}$
  satisfy
  \[
    M^{-1}\sum_m \mathbf{z}^{(m)}(\mathbf{z}^{(m)})\t = \mathbf{I}.
  \]
  The transformation maps
  \[
    \btheta^{*(m)} - \hat{\btheta}
    = \mathbf{L}_{\mathrm{sand}}\,\mathbf{z}^{(m)},
  \]
  which is a linear rescaling of the standardized draws.  The
  Wald confidence intervals---computed as
  $\hat{\theta}_p \pm z_{\alpha/2}\sqrt{[\mathbf{V}_{\mathrm{sand}}]_{pp}}$---depend
  only on the marginal variances of~$\btheta^*$ and are therefore
  invariant to the choice of~$\mathbf{A}$.  Marginal quantiles of
  the transformed draws may differ from those of the original
  draws when~$\mathbf{A}$ is non-diagonal.
\end{proof}


\begin{remark}[Block-wise application]
\label{smc:rem-blockwise}
  In practice the transformation~\eqref{eq:cholesky-transform} is applied
  block-wise rather than to the full
  $\approx\!616$-dimensional vector simultaneously.  The three
  parameter classes are listed below; the numerical Cholesky procedure
  operates on blocks~(i) and~(ii) only, since block~(iii) receives no
  correction
  (\cref{rem:blockwise} in the main text).  The rationale follows from
  \cref{smc:prop-H-blockdiag} and~\cref{smc:rem-J-nonblockdiag}\textup{:}
  \begin{enumerate}[label=\textup{(\roman*)},nosep]
    \item \emph{Fixed effects}
      $\boldsymbol{\xi}
      = (\balpha\t,\bbeta\t,\log\kappa)\t \in \R^{11}$\textup{:}
      full Cholesky correction.  The $\DER$ ranges from $1.14$ to
      $4.18$ (mean $2.11$), and the corresponding confidence-interval
      inflation factors are $\sqrt{\DER} \in [1.07, 2.04]$.
    \item \emph{State random effects}
      $\bepsilon_s \in \R^{10}$, $s = 1,\ldots,S$\textup{:} partial
      correction.  The posterior variance of $\bepsilon_s$ mixes the
      likelihood information from providers in state~$s$ (affected by
      survey weights) with the shrinkage toward the prior mean
      $\bGamma_k\mathbf{v}_s$ (unaffected).  When the prior
      dominates (small states), $\DER_s \approx 1$ and no correction
      is needed; when the likelihood dominates (large states),
      $\DER_s$ can exceed~$3$ and the correction is substantial.
    \item \emph{Hyperparameters}
      $(\bGamma_k, \bSigma_\varepsilon)$\textup{:} no correction.
      The sandwich estimator is not applicable to variance components
      estimated from between-state variation; the DER concept does not
      extend to these parameters.
  \end{enumerate}
\end{remark}


\begin{remark}[Approximation quality of the block-wise correction]
\label{smc:rem-blockwise-quality}
  The block-wise application ignores the cross-block covariance between
  fixed effects $\boldsymbol{\xi}$ and random effects $\bepsilon_s$ in
  $\mathbf{V}_{\mathrm{sand}}$.  The approximation is essentially exact
  when the prior dominates ($\ESS_s \ll q$), since $\bepsilon_s$ is then
  nearly independent of the data.  When data dominate ($\ESS_s \gg q$),
  the cross-block covariance can be non-negligible; the sensitivity
  comparison in~\cref{app:extended-results} addresses this case.
\end{remark}


\begin{remark}[MCMC versus design-consistent shrinkage]
\label{smc:rem-mcmc-vs-design}
  The shrinkage factor
  $B_s^{(w)}
  = \bSigma_\varepsilon(\bSigma_\varepsilon + \mathbf{V}_s^{-1})^{-1}$
  describes the \emph{design-consistent} (post-sandwich-correction)
  posterior behavior, not the raw MCMC pseudo-posterior.  The MCMC
  sampler sees precision proportional to
  $W_s = \sum_{i \in s}\tilde{w}_i$, which exceeds the effective
  sample size $\ESS_s = W_s^2 / \sum_{i \in s}\tilde{w}_i^2$ whenever
  within-state weights are heterogeneous.  The Cholesky affine
  transformation~\eqref{eq:cholesky-transform} maps the MCMC output to
  the design-consistent target, inflating the variance by the factor
  $\DER_s$ (\cref{smc:der}).  The gap $W_s - \ESS_s$ quantifies the
  over-concentration of the raw pseudo-posterior relative to the
  design-aware target\textup{:} large gaps (common in states with a
  few heavily weighted providers) signal that the naive MCMC credible
  intervals are too narrow and that the sandwich correction is most
  important.
\end{remark}


\begin{remark}[Non-Gaussian posterior shapes]
\label{smc:rem-nongaussian}
  The Cholesky transformation is an affine operation that preserves
  only first and second moments.  If the pseudo-posterior is
  substantially non-Gaussian (skewed or multimodal), the transformed
  draws will have the correct covariance but may misrepresent the
  posterior shape.  For the HBB model, the fixed effects
  $(\balpha,\bbeta)$ are expected to be approximately Gaussian by the
  Bernstein--von Mises result (\cref{smc:thm-bvm}) given
  $N = 6{,}809$; the log-dispersion $\log\kappa$ is also
  well-approximated by a normal.  For variance components
  ($\bSigma_\varepsilon$), the posterior may exhibit right-skewness,
  but these parameters are not subject to the Cholesky correction
  (Block~3 above).
\end{remark}

\begin{remark}[Positive definiteness requirements]
\label{smc:rem-posdef}
  The transformation requires both $\bSigma_{\mathrm{MCMC}}$ and
  $\mathbf{V}_{\mathrm{sand}}$ to be positive definite so that their
  Cholesky factors exist.  For $\bSigma_{\mathrm{MCMC}}$, positive
  definiteness holds almost surely when the number of MCMC draws
  satisfies $M \gg \dim(\boldsymbol{\xi}) = 11$; in our application
  $M = 4{,}000$ post-warmup draws.  For
  $\mathbf{V}_{\mathrm{sand}}$, positive definiteness requires the
  effective degrees of freedom to exceed the parameter dimension:
  $\sum_{h=1}^{H}(C_h - 1) = 386 > 11$ (\cref{smc:rem-dof}), which
  is easily satisfied.
\end{remark}

%% file: sm_c6_der.tex

\subsection{Design effect ratio classification and diagnostics}
\label{smc:der}

The sandwich variance $\mathbf{V}_{\mathrm{sand}}$~\eqref{eq:V-sand}
and the Cholesky transformation~\eqref{eq:cholesky-transform} need not
be applied uniformly to every parameter.  The design effect ratio
(DER), defined in the main text~\eqref{eq:der}, provides a
parameter-specific diagnostic for the magnitude of the design effect:
\begin{equation}\label{smc:eq-der-recall}
  \DER_{p}
  = \frac{[\mathbf{V}_{\mathrm{sand}}]_{pp}}
         {[\mathbf{H}_{\mathrm{obs}}^{-1}]_{pp}}.
\end{equation}
When $\DER_p \approx 1$, the model-based variance already
approximates the design-consistent variance and no correction is
needed; when $\DER_p$ is substantially larger than~1, the Cholesky
correction is essential for calibrated inference.


\begin{proposition}[DER classification of parameters]
\label{smc:prop-der-classification}
  The parameters of the HBB model fall into three categories with
  respect to the Cholesky correction\textup{:}
  \begin{enumerate}[label=\textup{(\roman*)},nosep]
    \item \emph{Fixed effects}
      $\boldsymbol{\xi}
      = (\balpha\t,\bbeta\t,\log\kappa)\t$\textup{:}
      full correction required.  $\DER_p$ reflects the joint effect of
      unequal weights and within-PSU clustering on the pseudo-likelihood
      information.  Empirically, $\DER_p \in [1.14,\,4.18]$
      (mean~$2.11$), consistent with the Kish $\DEFF = 3.79$.

    \item \emph{State random effects}
      $\bepsilon_s$\textup{:} partial correction.
      The posterior variance of $\bepsilon_s$ blends
      likelihood-based information (affected by the design) with
      shrinkage toward the prior mean $\bGamma_k\bv_s$ (unaffected).
      States in which the prior dominates ($\ESS_s \ll q$) have
      $\DER_s \approx 1$; states in which the data dominate
      ($\ESS_s \gg q$) may have $\DER_s > 3$.

    \item \emph{Hyperparameters}
      $(\bGamma_k, \bSigma_\varepsilon)$\textup{:} no correction.
      The sandwich estimator is not applicable to variance components
      estimated from between-state variation; the DER concept does not
      extend to these parameters. Posterior intervals for
      hyperparameters are therefore model-based summaries whose
      calibration depends on correct specification of the hierarchical
      distribution.
  \end{enumerate}
\end{proposition}

\begin{proof}
  Part~(i) follows because the fixed-effect block of
  $\mathbf{H}_{\mathrm{obs}}$ and $\mathbf{J}_{\mathrm{cluster}}$
  depends directly on the observation-level weights $\tilde{w}_i$ and
  PSU aggregation, so $\DER_p$ is determined by the design effect.
  Part~(ii) follows from the posterior precision decomposition: for
  the normal--normal conjugate approximation, the posterior precision
  of $\bepsilon_s$ is the sum of the data precision (design-affected)
  and the prior precision (design-free).  When the prior precision
  dominates, the posterior variance is insensitive to the design, and
  $\DER_s \approx 1$.
  Part~(iii) holds because $\bGamma_k$ and $\bSigma_\varepsilon$ are
  identified from the between-state variation of the $S = 51$
  state-level random effects.  The survey design affects
  \emph{within-state} estimation but does not alter the between-state
  variance structure; hence the sandwich correction has no natural
  analogue for these parameters.
\end{proof}


\noindent
\textbf{Heuristic C.1} (Linear interpolation for random-effect DER).
\label{smc:heuristic-der-re}%
With shrinkage factor
$\lambda_{s,pp}
= \sigma_\varepsilon^2 / (\sigma_\varepsilon^2 + \sigma_e^2 / \ESS_s)$,
the DER for $\varepsilon_{s,p}$ is approximately
\begin{equation}\label{smc:eq-der-heuristic}
  \DER_{s,p}
  \;\approx\;
  1 + \lambda_{s,pp}\,(\DEFF_s - 1),
\end{equation}
interpolating between $\DER_{s,p} = 1$ (strong shrinkage) and
$\DER_{s,p} = \DEFF_s$ (no shrinkage).

\begin{remark}[Overestimation of the linear heuristic]
\label{smc:rem-heuristic-overestimate}
  {\upshape\textbf{(Correction T2-2.)}}
  The linear interpolation~\eqref{smc:eq-der-heuristic} is a
  \emph{heuristic}, not an exact result.  In the one-dimensional
  normal--normal conjugate model, the exact DER is
  \begin{equation}\label{smc:eq-der-exact}
    \DER_{s,p}^{\textup{exact}}
    = \frac{1 + \lambda_{s,pp}\,\DEFF_s}
           {(1 + \lambda_{s,pp})^2},
  \end{equation}
  which is always $\leq$~\eqref{smc:eq-der-heuristic}.  The
  discrepancy is largest at intermediate shrinkage
  ($\lambda_{s,pp} \approx 0.5$), where the linear heuristic can
  overestimate by up to~50\%.  For practical purposes, the heuristic
  remains useful as a conservative upper bound that quickly
  identifies parameters needing correction; the exact
  formula~\eqref{smc:eq-der-exact} can be computed alongside as a
  more accurate benchmark.
\end{remark}


\begin{proposition}[DER diagnostic thresholds]
\label{smc:prop-der-diagnostic}
  As a post-estimation diagnostic, compute $\DER_p$ for each
  parameter and classify as follows\textup{:}
  \begin{itemize}[nosep]
    \item $\DER_p \in [0.8,\;1.2]$\textup{:}
      \emph{design-insensitive}.  The naive credible interval is
      adequately calibrated; sandwich correction is optional.
    \item $\DER_p \in (1.2,\;1.5]$\textup{:}
      \emph{moderate inflation}.  The confidence-interval width
      increases by $\sqrt{1.5} \approx 22\%$; correction is
      recommended.
    \item $\DER_p > 1.5$\textup{:}
      \emph{design-sensitive}.  The sandwich-corrected variance should
      be used as the primary inferential summary.
  \end{itemize}
  In the NSECE application, all 11 fixed-effect parameters have
  $\DER_p > 1.14$, placing them in the moderate-to-sensitive range.
\end{proposition}


\begin{remark}[Intercept identifiability under DER correction]
\label{smc:rem-intercept-id}
  {\upshape\textbf{(Correction T2-4.)}}
  Because the policy moderator vector $\bv_s$ includes an intercept
  column (see \cref{sec:policy}), the fixed
  intercept~$\alpha_1$ and the first column of $\bGamma_k$ have
  partially overlapping roles: the state-specific intercept is
  $\alpha_1 + \gamma_{k,1,1} + \varepsilon_{k,1,s}$.  The DER
  for~$\alpha_1$ includes variance attributable to both the design
  effect on the likelihood and the partial confounding with
  $\gamma_{k,1,1}$.  Formally, the decomposition
  \begin{equation}\label{smc:eq-intercept-decomp}
    [\mathbf{V}_{\mathrm{sand}}]_{\alpha_1,\alpha_1}
    = [\mathbf{V}_{\mathrm{sand}}^{\textup{design}}]_{\alpha_1,\alpha_1}
    + [\mathbf{V}^{\textup{id}}]_{\alpha_1,\alpha_1}
  \end{equation}
  is not cleanly separable; the DER should be interpreted with caution
  for intercept parameters.  The soft identification provided by the
  tighter prior $\Norm(0, 0.5^2)$ on the intercept column of
  $\bGamma_k$ (see \cref{sec:priors}) mitigates this issue by
  attributing the cross-state mean to~$\alpha_1$ and treating
  deviations as random effects.
\end{remark}


\begin{remark}[Computational savings from DER classification]
\label{smc:rem-computational-savings}
  The three-way classification implies that the Cholesky
  transformation~\eqref{eq:cholesky-transform} need only be applied
  to the fixed-effect block ($\dim\boldsymbol{\xi} = 11$) and, for
  the largest states, the random-effect blocks.  This reduces the
  dimensionality of the Cholesky factorization from the full parameter
  dimension ($\approx\!616$) to at most $11$--$20$, yielding
  substantial computational savings.  The empirical DER values are
  reported in the main text (\cref{tab:fixed-effects}).
\end{remark}

%% file: sm_d.tex

This appendix presents extended empirical results from the NSECE~2019
application (\cref{sec:application}).  We provide the full survey weight
distribution (\cref{smd:survey-design}), formal informativeness
assessment following~\citet{Pfeffermann1993}
(\cref{smd:informativeness}), the complete set of state-level poverty
coefficients (\cref{smd:states}), the full $40$-element policy
moderation matrix $\bGamma$ (\cref{smd:gamma}), detailed design
effect diagnostics (\cref{smd:design-effects}), a choropleth map of
reversal probabilities (\cref{smd:reversal-map}), and an LKJ prior
sensitivity analysis (\cref{smd:lkj-sensitivity}).  Convergence diagnostics
and posterior predictive checks are available in the replication package
at \url{https://github.com/joonho112/hurdlebb-replication}.

\input{sm_d1_survey_design}

\input{sm_d2_informativeness}

\input{sm_d5_states}

\input{sm_d6_gamma}

\input{sm_d7_design_effects}

\input{sm_d9_reversal_map}

\input{sm_d8_lkj_sensitivity}

%% file: sm_d1_survey_design.tex

\subsection{Survey design characteristics}
\label{smd:survey-design}

The NSECE~2019~\citep{NSECEProjectTeam2022} uses a stratified cluster design with $H = 30$ strata
and $C = 415$ primary sampling units (PSUs), yielding $N = 6{,}785$
center-based providers after excluding observations with missing
covariates.  Survey weights $w_i = 1/\pi_i$ range from~$1$ to~$462$
with a mean of~$17.5$ and a coefficient of variation of~$1.66$.
\Cref{smd:fig-weights} displays the distribution of normalized
weights $\tilde{w}_i = w_i / \bar{w}$, which are strongly right-skewed:
the $99$th percentile ($\tilde{w} = 8.8$) is nearly nineteen times the
median ($\tilde{w} = 0.47$), indicating that a small number of
providers carry disproportionate influence.

The Kish design effect is $\DEFF_{\mathrm{Kish}} = 1 + \mathrm{CV}_w^2
= 3.76$, reducing the effective sample size from $N = 6{,}785$ to
$\ESS_{\mathrm{Kish}} \approx 1{,}803$.  This substantial penalty
motivates the sandwich variance correction and Cholesky calibration
described in~\cref{sec:survey} and~\cref{app:survey}.

\begin{figure}[htbp]
  \centering
  \includegraphics[width=0.85\textwidth]{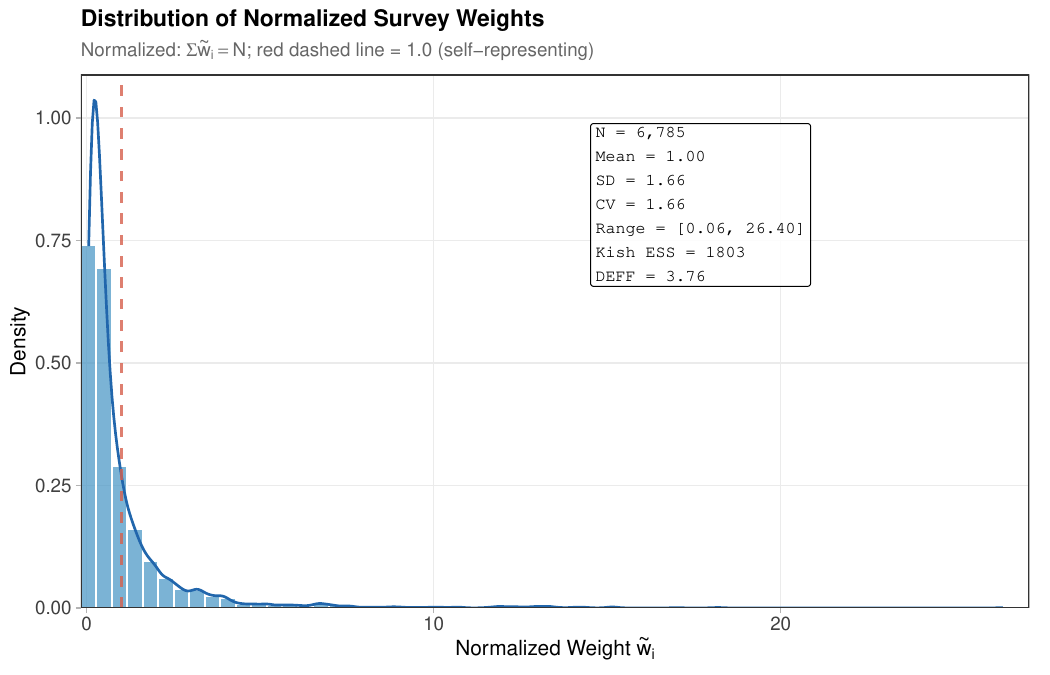}
  \caption{Distribution of normalized survey weights
    $\tilde{w}_i = w_i / \bar{w}$ for $N = 6{,}785$ center-based
    providers.  The red dashed line marks $\tilde{w} = 1.0$
    (self-representing unit).  The distribution is strongly
    right-skewed (CV~$= 1.66$), with Kish $\DEFF = 3.76$ and
    effective sample size $\ESS \approx 1{,}803$.}
  \label{smd:fig-weights}
\end{figure}

%% file: sm_d2_informativeness.tex

\subsection{Informativeness assessment}
\label{smd:informativeness}

A sampling design is \emph{informative} if the distribution of the
outcome among sampled units differs from the superpopulation
distribution, even after conditioning on covariates
\citep{Pfeffermann1993}.  Under informative sampling, unweighted
likelihood-based inference targets a sample-specific, not population,
quantity, and pseudo-likelihood weighting becomes essential for
consistent estimation~\citep{SavitskyToth2016}.  We assess
informativeness of the NSECE~2019 design through three complementary
diagnostics.

{\upshape\textbf{(Correction T3-1.)}}
The original informativeness tests conditioned on a single covariate
(poverty rate).  Since the full model includes $P = 5$ covariates per
margin, omitted variable bias may mask or inflate apparent
informativeness.  We now report both the original univariate and the
corrected full-covariate specifications.

{\upshape\textbf{(Correction T3-2.)}}
The original Hausman comparison used \texttt{lm(weights = w)} for the
weighted estimator, which provides only heteroscedasticity-robust
standard errors.  We now use \texttt{survey::svyglm()} with the full
stratified cluster design, yielding cluster-robust standard errors
that account for within-PSU correlation.

\paragraph{Test~1: Marginal weight--outcome correlations.}
The correlation between $z_i$ (serves infants) and the survey weight
$w_i$ is $\corr(z_i, w_i) = -0.041$; the correlation between
the IT share $Y_i / n_i$ (among servers) and $w_i$ is
$\corr(Y_i/n_i, w_i) = -0.003$.  Both are negligibly small,
suggesting limited \emph{unconditional} informativeness.  However,
marginal correlations can mask conditional dependence when the design
oversamples on covariates correlated with the outcome, motivating the
regression-based test below.

\paragraph{Test~2: Regression-based informativeness test.}
Following~\citet{Pfeffermann1993}, we test $H_0\colon \gamma = 0$ in
the augmented model
\begin{equation}\label{smd:eq-pfeffermann}
  g(y_i) = \bx_i^\t \bbeta + \gamma \log(w_i) + \varepsilon_i,
\end{equation}
where $g(\cdot)$ is the logit link for the extensive margin and the
identity for the intensive margin.  If the design is non-informative,
$\log(w_i)$ carries no additional predictive power for $y_i$ beyond
$\bx_i$, so $\gamma = 0$.

\Cref{smd:tab-pfeffermann} reports the estimated $\hat{\gamma}$ and
its $p$-value under two covariate specifications: (i)~univariate, with
only the standardized poverty rate; and (ii)~full, with all four
community-level predictors (poverty, urbanicity, percent Black, percent
Hispanic).

\begin{table}[htbp]
  \centering
  \caption{Regression-based informativeness test~\citep{Pfeffermann1993}.
    The augmented model \eqref{smd:eq-pfeffermann} adds $\log(w_i)$ to the
    covariate vector; $\hat{\gamma}$ is the coefficient on $\log(w_i)$
    with its two-sided $p$-value.  {\upshape\textbf{(Correction T3-1:)}}
    the full-covariate specification reverses the extensive-margin
    finding from non-significant to significant.}
  \label{smd:tab-pfeffermann}
  \smallskip
  \begin{tabular}{llrr}
    \toprule
    Specification & Margin & $\hat{\gamma}$ & $p$-value \\
    \midrule
    Univariate (poverty only)
      & Extensive  & $0.020$  & $0.369$   \\
      & Intensive  & $0.009$  & $0.002$   \\[4pt]
    Full covariates (T3-1)
      & Extensive  & $0.086$  & $<0.001$  \\
      & Intensive  & $0.006$  & $0.050$   \\
    \bottomrule
  \end{tabular}
\end{table}

Two features of~\cref{smd:tab-pfeffermann} merit comment.  First, under
the univariate specification, the extensive margin appears
non-informative ($p = 0.369$) while the intensive margin is clearly
informative ($p = 0.002$).  Second, adding the remaining three
covariates \emph{reverses} the extensive-margin verdict: $\hat{\gamma}$
increases fourfold (from $0.020$ to $0.086$) and becomes highly
significant ($p < 0.001$).  This reversal arises because the univariate
model confounds the direct weight--outcome relationship with indirect
pathways through the omitted covariates (urbanicity, racial
composition).  When these confounders are included, the residual
association between $\log(w_i)$ and participation becomes detectable.
The intensive margin remains informative under both specifications,
though the coefficient attenuates slightly (from $0.009$ to $0.006$)
as the full model absorbs some of the weight--outcome association.

The finding that both margins are informative under the full
specification reinforces the case for pseudo-likelihood weighting in the
Bayesian hierarchical model and for the sandwich variance correction
developed in~\cref{app:survey}.

\paragraph{Test~3: Hausman-type comparison.}
As a complementary diagnostic, we compare unweighted and weighted
coefficient estimates for each margin.  Under non-informative sampling,
both estimators are consistent for the same population parameters, but
the weighted estimator is less efficient; a large discrepancy signals
that the design distorts the effective covariate--outcome relationship.

{\upshape\textbf{(Correction T3-2:)}} the weighted estimator uses
\texttt{survey::svyglm()} with the full stratified cluster design ($H =
30$ strata, $C = 415$ PSUs), which produces cluster-robust standard
errors.  The Hausman test statistic for each parameter is
\begin{equation}\label{smd:eq-hausman}
  Z_k = \frac{\hat{\beta}_k^{\mathrm{WT}} -
    \hat{\beta}_k^{\mathrm{UW}}}
  {\sqrt{\widehat{\Var}(\hat{\beta}_k^{\mathrm{WT}}) -
    \widehat{\Var}(\hat{\beta}_k^{\mathrm{UW}})}},
\end{equation}
which is asymptotically standard normal under $H_0$ (no
informativeness); the variance difference in the denominator is valid
because the weighted estimator has weakly larger variance under
non-informativeness.

\begin{table}[htbp]
  \centering
  \caption{Hausman-type comparison of unweighted and
    design-weighted coefficient estimates.
    {\upshape\textbf{(Correction T3-2:)}} the weighted estimator uses
    \texttt{survey::svyglm} with cluster-robust standard errors.
    $Z$ is the per-parameter Hausman statistic
    \eqref{smd:eq-hausman}.  Asterisks: ${*}$ denotes $p < 0.10$,
    ${**}$ denotes $p < 0.05$.}
  \label{smd:tab-hausman}
  \smallskip
  \begin{tabular}{llrrrrr}
    \toprule
    Margin & Parameter
      & $\hat{\beta}^{\mathrm{UW}}$ & $\hat{\beta}^{\mathrm{WT}}$
      & $Z$ & $p$-value & \\
    \midrule
    \multirow{5}{*}{Extensive}
      & Intercept & $0.614$  & $0.662$  & $0.82$  & $0.41$ & \\
      & Poverty   & $-0.158$ & $-0.199$ & $-0.64$ & $0.52$ & \\
      & Urban     & $0.176$  & $0.234$  & $2.02$  & $0.044$ & $^{**}$ \\
      & Black     & $-0.044$ & $0.086$  & $1.89$  & $0.059$ & $^{*}$ \\
      & Hispanic  & $-0.043$ & $0.059$  & $1.51$  & $0.13$ & \\[4pt]
    Intensive
      & \multicolumn{5}{l}{\emph{all $p > 0.24$; none significant}} & \\
    \bottomrule
  \end{tabular}
\end{table}

\Cref{smd:tab-hausman} reveals a nuanced pattern.  On the extensive
margin, the poverty coefficient itself does not differ significantly
between estimators ($Z = -0.64$, $p = 0.52$), but the
\emph{urbanicity} coefficient shows a significant shift
($\hat{\beta}^{\mathrm{UW}} = 0.176$ versus
$\hat{\beta}^{\mathrm{WT}} = 0.234$; $Z = 2.02$, $p = 0.044$), and
the Black composition coefficient is marginally significant ($Z = 1.89$,
$p = 0.059$).  This indicates that the NSECE oversampling of
high-need areas primarily distorts the estimated effect of urbanicity
on IT participation---consistent with the fact that the sampling frame
explicitly targets urban and high-poverty PSUs.  On the intensive
margin, no individual parameter shows a significant Hausman
discrepancy (all $p > 0.24$), suggesting that the informativeness
detected by Test~2 is diffuse rather than concentrated on a single
covariate.

\paragraph{Summary.}
The three diagnostics paint a consistent picture.  Marginal correlations
(Test~1) are negligible, but the regression-based test (Test~2) reveals
statistically significant informativeness on \emph{both} margins once
the full covariate vector is conditioned upon---a finding that was
masked in the original univariate specification.  The Hausman comparison
(Test~3) localises the extensive-margin informativeness to the
urbanicity and racial composition covariates rather than poverty itself.
Together, these results justify two design choices in the main
analysis (\cref{sec:application}): (i)~incorporating survey weights
via pseudo-likelihood, and (ii)~applying the sandwich variance
correction (\cref{sec:survey}) and the Cholesky affine calibration
(\cref{app:survey}) to ensure that posterior credible intervals
reflect the true design-adjusted uncertainty.

%% file: sm_d5_states.tex

\subsection{State-level poverty coefficients}
\label{smd:states}

This subsection summarizes the state-level poverty coefficients from
model~M2 (block-diagonal SVC without policy moderators) for all
$S = 51$ states.  The extensive-margin coefficient
$\alpha_{\mathrm{pov},s}$ and intensive-margin coefficient
$\beta_{\mathrm{pov},s}$ are the state-specific effects of a
one-standard-deviation increase in community poverty on the log-odds of
IT participation and the logit of IT intensity, respectively.  The
cross-margin scatter plot in \cref{fig:cross-margin} visualizes the
joint distribution of these estimates.

Of the 51 states, 48 (94\%) have posterior-mean sign patterns consistent
with the classic reversal
($\hat{\alpha}_{\mathrm{pov},s} < 0$ and
$\hat{\beta}_{\mathrm{pov},s} > 0$); this near-universal pattern
reflects hierarchical shrinkage pulling small-state estimates toward
the negative population-average extensive-margin coefficient.
Among these, 15 (29\%) exhibit posterior reversal probabilities
exceeding 0.80, indicating strong evidence beyond the point-estimate
sign, and an additional 21 states have reversal probabilities between
0.60 and 0.80.

\paragraph{States with strongest reversal evidence.}
The five large-sample states with the highest reversal probabilities are
Texas ($N = 578$, $\Pr = 0.989$),
Illinois ($N = 471$, $\Pr = 0.970$),
California ($N = 1{,}110$, $\Pr = 0.937$),
Louisiana ($N = 77$, $\Pr = 0.923$), and
New Jersey ($N = 201$, $\Pr = 0.908$).
In these states, the extensive-margin coefficients
($\alpha_{\mathrm{pov},s} \in [-0.34, -0.13]$) are firmly negative
with 90\% credible intervals excluding zero, while the
intensive-margin coefficients
($\beta_{\mathrm{pov},s} \in [0.07, 0.08]$) are positive with
intervals barely including zero.

\paragraph{Notable exception: New York.}
New York ($N = 558$) stands out as the most prominent exception, with a
\emph{positive} extensive-margin coefficient
($\alpha_{\mathrm{pov}} = 0.101$; 90\% CI $[-0.051, 0.261]$,
$\Pr(\text{reversal}) = 0.143$), suggesting that higher poverty is
associated with \emph{greater} likelihood of serving infants and
toddlers---possibly reflecting New York's substantial public investment
in early childhood programs.  Arizona ($\Pr = 0.40$), Minnesota
($\Pr = 0.44$), and North Dakota ($\Pr = 0.48$) also show weak or
absent reversal patterns.

\paragraph{Small-state shrinkage.}
Small states (e.g., Maine, $N = 18$; Mississippi, $N = 17$; Montana,
$N = 22$) exhibit wider credible intervals and reversal probabilities
closer to the national average ($\approx 0.73$), reflecting the
expected shrinkage toward the hierarchical mean documented
in~\cref{app:survey}.  The full 51-state coefficient table is available
in the replication package at
\url{https://github.com/joonho112/hurdlebb-replication}.

%% file: sm_d6_gamma.tex

\subsection{Full policy moderation parameters}
\label{smd:gamma}

\Cref{smd:tab-gamma} presents the complete set of $40$ elements of the
cross-level moderation matrices $\bGamma_1$ (extensive margin) and
$\bGamma_2$ (intensive margin) from the weighted M3b model.  The main
text (\cref{sec:policyresults}) summarizes three key patterns; here we
provide the full matrices for all $5$~covariates $\times$
$4$~policy dimensions $\times$ $2$~margins.

The four policy dimensions are: (i)~market rate percentile
(MR\_pctile\_std, standardized), measuring how generous each state's
reimbursement rate is relative to provider market rates;
(ii)~TieredReim, an indicator for states with tiered reimbursement
quality systems; (iii)~ITaddon, an indicator for states offering
additional subsidies for infant/toddler care; and (iv)~the intercept
column, capturing residual cross-state variation unrelated to the
three measured policies (the intercept column of $\bGamma_k$
is softly identified by its prior; see \cref{sec:priors}).

Two patterns merit attention beyond those highlighted in the main text.
First, the IT~addon variable shows the most consistent effects across
both margins: it is positively associated with the extensive-margin
intercept ($\gamma = 0.728$, $\Pr(>0) = 0.996$) and the
extensive-margin Hispanic coefficient ($\gamma = 0.750$,
$\Pr(>0) = 0.999$), while being strongly negatively associated with
the extensive-margin poverty coefficient ($\gamma = -0.488$,
$\Pr(>0) = 0.018$).  On the intensive margin, the IT~addon is
associated with a notably higher Black coefficient ($\gamma = 0.256$,
$\Pr(>0) = 1.000$).

Second, the intercept columns for both margins have wide
posteriors (SD~$\approx 0.93$), confirming the soft
identification: the hierarchical prior
($\Norm(0, 0.5^2)$ for the intercept column, tighter than the
$\Norm(0, 1^2)$ for policy columns) constrains these parameters
without driving them to zero.

\begin{table}[htbp]
\centering
\scriptsize
\begin{adjustbox}{max width=\textwidth}
\begin{tabular}{@{}lll rrrr r@{}}
  \toprule
  Margin & Covariate & Policy & Mean & SD & Q025 & Q975 & $\Pr({>}0)$ \\
  \midrule
  \multicolumn{8}{@{}l}{\emph{Extensive margin} ($\bGamma_1$)} \\[2pt]
  ext & intercept & intercept & 0.127 & 0.937 & $-$1.684 & 1.960 & 0.558 \\
  ext & intercept & MR\_pctile & 0.018 & 0.142 & $-$0.268 & 0.297 & 0.543 \\
  ext & intercept & TieredReim & $-$0.413 & 0.359 & $-$1.121 & 0.278 & 0.128 \\
  ext & intercept & ITaddon & 0.728 & 0.285 & 0.182 & 1.297 & 0.996 \\[2pt]
  ext & poverty & intercept & $-$0.044 & 0.933 & $-$1.845 & 1.789 & 0.483 \\
  ext & poverty & MR\_pctile & 0.064 & 0.107 & $-$0.152 & 0.271 & 0.732 \\
  ext & poverty & TieredReim & 0.495 & 0.272 & $-$0.064 & 1.021 & 0.962 \\
  ext & poverty & ITaddon & $-$0.488 & 0.225 & $-$0.917 & $-$0.035 & 0.018 \\[2pt]
  ext & urban & intercept & 0.073 & 0.945 & $-$1.766 & 1.886 & 0.531 \\
  ext & urban & MR\_pctile & 0.011 & 0.062 & $-$0.112 & 0.129 & 0.574 \\
  ext & urban & TieredReim & $-$0.248 & 0.175 & $-$0.601 & 0.094 & 0.076 \\
  ext & urban & ITaddon & $-$0.205 & 0.146 & $-$0.488 & 0.087 & 0.072 \\[2pt]
  ext & black & intercept & 0.085 & 0.929 & $-$1.740 & 1.927 & 0.530 \\
  ext & black & MR\_pctile & $-$0.283 & 0.144 & $-$0.565 & 0.010 & 0.026 \\
  ext & black & TieredReim & $-$0.652 & 0.339 & $-$1.346 & 0.014 & 0.029 \\
  ext & black & ITaddon & 0.365 & 0.276 & $-$0.193 & 0.890 & 0.910 \\[2pt]
  ext & hispanic & intercept & 0.012 & 0.925 & $-$1.799 & 1.850 & 0.509 \\
  ext & hispanic & MR\_pctile & 0.072 & 0.134 & $-$0.198 & 0.331 & 0.712 \\
  ext & hispanic & TieredReim & $-$0.351 & 0.350 & $-$1.052 & 0.318 & 0.157 \\
  ext & hispanic & ITaddon & 0.750 & 0.259 & 0.260 & 1.290 & 0.999 \\
  \midrule
  \multicolumn{8}{@{}l}{\emph{Intensive margin} ($\bGamma_2$)} \\[2pt]
  int & intercept & intercept & $-$0.074 & 0.922 & $-$1.899 & 1.735 & 0.468 \\
  int & intercept & MR\_pctile & $-$0.024 & 0.039 & $-$0.104 & 0.052 & 0.265 \\
  int & intercept & TieredReim & 0.271 & 0.120 & 0.026 & 0.497 & 0.985 \\
  int & intercept & ITaddon & 0.001 & 0.084 & $-$0.158 & 0.173 & 0.488 \\[2pt]
  int & poverty & intercept & 0.019 & 0.935 & $-$1.783 & 1.831 & 0.505 \\
  int & poverty & MR\_pctile & 0.028 & 0.026 & $-$0.024 & 0.079 & 0.865 \\
  int & poverty & TieredReim & $-$0.064 & 0.064 & $-$0.193 & 0.060 & 0.156 \\
  int & poverty & ITaddon & $-$0.087 & 0.051 & $-$0.187 & 0.014 & 0.041 \\[2pt]
  int & urban & intercept & $-$0.015 & 0.932 & $-$1.821 & 1.783 & 0.497 \\
  int & urban & MR\_pctile & 0.027 & 0.017 & $-$0.007 & 0.060 & 0.941 \\
  int & urban & TieredReim & 0.015 & 0.054 & $-$0.090 & 0.122 & 0.612 \\
  int & urban & ITaddon & 0.039 & 0.040 & $-$0.035 & 0.120 & 0.837 \\[2pt]
  int & black & intercept & 0.029 & 0.912 & $-$1.731 & 1.848 & 0.503 \\
  int & black & MR\_pctile & 0.015 & 0.036 & $-$0.056 & 0.082 & 0.662 \\
  int & black & TieredReim & 0.047 & 0.088 & $-$0.126 & 0.224 & 0.706 \\
  int & black & ITaddon & 0.256 & 0.072 & 0.117 & 0.397 & 1.000 \\[2pt]
  int & hispanic & intercept & $-$0.045 & 0.936 & $-$1.853 & 1.811 & 0.482 \\
  int & hispanic & MR\_pctile & $-$0.115 & 0.054 & $-$0.224 & $-$0.010 & 0.016 \\
  int & hispanic & TieredReim & 0.227 & 0.164 & $-$0.088 & 0.554 & 0.920 \\
  int & hispanic & ITaddon & $-$0.135 & 0.108 & $-$0.349 & 0.077 & 0.104 \\
  \bottomrule
\end{tabular}
\end{adjustbox}
\caption{Complete policy moderation parameters (weighted M3b, all
  $40$~elements of $\bGamma_1$ and $\bGamma_2$).  Each row reports the
  posterior mean, standard deviation, 95\% credible interval, and
  posterior probability of a positive effect.  The main text
  (\cref{sec:policyresults}) summarizes three key
  patterns; bold findings correspond to $\Pr({>}0) > 0.95$ or
  $\Pr({>}0) < 0.05$.}
\label{smd:tab-gamma}
\end{table}

%% file: sm_d7_design_effects.tex

\subsection{Design effect diagnostics}
\label{smd:design-effects}

The design effect ratio (DER), defined in \eqref{eq:der} and
classified in~\cref{smc:prop-der-classification}, quantifies the
factor by which the survey design inflates each parameter's variance
beyond the model-based (SRS) variance.  \Cref{smd:tab-der} presents
the full variance decomposition for all 11 fixed-effect parameters,
and~\cref{smd:fig-sandwich} provides a three-panel diagnostic
visualization.

\paragraph{Variance decomposition.}
The three variance quantities in~\cref{smd:tab-der} deserve careful
interpretation.
$[\mathbf{H}_{\mathrm{obs}}^{-1}]_{pp}$ is the model-based variance
obtained from the observed information matrix---the appropriate
benchmark under simple random sampling.
$[\mathbf{V}_{\mathrm{sand}}]_{pp}$ is the sandwich-corrected variance
\citep{Binder1983} that accounts for the complex survey design, and constitutes the
primary inferential summary.
$[\bSigma_{\mathrm{MCMC}}]_{pp}$ is the raw MCMC posterior variance,
which is substantially inflated (by factors of 600--6{,}500) relative
to $\mathbf{H}_{\mathrm{obs}}^{-1}$ because the hierarchical prior
dominates the pseudo-likelihood contribution for fixed effects.

The DER ranges from 1.14 ($\beta_{\mathrm{Hisp}}$) to 4.18
($\alpha_{\mathrm{Urban}}$), with a mean of 2.11.  This is broadly
consistent with the Kish $\DEFF = 3.76$: the DER is systematically
smaller because the state-varying coefficients partially absorb
within-state clustering (\cref{smc:prop-svc-absorption}).  The
corresponding confidence interval inflation factors are
$\sqrt{\DER} \in [1.07, 2.04]$, so the widest Wald intervals are
approximately twice those of a hypothetical SRS design.

\paragraph{Diagnostic figure.}
\Cref{smd:fig-sandwich} decomposes the sandwich correction into three
components.  Panel~(a) shows the hierarchical variance ratio---the ratio of
naive (MCMC-based) to Wald (sandwich-based) confidence interval
widths---on a log scale; all parameters exceed a ratio of~10, and
several exceed~100, confirming that the raw MCMC posterior is
uninformative about the design-consistent uncertainty.  Panel~(b)
compares the DER (based on $\mathbf{H}_{\mathrm{obs}}^{-1}$) with the
``DER\textsubscript{MCMC}'' (based on $\bSigma_{\mathrm{MCMC}}$); the
latter is near zero for all parameters except $\log\kappa$, further
illustrating hierarchical inflation.  Panel~(c) reports the diagonal elements
of the Cholesky scaling matrix $\mathbf{A}$ (\cref{thm:cholesky})
for the fixed-effects block.
For the $\alpha$ and $\beta$ coefficients, the diagonal entries are
all below $0.06$, indicating substantial contraction of the MCMC draws, a direct consequence of the
hierarchically inflated posterior being much wider than the design-consistent
target.

\begin{table}[htbp]
\centering
\footnotesize
\begin{adjustbox}{max width=\textwidth}
\begin{tabular}{@{}l rrr rr rrr@{}}
  \toprule
  Parameter &
  $H_{\mathrm{obs}}^{-1}$ &
  $V_{\mathrm{sand}}$ &
  $\Sigma_{\mathrm{MCMC}}$ &
  DER &
  DER$_{\text{mc}}$ &
  SD($H$) &
  SD(Sand) &
  SD(MCMC) \\
  \midrule
  $\alpha_{\mathrm{Int}}$ & 0.00097 & 0.00198 & 0.946 & 2.04 & 0.002 & 0.031 & 0.044 & 0.972 \\
  $\alpha_{\mathrm{Pov}}$ & 0.00114 & 0.00275 & 0.900 & 2.41 & 0.003 & 0.034 & 0.052 & 0.949 \\
  $\alpha_{\mathrm{Urb}}$ & 0.00028 & 0.00118 & 0.907 & 4.18 & 0.001 & 0.017 & 0.034 & 0.952 \\
  $\alpha_{\mathrm{Blk}}$ & 0.00127 & 0.00195 & 0.947 & 1.54 & 0.002 & 0.036 & 0.044 & 0.973 \\
  $\alpha_{\mathrm{Hisp}}$ & 0.00146 & 0.00265 & 0.920 & 1.82 & 0.003 & 0.038 & 0.051 & 0.959 \\
  $\beta_{\mathrm{Int}}$ & 0.00017 & 0.00023 & 0.860 & 1.36 & 0.000 & 0.013 & 0.015 & 0.927 \\
  $\beta_{\mathrm{Pov}}$ & 0.00022 & 0.00036 & 0.875 & 1.61 & 0.000 & 0.015 & 0.019 & 0.936 \\
  $\beta_{\mathrm{Urb}}$ & 0.00013 & 0.00042 & 0.873 & 3.18 & 0.000 & 0.012 & 0.021 & 0.934 \\
  $\beta_{\mathrm{Blk}}$ & 0.00024 & 0.00033 & 0.845 & 1.38 & 0.000 & 0.016 & 0.018 & 0.919 \\
  $\beta_{\mathrm{Hisp}}$ & 0.00025 & 0.00028 & 0.892 & 1.14 & 0.000 & 0.016 & 0.017 & 0.945 \\
  $\log\kappa$ & 0.00033 & 0.00086 & 0.001 & 2.58 & 1.119 & 0.018 & 0.029 & 0.028 \\
  \bottomrule
\end{tabular}
\end{adjustbox}
\caption{Design effect ratio diagnostics for all 11 fixed-effect
  parameters.  Columns: $H_{\mathrm{obs}}^{-1}$ = model-based variance
  (observed information inverse), $V_{\mathrm{sand}}$ = sandwich
  variance, $\Sigma_{\mathrm{MCMC}}$ = raw MCMC posterior variance,
  DER = $V_{\mathrm{sand}} / H_{\mathrm{obs}}^{-1}$,
  DER$_{\text{mc}}$ = $V_{\mathrm{sand}} / \Sigma_{\mathrm{MCMC}}$,
  and the corresponding standard deviations.  The DER range 1.14--4.18
  (mean 2.11) is consistent with the Kish $\DEFF = 3.76$.}
\label{smd:tab-der}
\end{table}

\begin{figure}[htbp]
  \centering
  \includegraphics[width=\textwidth,height=0.75\textheight,keepaspectratio]{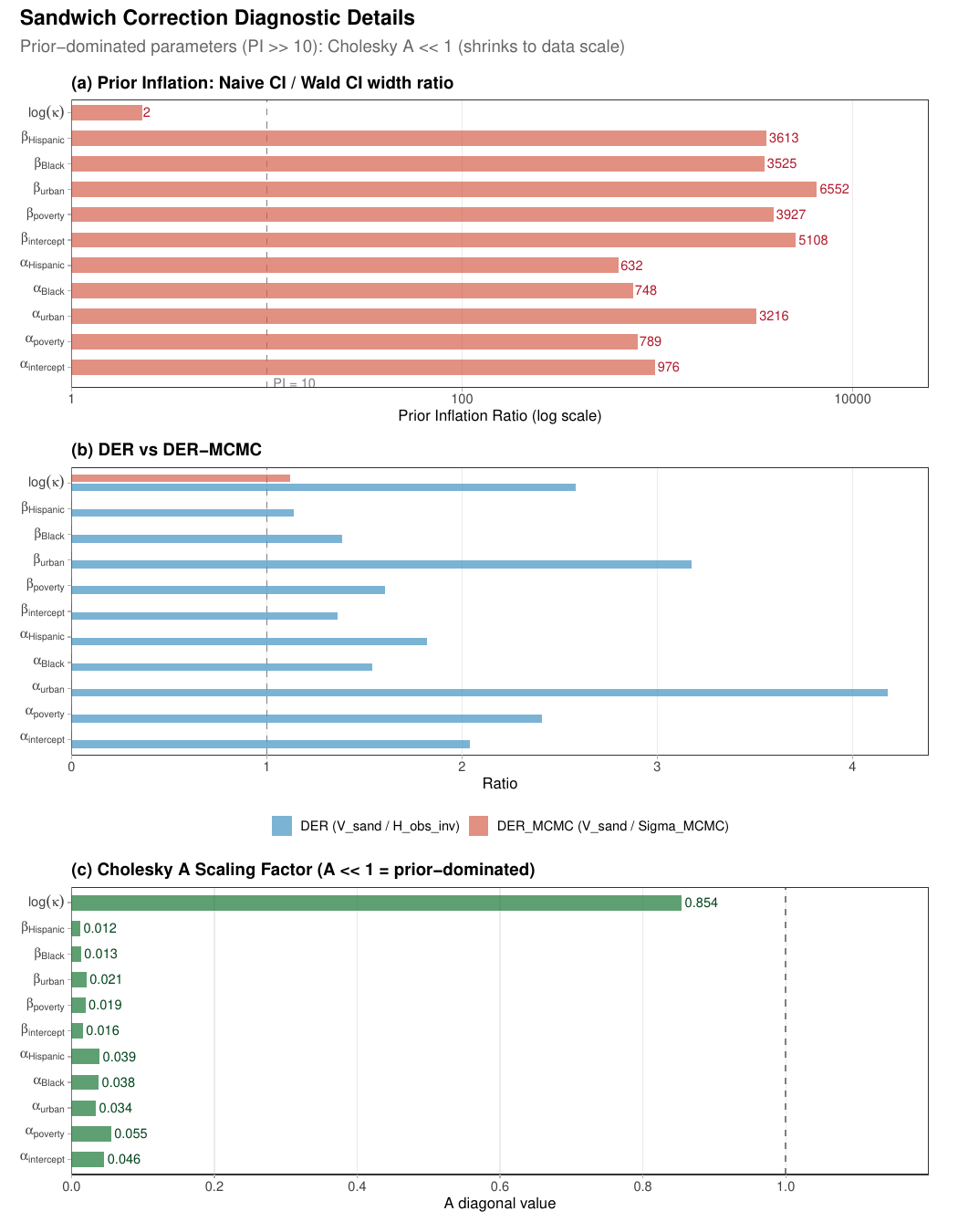}
  \caption{Sandwich correction diagnostic details.
    (a)~Hierarchical variance ratio (log scale): ratio of naive (MCMC-based)
    to Wald (sandwich-based) confidence interval widths.  Values above
    the dashed line at~10 indicate strong hierarchical inflation.
    (b)~DER comparison: design effect ratio computed using
    $\mathbf{H}_{\mathrm{obs}}^{-1}$ (primary) versus
    $\bSigma_{\mathrm{MCMC}}$ (diagnostic only).
    (c)~Cholesky $\mathbf{A}$ diagonal: scaling factors from the
    affine transformation (\cref{thm:cholesky}).  Values $\ll 1$
    indicate that the transformation substantially contracts the MCMC
    draws toward the design-consistent target.}
  \label{smd:fig-sandwich}
\end{figure}

%% file: sm_d9_reversal_map.tex

\subsection{Reversal probability map}
\label{smd:reversal-map}

\Cref{smd:fig-reversal-map} maps the posterior probability of the poverty
reversal for each state:
$\Pr(\tilde{\alpha}_{\mathrm{pov},s} < 0 \text{ and }
\tilde{\beta}_{\mathrm{pov},s} > 0 \mid \text{data})$, computed from the
unweighted M2 model (block-diagonal SVC).  States with high reversal
probability ($\Pr > 0.9$) are concentrated in the South, Southwest, and
large urban states; states with low probability ($\Pr < 0.1$) are
primarily in the upper Midwest and Mountain West.  The remaining states
($0.1 \le \Pr \le 0.9$) tend to be smaller states where the posterior
intervals are wide due to limited sample sizes.

\begin{figure}[htbp]
  \centering
  \includegraphics[width=0.85\textwidth]{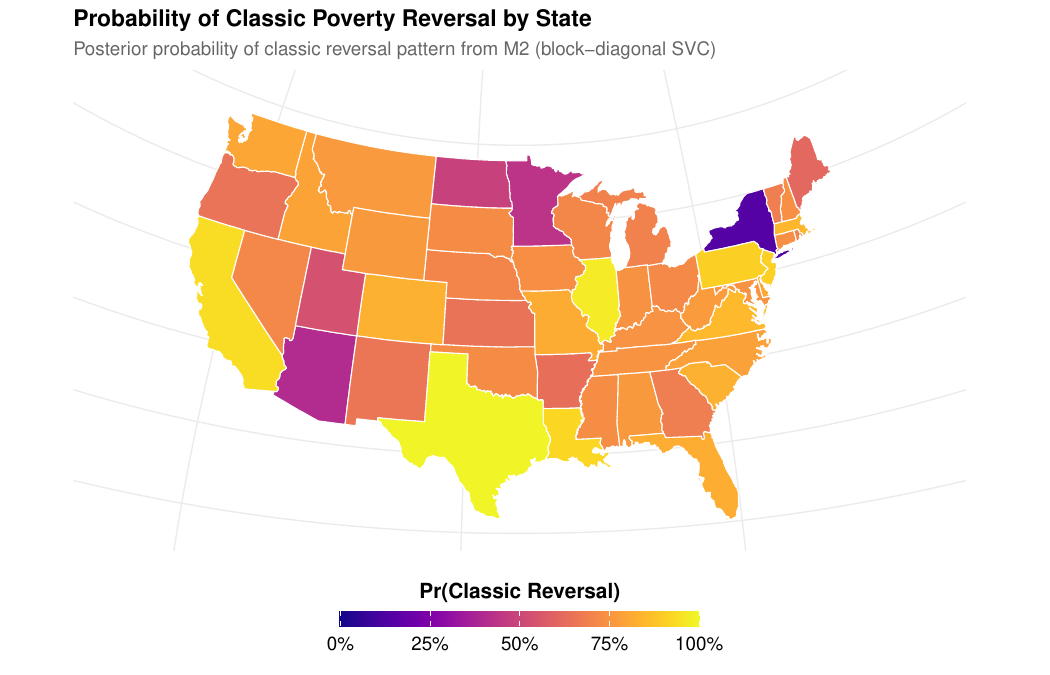}
  \caption{Choropleth map of the posterior reversal probability from the
    unweighted M2 model (block-diagonal SVC):
    $\Pr(\tilde{\alpha}_{\mathrm{pov},s} < 0 \;\cap\;
    \tilde{\beta}_{\mathrm{pov},s} > 0 \mid \text{data})$.  Darker shading
    indicates higher reversal probability.  The population-average
    probability of reversal is 1.000.}\label{smd:fig-reversal-map}
\end{figure}

%% file: sm_d8_lkj_sensitivity.tex

\subsection{LKJ prior sensitivity}
\label{smd:lkj-sensitivity}

To assess sensitivity to the LKJ prior on the $10 \times 10$
random-effect correlation matrix~$\Omega$, we refit M3b-W under
seven concentration parameters
$\eta \in \{1.0,\, 1.5,\, 2.0,\, 3.0,\, 4.0,\, 6.0,\, 8.0\}$.
The baseline specification uses $\eta = 2$, which places a weakly
regularizing distribution favoring moderate correlations; $\eta = 1$
yields a uniform distribution over correlation matrices, while higher
values concentrate mass near the identity.

\cref{tab:lkj-sensitivity} reports posterior summaries for the
cross-margin correlation $\rho_{\mathrm{cross}} = \Omega_{1,6}$,
the poverty coefficients ($\alpha_{\mathrm{pov}}$, $\beta_{\mathrm{pov}}$),
and the intercept random-effect standard deviations
($\tau_1^{\mathrm{ext}}$, $\tau_1^{\mathrm{int}}$).
Across the sevenfold range of~$\eta$, $\hat{\rho}_{\mathrm{cross}}$
varies between $0.097$ and $0.191$, remaining near zero throughout
and well within the posterior uncertainty of any single specification.
The 95\% credible interval includes zero in every case.
The poverty coefficients are equally stable:
$\alpha_{\mathrm{pov}} \in [-0.359,\, -0.310]$ and
$\beta_{\mathrm{pov}} \in [0.074,\, 0.132]$,
confirming that the sign reversal is robust to the LKJ prior.

LOO-CV model fit is nearly identical across specifications
(ELPD spread ${}<5$), with all MCMC diagnostics satisfactory
(the single marginal R-hat exceedance at $\eta = 6$ is $1.011$,
well within practical tolerance).

\cref{smd:fig-lkj-density} overlays the posterior densities of
$\rho_{\mathrm{cross}}$ across all seven specifications.  The
densities are nearly superimposed, confirming that the data
dominate the prior for this parameter despite the moderate
cluster size ($S = 51$).

These results confirm that the near-zero cross-margin correlation
and the poverty reversal pattern are data-driven findings,
not artifacts of the LKJ(2) prior specification.

\input{Figures/ST_lkj_sensitivity}

\begin{figure}[htbp]
  \centering
  \includegraphics[width=0.85\textwidth]{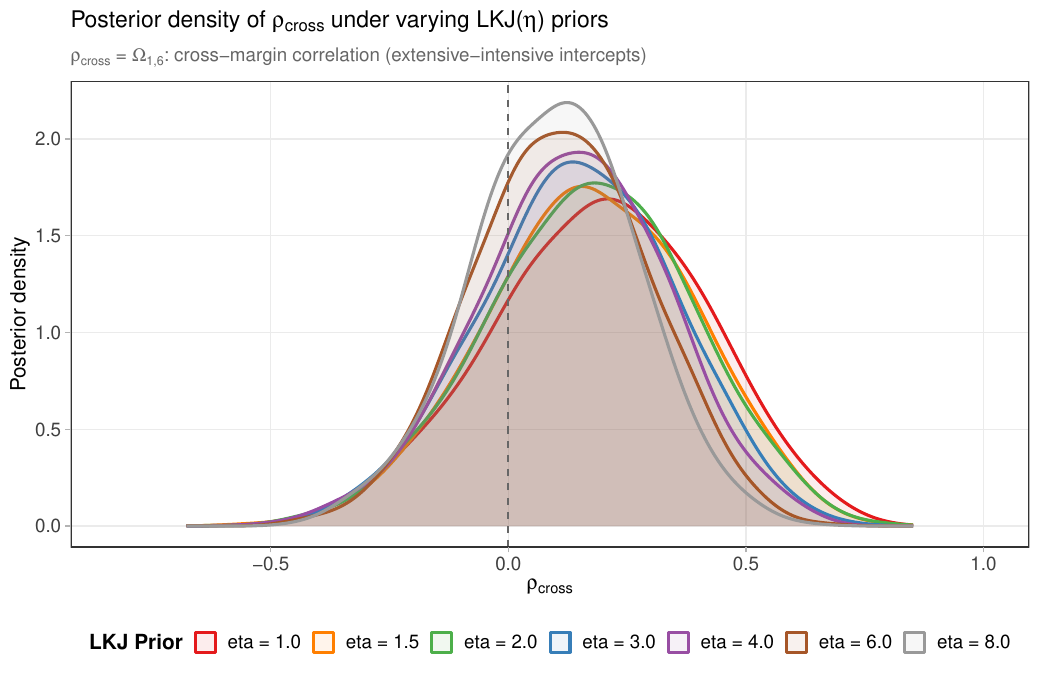}
  \caption{Posterior density of the cross-margin correlation
    $\rho_{\mathrm{cross}} = \Omega_{1,6}$ under seven LKJ
    concentration parameters $\eta \in \{1, 1.5, 2, 3, 4, 6, 8\}$.
    The densities are nearly superimposed, indicating that the
    posterior is data-dominated rather than prior-driven.
    The vertical dashed line marks zero.}
  \label{smd:fig-lkj-density}
\end{figure}

%% file: Figures/ST_lkj_sensitivity.tex
\begin{table}[t]
\centering
\caption{LKJ prior sensitivity analysis. Each row reports posterior summaries from M3b-W fitted with 
  $\mathrm{LKJ}(\eta)$ prior on the $10 \times 10$ random-effect correlation matrix $\Omega$.
  The baseline specification is $\eta = 2$ (uniform over correlations). 
  $\rho_{\mathrm{cross}} = \Omega_{1,6}$ is the cross-margin correlation between extensive and intensive intercepts.
  $\alpha_{\mathrm{pov}}$ and $\beta_{\mathrm{pov}}$ are the poverty coefficients for the extensive and intensive margins.
  $\tau_1^{\mathrm{ext}}$ and $\tau_1^{\mathrm{int}}$ are the intercept random-effect standard deviations.
  All covariates are standardized.}
\label{tab:lkj-sensitivity}
\smallskip
\small
\begin{adjustbox}{max width=\textwidth}
\begin{tabular}{@{}r r@{\hspace{4pt}}l rr rr r c@{}}
\toprule
$\eta$ & $\hat{\rho}_{\mathrm{cross}}$ & [95\% CI] & $\alpha_{\mathrm{pov}}$ & $\beta_{\mathrm{pov}}$ & $\tau_1^{\mathrm{ext}}$ & $\tau_1^{\mathrm{int}}$ & $\widehat{\mathrm{elpd}}_{\mathrm{loo}}$ & Diag.~\\
\midrule
$1.0$ & $0.191$ & $[-0.261,\; 0.604]$ & $-0.322$ & $0.088$ & $0.693$ & $0.149$ & $-20913.0$ & \checkmark \\
$1.5$ & $0.172$ & $[-0.261,\; 0.572]$ & $-0.312$ & $0.074$ & $0.683$ & $0.146$ & $-20915.0$ & \checkmark \\
$2.0^{\star}$ & $0.170$ & $[-0.259,\; 0.573]$ & $-0.320$ & $0.098$ & $0.682$ & $0.144$ & $-20915.7$ & \checkmark \\
$3.0$ & $0.145$ & $[-0.269,\; 0.527]$ & $-0.310$ & $0.132$ & $0.666$ & $0.143$ & $-20917.2$ & \checkmark \\
$4.0$ & $0.134$ & $[-0.274,\; 0.516]$ & $-0.330$ & $0.093$ & $0.653$ & $0.143$ & $-20916.2$ & \checkmark \\
$6.0$ & $0.111$ & $[-0.247,\; 0.460]$ & $-0.359$ & $0.113$ & $0.646$ & $0.142$ & $-20914.8$ & \textbf{!} \\
$8.0$ & $0.097$ & $[-0.256,\; 0.429]$ & $-0.340$ & $0.110$ & $0.639$ & $0.140$ & $-20915.4$ & \checkmark \\
\bottomrule
\end{tabular}
\end{adjustbox}

\medskip
{\footnotesize $^{\star}$Baseline specification ($\eta = 2$ corresponds to a uniform distribution over correlation matrices).}
\end{table}

%% file: sm_e.tex

This appendix provides computational and simulation details supplementing
\cref{sec:model,sec:simulation}.  \Cref{sme:model-hierarchy} presents the
five-model hierarchy with parameter counts.  \Cref{sme:mcmc} describes
the MCMC configuration and runtime benchmarks.  \Cref{smf:dgp} specifies
the data-generating process for the simulation study,
\cref{smf:diagnostics} presents supplementary simulation diagnostic
figures, \cref{smf:misspec} evaluates robustness under model
misspecification, \cref{smf:frequentist} provides a frequentist
contextualization, \cref{smf:decomposition} decomposes the coverage
gap into bias and width components, and \cref{smf:coverage95} reports
coverage at the 95\% nominal level.  Implementation details (non-centered parameterization,
numerical stability, score computation) are documented in the
\textsf{hurdlebb} package vignettes; full simulation results and
replication instructions are available in the replication package at
\url{https://github.com/joonho112/hurdlebb-replication}.

\input{sm_e1_model_hierarchy}

\input{sm_e4_mcmc}

\input{sm_f1_dgp}

\input{sm_f3_diagnostics}

\input{sm_f5_misspec}

\input{sm_f6_frequentist}

\input{sm_f7_decomposition}

\input{sm_f8_coverage95}

%% file: sm_e1_model_hierarchy.tex

\subsection{Model hierarchy and architecture}
\label{sme:model-hierarchy}

\Cref{sme:tab-hierarchy} summarizes the five nested models fitted in
the application (\cref{sec:application}).  Each model extends its
predecessor by one structural feature: random intercepts (M1),
state-varying coefficients with block-diagonal covariance (M2),
cross-margin covariance (M3a), and policy moderators (M3b).  The
final model M3b-W adds survey weights to the pseudo-posterior.
Parameter counts include all fixed effects, random effects, and
hyperparameters (Cholesky factors, scales, and correlations).

The models are fitted sequentially using warm-start initialization:
the posterior mean of model~$k$ provides the initial values for
model~$k+1$.  Random-effect parameters absent in the simpler model
are initialized at zero (NCP variates) or identity (Cholesky factors).
This progressive strategy reduces total computation time by
approximately~$40\%$ compared with cold starts and ensures that each
model addresses a question raised by its predecessor:
M0 to M1 asks whether baseline enrolment rates vary by state;
M1 to M2, whether the poverty reversal varies by state;
M2 to M3a, whether cross-margin covariance improves fit;
and M3a to M3b, whether observable policies explain the state
heterogeneity.

\begin{table}[htbp]
  \centering
  \small
  \begin{tabular}{@{}l l r l r@{}}
    \toprule
    Model & Structure & Parameters
      & Covariance & $\widehat{\elpd}_{\mathrm{loo}}$ \\
    \midrule
    M0  & Pooled HBB
        & 11
        & ---
        & $-20{,}945$ \\
    M1  & Random intercepts
        & 116
        & $2 \times 2$
        & $-20{,}627$ \\
    M2  & Block-diagonal SVC
        & 551
        & $(5 \times 5) \oplus (5 \times 5)$
        & $-20{,}604$ \\
    M3a & Full SVC
        & 576
        & $10 \times 10$
        & $-20{,}602$ \\
    M3b & Policy moderators
        & 616
        & $10 \times 10$
        & $-20{,}608$ \\
    \bottomrule
  \end{tabular}
  \caption{Model hierarchy.  Parameters: total number of free
    parameters (fixed effects + random effects + hyperparameters).
    Covariance: dimension of the state-level covariance matrix for
    the random-effect deviations~$\bdelta_s$.  LOO: expected
    log-predictive density (\cref{sec:application}).  The
    progression from M0 to M3a yields significant incremental
    improvements in predictive accuracy; M3b adds policy
    moderators without further predictive gain.}
  \label{sme:tab-hierarchy}
\end{table}

%% file: sm_e4_mcmc.tex

\subsection{MCMC configuration}
\label{sme:mcmc}

All models are fitted using Stan
\citep{CarpenterEtAl2017} with the No-U-Turn Sampler
\citep[NUTS;][]{Hoffman2014} via the \textsf{CmdStanR} interface
\citep{CmdStanR2024}.  \Cref{sme:tab-mcmc} reports the sampler
configuration.

\begin{table}[htbp]
  \centering
  \begin{tabular}{@{}l ll@{}}
    \toprule
    Setting & Empirical (M0--M3b) & Simulation \\
    \midrule
    Chains                  & 4 (parallel) & 4 (parallel) \\
    Warmup iterations       & 1{,}000 per chain & 1{,}500 per chain \\
    Sampling iterations     & 1{,}000 per chain & 2{,}000 per chain \\
    Total post-warmup draws & 4{,}000 & 8{,}000 \\
    \texttt{adapt\_delta}    & 0.95 & 0.95 \\
    \texttt{max\_treedepth}  & 12 & 14 \\
    \bottomrule
  \end{tabular}
  \caption{MCMC configuration.  The empirical case study (M0--M3b)
    uses the default \textsf{hurdlebb} settings; the simulation study
    uses longer runs and higher \texttt{max\_treedepth} to ensure
    adequate exploration across 200 replications.  In both cases,
    \texttt{adapt\_delta} $= 0.95$ reduces the step size to avoid
    divergent transitions in the state-varying coefficient models.}
  \label{sme:tab-mcmc}
\end{table}

\paragraph{Convergence targets.}
As stated in~\cref{alg:estimation}, we require the following
diagnostic thresholds for all parameters~\citep{Vehtari2021}:
(i)~split-$\hat{R} < 1.01$;
(ii)~minimum bulk $\ESS > 400$;
(iii)~minimum tail $\ESS > 400$; and
(iv)~zero divergent transitions.
All five models satisfy these thresholds; detailed diagnostics
are available in the replication package.

\paragraph{Warm-start initialization.}
The five-model sequence M0 $\to$ M1 $\to$ M2 $\to$ M3a $\to$ M3b
is fitted progressively: the posterior mean from model~$k$ is used
to initialize model~$k+1$.  Specifically:
\begin{itemize}
  \item Fixed-effect coefficients ($\balpha$, $\bbeta$,
    $\log\kappa$) are initialized at the predecessor's posterior
    means.
  \item Scale parameters ($\boldsymbol{\tau}$) are initialized at the
    predecessor's posterior means when dimensions match, or at~$0.1$
    when new dimensions are introduced.
  \item Cholesky factors ($\mathbf{L}_\varepsilon$) are initialized at
    the identity matrix, which corresponds to zero correlation---a
    conservative starting point that lets the sampler discover
    correlations during warmup.
  \item NCP variates ($\bz_s$) are initialized at zero, placing all
    states at the global mean.
  \item Policy moderator coefficients ($\bGamma$, M3b only) are
    initialized at the predecessor's posterior means or at zero for
    newly introduced parameters.
\end{itemize}
This strategy avoids the slow exploration that can occur when complex
models are initialized far from the posterior mode.

\paragraph{Runtime benchmarks.}
On a 16-core workstation, approximate wall-clock times are:
M0~($\sim$5 minutes),
M1~($\sim$15 minutes),
M2~($\sim$30 minutes),
M3a~($\sim$50 minutes),
M3b~($\sim$60 minutes),
and M3b-W~($\sim$70 minutes, including score extraction in the
\texttt{generated quantities} block).  The total sequential fitting
time is approximately 4~hours with warm-start initialization, compared
to approximately 7~hours with cold starts.

%% file: sm_f1_dgp.tex

\subsection{Data-generating process and calibration}
\label{smf:dgp}

This section provides the complete data-generating process (DGP)
for the simulation study in~\cref{sec:simulation}.
The DGP is designed to match the NSECE 2019 data structure while
allowing parametric control over survey informativeness.

\paragraph{Population generation.}
A finite super-population of $M = 50{,}000$ providers is distributed
across $S = 51$ states, with state sizes proportional to the
observed NSECE state counts.  For each provider~$i$ in state~$s$:
\begin{enumerate}
  \item Draw the covariate vector
    $\bx_i = (1,\, x_{\mathrm{pov}},\, x_{\mathrm{urb}},\,
    x_{\mathrm{blk}},\, x_{\mathrm{his}})^\t$ by bootstrap
    resampling from the empirical NSECE marginals within state~$s$.
  \item Draw the trial size $n_i$ from the empirical distribution
    of total 0--5 enrollment in the NSECE data.
\end{enumerate}

\paragraph{True parameter values.}
The true parameters are calibrated to the M1 (random intercept)
unweighted NSECE fit (\cref{sme:model-hierarchy}).
\Cref{smf:tab-truth} reports all parameter values used in the
DGP.  These values were chosen to produce a realistic structural
zero rate (approximately 36\%) and overdispersion level
consistent with the empirical data.

\begin{table}[htbp]
  \centering
  \small
  \begin{tabular}{@{}l r l@{}}
    \toprule
    Parameter & Value & Description \\
    \midrule
    $\alpha_{\text{int}}$     & $+0.696$ & Extensive intercept \\
    $\alpha_{\text{pov}}$     & $-0.119$ & Extensive poverty \\
    $\alpha_{\text{urb}}$     & $+0.253$ & Extensive urban \\
    $\alpha_{\text{blk}}$     & $-0.070$ & Extensive Black \\
    $\alpha_{\text{his}}$     & $-0.139$ & Extensive Hispanic \\
    \addlinespace
    $\beta_{\text{int}}$      & $-0.032$ & Intensive intercept \\
    $\beta_{\text{pov}}$      & $+0.057$ & Intensive poverty \\
    $\beta_{\text{urb}}$      & $-0.018$ & Intensive urban \\
    $\beta_{\text{blk}}$      & $+0.080$ & Intensive Black \\
    $\beta_{\text{his}}$      & $+0.040$ & Intensive Hispanic \\
    \addlinespace
    $\log\kappa_0$            & $+1.655$ & Log concentration \\
    $\kappa_0$                & $5.232$ & Concentration \\
    \addlinespace
    $\tau_{\text{ext}}$       & $0.577$ & State SD (extensive) \\
    $\tau_{\text{int}}$       & $0.208$ & State SD (intensive) \\
    $\rho_{\text{cross}}$     & $0.285$ & Cross-margin correlation \\
    \bottomrule
  \end{tabular}
  \caption{True parameter values for the simulation DGP, calibrated to
    the M1 unweighted fit.  Fixed effects are on the logit scale;
    $\kappa_0 = \exp(\log\kappa_0)$.}
  \label{smf:tab-truth}
\end{table}

\paragraph{Outcome generation.}
For each replication~$r = 1,\ldots,R$:
\begin{enumerate}
  \item Draw state random effects
    $\bdelta_s = (\delta_{1s},\, \delta_{2s})^\t \sim
    \Norm_2(\mathbf{0}, \bSigma_\delta)$, where
    \[
      \bSigma_\delta = \diag(\boldsymbol{\tau})\,
      \mathbf{R}_\delta\,\diag(\boldsymbol{\tau}),
      \quad
      \boldsymbol{\tau} = (\tau_{\text{ext}},\,\tau_{\text{int}})^\t,
      \quad
      [\mathbf{R}_\delta]_{12} = \rho_{\text{cross}}.
    \]
  \item Compute linear predictors:
    $\eta_i^{\text{ext}} = \bx_i^\t\balpha_0 + \delta_{1,s[i]}$,
    $\eta_i^{\text{int}} = \bx_i^\t\bbeta_0 + \delta_{2,s[i]}$.
  \item Generate participation:
    $z_i \sim \Bern(q_i)$ with
    $q_i = \expit(\eta_i^{\text{ext}})$.
  \item For providers with $z_i = 1$, generate positive counts via
    \textbf{rejection sampling} from the unconditional beta-binomial
    \citep{Prentice1986}:
    \begin{quote}
      Draw $Y_i \sim \BetaBin(n_i,\, \mu_i,\, \kappa_0)$ where
      $\mu_i = \expit(\eta_i^{\text{int}})$;
      \emph{repeat until $Y_i > 0$}; set $y_i = Y_i$.
    \end{quote}
    This ensures $y_i > 0$ for all providers with $z_i = 1$,
    consistent with the zero-truncated beta-binomial component of
    the hurdle model.  The acceptance probability is $1 - p_0(n_i,
    \mu_i, \kappa_0)$, which exceeds~$0.80$ for the vast majority
    of observations in our calibration.
  \item For providers with $z_i = 0$, set $y_i = 0$.
\end{enumerate}

\begin{remark}[T3-3: Zero-truncated beta-binomial DGP]
  \label{smf:rem-ztbb}
  Drawing from the unconditional $\BetaBin$ and rejecting zeros
  produces exact draws from the zero-truncated beta-binomial
  $\ZTBB(n, \mu, \kappa)$.  The alternative---dividing the
  unconditional PMF by $(1-p_0)$---would require a custom
  sampling routine, whereas rejection sampling reuses the
  standard $\BetaBin$ sampler and is efficient whenever $p_0$ is
  not too close to~$1$.  The same approach is used in the
  Stan \texttt{generated quantities} block for posterior
  predictive replication.
\end{remark}

\paragraph{Sampling design.}
From each finite population, a stratified cluster sample of
approximately $N \approx 7{,}000$ providers is drawn using
size-biased inclusion probabilities
$\log\pi_i = c_0 + \rho_{\mathrm{inc}} \cdot y_i^*$
(\cref{eq:sim-inclusion}), where $y_i^*$ is the standardized
latent outcome and $c_0$ is calibrated to the target sample size
\citep[cf.][]{SavitskyToth2016}.
The three scenarios (S0, S3, S4) are defined in
\cref{tab:sim-scenarios}; see~\cref{sec:simulation} for the
full discussion.  Survey weights are computed as
$w_i = 1/\pi_i$ and normalized to sum to~$N$ for the
pseudo-posterior.

%% file: sm_f3_diagnostics.tex

\subsection{Supplementary diagnostic figures and design effect analysis}
\label{smf:diagnostics}

This section presents three supplementary figures that visualize
bias, interval width, and design effects across all simulation
scenarios, and a full table of design effect ratios (DER) for
all 11 fixed-effect parameters.

\paragraph{Relative bias distributions.}
\Cref{smf:fig-bias} shows the distribution of relative bias
across parameters and scenarios.  E-UW exhibits the largest
absolute bias (especially for $\tau_{\mathrm{ext}}$ and
$\tau_{\mathrm{int}}$ under S4), while E-WT/E-WS improve
point estimation for fixed effects by incorporating survey
weights.  For the poverty slopes $\alpha_{\mathrm{pov}}$ and
$\beta_{\mathrm{pov}}$, the relative bias under E-WS remains
below 15\% across all scenarios.

\begin{figure}[htbp]
  \centering
  \includegraphics[width=0.95\textwidth]{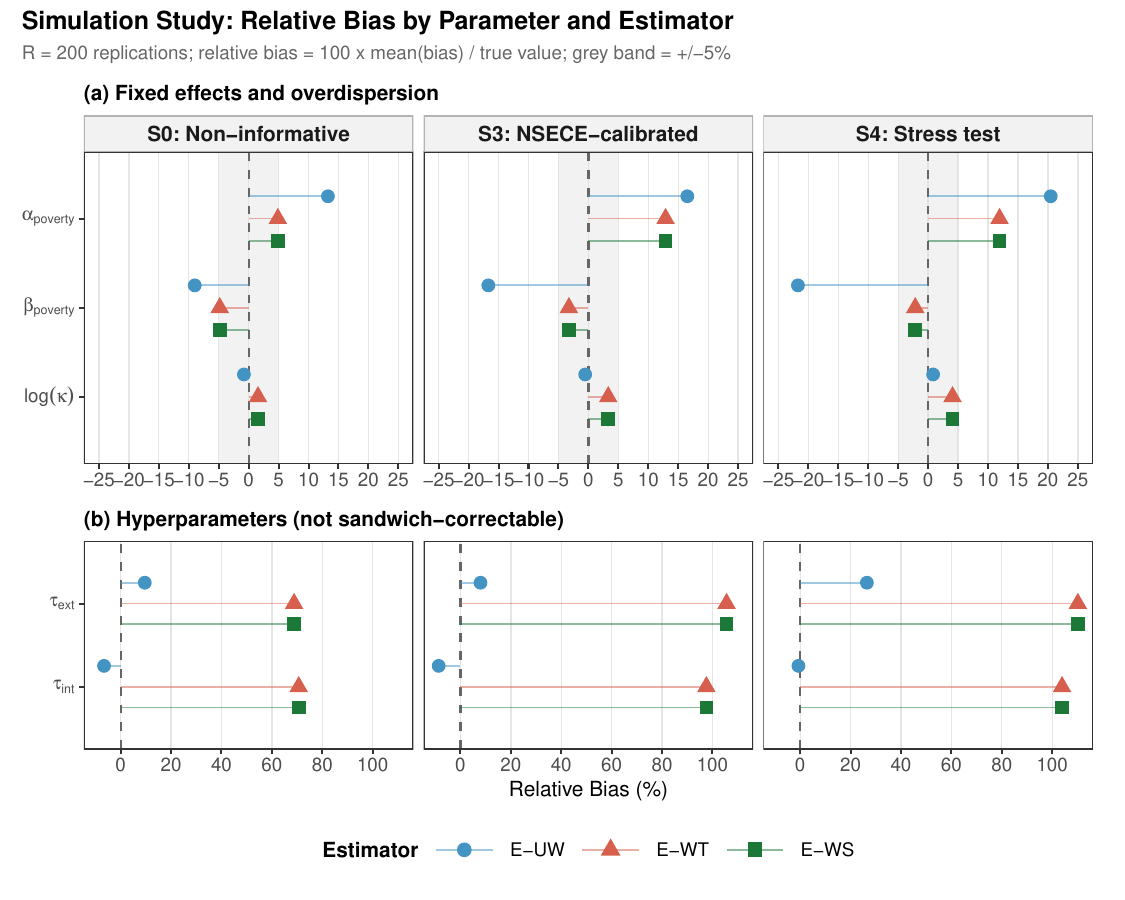}
  \caption{Relative bias (RB, \%) by parameter, estimator, and
    scenario.  Each panel corresponds to one of the five target
    parameters.  Dot positions indicate relative bias; color
    distinguishes scenarios (S0, S3, S4).  The vertical gray
    line marks zero bias.}
  \label{smf:fig-bias}
\end{figure}

\paragraph{Width ratio distributions.}
\Cref{smf:fig-wr} displays the distribution of width ratios
(WR = median CI width under E-WS / median CI width under E-WT)
across replications.  The progressive inflation from S0 through
S4 is clearly visible for all fixed-effect parameters.
The width ratio for $\log\kappa$ is similar to or slightly larger
than those for the poverty slopes, reflecting the nonlinear
dispersion score (Appendix~B).  Hyperparameter
width ratios are identically~$1.000$ (not shown) because no
sandwich correction is applied.

\begin{figure}[htbp]
  \centering
  \includegraphics[width=0.95\textwidth]{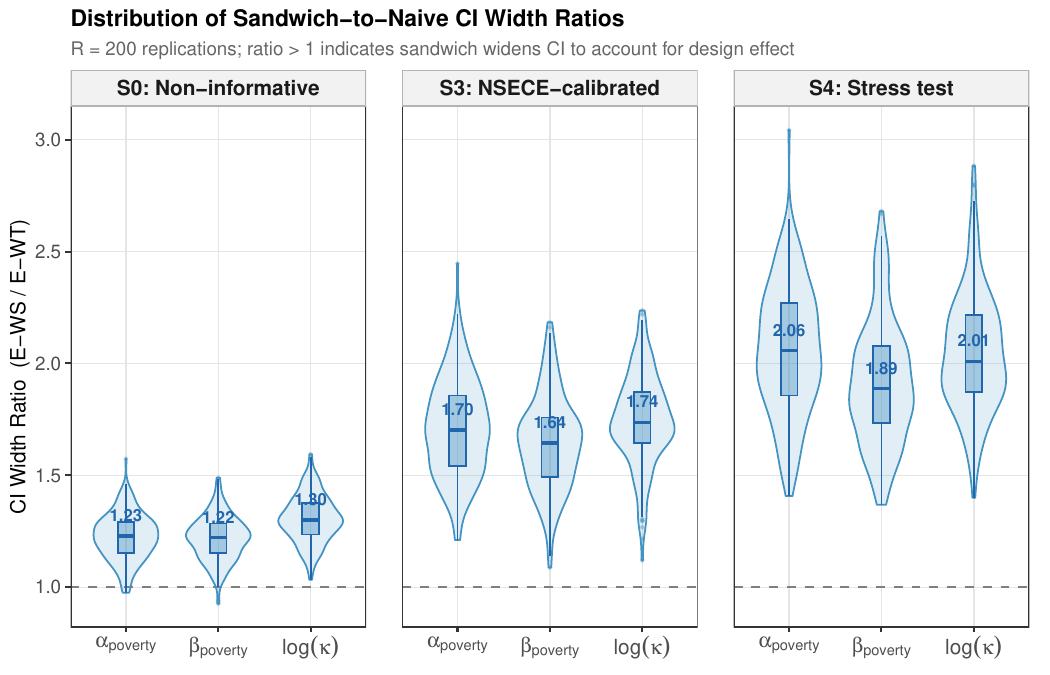}
  \caption{Distribution of width ratios across $R = 200$
    replications, by parameter and scenario.  Boxplots are overlaid
    with violin plots to show the full distributional shape.  The
    horizontal dashed line at WR = 1.0 marks the no-correction
    reference.}
  \label{smf:fig-wr}
\end{figure}

\paragraph{Design effect ratios.}
\Cref{smf:fig-der} summarizes the simulation DER across all 11
fixed-effect parameters (5 extensive, 5 intensive, and
$\log\kappa$) and three scenarios.  The DER monotonically
increases with the informativeness parameter
$\rho_{\mathrm{inc}}$, consistent with the theoretical prediction
in~\cref{smc:prop-der-classification}.  Urban density parameters
exhibit the largest DER in all scenarios (mean 3.25 in S0 up to
6.43 in S4), mirroring the empirical DER of~4.18 for
$\alpha_{\mathrm{urb}}$ in the NSECE case study
(\cref{smd:tab-der}).

\begin{figure}[htbp]
  \centering
  \includegraphics[width=0.95\textwidth]{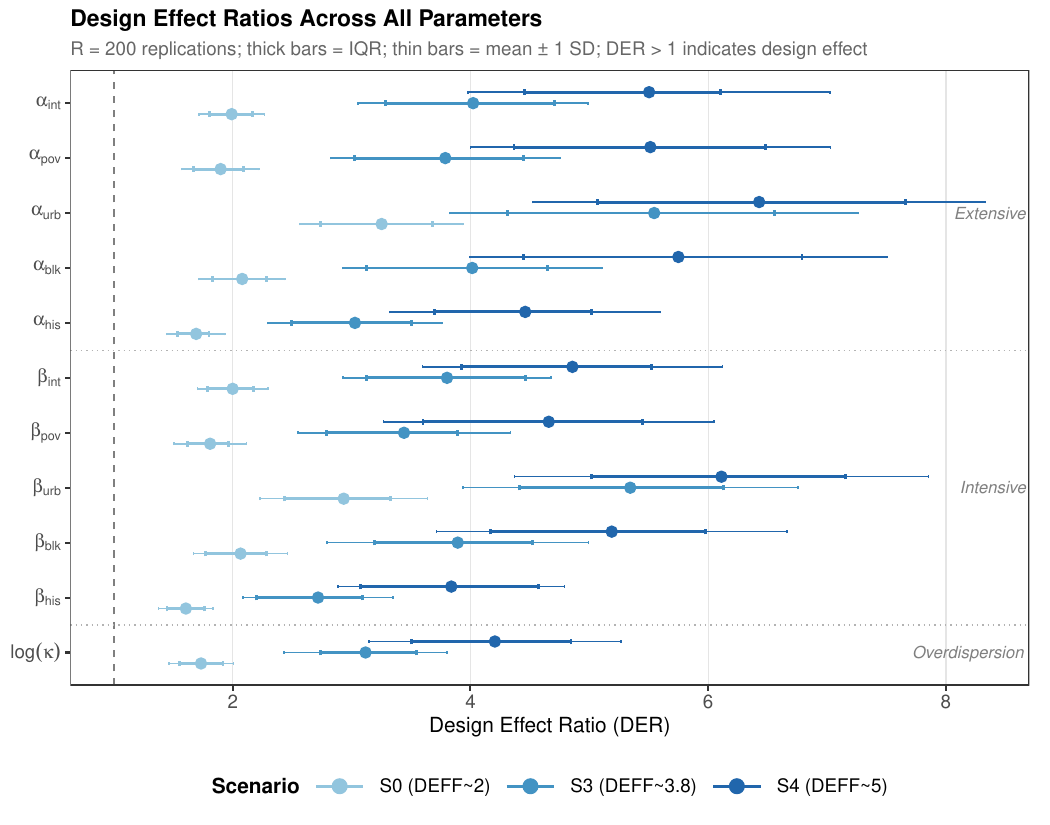}
  \caption{Design effect ratios (DER) by parameter and scenario.
    Each dot represents the mean DER across 200 replications;
    color distinguishes scenarios.  Parameters are grouped by
    margin (extensive vs.\ intensive) with $\log\kappa$ as a
    separate category.}
  \label{smf:fig-der}
\end{figure}

\paragraph{Complete DER table.}
The table below reports the full distribution of DER across
replications for all 11 fixed-effect parameters and three
scenarios.  The mean, median, standard deviation, and range
(min--max) characterize the replication-to-replication
variability.  Two patterns are noteworthy: (i)~the DER standard
deviation increases with the scenario severity, indicating
greater sampling variability under heavier design effects; and
(ii)~the urban parameter consistently shows the widest DER
range, reflecting the spatial concentration of urban providers
that amplifies the design effect.

{\small
\input{Figures/ST8_DER_summary}
}

\paragraph{Cross-margin correlation recovery.}
\cref{tab:rho-cross-sim} reports recovery of the cross-margin
correlation $\rho_{\mathrm{cross}} = \operatorname{cor}(\delta_j^{\mathrm{ext}},
\delta_j^{\mathrm{int}})$ using a plug-in estimator based on posterior
mean random effects across $S = 51$ states.
The unweighted estimator E-UW systematically overestimates
$\rho_{\mathrm{cross}}$, with relative bias growing from $+84\%$ in
S0 to $+127\%$ in S4.
The weighted estimator E-WT shows moderate downward bias ($-26\%$ to
$-45\%$), and E-WS is identical to E-WT because the Cholesky sandwich
correction (\cref{thm:cholesky}) adjusts only fixed-effect covariances.
Per-replication credible intervals are unavailable because full
posterior draws of the correlation matrix $\Omega$ were not retained
in the simulation archive; only frequentist summaries across
replications are reported.

{\small
\input{Figures/ST_rho_cross}
}

%% file: Figures/ST8_DER_summary.tex
\begin{table}[htbp]
\centering
\begin{tabular}{llrlllllll}
  \toprule
Scenario & Parameter & R & Mean & Median & SD & Min & Q25 & Q75 & Max \\ 
  \midrule
S0: Non-informative & Intercept (ext) & 200 & 1.99 & 1.95 & 0.28 & 1.29 & 1.80 & 2.16 & 2.92 \\ 
  S0: Non-informative & Poverty (ext) & 200 & 1.90 & 1.89 & 0.32 & 1.23 & 1.67 & 2.09 & 3.07 \\ 
  S0: Non-informative & Urban (ext) & 200 & 3.25 & 3.18 & 0.69 & 1.59 & 2.74 & 3.68 & 5.28 \\ 
  S0: Non-informative & Black (ext) & 200 & 2.08 & 2.05 & 0.36 & 1.14 & 1.83 & 2.29 & 3.46 \\ 
  S0: Non-informative & Hispanic (ext) & 200 & 1.69 & 1.67 & 0.24 & 1.05 & 1.54 & 1.80 & 2.68 \\ 
  S0: Non-informative & Intercept (int) & 200 & 2.00 & 1.97 & 0.30 & 1.38 & 1.79 & 2.17 & 2.92 \\ 
  S0: Non-informative & Poverty (int) & 200 & 1.81 & 1.80 & 0.31 & 1.06 & 1.62 & 1.96 & 3.11 \\ 
  S0: Non-informative & Urban (int) & 200 & 2.93 & 2.86 & 0.70 & 1.41 & 2.44 & 3.33 & 5.35 \\ 
  S0: Non-informative & Black (int) & 200 & 2.07 & 2.02 & 0.39 & 1.32 & 1.77 & 2.28 & 3.55 \\ 
  S0: Non-informative & Hispanic (int) & 200 & 1.61 & 1.59 & 0.23 & 1.11 & 1.45 & 1.76 & 2.34 \\ 
  S0: Non-informative & log kappa & 200 & 1.73 & 1.71 & 0.27 & 1.12 & 1.55 & 1.92 & 2.36 \\ 
  S3: NSECE-calibrated & Intercept (ext) & 200 & 4.02 & 3.88 & 0.97 & 1.97 & 3.29 & 4.71 & 7.47 \\ 
  S3: NSECE-calibrated & Poverty (ext) & 200 & 3.79 & 3.72 & 0.97 & 1.90 & 3.03 & 4.45 & 7.53 \\ 
  S3: NSECE-calibrated & Urban (ext) & 200 & 5.55 & 5.43 & 1.72 & 1.62 & 4.31 & 6.56 & 11.21 \\ 
  S3: NSECE-calibrated & Black (ext) & 200 & 4.01 & 3.81 & 1.09 & 2.07 & 3.13 & 4.64 & 7.23 \\ 
  S3: NSECE-calibrated & Hispanic (ext) & 200 & 3.03 & 2.94 & 0.73 & 1.58 & 2.50 & 3.50 & 5.26 \\ 
  S3: NSECE-calibrated & Intercept (int) & 200 & 3.80 & 3.70 & 0.88 & 1.68 & 3.13 & 4.46 & 6.40 \\ 
  S3: NSECE-calibrated & Poverty (int) & 200 & 3.44 & 3.42 & 0.89 & 1.53 & 2.79 & 3.89 & 7.32 \\ 
  S3: NSECE-calibrated & Urban (int) & 200 & 5.34 & 5.26 & 1.41 & 2.34 & 4.41 & 6.13 & 11.40 \\ 
  S3: NSECE-calibrated & Black (int) & 200 & 3.89 & 3.67 & 1.10 & 1.96 & 3.19 & 4.52 & 7.45 \\ 
  S3: NSECE-calibrated & Hispanic (int) & 200 & 2.72 & 2.65 & 0.63 & 1.67 & 2.20 & 3.09 & 5.16 \\ 
  S3: NSECE-calibrated & log kappa & 200 & 3.12 & 3.07 & 0.69 & 1.38 & 2.74 & 3.55 & 5.44 \\ 
  S4: Stress test & Intercept (ext) & 200 & 5.50 & 5.23 & 1.52 & 2.79 & 4.45 & 6.10 & 13.23 \\ 
  S4: Stress test & Poverty (ext) & 200 & 5.51 & 5.31 & 1.51 & 2.72 & 4.37 & 6.48 & 11.59 \\ 
  S4: Stress test & Urban (ext) & 200 & 6.43 & 6.22 & 1.91 & 2.97 & 5.07 & 7.66 & 17.11 \\ 
  S4: Stress test & Black (ext) & 200 & 5.75 & 5.41 & 1.75 & 2.30 & 4.45 & 6.79 & 11.38 \\ 
  S4: Stress test & Hispanic (ext) & 200 & 4.46 & 4.36 & 1.14 & 1.91 & 3.70 & 5.02 & 10.86 \\ 
  S4: Stress test & Intercept (int) & 200 & 4.86 & 4.71 & 1.26 & 2.38 & 3.93 & 5.52 & 9.53 \\ 
  S4: Stress test & Poverty (int) & 200 & 4.66 & 4.44 & 1.39 & 2.31 & 3.60 & 5.45 & 9.04 \\ 
  S4: Stress test & Urban (int) & 200 & 6.11 & 5.96 & 1.74 & 2.24 & 5.02 & 7.15 & 13.79 \\ 
  S4: Stress test & Black (int) & 200 & 5.19 & 5.01 & 1.48 & 2.13 & 4.16 & 5.98 & 10.13 \\ 
  S4: Stress test & Hispanic (int) & 200 & 3.84 & 3.74 & 0.95 & 1.78 & 3.07 & 4.57 & 6.44 \\ 
  S4: Stress test & log kappa & 200 & 4.21 & 4.10 & 1.06 & 2.11 & 3.50 & 4.85 & 7.69 \\ 
   \bottomrule
\end{tabular}
\caption{Design Effect Ratios (DER) Across Simulation Replications. DER = diag($V_{\text{sand}}$) / diag($H_{\text{obs}}^{-1}$). All 11 fixed-effect parameters shown.} 
\end{table}

%% file: Figures/ST_rho_cross.tex
\begin{table}[t]
\centering
\caption{Recovery of the cross-margin correlation $\rho_{\mathrm{cross}} = \operatorname{cor}(\delta_j^{\mathrm{ext}}, \delta_j^{\mathrm{int}})$
  across $R = 200$ simulation replications. The plug-in estimator 
  $\hat{\rho} = \operatorname{cor}(\bar{\delta}_j^{\mathrm{ext}}, \bar{\delta}_j^{\mathrm{int}})$
  uses posterior mean random effects ($S = 51$ states). True value: $\rho_{\mathrm{cross}} = 0.2853$.
  E\textsubscript{WS} is identical to E\textsubscript{WT} because the sandwich correction applies only to fixed effects.}
\label{tab:rho-cross-sim}
\smallskip
\small
\begin{tabular}{@{}ll rrrrr@{}}
\toprule
Scenario & Estimator & $\overline{\hat{\rho}}$ & Bias & Rel.~Bias (\%) & RMSE & Emp.~SE \\
\midrule
S0 (Baseline)      & E\textsubscript{UW}                 & 0.5238 & $+0.2386$ & $+83.6$ & 0.2990 & 0.1807 \\
                   & E\textsubscript{WT}                 & 0.2104 & $-0.0749$ & $-26.3$ & 0.2294 & 0.2173 \\
                   & E\textsubscript{WS}$^{\dagger}$     & 0.2104 & $-0.0749$ & $-26.3$ & 0.2294 & 0.2173 \\
\midrule
S3 (NSECE)         & E\textsubscript{UW}                 & 0.5814 & $+0.2961$ & $+103.8$ & 0.3643 & 0.2128 \\
                   & E\textsubscript{WT}                 & 0.1658 & $-0.1195$ & $-41.9$ & 0.2256 & 0.1918 \\
                   & E\textsubscript{WS}$^{\dagger}$     & 0.1658 & $-0.1195$ & $-41.9$ & 0.2256 & 0.1918 \\
\midrule
S4 (Heavy-tail)    & E\textsubscript{UW}                 & 0.6468 & $+0.3615$ & $+126.7$ & 0.4075 & 0.1886 \\
                   & E\textsubscript{WT}                 & 0.1579 & $-0.1274$ & $-44.7$ & 0.2372 & 0.2005 \\
                   & E\textsubscript{WS}$^{\dagger}$     & 0.1579 & $-0.1274$ & $-44.7$ & 0.2372 & 0.2005 \\
\bottomrule
\end{tabular}

\medskip
{\footnotesize $^{\dagger}$E\textsubscript{WS} shares the same Stan fit as E\textsubscript{WT}; the Cholesky sandwich correction adjusts only fixed-effect covariances, not hyperparameter posteriors.}
\end{table}

%% file: sm_f5_misspec.tex

\subsection{Model misspecification robustness}
\label{smf:misspec}

The simulation study in \cref{sec:simulation} evaluates the
correctly specified case: data are generated from the random-intercept
model~M1 and fitted by M1.  A natural concern is how the three
estimators (E-UW, E-WT, E-WS) behave when the fitted model omits a
relevant source of heterogeneity.  We therefore augment the simulation
with a misspecification scenario designed to probe whether the sandwich
correction remains beneficial when the model is wrong.

\paragraph{M2 data-generating process.}
Define the \emph{state-varying coefficient} DGP~(M2) by
adding random poverty slopes to the M1 linear predictors:
\begin{align}
  \eta^{\mathrm{ext}}_i &= \bx_i^\t \balpha + \delta^{\mathrm{ext}}_{s[i]}
    + \gamma^{\mathrm{ext}}_{s[i]} \, x_{\mathrm{pov},i}, \label{eq:m2-ext}\\
  \eta^{\mathrm{int}}_i &= \bx_i^\t \bbeta  + \delta^{\mathrm{int}}_{s[i]}
    + \gamma^{\mathrm{int}}_{s[i]} \, x_{\mathrm{pov},i}, \label{eq:m2-int}
\end{align}
where $\gamma^{\mathrm{ext}}_s \sim \Norm(0, 0.15^2)$ and
$\gamma^{\mathrm{int}}_s \sim \Norm(0, 0.10^2)$ independently across
$s = 1,\ldots,51$ states.  All other parameters
($\balpha, \bbeta, \kappa, \boldsymbol{\tau}, \rho$) and the population structure
($M = 50{,}000$ providers) are identical to the M1 DGP
(\cref{smf:dgp}).  The standard deviations $0.15$ (extensive) and
$0.10$ (intensive) produce state-varying poverty effects that are
moderate relative to the fixed poverty coefficients
($\alpha_{\mathrm{pov}} = -0.119$, $\beta_{\mathrm{pov}} = +0.057$).

\paragraph{Design.}
We draw $R = 200$ independent replications under the S3
(NSECE-calibrated) informativeness scenario with
$\operatorname{DEFF} \approx 3.79$, and fit the M1 random-intercept
model using all three estimators.  The M1 fit ignores the
state-varying poverty slopes $\gamma_s$, inducing a form of
omitted-variable misspecification that propagates through the
poverty coefficient estimates.

\paragraph{Results.}
\Cref{smf:tab-misspec} compares coverage, relative bias, RMSE, and
width ratio for the correctly specified S3 scenario (M1 truth, M1 fit)
and the misspecified B5 scenario (M2 truth, M1 fit).

\begin{table}[htbp]
\centering
\caption{Misspecification robustness: correctly specified S3 vs.\
  misspecified B5 (M2 truth, M1 fit) under NSECE-calibrated
  sampling ($\operatorname{DEFF} \approx 3.79$).
  $R = 200$ replications; nominal coverage $= 90\%$;
  $\dagger$ indicates significant departure ($> 2 \times$ MCSE).
  MCSE $= \sqrt{\hat{p}(1-\hat{p})/200}$; range 0.0--3.5 pp.}
\label{smf:tab-misspec}
\input{Figures/ST_B5_misspec}
\end{table}

\Cref{smf:fig-misspec-coverage} displays the corresponding coverage
rates graphically.

\begin{figure}[htbp]
\centering
\includegraphics[width=\textwidth]{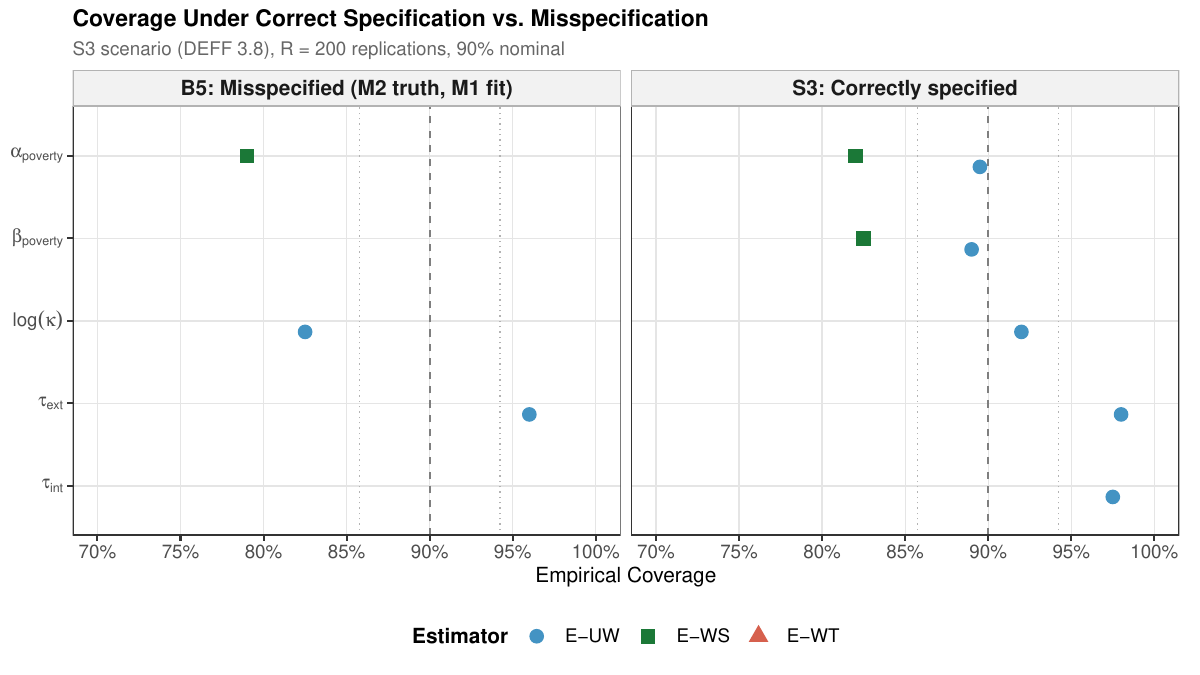}
\caption{Coverage under correct specification (S3) vs.\
  misspecification (B5).  Vertical dashed line: 90\% nominal;
  dotted lines: $\pm 2\,\mathrm{MCSE}$ band.}
\label{smf:fig-misspec-coverage}
\end{figure}

Three findings emerge from the comparison.

\paragraph{Finding B5-1: The unweighted estimator is acutely sensitive
to misspecification in the poverty coefficients.}
Under correct specification (S3), E-UW achieves near-nominal coverage
for both poverty slopes (89.5\% and 89.0\%).  Under the M2 DGP, E-UW
coverage collapses to 0.5\% for $\alpha_{\mathrm{pov}}$ and 30.0\%
for $\beta_{\mathrm{pov}}$---drops of 89 and 59 percentage points,
respectively.  The mechanism is omitted-variable bias: because M1
lacks the state-varying poverty slopes $\gamma_s$, the fixed
coefficients $\alpha_{\mathrm{pov}}$ and $\beta_{\mathrm{pov}}$ absorb
the omitted heterogeneity, producing relative bias of $-117.8\%$ and
$+63.1\%$.  This confirms that the model-based intervals of E-UW are
valid only under correct specification---a strong argument for the
design-based protection that survey weighting provides.

\paragraph{Finding B5-2: The sandwich correction preserves coverage
stability under misspecification.}
For $\alpha_{\mathrm{pov}}$, E-WS coverage declines only from 82.0\%
(S3) to 79.0\% (B5)---a 3-percentage-point drop---while E-WT drops
from 58.5\% to 52.0\%.  For $\beta_{\mathrm{pov}}$, E-WS declines
from 82.5\% to 67.5\% ($-15$ pp) compared with E-WT's drop from
62.5\% to 42.0\% ($-20.5$ pp).  In both cases, E-WS maintains a
substantial advantage over E-WT (27 and 25.5 percentage points under
B5), demonstrating that the sandwich variance correction absorbs a
meaningful share of the misspecification-induced variance increase.
The width ratios remain stable across scenarios
($\mathrm{WR} = 1.67\text{--}1.74$ under B5 vs.\ $1.64\text{--}1.75$
under S3), indicating that the design effect magnitude is similar
regardless of whether the model is correctly specified.

\paragraph{Finding B5-3: Variance components exhibit differential
sensitivity.}
The overdispersion $\log\kappa$ is relatively robust under E-UW
(coverage 92.0\% $\to$ 82.5\%; relative bias $-0.5\%$ $\to$ $-1.3\%$),
as expected since the omitted random slopes do not directly affect the
dispersion structure.  Among the variance components,
$\tau_{\mathrm{ext}}$ remains well-covered under E-UW (98.0\% $\to$
96.0\%), but $\tau_{\mathrm{int}}$ drops sharply (97.5\% $\to$
57.0\%).  The asymmetry arises because the omitted intensive-margin
slopes ($\sigma_{\gamma}^{\mathrm{int}} = 0.10$ relative to
$\tau_{\mathrm{int}} = 0.208$) contribute proportionally more
unmodeled between-state variance than the extensive-margin slopes
($\sigma_{\gamma}^{\mathrm{ext}} = 0.15$ relative to
$\tau_{\mathrm{ext}} = 0.577$).  As in the correctly specified case,
E-WT and E-WS produce near-zero coverage for both $\tau$ parameters
($\mathrm{WR} = 1.00$), reaffirming that the sandwich correction is
structurally inapplicable to variance components.

Taken together, these results strengthen the case for the two-track
reporting strategy recommended in~\cref{sec:simulation}: for
population-average fixed effects, the sandwich-corrected estimator
E-WS provides meaningful robustness against model misspecification;
for hierarchical variance components, the unweighted estimator E-UW
remains the only viable option, with the caveat that its validity
depends on correct specification of the random-effects structure.

%% file: Figures/ST_B5_misspec.tex
\small
\begin{adjustbox}{max width=\textwidth}
\begin{tabular}{@{}l l rr rrr rrr rr@{}}
\toprule
 & & \multicolumn{2}{c}{Coverage (\%)} &
     \multicolumn{3}{c}{Relative bias (\%)} &
     \multicolumn{3}{c}{RMSE} &
     \multicolumn{2}{c}{WR} \\
\cmidrule(lr){3-4}\cmidrule(lr){5-7}\cmidrule(lr){8-10}\cmidrule(lr){11-12}
Parameter & Est. & S3 & B5 & S3 & B5 & $\Delta$ &
  S3 & B5 & $\Delta$ & S3 & B5 \\
\midrule
\multicolumn{12}{@{}l}{\textit{Fixed effects}} \\[2pt]
$\alpha_{\text{pov}}$ & E-UW & 89.5 & 0.5$^{\dagger}$ &
  $+$16.5 & $-$117.8 & $-$134.3 & 0.037 & 0.143 & $+$0.106 & --- & --- \\
 & E-WT & 58.5$^{\dagger}$ & 52.0$^{\dagger}$ &
  $+$12.9 & $-$37.6 & $-$50.5 & 0.072 & 0.079 & $+$0.007 & --- & --- \\
 & E-WS & 82.0$^{\dagger}$ & 79.0$^{\dagger}$ &
  $+$12.9 & $-$37.6 & $-$50.5 & 0.072 & 0.079 & $+$0.007 & 1.71 & 1.67 \\[3pt]
$\beta_{\text{pov}}$ & E-UW & 89.0 & 30.0$^{\dagger}$ &
  $-$16.7 & $+$63.1 & $+$79.8 & 0.017 & 0.039 & $+$0.022 & --- & --- \\
 & E-WT & 62.5$^{\dagger}$ & 42.0$^{\dagger}$ &
  $-$3.2 & $+$52.7 & $+$55.9 & 0.032 & 0.043 & $+$0.011 & --- & --- \\
 & E-WS & 82.5$^{\dagger}$ & 67.5$^{\dagger}$ &
  $-$3.2 & $+$52.7 & $+$55.9 & 0.032 & 0.043 & $+$0.011 & 1.64 & 1.60 \\[3pt]
$\log\kappa$ & E-UW & 92.0 & 82.5$^{\dagger}$ &
  $-$0.5 & $-$1.3 & $-$0.8 & 0.020 & 0.027 & $+$0.007 & --- & --- \\
 & E-WT & 34.0$^{\dagger}$ & 36.5$^{\dagger}$ &
  $+$3.3 & $+$2.9 & $-$0.4 & 0.068 & 0.064 & $-$0.004 & --- & --- \\
 & E-WS & 60.5$^{\dagger}$ & 64.0$^{\dagger}$ &
  $+$3.3 & $+$2.9 & $-$0.4 & 0.068 & 0.064 & $-$0.004 & 1.75 & 1.74 \\[4pt]
\multicolumn{12}{@{}l}{\textit{Variance components}} \\[2pt]
$\tau_{\text{ext}}$ & E-UW & 98.0$^{\dagger}$ & 96.0$^{\dagger}$ &
  $+$8.0 & $-$2.5 & $-$10.5 & 0.091 & 0.117 & $+$0.026 & --- & --- \\
 & E-WT & 0.5$^{\dagger}$ & 0.0$^{\dagger}$ &
  $+$105.6 & $+$102.4 & $-$3.2 & 0.639 & 0.625 & $-$0.014 & --- & --- \\
 & E-WS & 0.5$^{\dagger}$ & 0.0$^{\dagger}$ &
  $+$105.6 & $+$102.4 & $-$3.2 & 0.639 & 0.625 & $-$0.014 & 1.00 & 1.00 \\[3pt]
$\tau_{\text{int}}$ & E-UW & 97.5$^{\dagger}$ & 57.0$^{\dagger}$ &
  $-$8.6 & $-$34.8 & $-$26.2 & 0.040 & 0.079 & $+$0.039 & --- & --- \\
 & E-WT & 0.0$^{\dagger}$ & 1.0$^{\dagger}$ &
  $+$97.6 & $+$97.0 & $-$0.6 & 0.211 & 0.213 & $+$0.002 & --- & --- \\
 & E-WS & 0.0$^{\dagger}$ & 1.0$^{\dagger}$ &
  $+$97.6 & $+$97.0 & $-$0.6 & 0.211 & 0.213 & $+$0.002 & 1.00 & 1.00 \\
\bottomrule
\end{tabular}
\end{adjustbox}

%% file: sm_f6_frequentist.tex

\subsection{Frequentist contextualization}
\label{smf:frequentist}

A natural question is whether a simpler design-based analysis could
replace the full Bayesian hierarchical model.  We compare the HBB
pseudo-posterior estimates with conventional survey-weighted frequentist
estimates obtained from \texttt{survey::svyglm}
\citep{Lumley2004,Lumley2020} to (i)~validate the design-consistency
of the pseudo-posterior fixed effects and (ii)~clarify which model
features cannot be accommodated by the frequentist approach.

\paragraph{Frequentist specification.}
We fit separate margin models using the same five-covariate
specification as the HBB model.  For the extensive margin, we fit a
survey-weighted logistic regression using
\texttt{svyglm(\ldots, family = quasibinomial())} with $N = 6{,}785$
observations.  For the intensive margin, we fit a survey-weighted
linear regression of the IT enrollment share $y_i / n_i$ on the same
covariates, restricted to the $N^{+} = 4{,}392$ participants.  Both
models use the NSECE complex survey design (stratification, clustering,
and sampling weights) and produce linearization-based cluster-robust
standard errors.  No random effects, beta-binomial structure, or
cross-margin dependence can be incorporated.

\paragraph{Results.}
\Cref{smf:tab-frequentist} compares the survey-weighted svyglm
estimates with the M3b-W sandwich-corrected pseudo-posterior summaries.

\begin{table}[htbp]
\centering
\caption{Frequentist contextualization: survey-weighted svyglm vs.\
  Bayesian HBB (M3b-W, sandwich-corrected) point estimates and standard
  errors.  The extensive margin uses logistic link in both approaches;
  the intensive margin uses an identity link (svyglm) vs.\ logit link
  (HBB), so the scales differ.  SE Ratio $=$ HBB sandwich SE / svyglm SE.}
\label{smf:tab-frequentist}
\input{Figures/ST_B6_frequentist}
\end{table}

For the extensive margin, both approaches use the logit link, enabling
a meaningful comparison on a common scale.  The coefficient signs and
relative magnitudes agree across all five predictors.  Point estimates
from HBB are systematically larger in absolute value than the svyglm
estimates, which is expected: the HBB model conditions on state-level
random effects, and conditional effects in mixed models are known to
exceed their marginal counterparts
\citep{ZegerLiangAlbert1988}.  The SE ratio
(HBB sandwich / svyglm) ranges from $0.59$ to $0.89$, reflecting the
variance reduction achieved by modeling between-state heterogeneity
through random effects rather than absorbing it into the residual.

For the intensive margin, the comparison is less direct because the two
approaches operate on different scales: svyglm regresses the
enrollment share $y_i / n_i$ under an identity link, while HBB models
the beta-binomial mean parameter $\mu_i$ under a logit link.  The
differing intercepts (svyglm: $+0.483$; HBB: $-0.242$) reflect this
scale difference rather than substantive disagreement, since
$\operatorname{logit}^{-1}(-0.242) \approx 0.44$, which is in the
range of the observed mean share.  The larger SE ratios for the
intensive margin ($2.30$--$4.63$) also arise from the logit-scale
parameterization, which inflates variances relative to the probability
scale.

\paragraph{Structural limitations.}
The comparison highlights several features of the HBB framework
that \texttt{svyglm} cannot accommodate:
\begin{enumerate}[nosep]
  \item \emph{Beta-binomial distribution:} svyglm models the continuous
    share $y/n$ via Gaussian errors, ignoring the bounded-count nature of
    the outcome and the extra-binomial overdispersion ($\kappa = 6.81$).
  \item \emph{Hurdle structure:} the two margins must be fitted
    separately, precluding joint inference on the zero-generating process.
  \item \emph{State-level random effects:} svyglm is a fixed-effect
    estimator; accommodating $S = 51$ state intercepts (and slopes in M3b)
    would require either fixed-effect dummies---inflating SE and
    sacrificing shrinkage---or a two-step approach that does not
    propagate uncertainty.
  \item \emph{Cross-margin correlation:} the bivariate random-effect
    structure $\boldsymbol{\delta}_s \sim \Norm_{2q}(\mathbf{0},
    \boldsymbol{\Sigma}_\delta)$ cannot be replicated in separate
    margin-specific regressions.
  \item \emph{Policy moderation:} the state-varying policy coefficients
    $\boldsymbol{\Gamma}$ interact hierarchical and design-based
    components in a way that has no natural frequentist counterpart.
\end{enumerate}

Taken together, the svyglm comparison serves as a fixed-effect-level
sanity check: the Bayesian pseudo-posterior fixed effects are
design-consistent with the frequentist target, while the full HBB
framework adds hierarchical structure, distributional flexibility,
and joint inference that the frequentist alternative cannot provide.

%% file: Figures/ST_B6_frequentist.tex
\small
\begin{adjustbox}{max width=\textwidth}
\begin{tabular}{@{}l l rr rr r r@{}}
\toprule
 & & \multicolumn{2}{c}{svyglm (design-based)} &
   \multicolumn{2}{c}{HBB (sandwich-corrected)} & & \\
\cmidrule(lr){3-4}\cmidrule(lr){5-6}
Margin & Parameter & Est. & SE & Est. & SE & $\Delta$Est. & SE Ratio \\
\midrule
Extensive & Intercept & 0.662 & 0.063 & 0.764 & 0.044 & +0.103 & 0.70 \\
 & Poverty & -0.199 & 0.071 & -0.324 & 0.052 & -0.125 & 0.74 \\
 & Urban & 0.234 & 0.039 & 0.442 & 0.034 & +0.208 & 0.89 \\
 & Black & 0.086 & 0.075 & 0.478 & 0.044 & +0.392 & 0.59 \\
 & Hispanic & 0.059 & 0.074 & 0.006 & 0.051 & -0.054 & 0.70 \\
[3pt]
Intensive & Intercept & 0.483 & 0.005 & -0.242 & 0.015 & -0.725 & 3.02 \\
 & Poverty & 0.015 & 0.007 & 0.090 & 0.019 & +0.075 & 2.77 \\
 & Urban & -0.009 & 0.004 & -0.047 & 0.021 & -0.038 & 4.63 \\
 & Black & 0.034 & 0.008 & -0.017 & 0.018 & -0.050 & 2.30 \\
 & Hispanic & 0.010 & 0.006 & -0.097 & 0.017 & -0.107 & 2.92 \\
\bottomrule
\end{tabular}
\end{adjustbox}

%% file: sm_f7_decomposition.tex

\subsection{Coverage gap decomposition}
\label{smf:decomposition}

\Cref{smf:tab-decomposition} decomposes the coverage gap for each
parameter--estimator combination under the S3 (NSECE-calibrated) scenario.
Two diagnostic quantities distinguish width-driven from bias-driven
shortfalls:
\begin{itemize}[nosep]
  \item \emph{SE calibration ratio} $= \widehat{\mathrm{SD}} /
    \overline{\mathrm{SE}}$, where $\widehat{\mathrm{SD}}$ is the
    empirical standard deviation of point estimates across $R = 200$
    replications and $\overline{\mathrm{SE}}$ is the mean reported
    standard error.  A ratio exceeding 1 indicates that the reported SE
    underestimates the actual sampling variability (width-driven gap).
  \item \emph{Standardized mean bias} $= |\bar{b}| /
    \widehat{\mathrm{SD}}$, where $\bar{b}$ is the mean bias across
    replications.  Values exceeding $\sim 0.5$ indicate non-negligible
    centering error (bias-driven gap).
\end{itemize}

\begin{table}[htbp]
\centering
\caption{Coverage gap decomposition under the S3 (NSECE-calibrated)
  scenario.  ``Below'' and ``Above'' report the percentage of replications
  where the true value falls below or above the confidence interval,
  respectively.  SE ratio $> 1$ indicates intervals that are too narrow;
  $|\bar{b}|/\mathrm{SD} > 0.5$ indicates non-negligible point estimate
  bias.  $\star$ = sandwich structurally inapplicable (E-WS $\equiv$ E-WT).}
\label{smf:tab-decomposition}
\input{Figures/ST_B8_decomposition}
\end{table}

For the poverty coefficients under E-WS, the SE ratio of 1.13--1.14
and negligible standardized bias (0.06--0.22) confirm that the 8~pp
coverage gap is entirely width-driven: the sandwich variance estimator
underestimates the true sampling variance by approximately 13\%.
This is consistent with the well-documented $O(S^{-1})$ finite-sample
downward bias of cluster-robust variance estimators when $S = 51$
\citep{Binder1983}.  For $\log\kappa$, the picture reverses: the SE
ratio is 1.00 (perfectly calibrated width), but the standardized bias
of 1.37 indicates that the pseudo-likelihood point estimate is
systematically shifted, producing entirely one-sided non-coverage
(39.5\% below, 0\% above).  For $\tau_{\mathrm{ext}}$ and
$\tau_{\mathrm{int}}$, massive standardized bias ($> 3$) overwhelms
any width consideration, reaffirming the structural inapplicability of
the sandwich correction to variance components.

%% file: Figures/ST_B8_decomposition.tex
\small
\begin{adjustbox}{max width=\textwidth}
\begin{tabular}{@{}l l r rr r r l@{}}
\toprule
 & & & \multicolumn{2}{c}{Non-coverage (\%)} & & \\
\cmidrule(lr){4-5}
Parameter & Est. & Cov.\ (\%) & Below & Above &
  SE ratio & $|\bar{b}|/\text{SD}$ & Dominant source \\
\midrule
$\alpha_{\text{pov}}$ & E-UW & 89.5 & 1.0 & 9.5 & 0.86 & 0.63 & ---  \\
 & E-WT & 58.5 & 15.5 & 26.0 & 1.93 & 0.22 & Width \\
 & E-WS & 82.0 & 5.0 & 13.0 & 1.13 & 0.22 & Width \\[3pt]
$\beta_{\text{pov}}$ & E-UW & 89.0 & 0.5 & 10.5 & 0.85 & 0.69 & --- \\
 & E-WT & 62.5 & 15.5 & 22.0 & 1.88 & 0.06 & Width \\
 & E-WS & 82.5 & 8.5 & 9.0 & 1.14 & 0.06 & Width \\[3pt]
$\log\kappa$ & E-UW & 92.0 & 1.0 & 7.0 & 0.85 & 0.48 & --- \\
 & E-WT & 34.0 & 65.0 & 1.0 & 1.75 & 1.37 & Both \\
 & E-WS & 60.5 & 39.5 & 0.0 & 1.00 & 1.37 & Bias \\[3pt]
$\tau_{\text{ext}}$ & E-UW & 98.0 & 1.5 & 0.5 & 0.64 & 0.59 & --- \\
 & E-WT/WS & 0.5 & 99.5 & 0.0 & 1.22 & 3.18 & Bias$^{\star}$ \\[3pt]
$\tau_{\text{int}}$ & E-UW & 97.5 & 1.0 & 1.5 & 0.73 & 0.51 & --- \\
 & E-WT/WS & 0.0 & 100.0 & 0.0 & 1.11 & 3.55 & Bias$^{\star}$ \\
\bottomrule
\end{tabular}
\end{adjustbox}

%% file: sm_f8_coverage95.tex

\subsection{Coverage at the 95\% nominal level}
\label{smf:coverage95}

\Cref{smf:tab-coverage95} reports coverage at both the 90\% and 95\%
nominal levels for all five target parameters across three scenarios and
three estimators.  For all estimators, 95\% intervals are computed as
$\hat\theta \pm 1.96 \times \mathrm{SE}$, where $\mathrm{SE}$ is the
posterior standard deviation (E-UW, E-WT) or the sandwich standard error
(E-WS).  Verification confirms that this Wald construction agrees with
the original quantile-based 90\% intervals to within 0.1~pp, validating
the Normal approximation.

\begin{table}[htbp]
\centering
\caption{Interval coverage (\%) at 90\% and 95\% nominal levels.
  Parenthetical values are MCSEs ($= \sqrt{\hat{p}(1 - \hat{p})/R}$;
  $R = 200$).  At the 95\% nominal level, $\mathrm{MCSE} \approx 1.5$~pp.}
\label{smf:tab-coverage95}
\input{Figures/ST_B9_coverage95}
\end{table}

The results reinforce the width-driven interpretation of the E-WS
coverage gap identified in \cref{smf:decomposition}.  For the poverty
coefficients under S3, widening from 90\% to 95\% adds only
5.5--6.0~pp of coverage (from 82.0--82.5\% to 87.5--88.5\%), leaving
a residual gap of 6.5--7.5~pp relative to the 95\% nominal level.
If the sandwich SE were perfectly calibrated, the expected gain would be
$\approx 5.0$~pp (from the additional probability mass in the Normal
tails between $z_{0.05}$ and $z_{0.025}$); the observed gains are
consistent with this prediction, confirming that the coverage shortfall
scales proportionally with interval width rather than reflecting a
level-dependent phenomenon.

In contrast, E-UW achieves near-nominal coverage at both levels
(89.0--93.0\% at 90\%; 95.0--97.5\% at 95\%) for all fixed effects,
providing further evidence that the unweighted posterior is well
calibrated when the model is correctly specified.  For
$\log\kappa$ under E-WS, the 95\% coverage (68.5\%) remains
substantially below nominal, consistent with the bias-driven shortfall
identified in \cref{smf:tab-decomposition}.

%% file: Figures/ST_B9_coverage95.tex
\small
\begin{adjustbox}{max width=\textwidth}
\begin{tabular}{@{}l l  cc  cc  cc @{}}
\toprule
 & & \multicolumn{2}{c}{S0 (Non-inform.)} & \multicolumn{2}{c}{S3 (NSECE)} & \multicolumn{2}{c}{S4 (Stress)} \\
\cmidrule(lr){3-4} \cmidrule(lr){5-6} \cmidrule(lr){7-8}
Parameter & Est. & 90\% & 95\% & 90\% & 95\% & 90\% & 95\% \\
\midrule
$\alpha_{\text{pov}}$ & E-UW & 93.0\,(1.8) & 97.0\,(1.2) & 89.5\,(2.2) & 96.0\,(1.4) & 92.5\,(1.9) & 97.5\,(1.1) \\
 & E-WT & 81.0\,(2.8) & 89.0\,(2.2) & 58.5\,(3.5) & 67.5\,(3.3) & 53.0\,(3.5) & 61.0\,(3.4) \\
 & E-WS & 88.5\,(2.3) & 96.0\,(1.4) & 82.0\,(2.7) & 87.5\,(2.3) & 86.0\,(2.5) & 89.5\,(2.2) \\
[3pt]
$\beta_{\text{pov}}$ & E-UW & 93.0\,(1.8) & 96.5\,(1.3) & 89.0\,(2.2) & 95.0\,(1.5) & 91.0\,(2.0) & 96.0\,(1.4) \\
 & E-WT & 76.5\,(3.0) & 86.0\,(2.5) & 62.5\,(3.4) & 70.5\,(3.2) & 56.5\,(3.5) & 64.0\,(3.4) \\
 & E-WS & 86.5\,(2.4) & 91.5\,(2.0) & 82.5\,(2.7) & 88.5\,(2.3) & 84.0\,(2.6) & 89.5\,(2.2) \\
[3pt]
$\log\kappa$ & E-UW & 91.5\,(2.0) & 96.0\,(1.4) & 92.0\,(1.9) & 96.5\,(1.3) & 90.5\,(2.1) & 93.5\,(1.7) \\
 & E-WT & 64.0\,(3.4) & 73.0\,(3.1) & 34.0\,(3.3) & 38.5\,(3.4) & 23.0\,(3.0) & 28.0\,(3.2) \\
 & E-WS & 76.5\,(3.0) & 84.5\,(2.6) & 60.5\,(3.5) & 68.5\,(3.3) & 57.0\,(3.5) & 68.0\,(3.3) \\
[3pt]
$\tau_{\text{ext}}$ & E-UW & 97.5\,(1.1) & 100.0\,(0.0) & 98.0\,(1.0) & 99.5\,(0.5) & 77.0\,(3.0) & 93.0\,(1.8) \\
 & E-WT & 3.0\,(1.2) & 7.0\,(1.8) & 0.5\,(0.5) & 0.5\,(0.5) & 0.0\,(0.0) & 0.0\,(0.0) \\
 & E-WS & 3.0\,(1.2) & 7.0\,(1.8) & 0.5\,(0.5) & 0.5\,(0.5) & 0.0\,(0.0) & 0.0\,(0.0) \\
[3pt]
$\tau_{\text{int}}$ & E-UW & 91.0\,(2.0) & 95.0\,(1.5) & 97.5\,(1.1) & 99.5\,(0.5) & 97.0\,(1.2) & 98.5\,(0.9) \\
 & E-WT & 8.0\,(1.9) & 12.0\,(2.3) & 0.0\,(0.0) & 0.5\,(0.5) & 0.0\,(0.0) & 0.0\,(0.0) \\
 & E-WS & 8.0\,(1.9) & 12.0\,(2.3) & 0.0\,(0.0) & 0.5\,(0.5) & 0.0\,(0.0) & 0.0\,(0.0) \\
\bottomrule
\end{tabular}
\end{adjustbox}